\documentclass[11pt,english]{article}
\usepackage[utf8]{inputenc}
\usepackage[latin9]{luainputenc}
\usepackage{geometry}
\usepackage{ragged2e}
\usepackage{amsmath}
\usepackage{amssymb}
\usepackage{amsthm}
\usepackage{babel}
\usepackage{graphicx}
\usepackage{epstopdf}
\graphicspath{ {Inkscape/} }
\usepackage{slashed}
\usepackage[math]{blindtext}
\usepackage{bm}
\usepackage[usenames, dvipsnames]{color}
\usepackage{calc}
\usepackage{setspace}
\usepackage{enumerate}
\usepackage{comment}
\usepackage{multicol}
\usepackage{marginnote}
\usepackage{wrapfig}
\usepackage{scrextend}
\usepackage{csquotes}
\deffootnote[-1em]{-1em}{1em}{\textsuperscript{\thefootnotemark}\,}

\newtheorem{theorem}{Theorem}

\newtheorem{lemma}{Lemma}
\newtheorem{proposition}{Proposition}
\newtheorem{remark}{Remark}

\newtheorem{corollary}{Corollary}

\usepackage{chngcntr}
\counterwithin{theorem}{section}
\counterwithin{definition}{section}
\counterwithin{lemma}{section}
\counterwithin{proposition}{section}
\counterwithin{remark}{section}
\counterwithin{example}{section}
\counterwithin{corollary}{section}

\newcommand{\abs}[1]{\left\lvert#1\right\rvert}
\newcommand{\norm}[1]{\left\lVert#1\right\rVert}
\newcommand{\x}{\mathcal{X}}
\newcommand{\h}{\mathcal{H}^+}
\newcommand{\br}[1]{\left \lbrace #1 \right \rbrace}
\newcommand{\p}{\Psi_i}
\newcommand{\dd}{\bm{\slashed{\mathcal{D}}}_1}
\newcommand{\DD}{\bm{\slashed{\mathcal{D}}}_2}
\newcommand{\slas}[1]{\slashed{#1}}
\newcommand{\I}{\check{I}_{\tau_1}^{\tau_2}}

\newcommand{\el}[1]{{}^{(F)}\hspace{-0.08cm} {#1}}

\usepackage[
    natbib=true,
    style=numeric,
    sorting=anyt
]{biblatex}
\bibliography{references}

\usepackage{hyperref}
\hypersetup{
	colorlinks=true,
	linkcolor=blue,
	filecolor=magenta,      
	urlcolor=cyan,
}

\begin{document}
	\title{ Instability of gravitational and electromagnetic perturbations of extremal Reissner--Nordstr\"om spacetime}
	\author{Marios Antonios Apetroaie \footnote{University of Toronto, marios.apetroaie@mail.utoronto.ca}}
	\date{}
	\maketitle
	
	\begin{abstract}We study the linear stability problem to gravitational and electromagnetic perturbations of the \textit{extremal},  $ \abs{\mathcal{Q}}=M, $  Reissner--Nordstr\"om spacetime, as a solution to the Einstein-Maxwell equations. Our work uses and extends the framework \cite{giorgi2019boundedness,giorgie2020boundedness} of Giorgi, and contrary to the subextremal case we prove that instability results hold for a set of gauge invariant quantities along the event horizon $ \h $. In particular, we prove decay, non-decay, and polynomial blow-up estimates asymptotically along $ \h $, the exact behavior depending on the number of translation invariant  derivatives that we take. As a consequence, we show that for generic initial data, solutions to the generalized Teukolsky system of positive and negative spin satisfy both stability and instability results. It is worth mentioning that the negative spin solutions are significantly more unstable, with the extreme curvature component $ \underline{\alpha} $ not decaying asymptotically along the event horizon $ \h, $ a result previously unknown in the literature. 
	\end{abstract}

{
	\hypersetup{linkcolor=black}
	\tableofcontents
}
	\setstretch{1}
	\section{Introduction} 
	The question of \textit{stability} of black holes, as solutions to the Einstein equation, has led to a vast interdisciplinary research work addressing this problem for various spacetime models. Some of the most recent results include the proof of non-linear stability of Schwarzschild \cite{dafermos2021non}, and Kerr \cite{giorgi2022wave,klainerman2021kerr,klainerman2022brief} for small angular momentum, i.e. $ \abs{a}/m\ll 1 $. This was done in the spirit of the seminal work of Christodoulou--Klainerman \cite{christodoulou1993global}, proving the non-linear stability of Minkowski spacetime.
	
	In this paper, we are interested in	the Reissner--Nordstr\"om family of spacetimes $ (\mathcal{M}, g_{_{M,Q}}) $, 
	which in local coordinates takes the form \begin{align}
		g_{_{M,Q}} = - \left(1-\frac{2M}{r}+\dfrac{\abs{Q}^{2}}{r^{2}}\right)dt^{2}+ \left(1-\frac{2M}{r}+\dfrac{\abs{Q}^{2}}{r^{2}}\right)^{-1} dr^{2}+ r^{2}\left(d\theta^{2}+ \sin^{2}\theta d\phi^{2}\right),
	\end{align}
	and represent the spacetime outside a non-rotating, spherical symmetric, charged black hole of mass $ M $ and charge $ Q, $ with $ \abs{Q}\leq M. $
	It stands as the \textit{unique} spherically symmetric, asymptotically flat solution of the Einstein-Maxwell equations \begin{align}
		\begin{aligned}
			Ric(g)_{\mu\nu} &= 2F_{\mu \lambda}F_{\nu}^{\lambda}- \dfrac{1}{2}g_{\mu\nu}F_{\alpha\beta}F^{\alpha\beta}\\
			D_{[\alpha}F_{\beta \gamma]} &=0, \hspace{1cm} D^{\alpha}F_{\alpha\beta}=0.
		\end{aligned}
	\end{align}
	where $ D $ is the Levi--Civita connection associated to the metric $ g $, and the 2-form $ F $ is the electromagnetic tensor verifying the Maxwell equations.
	
	In a series of papers \cite{giorgi2019boundedness,giorgi2020linear,giorgi2020linearfull,giorgie2020boundedness}, the author concluded the \textit{linear} stability of the full \textit{subextremal} range of Reissner--Nordstr\"om spacetimes as solutions to the Einstein-Maxwell equations. Roughly, this means that all solutions to the linearized Einstein-Maxwell equations around a Reissner--Nordstr\"om solution, $ g_{_{M,Q}} $ with $ \abs{Q}<M $, arising from regular asymptotically flat initial data remain \textit{uniformly bounded} in the exterior, and \textit{decay} to a linearized Kerr-Newman solution after adding a pure gauge solution.
	However, some of the results developed in this work no longer hold in the \textbf{extremal} case, $ \abs{Q}=M, $ because of the degeneracy of the
	redshift effect at the event horizon $ \h $, a necessary ingredient in showing linear stability in \cite{giorgi2019boundedness,giorgie2020boundedness}. In addition, in view of the \textit{Aretakis instabilities} that manifest on the horizon $ \h $  for the homogeneous wave equation \cite{aretakis2011stabilityII,aretakis2011stability}, one expects that similar instabilities arise in the linearized gravity as well.
	
	The purpose of this paper is to address the linear stability problem for the extreme Reissner--Nordstr\"om spacetime of maximally charged black holes. Our  results are at the level of \textit{gauge--invariant quantities}, characterized by the fact that they vanish in any \textit{pure gauge} solution. We are looking at quantities that arise naturally in the linearization procedure and are shown to satisfy a generalized version of the so-called \textit{Teukolsky equation}; see \cite{giorgi2019boundedness,giorgie2020boundedness}. These wave-type equations govern the gravitational and electromagnetic perturbations of \textit{ERN} and decouple completely from the full set of linearized Einstein-Maxwell system when written in a null frame, and thus can be studied independently. 
	\subsection{The Teukolsky system and the instability result}
	In the work of linear stability  of Schwarzschild and Kerr spacetimes, or even in the case of Maxwell equations in these backgrounds, the resulting Teukolsky equations are independent for each corresponding extreme component (gravitational or electromagnetic); see \cite{dafermos2019linear,shlapentokh2020boundedness,blue2008decay,pasqualotto2019spin}. 
	However, in the case of Reissner--Nordstr\"om, we obtain generalized Teukolsky equations that are heavily coupled with each other. This is, roughly, due to the initiation of both gravitational and electromagnetic perturbations in the presence of charge. The Teukolsky system of $ \pm $ spin is
	satisfied by a pair of extreme curvature components $\alpha_{AB}, \underline{a}_{AB}$, and two pairs $\mathfrak{f}_{AB}, \underline{\mathfrak{f}}_{AB}, \ \tilde{\beta}_{A}, \underline{\tilde{\beta}}_{A}$,  defined in terms of both Ricci coefficients and curvature/electromagnetic components.
	In particular, let $ \mathcal{T}_i:= 	\square_{g_{M,Q}} + c_i(r)\cdot \slas{\nabla}_{e_{3}} + d_i(r)\cdot \slas{\nabla}_{e_{4}} + V_i(r) ,$  be a Teukolsky type operator with coefficients depending on $M, Q $, then the generalized Teukolsky system is schematically given by 
	\begin{align*}
		\begin{aligned}[c]
			\mathcal{T}_1(\alpha)\ &=\ w_1(r)\ \slas{\nabla}_{e_{4}}\mathfrak{f} + z_1(r)\ \mathfrak{f} \\
			\mathcal{T}_2(\mathfrak{f})\ &=\ w_2(r)\ \slas{\nabla}_{e_{3}}\alpha + z_2(r)\ \alpha 
		\end{aligned}
		\hspace{1.5cm} \text{and} \hspace{1.5cm} 
		\begin{aligned}[c]
			\mathcal{T}_3(\tilde{\beta})\ &=\ w_3(r) \ \slas{div} \mathfrak{f} + z_3(r) \ \slas{div} \alpha,
		\end{aligned}
	\end{align*}	
	written with respect to a null frame  $\br{e_3,e_4,e_A}_{_{A=1,2}} $,  and by $ \slas{\nabla}_{e_i} $ we denote the projection of spacetime covariant derivative $ D_{e_i} $ on the section spheres.
	
	As a central result in this paper, we obtain estimates for the Teukolsky system in the exterior of extreme Reissner--Nordstr\"om spacetime, up to and including the event horizon $ \h. $ A set of conservation laws that hold for induced gauge--invariant quantities, along the event horizon $ \h,$  yields an analogue of Aretakis instability that carries up to the level of Teukolsky solution. Both stability and instability results coexist, which can be summarized in the following theorem.
	
	\begin{flushleft}
		{\large \textbf{Theorem.}}\ \normalsize \textit{(Rough version)} Let $ \alpha, \mathfrak{f} ,  \tilde{\beta}  $ and $ \underline{\alpha}, \underline{\mathfrak{f}}, \underline{\tilde{\beta}} $ be solutions to the generalized Teukolsky system of $ \pm $ spin on the extreme Reissner--Nordstr\"om exterior, and let $ Y $ denote a transversal invariant derivative, then for generic initial data
		\begin{enumerate}[i)]
			\item Away from the event horizon $ H^{+} \equiv \br{r=M}, $ i.e. $ \br{r\geq r_0} $ for any $ r_0>M $, Teukolsky solutions \textbf{decay} with respect to the time function of a suitable foliation of the exterior,
			\item The following pointwise decay, non-decay and  blow-up estimates hold asymptotically along the event horizon $ \h $ \footnote{$ \norm{\xi}_{\infty}(\tau) := \norm{\xi}_{L^{\infty}(S^{2}_{\tau,M})}$,  $ \norm{\xi}_{S^{2}_{\tau,M}} := \norm{\xi}_{L^{2}(S^{2}_{\tau,M}).} $, and $ f(\tau) \sim_{_{p}} g(\tau) \ \Rightarrow \displaystyle \lim_{\tau\to \infty} \frac{f(\tau)}{g(\tau)} = c,$ with $ c $ depending on $ p. $ }
			\begin{enumerate}
				\item For the positive spin solutions, we have 
				\begin{itemize}\item 
					$ \norm{\slas{\nabla}^{m}_{Y}\mathfrak{f}}_{\infty}\hspace{-0.1cm}(\tau) ,\  \norm{\slas{\nabla}^{m}_{Y}\tilde{\beta}}_{\infty}\hspace{-0.2cm}(\tau)$, and $\norm{\slas{\nabla}^{n}_{Y}\alpha}_{\infty} \hspace{-0.1cm}(\tau)$  \textbf{decay} for any $ m\leq 2, n\leq 4. $
					\vspace{0.1cm}
					\item $ \norm{\slas{\nabla}^{3}_{Y}\mathfrak{f}}_{S^{2}_{\tau,M}}, \ \norm{\slas{\nabla}^{3}_{Y}\tilde{\beta}}_{S^{2}_{\tau,M}}, $ and $ \norm{\slas{\nabla}^{5}_{Y}\alpha}_{S^{2}_{\tau,M}} $ do \textbf{not} decay along $ \h. $ \vspace{0.1cm}
					\item  $ \norm{\slas{\nabla}^{k+3}_{Y}\xi}_{S^{2}_{\tau,M}}\thicksim_{_{k}} \tau^{k},$ and $\norm{\slas{\nabla}^{k+5}_{Y}\alpha}_{S^{2}_{\tau,M}} \thicksim_{_{k}} \tau^{k} $, as $ \tau \rightarrow \infty $, for any $ k\in \mathbb{N},\  \xi \in \br{\mathfrak{f},\tilde{\beta}}. $
				\end{itemize}
				\item For the negative spin solutions, we have   \begin{itemize}
					\item Decay,  \quad $\norm{\underline{\mathfrak{f}}}_{\infty}\hspace{-0.1cm}(\tau)+ \norm{\underline{\tilde{\beta}}}_{\infty}\hspace{-0.1cm}(\tau)  \xrightarrow{\tau \to \infty}\ 0 $
					\vspace{0.1cm}
					\item $ \norm{\slas{\nabla}_{Y}\underline{\mathfrak{f}}}_{S^{2}_{\tau,M}}, \ \norm{\slas{\nabla}_{Y}\underline{\tilde{\beta}}}_{S^{2}_{\tau,M}}, $ and $ \norm{\underline{\alpha}}_{S^{2}_{\tau,M}} $ do \textbf{not} decay along $ \h, $ \vspace{0.1cm}
					\item 
					$ \norm{\slas{\nabla}^{k+1}_{Y}\xi}_{S^{2}_{\tau,M}}\thicksim_{_{k}} \tau^{k},$ and $\norm{\slas{\nabla}^{k}_{Y}\alpha}_{S^{2}_{\tau,M}} \thicksim_{_{k}} \tau^{k} $, as $ \tau \rightarrow \infty $, for any $ k\in \mathbb{N},\  \xi \in \br{\underline{\mathfrak{f}},\underline{\tilde{\beta}}}. $
				\end{itemize}
			\end{enumerate}
		\end{enumerate}
	\end{flushleft}
	\setstretch{1}
	\begin{remark} The extreme curvature component $ \underline{\alpha} $ does not decay, itself \footnote{This is due to the appearance of a transversal invariant derivative of $ \underline{\mathfrak{f}} $ on the right-hand side of the Teukolsky equation satisfied by $  \underline{\alpha} $, which acts as a source and it does not decay along the horizon $ \h. $ }, along $ \h $. This is \textbf{not} the case for axisymmetric
		linear perturbations of \textbf{extreme Kerr}, where it has been shown numerically \cite{burko2021scalar} that $ \underline{\alpha}  $ \textbf{decays} asymptotically on $ \mathcal{H}^{+}$; see relevant works \cite{lucietti2012gravitational,casals2016horizon}. This suggests that extreme Reissner--Nordstr\"om spacetimes are linearly more unstable than extreme Kerr ones in the axisymmetric setting.
	\end{remark}
\begin{remark}
For the positive spin equations, in both the extreme Reissner--Nordstr\"om and extreme Kerr (axisymmetric perturbations), it takes \textbf{five} transversal invariant derivatives of the curvature component $  \alpha $ not to decay asymptotically along $ \h $, as seen in \cite{lucietti2012gravitational}. Moreover, they show that any additional transversal derivative of $ \alpha $ yields asymptotic blow-up estimates with the same rates as in the Theorem above.
\end{remark}

	In order to arrive at these estimates, we rely on the resolution introduced in the proof of linear stability of Schwarzschild \cite{dafermos2019linear}, and then adapted in the case of Reissner--Nordstr\"om \cite{giorgi2019boundedness,giorgie2020boundedness}. The key part is a set of physical space transformations of Teukolsky solutions to higher order gauge invariant quantities which satisfy a system of generalized Regge--Wheeler equations (\ref{csystem}). To study the latter, we follow standard techniques and ideas developed in \cite{dafermos2013lectures,dafermos2009red,dafermos2010decay}.
	
	\subsection{Previous Works on Extreme Black Holes} 
	A first, rigorous study of the horizon instability pertaining to extreme Reissner--Nordstr\"om  spacetimes were initiated by Aretakis in \cite{aretakis2011stabilityII,aretakis2011stability}. A series of works followed, addressing the dynamics of wave equation models in the extreme Reissner--Nordstr\"om \cite{angelopoulos2020non,angelopoulos2021price,gajic2019interior,gajic2021quasinormal}, and  extreme Kerr \cite{aretakis2012decay,teixeira2020mode} spacetime. At the same time, Aretakis' results inspired several heuristics and numerical works that shed light on the stability of extreme black holes including \cite{murata2013happens,marolf2010dangers,lucietti2012gravitational,zimmerman2017horizon,casals2016horizon,gralla2018scaling,burko2021scalar,ori2013late}. Closely related to our paper are the results of \cite{lucietti2013horizon}, where the authors derive a horizon instability for Regge--Wheeler type equations in the linearized gravity of extreme Reissner--Nordstr\"om. The number of transversal invariant derivatives that appear in their corresponding  Aretakis constants agree with the ones we obtain in Section  \ref{Conservation Section}.

	\subsection{Outline of the paper} We present the structure of this paper, along with the main Theorems in each section.
	
	In Section \ref{Geometry}, we introduce the main coordinate systems and foliations we are using throughout the paper. In Section \ref{set up section}, we briefly review the set-up and the Teukolsky equations of \cite{giorgi2019boundedness,giorgie2020boundedness}, on which our work is based. We write the transformation theory, adapted in the extreme Reissner--Nordstr\"om case, which yields the Regge--Wheeler system, and after we consider its spherical harmonic decomposition we decouple it.
	
	In Section \ref{Regge estimates}, we study the induced decoupled Regge--Wheeler equations, and we prove Morawetz type estimates in Theorems \ref{Morawetz-Spacetime}, \ref{Spacetime non degenerate photon estimate}, and a degenerate redshift estimate with a degeneracy of the transversal derivative on the horizon $ \h $, as in Theorem \ref{Nuniform }. In Section \ref{EE first}, we remove the aforementioned degeneracy as shown in Theorem \ref{remove_degeneracy}.
	
	In section \ref{Conservation Section}, we show that solutions to the Regge--Wheeler equations are subject to a conservation law on the horizon $ \h $, Theorem \ref{conservation horizon}. We proceed with Section \ref{EE higher}, where we obtain higher order transversal invariant derivative estimates, Theorem \ref{high-order horizon estimates}. Last, we derive $ r^{p}- $hierarchy estimates in Section \ref{r-p}, which allows us to prove energy decay estimates in the exterior.
	
	Using the energy decay results and the conservation laws, we prove pointwise decay, non-decay, and blow-up estimates along the event horizon $ \h $, as in Theorems \ref{PWD-estimates}, \ref{Scalar blow up}. As a consequence, in Proposition \ref{finished estimates of Regge} we derive estimates for the initial Regge--Wheeler system.
	
	Finally, in the last two Sections \ref{Teukolsky positive section}  and \ref{Teukolsky negative section}, we use the transformations of Section \ref{set up section} to derive decay estimates for the Teukolsky system of $ \pm $ spin away from the horizon; Corollaries \ref{Decay teukolsky +}, \ref{Decay teukolsky -}. Estimates along the horizon $ \h $ are shown in Theorems \ref{f,b big Theorem}, \ref{a big Theorem} and Theorems \ref{f,b undeline estimates}, \ref{underline a estimate}.
	
	\subsection{Acknowledgments} I would like to express my gratitude to Aretakis Stefanos for introducing me to this problem and for his fruitful advice, insights, and comments. I am particularly grateful to Elena Giorgi for her constant support and invaluable ideas while this paper was being written. Finally, I would like to thank Mihalis Dafermos and Gustav Holzegel for assisting with helpful discussions and suggestions.

	\section{Extreme Reissner--Nordstr\"om Spacetime} \label{Geometry}
	In this section, we introduce the main coordinate systems we will be working with, and the foliations we use to derive the energy estimates. In addition, we briefly go through relevant elliptic notions and identities that are frequently used throughout the paper.

	\subsection{Differential structure and metric}
	With respect to the Boyer–Lindquist coordinates $ (t,r,\vartheta,\varphi) $ 
	the Reissner--Nordstr\"om metric takes the form $ g=g_{M,\mathcal{Q}} = -D dt^2 + \frac{1}{D}dr^2 +r^2 g_{\mathbb{S}^2}$, where $ D(r)= 1- \frac{2M}{r}+\frac{\abs{\mathcal{Q}}^2}{r^2} $ and $ g_{\mathbb{S}^2}  $ is the standard round metric on $ \mathbb{S}^2. $ In this paper, we are interested in the case where $ \abs{\mathcal{Q}} = M $, i.e the extremal Reissner--Nordstr\"om spacetime (ERN).
	
	Our main goal is to capture the behavior of gauge--invariant quantities up to and \textbf{including} the event horizon $ \h $, so we introduce a different coordinate system that extends regularly on $ \h. $ 
	To do so, we first  define the so-called \textbf{double null} coordinates in ERN. Consider the tortoise coordinate $ 
	r^{\star}(r):=r+2M\log(r-M)-\frac{M^2}{r-M}+C,
	$ for a constant $ C, $ and note
	$ r^{\star} $ satisfies $ \frac{\partial r^{\star}}{\partial r}= \frac{1}{D}. $ Then, the null coordinates are defined via $ u:=t-r^{\star}, \ v:=t+r^{\star}, $ with respect to which the ERN metric is given by \begin{align*}
		g_{_{ERN}} = -D(r) dv du + r^{2}g_{\mathbb{S}^{2}}
	\end{align*}
	We use this coordinate system to produce the $ r^{p}- $hierarchy estimates in Section \ref{r-p}, however, it doesn't extend regularly to $ \h. $
	One way to extend the metric beyond the event horizon is to consider the so-called \textbf{ingoing Eddington--Finkelstein} coordinates $ (v,r,\vartheta,\varphi) $ and with respect to that system the ERN metric is given by \begin{align} \label{ingoing metric}
		g_{_{ERN}} = - D dv^2 + 2dvdr +r^2 g_{\mathbb{S}^2}, \hspace{1cm} D= \left (1-\frac{M}{r}\right )^2.
	\end{align}
	In this setting, the event horizon is captured by $ \mathcal{H}^{+}\equiv \br{r=M} $, while the coordinate $ v \in (-\infty, \infty)$ traverses it. While the ingoing coordinates are regular up to $ r>0, $ we focus on producing estimates only on the domain of outer communication $ \mathcal{M} $, where $ r\geq M,$ i.e. \begin{align*}
		\mathcal{M} =\left((-\infty,\infty)\times [M,+\infty)\times \mathbb{S}^{2}\right).
	\end{align*}
	
	\paragraph{Foliations.} Here, we introduce the two foliations we are going to be using in this paper. 
	
	\begin{itemize}
		\item 	
		For the first one, let $ \Sigma_{0} $ be a flat $ SO(3)- $ invariant spacelike hypersurface terminating at $ i^{o} $ and crossing the event horizon $ \h $ with $ \partial \Sigma_0 = \Sigma_0\cap \h.$ Note, $ \Sigma_0 $ can be chosen such that its unit normal future directed vectorfield $ n_{\Sigma_0} $,  satisfies everywhere \[ \frac{1}{C}< -g(n_{\Sigma_0},n_{\Sigma_0}) < C, \hspace{1cm}  \frac{1}{C} < - g(n_{\Sigma_0},T) < C, \]
		for a positive constant $ C>0 $, where $ T=\partial_v $ is the global killing vector field in ERN, with respect to (\ref{ingoing metric}). Let  $ \phi_{\tau}^{T} $ be the one-parameter family of isomorphisms corresponding to the killing field T, and define the foliation $ \Sigma_{\tau} := \phi_{\tau}^T(\Sigma_0) $. Then, the hypersurfaces $ \Sigma_{\tau} $ are isometric to $ \Sigma_0 $ and the coercive relations above hold uniformly in $ \tau,
		$ for the same constant $ C. $ Estimates associated with this foliation take place in the region $ \quad
		\mathcal{R}(0,\tau):= \cup_{0\leq \tilde{\tau}\leq \tau} \Sigma_{\tilde{\tau}} $
		
		\begin{center}
			
			\includegraphics[scale=0.4]{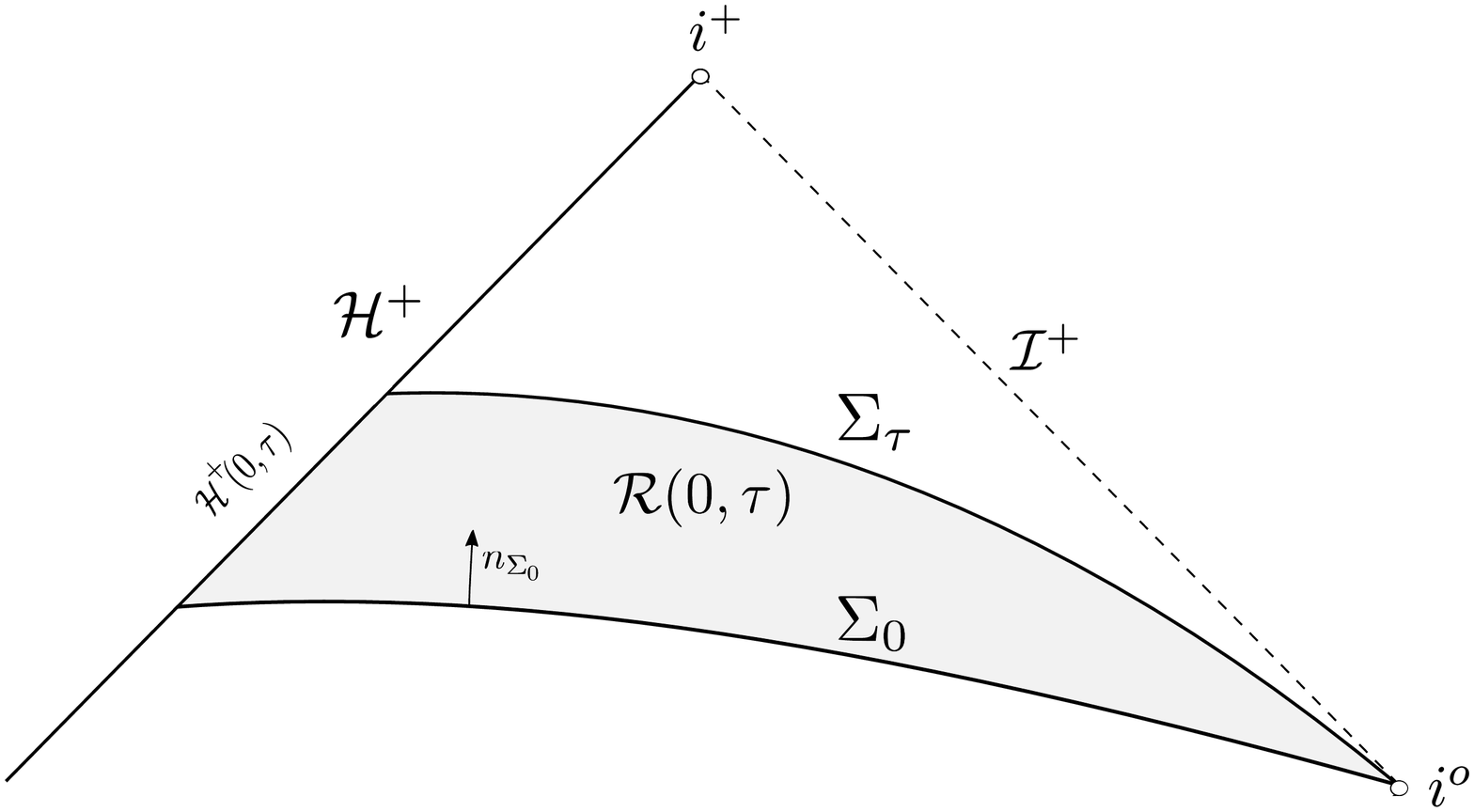}
			
		\end{center}
		\item Next, we introduce a foliation $ \check{\Sigma}_{\tau} $ that captures the radiating energy towards future null infinity $ \mathcal{I}^{+} $, and is ultimately used to obtain energy decay estimates. For that, fix $ R >2M $ and take the tortoise coordinate $ r^{\star} $ introduced earlier for $ C= -R-4M\log(R-M)- \frac{2M^2}{R-M}$. Define also the coordinate  $ t^{\star}:=t+ 2M\log(r-M) - \frac{M^2}{r-M}, $ and consider the following hypersurfaces for all $ \tau \in \mathbb{R} $
		\begin{align*}
			\check{\Sigma}_{\tau}:= \begin{cases}
				\lbrace t^{\star}=\tau \rbrace, \hspace{1cm} \text{for }\ M\leq r\leq R \\
				\\
				\lbrace u = \tau \rbrace, \hspace{1.1cm} \text{for} \ \ r\geq R,
			\end{cases}
		\end{align*}
		where $ u $ is the advanced null coordinate we saw above.
		With the specific choice of $ C $ in the definition of the tortoise coordinate $ r^{\star}(r) $, the hypersurfaces $ \check{\Sigma}_{\tau} $ are well defined on $ \br{r=R} $ for all $ \tau \in \mathbb{R}. $ Moreover, $ \check{\Sigma}_{\tau} $ crosses the event horizon $ \mathcal{H}^+ $ and terminates at future null infinity $ \mathcal{I}^+ $ for every $ \tau \in \mathbb{R} $, as seen in the figure below. Note, along $ \h $ the parameter of the foliation above satisfies:  $ \tau = v+ (M + C)$.
		\begin{center}
			\includegraphics[scale=0.4]{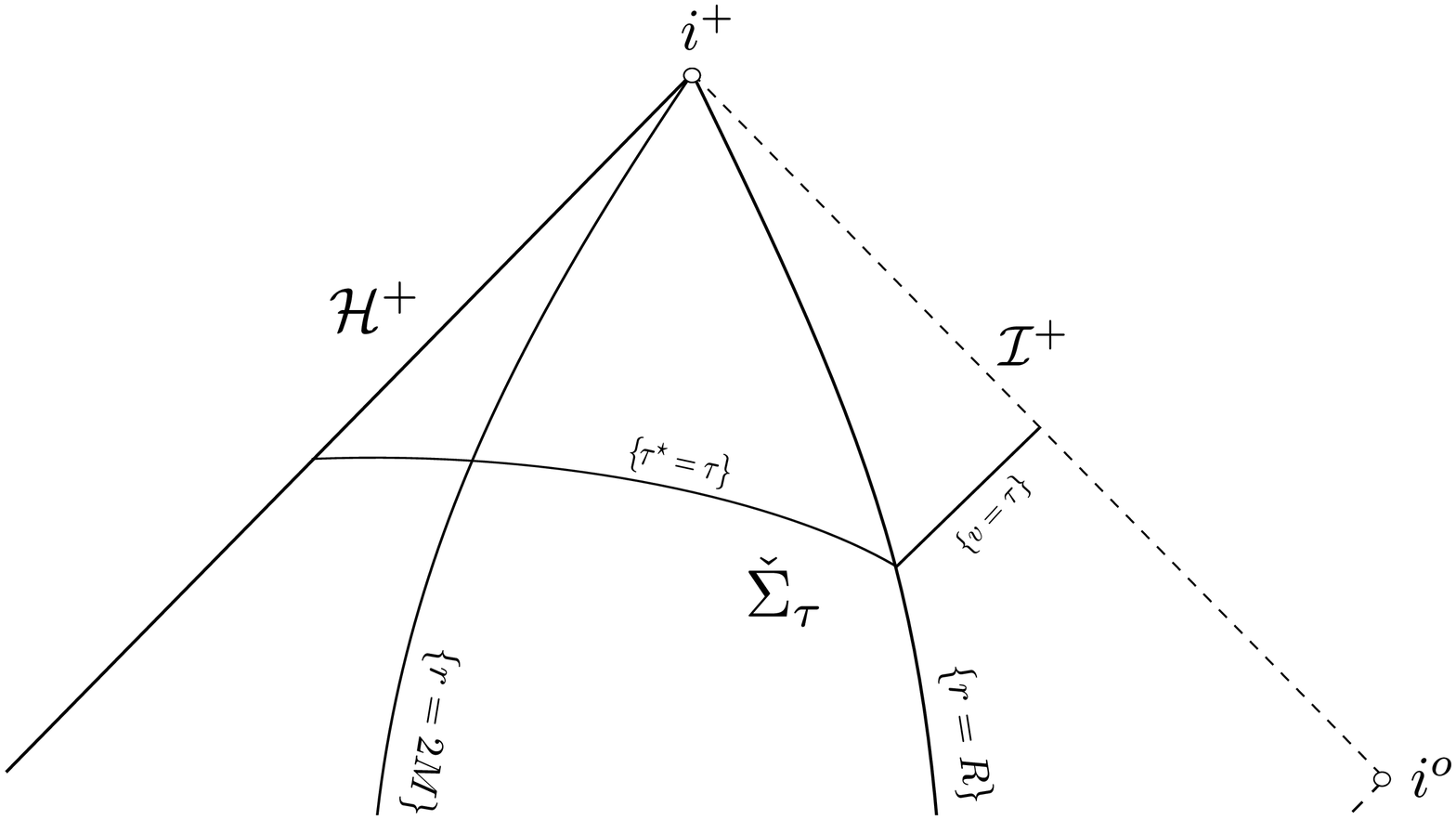}
		\end{center}
	\end{itemize}	
	For both foliations above, it is rather useful to consider a coordinate system associated to them. Of course, there is a natural way to define an induced one; for any $ P\in \Sigma_{\tau} $ (or $ \check{\Sigma}_{\tau} $) with coordinates $ P=(v_P,r_P,\omega_P) $,
	let $ \rho = r_p $ and $ \omega = \omega_P $, where $ \omega_P $ corresponds to the spherical coordinates. Thus, for each point on the hypersurface there is an associated pair $ (\rho,\omega) $ for $\rho\geq M, \ \omega\in \mathbb{S}^{2}. $ In addition, by the construction of the foliations, we have $ \partial_{\rho} = k(r)\cdot \partial_v + \partial_r $, for a bounded function $ k(r) $. Hence, in view of $ [\partial_v,\partial_{\rho}]=[\partial_{\rho},\partial_{\theta}] = [\partial_{\rho},\partial_{\phi}]= [\partial_{\theta},\partial_{\phi}] = 0 ,  $ and Frobenius theorem, we have that $ (\rho,\omega) $ defines a Lie propagated coordinate system on the foliation $ \Sigma_{\tau} $ (similarly for $ \check{\Sigma}_{\tau} $).

	\subsection{The \texorpdfstring{$ S^2_{v,r} $}{PDFstring}--tensor algebra and commutation formulae in ERN.}
	We briefly present relevant angular operators for $ S^2_{v,r} -$tensors and useful identities they satisfy, adopted on the ERN spacetime.  We express everything with respect to the ingoing Eddington--Finkelstein coordinates $ (v,r,\vartheta,\varphi). $
	
	Let $ \slashed{\nabla} $ be the covariant derivative associated to the round  metric $ \slashed{g} $ on the spheres $ S^2_{v,r}. $ Next, we denote by $ \slashed{\nabla}_{\partial_r} $ the projection to $ S^2_{v,r} $ of the spacetime covariant derivative $ \nabla_{\frac{\partial}{\partial r}} $. For the following  definitions of angular operators, let $ \xi_{A} $ be  any one-form and $ \theta_{AB} $ be any symmetric traceless 2-tensor on $ S^2_{v,r} $ . Then,  \begin{enumerate}[$ \diamond $ ]
		\item $ \dd $ takes  $ \xi $ to the pair of scalars $ (\slashed{div}\xi, \slashed{curl}\xi) ,$ where \[ \slashed{div}\xi = \slashed{\nabla}^{A}\xi_A, \hspace{1cm}   \slashed{curl}\xi = \slashed{\epsilon}^{AB}\slashed{\nabla}_{A}\xi_B\]
		\item $ \dd^{\star} $ is the formal $ L^2- $adjoint of $ \dd $, which takes any pair of scalars $ (f,g) $ into the one-form $ -\slashed{\nabla}_A f + \slashed{\epsilon}_{AB}\slashed{\nabla}^{B}g $. 
		\item $ \DD  $ takes the tensor $ \theta $ to the one-form 	$ (\slashed{div}  \theta) _{A} = \slashed{\nabla}^{C}\theta_{CA}$. 
		\item $ \DD^{\star} $ is the formal $ L^2- $adjoint of $ \DD $ which takes a one form $ \xi $ to the symmetric traceless two tensor \[ (\DD^{\star}\xi) _{AB} = -\frac{1}{2}\left(\slashed{\nabla}_B\xi_A+ \slashed{\nabla}_A\xi_B- (\slashed{div}\xi)\slashed{g}_{AB}\right). \]
	\end{enumerate}
	The first set of identities relating the above operators to the Laplacian $ \slashed{\Delta} $ on $ S^2_{v,r} $ can be easily checked \begin{align}
		\begin{aligned}
			&\dd \dd^{\star} = - \slashed{\Delta}_0, \hspace{1cm} &\dd^{\star}\dd = - \slas{\Delta}_1
			+ K  \\
			&\DD \DD^{\star} = -\frac{1}{2}\slas{\Delta}_1 - \frac{1}{2}K, \hspace{1cm} &\DD^{\star}\DD = -\frac{1}{2} \slas{\Delta}_2 +  K,  \label{elliptic 1}
		\end{aligned}
	\end{align}
	where $ K $ is the Gauss curvature of $ S^2_{v,r} $ and $ \slas{\Delta}_k $ is the Laplacian acting on $ k- $tensors respectively. 
	
	In addition, for  any $ \xi_{A_1,\dots,A_n} $, an $ n-$covariant $ S^2_{v,r} $-tensor, we have \begin{align}
		\left 	(\slas{\nabla}_{\partial_r}\slas{\nabla}_B -\slas{\nabla}_B \slas{\nabla}_{\partial_r}\right )\xi_{A_1\dots A_n} = - \frac{1}{r}\slas{\nabla}_B  \xi_{A_1\dots A_n}  
	\end{align} 
	from which we obtain the commutation identities \begin{align}
		[\slas{\nabla}_{\partial_r}, \dd]\xi = - \dfrac{1}{r}\dd \xi, \hspace{1cm}   [\slas{\nabla}_{\partial_r}, \DD]\theta = - \dfrac{1}{r}\DD \theta  \label{cummutation formula}
	\end{align}

	\subsection{Spherical harmonics and elliptic identities.} 
	In this paragraph, we briefly recall the spherical harmonics, and using the operators introduced above we define the corresponding orthogonal decomposition of one forms $ \xi $ and traceless symmetric 2-tensors $ \theta $ on  $ S_{v,r}^2 $.
	
	Fix $ \ell \in \mathbb{N} $, $ \abs{m}\leq \ell $ and denote by $ \mathring{Y}_{m}^{\ell}(\vartheta,\varphi) $  the spherical harmonics on the unit sphere, i.e. \[ \slashed{\Delta}_{\mathring{g}} \mathring{Y}_{m}^{\ell} = - \ell(\ell+1)\mathring{Y}_{m}^{\ell}\]
	This family forms an orthogonal basis of $ L^2(\mathbb{S}^2) $ with respect to the standard inner product on the sphere. Since we will be working on spheres of radius $ r $, $ S^2_{v,r} $, we take the normalized spherical harmonics denoted by $  Y_{m}^{\ell} (r,\vartheta,\varphi) $ and satisfy \[ \slashed{\Delta} Y_{m}^{\ell}= - \dfrac{1}{r^2}\ell(\ell+1) Y_{m}^{\ell}. \]
	A function $ f $ is said to be supported on the fixed frequency $ \ell  $ if the following projections vanish \[ \int_{S^2_{v,r}} f\cdot Y_{m}^{\tilde{\ell}} = 0\ , \]
	for all $ \tilde{\ell}\neq \ell. $ 
	
	Now, we recall that any one-form $ \xi $ has a unique representation $ \xi= r \dd^{\star}(f,g) $ for two functions $ f,g  $ on the unit sphere with vanishing mean, i.e. $ f_{\ell=0}=g_{\ell=0}=0.$ 
	\\ Similarly, any traceless symmetric two-tensor $ \theta $ on $ S^2_{v,r} $ has a unique representation $ \theta = r^2 \DD^{\star}\dd^{\star}(f,g) $, where both scalars $ f,g $ are supported on $ \ell \geq 2.  $ 
	
	We say that $ \xi, $ or $ \theta, $ are supported on the fixed frequency $ \ell  $, if the scalars $ f,g $ in their unique representation are supported on the fixed frequency $ \ell. $
	Now, using the identities from relation (\ref{elliptic 1}) we can show  \begin{align}
		\int_{S^2_{v,r}} \left(\abs{\slas{\nabla}\xi}^2 + \dfrac{1}{r^2}\abs{\xi}^2\right)= \int_{S^2_{v,r}}\abs{\dd \xi}^2  \label{elliptic 2.5}\\
		\int_{S^2_{v,r}} \left(\abs{\slas{\nabla}\theta}^2 + \dfrac{2}{r^2}\abs{\theta}^2\right)= 2\int_{S^2_{v,r}}\abs{\DD \theta}^2, \label{elliptic 3}
	\end{align}
	where for any $ S^2_{v,r}-$ tensor of rank n, $ \xi_{A_1\dots A_n} $, we have \[ \abs{\xi}^2 := \xi^{A_1\dots A_n}\cdot \xi_{A_1\dots A_n}= \slas{g}^{A_1B_1}\cdots \slas{g}^{A_nB_n}\xi_{A_1\dots A_n}\xi_{B_1\dots B_n} \]
	In addition, if $ \xi, \theta $  are supported on the fixed angular frequency $ \ell $, we have the following elliptic identities \begin{align}
		\int_{S^2_{v,r}}\abs{\slas{\nabla}\xi}^2 = \int_{S^2_{v,r}}\dfrac{\ell(\ell+1)-1}{r^2}\abs{\xi}^2, \\
		\int_{S^2_{v,r}}\abs{\slas{\nabla}\theta}^2 = \int_{S^2_{v,r}}\dfrac{\ell(\ell+1)-4}{r^2}\abs{\theta}^2. \label{elliptic 4}
	\end{align} 
	In particular, for $ \xi $ supported on the fixed  angular frequency $ \ell\geq 1 $, we have \begin{align}
		\int_{S^2_{v,r}}\abs{\dd \xi}^2 = 	\int_{S^2_{v,r}} \dfrac{\ell(\ell+1)}{r^2}\abs{\xi}^2  \label{elliptic 5 }
	\end{align}
	and for a traceless symmetric tensor $ \theta $ supported on the angular frequency $ \ell\geq 2, $ using the above identity for $ \xi = \DD\theta $ and then identities (\ref{elliptic 3}), (\ref{elliptic 4}) we obtain  \begin{align}
		\int_{S^2_{v,r}} \abs{\dd \DD \theta }^2 = 	\int_{S^2_{v,r}} \dfrac{\ell(\ell+1)}{2r^2}\cdot \dfrac{\ell(\ell+1) -2}{r^2}\abs{\theta}^2 \label{elliptic 6}
	\end{align}

	\section{The generalized Teukolsky and Regge--Wheeler system and the gauge invariant hierarchy} \label{set up section}
	The main goal of this paper is to study induced gauge--invariant quantities, arising from linear gravitational and electromagnetic perturbations of ERN spacetime. To do so, we rely on the set-up and resulting linearized equations obtained in \cite{giorgi2019boundedness,giorgie2020boundedness}, however, we adapt all notions involved to the extremal case $ M = \abs{\mathcal{Q}} .$ 
	
	In the following subsections, we first give an overview of the general set-up used and briefly describe the quantities we will be working with throughout the paper. Next, we write down the generalized Teukolsky system that the aforementioned quantities satisfy; we also go through the transformation theory that allows us to study this system by reducing it to a set of generalized Regge--Wheeler equations. Finally, using operators introduced in Section \ref{Geometry}, we derive the resulting scalar system from the tensorial one, and we show that it is subject to decoupling after we consider its spherical harmonics decomposition.
	
	\subsection{The set-up and the gauge--invariant quantities} Let $ (\mathcal{M},\mathbf{g}) $ be a $ 3+1-$dimensional Lorentzian manifold satisfying the Einstein--Maxwell equations \begin{align} \label{EME}
		\begin{aligned}
			\bm{Ric}(\mathbf{g})_{\mu \nu }\ &=\ 2 \bm{F}_{\mu \lambda} \bm{F}^{\lambda}_{\nu} - \dfrac{1}{2}\mathbf{g}_{\mu\nu}\bm{F}^{\alpha\beta}\bm{F}^{\alpha\beta}\\
			\bm{D}_{[\alpha}\bm{F}_{\beta\gamma]} \ &= \ 0, \hspace{1cm} \bm{D}^{\alpha}\bm{F}_{\alpha\beta} \ = \ 0,
		\end{aligned}
	\end{align}
	where $ \bm{D} $ is the covariant derivative associated to the metric $ \mathbf{g}_{\mu\nu} $, and $ \bm{F}_{\mu\nu} $ is the electromagnetic tensor. The author in \cite{giorgi2020linear} builds upon the formalism of \textbf{null frames} and initiates a null decomposition of (\ref{EME}) for induced quantities, i.e. Ricci coefficients, curvature, electromagnetic components, and the equations they satisfy. This can be done since any Lorentzian manifold admits a foliation out of 2-surfaces ($ S,\slas{\mathbf{g}} $), where $ \slas{\mathbf{g}} $ is the pullback metric of $ \mathbf{g} $ on $ S,  $ and at each point in $ \mathcal{M} $ we can associate a local null frame\footnote[1]{A null frame is one that satisfies:  $ g(e_3,g_3)=g(e_4,e_4)=0, \hspace{0.2cm} g(e_3,e_4)=-2\ $ and $\ g(e_A,e_B) = \slas{g}_{AB} $}
	$ \mathcal{N} = \br{e_3,e_4, e_A} $, with $ e_A$ tangent to $ S,$ for  $  A\in \br{1,2} $. We briefly express relevant quantities with respect to this null frame: \begin{itemize}
		\item \textbf{Ricci coefficients}
		\vspace{-0.3cm}
		\begin{align*}
			\begin{aligned}
				\chi_{A B} &:=\mathbf{g}\left(\mathbf{D}_A e_4, e_B\right), & \underline{\chi}_{A B}:=\mathbf{g}\left(\mathbf{D}_A e_3, e_B\right) \\
				\eta_A: &=\frac{1}{2} \mathbf{g}\left(\mathbf{D}_3 e_4, e_A\right), & \underline{\eta}_A:=\frac{1}{2} \mathbf{g}\left(\mathbf{D}_4 e_3, e_A\right) \\
				\xi_A: &=\frac{1}{2} \mathbf{g}\left(\mathbf{D}_4 e_4, e_A\right), & \underline{\xi}_A:=\frac{1}{2} \mathbf{g}\left(\mathbf{D}_3 e_3, e_A\right) \\
				\omega: &=\frac{1}{4} \mathbf{g}\left(\mathbf{D}_4 e_4, e_3\right), & \underline{\omega}:=\frac{1}{4} \mathbf{g}\left(\mathbf{D}_3 e_3, e_4\right) \\
				\zeta_A &:=\frac{1}{2} \mathbf{g}\left(\mathbf{D}_A e_4, e_3\right),& 
			\end{aligned}		
		\end{align*}
		and we also denote by $ \hspace{1cm}\kappa : = tr \chi, \hspace{1cm} \underline{\kappa} := tr \underline{\chi}. $
		\item \textbf{Curvature components} 
		\vspace{-0.3cm}
		\begin{align*}
			\begin{aligned}
				\alpha_{A B} &:=\mathbf{W}\left(e_A, e_4, e_B, e_4\right), & & \underline{\alpha}_{A B}:=\mathbf{W}\left(e_A, e_3, e_B, e_3\right) \\
				\beta_A: &=\frac{1}{2} \mathbf{W}\left(e_A, e_4, e_3, e_4\right), & & \underline{\beta}_A:=\frac{1}{2} \mathbf{W}\left(e_A, e_3, e_3, e_4\right) \\
				\rho: &=\frac{1}{4} \mathbf{W}\left(e_3, e_4, e_3, e_4\right), & & \sigma:=\frac{1}{4}{ }^{\star} \mathbf{W}\left(e_3, e_4, e_3, e_4\right)
			\end{aligned}
		\end{align*}
		where $ \mathbf{W} $ is the Weyl curvature of $ \mathbf{g} $ and $ ^{\star}\mathbf{W} $ denotes the Hodge dual.
		\item \textbf{Electromagnetic components} \begin{align*}
			\begin{aligned}
				&^{(F)}\beta_{A}:=\mathbf{F}\left(e_A, e_4\right), \quad ^{(F)}\underline{\beta}_{A}:=\mathbf{F}\left(e_A, e_3\right)\\
				&{ }^{(F)} \rho:=\frac{1}{2} \mathbf{F}\left(e_3, e_4\right), \quad{ }^{(F)} \sigma:=\frac{1}{2}{ }^{\star} \mathbf{F}\left(e_3, e_4\right)
			\end{aligned}
		\end{align*}
		where $ ^{\star}\mathbf{F} $ denotes the Hodge dual of $ \mathbf{F}.$
		
	\end{itemize}
	\begin{flushleft}
		\hrulefill
	\end{flushleft}

	With the above in mind, consider a one-parameter family of Lorentian metrics $ \bm{g}(\epsilon) $, around the Reissner--Nordstr\"om (\textit{RN}) solution, $ \bm{g}(0)\equiv g_{M,\mathcal{Q}} $, solving (\ref{EME}) and written in the following \textbf{Bondi form}; see \cite{giorgi2020linear}.
	\begin{align*}
		\begin{aligned}
			\mathbf{g}(\epsilon) &=-2 \varsigma(\epsilon) d u d s-\varsigma(\epsilon)^2 \underline{\Omega}(\epsilon) d u^2 \\
			&+\slas{\mathbf{g}}_{A B}(\epsilon)\left(d \theta^A-\frac{1}{2} \varsigma(\epsilon) \underline{b}(\epsilon)^A d u\right)\left(d \theta^B-\frac{1}{2} \varsigma(\epsilon) \underline{b}(\epsilon)^B d u\right),
		\end{aligned}
	\end{align*}
	with $ \varsigma(0)=1, \ \underline{\Omega}(0)= \left(1-\frac{2M}{r} + \frac{\mathcal{Q}^{2}}{r^{2}}\right), \ \underline{b}_A(0)=0 $ and $ \slas{\mathbf{g}}_{AB}(0)\equiv r^{2}\gamma_{AB}. $
	
	\textit{Initiating a linear gravitational and electromagnetic perturbation of RN spacetime
		corresponds to linearizing the full system of equations obtained from the null decomposition of (\ref{EME}) in terms of $ \epsilon $, with respect to the associated null frame $ \mathcal{N}_{\epsilon} = \br{\bm{e}_3(\epsilon),\bm{e}_4(\epsilon), \bm{e}_A(\epsilon)} $} given by \begin{align*}
		e_3(\epsilon)=2 \varsigma^{-1}(\epsilon) \partial_u+\underline{\Omega}(\epsilon) \partial_s+\underline{b}(\epsilon)^A \partial_{\theta^A}, \quad e_4(\epsilon)=\partial_s, \quad e_A=\partial_{\theta^A}.
	\end{align*} For example, the linearization of the induced Bianchi equation satisfied by the extreme curvature component $ \bm{\alpha}(\epsilon) \equiv \ 0 \ + \epsilon\cdot \alpha $  yields the following equation for the linearized quantity $ \alpha $  \begin{align*}
		\slas{\nabla}_3 \alpha+\left(\frac{1}{2} \underline{\kappa}-4 \underline{\omega}\right) \alpha=-2 \slas{D}_{2}^{\star} \beta-3 \rho \hat{\chi}-2\el{\rho}\left(\slas{D}^{\star}_2 \el{\beta}+\el{\rho} \hat{\chi}\right),
	\end{align*}
	where all operators and scalars that appear above are taken with respect to the RN background metric; the rest are unknowns of the linearized system. For the complete set of linearized system and their list of unknowns, see section 4 in
	\cite{giorgi2020linear}.
	
	\paragraph{Gauge--invariant quantities.} In this paper, we  focus our analysis on a special set of derived unknowns called \textit{gauge--invariant quantities},  identified as the ones that vanish in any \textit{pure gauge}  solution. Pure gauge solutions are derived from linearizing families of metrics that correspond to a smooth coordinate transformation of \textit{RN}, preserving its Bondi form. 
	
	The first two such quantities are unknowns of the linearized system itself, $ \alpha $ and $ \underline{\alpha}, $ which correspond to the linearized quantities of the extreme curvature components $ \bm{\alpha}(\epsilon)$, $ 
	\underline{\bm{\alpha}}(\epsilon),$  respectively. In the linear theory of Einstein \textbf{vacuum} equations, these quantities satisfy decoupled wave equations, known as the Teukolsky equation, which is the starting point of the analysis.  In the Einstein-Maxwell case, they no longer appear alone in these equations since there is a new set of gauge invariant quantities acting as a source. These new quantities are defined as \begin{align}
		\mathfrak{f} := \DD^{\star}\el{\beta} + \el{\rho} \hat{\chi}, \hspace{2cm} \underline{\mathfrak{f}}  := \DD^{\star}\el{\underline{\beta}} - \el{\rho} \underline{\hat{\chi}}
	\end{align}
	where $ \el{\beta}, \el{\underline{\beta}} $ and $ \hat{\chi}, \underline{\hat{\chi}} $ are unknowns of the linearized system, while $ \el{\beta},\ \DD^{\star} $ are to be taken with respect to the background RN metric. It is quite remarkable that also $ \mathfrak{f} $ and $ \underline{\mathfrak{f}} $ satisfy Teukolsky type equations themselves, however, coupled with $ \alpha $ and $ \underline{\alpha} $, respectively, acting as a source.
	
	The last set of gauge invariant quantities we are interested in are given by \begin{align}
		\tilde{\beta} := 2 \el{\rho}\beta - 3 \rho \el{\beta}, \hspace{2cm} \tilde{\underline{\beta}} := 2 \el{\rho}\underline{\beta} - 3 \rho \el{\underline{\beta}}, \label{beta definition}
	\end{align}
	where $ \el{\beta}, \el{\underline{\beta}} $ and $ \beta, \underline{\beta} $ are unknowns. In the case of Maxwell equations in Schwarzschild  spacetime, we have $ \el{\rho} = 0$ and the quantities $ \el{\beta}, \el{\underline{\beta}} $ are gauge--invariant themselves, satisfying Teukolsky equations of $ \pm 1 $ spin. However, in our case where $ \el{\rho}\neq 0 $, we consider the modified quantities (\ref{beta definition}) which now are gauge--invariant and we will see they satisfy a generalized Teukolsky equation of $ \pm 1 $ spin. 
	
	\subsection{Rescaled null frame and regular quantities.} Throughout the papers \cite{giorgi2020linear,giorgi2020linearfull}, the author uses the null frame $ \mathcal{N}_{\epsilon}  $ defined in terms of the Bondi coordinates reducing  to the \textit{outgoing Eddington--Finkelstein} coordinates in the case of RN spacetime. However, much like the coordinate system itself, this null frame does \textbf{not} extend smoothly to the event horizon $ \h. $ Nevertheless, the rescaled  null  frame $ \mathcal{N}_{\star} : = \br{\underline{\Omega}^{-1}(\epsilon)\epsilon_{3}(\epsilon), \underline{\Omega}(\epsilon) e_{4}(\epsilon), e_{A} }$
	extends smoothly to $ \h $, and so do all induced gauge--invariant quantities when expressed with respect to this frame.
	
	Below, we write the rescaled frame in the extreme RN spacetime with respect to the \textit{Ingoing Eddington-Finkelstein} coordinates, and we give the appropriate rescalings of relevant quantities so that they extend smoothly on $ \h. $ In particular, in ERN we have $ \underline{\Omega} = D(r) = \left(1-\frac{M}{r}\right) ^{2}$, and in the ingoing coordinates (\ref{ingoing metric}) the rescaled null vectors are given by \begin{align}
		e_{3}^{\star} = - \partial_r, \hspace{1cm} e_{4}^{\star} = 2\partial_v + D\partial_r.
	\end{align}
	Then, the corresponding set of rescaled gauge--invariant quantities that extend smoothly on $ \h $, and the ones that we aim to obtain estimates for, are given by \begin{align}\begin{aligned}
			\alpha_{\star} &= D^{2}\cdot\alpha, \hspace{1cm} 	\mathfrak{f}_{\star} = D\cdot \mathfrak{f}, \hspace{1.5cm} \tilde{\beta} _{\star} = D\cdot \tilde{\beta} \\
			\underline{\alpha}_{\star}&= D^{-2} \cdot \underline{\alpha} 	, \hspace{1cm} \underline{\mathfrak{f}} _{\star} = D^{-1} \cdot \underline{\mathfrak{f}}, \hspace{0.9cm} \tilde{\underline{\beta}} _{\star} = D^{-1}\cdot \tilde{\underline{\beta}}
		\end{aligned}
	\end{align}
	From now on, all quantities with $ \star $ subscript are expressed with respect to the rescaled frame $ \mathcal{N}_{\star} $ defined above, and thus are regular up to and including the horizon $ \h. $

	\subsection{The Teukolsky and Regge--Wheeler system of \texorpdfstring{$ \pm $}{PDFstring} spin}\label{Teukolsky section} First, let us present the generalized Teukolsky equations of $ \pm 2 $ spin satisfied by the traceless, symmetric gauge--invariant quantities $ \alpha _{\star}, \mathfrak{f}_{\star} $, and $ \underline{\alpha}_{\star}, \underline{\mathfrak{f}}_{\star} $. We simply adapt the equations of \cite{giorgi2020linear} in the case of ERN and  express them in terms of the $ \mathcal{N}_{\star} $ null frame.
	For the $ +2 $ spin equations we write
	\begin{align}
		\begin{aligned}
			\square_{\mathbf{g}} \alpha_{\star}\ =\ &2(\kappa_{\star}+2 \omega_{\star}) \slas{\nabla}_3 \alpha_{\star}+\left(\frac{1}{2} \kappa_{\star} \underline{\kappa}_{\star}-4 \rho_{\star}+4\el{\rho}_{\star}^2+2 \omega _{\star}\underline{\kappa}_{\star}\right) \alpha_{\star} \\ & +4\el{\rho}_{\star}\left(\slas{\nabla}_{4^{\star}} \mathfrak{f}_{\star}+(\kappa_{\star}+2 \omega_{\star}) \mathfrak{f}_{\star}\right), \\
			\square_{\mathbf{g}}(r \mathfrak{f}_{\star})\ = \ &(\kappa_{\star}+2 \omega_{\star}) \slas{\nabla}_{3^{\star}}(r \mathfrak{f}_{\star})+\left(-\frac{1}{2} \kappa_{\star} \underline{\kappa}_{\star}-3 \rho_{\star}+\omega_{\star} \underline{\kappa}_{\star}\right) r \mathfrak{f}_{\star} \\ & -r\el{\rho}_{\star}\left(\slas{\nabla}_{3^{\star}} \alpha_{\star}+\underline{\kappa}_{\star} \alpha_{\star}\right).
		\end{aligned}
	\end{align}
	where $ \square_{\mathbf{g}} = \mathbf{g}^{\mu\nu} \nabla_{\nu}\nabla_{\mu} $
	is taken with respect to the extreme Reissner--Norstr\"om metric, and  we denote by $ \slas{\nabla}_{3^{\star}} \equiv \slas{\nabla}_{e_{3}^{\star}}  $, and similarly for $ e_{4}^{\star} $. All coefficients above correspond to the background ERN values when expressed in the $ \mathcal{N}_{\star} $ null frame, given by \begin{align*}
		\kappa _{\star}= \frac{2D}{r}, \hspace{1cm} \underline{\kappa}_{\star} = -\dfrac{2}{r}, \hspace{1cm} \omega_{\star}= -\dfrac{M}{r^{2}}\sqrt{D}, \hspace{1cm} \underline{\omega}_{\star}=0, \hspace{1cm} \el{\rho}_{\star}= \dfrac{M}{r^{2}}, \hspace{1cm} \rho _{\star}= -\dfrac{2M}{r^{3}}\sqrt{D}.
	\end{align*}
	Similarly, we have the $ -2 $ spin equations \begin{align}
		\begin{aligned}
			\square_{\mathbf{g}} \underline{\alpha}_{\star}=&-4 \omega \slas{\nabla}_{3^{\star}} \underline{\alpha}_{\star}+2\underline{\kappa} \slas{\nabla}_{4^{\star}} \underline{\alpha}_{\star}+\left(\frac{1}{2} \kappa \underline{\kappa}-4 \rho+4\el{\rho}^2-10 \omega \underline{\kappa}-4 \slas{\nabla}_{3^{\star}} \omega\right) \underline{\alpha} _{\star}\\
			&-4\el{\rho}\left(\slas{\nabla}_{3^{\star}} \underline{\mathfrak{f}}_{\star}+(\underline{\kappa}+2 \underline{\omega}) \underline{\mathfrak{f}}_{\star}\right), \\
			\square_{\mathbf{g}}(r \underline{\mathfrak{f}}_{\star})=&-2 \omega \slas{\nabla}_{3^{\star}}(r \underline{\mathfrak{f}}_{\star})+\underline{\kappa} \slas{\nabla}_{4^{\star}}(r \underline{\mathfrak{f}}_{\star})+\left(-\frac{1}{2} \kappa \underline{\kappa}-3 \rho-3 \omega \underline{\kappa}-2 \slas{\nabla}_{3^{\star}} \omega\right) r \underline{\mathfrak{f}}_{\star} \\
			&+r\el{\rho}\left(\slas{\nabla}_{4^{\star}} \underline{\alpha}_{\star}+(\kappa-4 \omega) \underline{\alpha}_{\star}\right)
		\end{aligned}
	\end{align}
	Regarding the $ \pm 1 $ spin generalized Teukolsky equations satisfied by $ \tilde{\beta}_{\star} $ and $ \underline{\tilde{\beta}}_{\star} $ we have \begin{align}
		\begin{aligned}
			\square_{\mathbf{g}}\left(r^3 \tilde{\beta}_{\star}\right)\ =\ &(\kappa_{\star}+2 \omega_{\star}) \slas{\nabla}_{3^{\star}}\left(r^3 \tilde{\beta}_{\star}\right)+\left(\frac{1}{4} \kappa_{\star} \underline{\kappa}_{\star}+\omega _{\star}\underline{\kappa}_{\star}-2 \rho_{\star}+3\el{\rho}_{\star}^2+2 \slas{\nabla}_{3^{\star}} \omega_{\star}\right) r^3 \tilde{\beta} _{\star}\\
			&-2 r^3 \underline{\kappa}_{\star}\el{\rho}_{
				\star}^2\left(\slas{\nabla}_{4^{\star}}\el{\beta}_{\star}+\left(\frac{3}{2} \kappa_{\star}+2 \omega_{\star}\right)\el{\beta}_{\star}-2\el{\rho}_{\star}\xi_{\star}\right)+8r^{3}\left (\el{\rho}_{\star}\right ) ^{2}\cdot \slas{div}\mathfrak{f}_{\star}
			\\
			\square_{\mathbf{g}}\left(r^3 \underline{\tilde{\beta}}_{\star}\right)=&-2 \omega_{\star} \slas{\nabla}_{3^{\star}}\left(r^3 \underline{\tilde{\beta}}_{\star}\right)+\underline{\kappa} \slas{\nabla}_{4^{\star}}\left(r^3 \underline{\tilde{\beta}}_{\star}\right)+\left(\frac{1}{4} \kappa_{\star} \underline{\kappa}_{\star}-3 \omega_{\star} \underline{\kappa}_{\star}-2 \rho+3^{(F)} \rho_{\star}^2\right) r^3 \underline{\tilde{\beta}}_{\star} \\
			&-2 r^3 \kappa_{\star}\el{\rho}_{\star}^2\left(\slas{\nabla}_{3^{\star}}\el{\underline{\beta}}_{\star}+\frac{3}{2}\kappa_{\star}\el{\underline{\beta}}_{\star}+2\el{\rho} _{\star}\underline{\xi}_{\star}\right)+8r^{3}\left (\el{\rho}_{\star}\right )^{2}\cdot \slas{div}\underline{\mathfrak{f}}_{\star},
		\end{aligned}
	\end{align}
	where here $ \el{\beta}_{\star}, \xi_{\star} $ and $ \el{\underline{\beta}}_{\star}, \underline{\xi}_{\star} $ are unknowns of the linearized system. 
	
	In addition to the Teukolsky equations above, the following relating equations have been derived in \cite{giorgi2019boundedness}, for $ \alpha, \mathfrak{f} $ and $ \tilde{\beta} $, and their underlined analogues.
	\begin{align}
		^{(F)}\rho_{\star} \dfrac{1}{\underline{\kappa}_{\star}} \slas{\nabla}_{3^{\star}}\left(r^3 \underline{\kappa}_{\star}^2 {\alpha}_{\star}\right) &\ =\  - \left(^{(F)}\rho_{\star}^{2} + 3 \rho_{\star}\right) r^3 \underline{\kappa}_{\star}{\mathfrak{f}}_{\star} - r^3 \underline{\kappa}_{\star}  \DD^{\star}\left(\tilde{{\beta}}_{\star}\right) \\
		^{(F)}\rho_{\star} \dfrac{1}{\kappa_{\star}} \slas{\nabla}_{4^{\star}}\left(r^3 \kappa_{\star}^2 \underline{\alpha}_{\star}\right) &\  =\  \left(^{(F)}\rho_{\star}^{2} + 3 \rho_{\star}\right) r^3 \kappa_{\star} \underline{\mathfrak{f}}_{\star} + r^3 \kappa_{\star} \DD^{\star}\left(\tilde{\underline{\beta}}_{\star}\right)
	\end{align}
	We will use these relations to obtain estimates for $ \alpha_{\star} $ and $ \underline{\alpha}_{\star} $, after we control the quantities of the right-hand side.

	Below, we present the transformation theory that allows us to obtain the generalized Regge--Wheeler equation. We write our equations in the ingoing coordinates of Section \ref{Geometry}.
	\paragraph{The transformation theory.} For any $ n- $rank $ S^{2}_{v,r}-$tensor $ \xi $, consider the following operators  \begin{align}\label{Transformation operators}
		\underline{P}(\xi) :=& \dfrac{1}{\underline{\kappa}_{\star}}\slas{\nabla}_{3^{\star}} (r\cdot\xi) \ = \ \dfrac{r}{2} \slas{\nabla}_{\partial_r}(r\cdot\xi) , \\
		P(\xi) :=& \dfrac{1}{\kappa_{\star}}\slas{\nabla}_{4^{\star}}(r\cdot \xi) \ = \ \dfrac{r}{2D} \slas{\nabla}_{2\partial_v + D\partial_r}(r\cdot\xi)
	\end{align}
	Using the above operators we define the following tensors \begin{align}
		\begin{aligned}
			\bm{q}^{F}\ :=& \ \underline{P}(r^{2}\underline{\kappa}_{\star}\cdot \mathfrak{f}_{\star}) = - r \slas{\nabla}_{\partial_r}(r^2\mathfrak{f}_{\star})	 
		\end{aligned}
		\\
		\begin{aligned}
			\bm{p}\ :=& \ \underline{P}(r^{4}\underline{\kappa}_{\star}\cdot \tilde{\beta}_{\star}) = - r \slas{\nabla}_{\partial_r}(r^4\tilde{\beta}_{\star})
		\end{aligned}
	\end{align}
	and similarly, 
	\begin{align}
		\begin{aligned}
			\underline{\bm{q}}^{F} \ :=& \ P(r^{2}\kappa_{\star}\underline{\mathfrak{f}}) = 
			\dfrac{1}{\kappa_{\star}} \slas{\nabla}_{4^{\star}}\left (r^3 \kappa_{\star}\cdot \underline{\mathfrak{f}}_{\star}\right )  \\
			\ =& \ r^{3}\slas{\nabla}_{\partial_{v}}\underline{\mathfrak{f}}_{\star} + r \slas{\nabla}_{\partial_{r}}\left(r^{2}D\cdot\underline{\mathfrak{f
			}}_{\star}\right),
		\end{aligned}
		\\
		\begin{aligned}
			\underline{\bm{p}} \ :=&\ P(r^{4}\kappa_{\star}\underline{
				\tilde{\beta}}_{\star}) =\dfrac{1}{\kappa_{\star}} \slas{\nabla}_{4^{\star}}\left (r^5 \kappa_{\star}\cdot \underline{\tilde{\beta}}_{\star}\right ) \\ \ =& \  r^{5}\slas{\nabla}_{\partial_{v}}\underline{\tilde{\beta}}_{\star} + r \slas{\nabla}_{\partial_{r}}\left(r^{4}D\cdot\underline{\tilde{\beta}}_{\star}\right).
		\end{aligned}
	\end{align}
	Note, all tensor $ \bm{q}^{F}, \bm{p}$ and $ \underline{\bm{q}}^{F}, \underline{
		\bm{p}} $, are regular on the horizon $ \h. $
	In \cite{giorgi2020linearfull}, it has been shown that the above symmetric traceless tensors satisfy the following generalized Regge--Wheeler equations, adapted in ERN spacetime 
	\begin{align}
		\begin{aligned} 
			\square_{g}\bm{p} -V_1 \bm{p} &= \dfrac{8 M^2}{r^2} \DD \bm{q^F} \\
			\square_{g}\bm{q^{F}} -V_2 \bm{q^{F}} &= \dfrac{1}{r^2} \DD^{\star}\bm{p}
		\end{aligned} \label{csystem}
	\end{align}
	where, $ g $ corresponds to the background ERN metric as in (\ref{ingoing metric}), and 
	\begin{align} \label{potentials of initial system}
		V_1= \dfrac{1}{r^2}\left ( 1-\dfrac{2M}{r}+\dfrac{6M^2}{r^2}\right ) &= \dfrac{1}{r^2}\left ( D(r) +\dfrac{5M^2}{r^2}\right ), \\
		V_2= \dfrac{4}{r^2}\left (1-\dfrac{2M}{r}+\dfrac{3M^2}{2r^2}\right ) &= \dfrac{4}{r^2}\left (D(r)+\dfrac{M^2}{2r^2}\right ), 
	\end{align}
	The same equations hold for $ \underline{\bm{q}}^{F} $ and $ \underline{\bm{p}} $ as well.
	In the next subsection, we derive the scalar system and we show how to decouple it. This decoupling is necessary in order to capture the dominant behavior of solutions to (\ref{csystem}) asymptotically along the event horizon $ \h. $

	\subsection{Derivation and decoupling of the scalar Regge--Wheeler system}
	System \eqref{csystem} involves two tensorial equations, of type 2 and 1. Instead, we use the $ S^{2}_{v,r} $ angular operators of Section \ref{Geometry} to derive the corresponding scalar system. \textit{From now on we only write the equations for the positive spin case, and all estimates we derive will automatically hold for the underline  counterparts since they satisfy the same equations.}
	\begin{proposition} The induced scalars  $\phi := r^2 \dd\DD \bm{q^{F}} $, and $ \psi:= r\dd\bm{p} $, satisfy the following coupled system
		\begin{align}
			\begin{aligned}
				\square_{g} \phi + \left (4K-V_2\right )\phi &= -\dfrac{1}{2r}\slashed{\Delta} \psi -\dfrac{K}{r}\psi  \\
				\square_{g} \psi + \left (K-V_1\right )\psi &= \dfrac{8 M^2}{r^3}\phi.  
			\end{aligned} \label{ceq}
		\end{align} 
		where $ K $ is the Gauss curvature of the section spheres, and $ V_1, V_2 $ as in (\ref{potentials of initial system}).
		
		\begin{proof}
			We use Lemma A.1.4. of \cite{giorgie2020boundedness}: \[ \left (-r \slashed{D}_2 \square_2 + \square_1 r \slashed{D}_2\right ) \bm{\Psi} = -3 K r \slashed{D}_2 \bm{\Psi} \]
			\[ \left (-r\slashed{D}_1 \square_1 + \square_0 r\slashed{D}_1 \right ) \bm{\Phi}=-K r \slashed{D}_1 \bm{\Phi}. \]
			\[ -\slashed{D}_1 \slashed{\Delta}_1 + \slashed{\Delta}_0 \slashed{D}_1 = - K \slashed{D}_1.\]
			In particular, commuting equations of system (\ref{csystem}) with $ \slashed{D}_1, \slashed{D}_2 $ and using the relations above we obtain the left hand side of the equations as seen above in (\ref{ceq}).
			
			For the right-hand side of the first equations we need to compute \[ \slashed{D}_1\slashed{D}_2\slashed{D}_2^{\star }\bm{p} \]
			We use the fact that $ \slashed{D}_2 \slashed{D}_2^{\star} = -\dfrac{1}{2}\slashed{\Delta}_1 - \dfrac{1}{2}K $, and the last relation of the Lemma above to  obtain \[  
			\slashed{D}_1\left (\slashed{D}_2\slashed{D}_2^{\star }\bm{p}\right )= \slashed{D}_1 \left (-\dfrac{1}{2}\slashed{\Delta}_1-\dfrac{1}{2}K \right )\bm{p}= -\dfrac{1}{2}\slashed{\Delta}_0 \left (\slashed{D}_1 \bm{p}\right )- K \slashed{D}_1 \bm{p},
			\]
			which concludes the proof.

		\end{proof}
	\end{proposition}

	\begin{remark}
		We saw earlier that $ \bm{p}, \bm{q^{F}} $ are regular quantities on the horizon $ H^{+} $, and thus, $ \phi$ and $ \psi $ are regular on the horizon as well. 
	\end{remark}
	\subsubsection*{The Cauchy problem for the scalar coupled system}
	The Teukolsky system we introduced earlier admits a well-posed Cauchy initial value problem; see \cite{giorgi2019boundedness,giorgie2020boundedness}. However, we will be studying the induced coupled scalar system (\ref{ceq}) independently as a system of its own. In particular, we consider general solutions $ (\phi,\psi) $ to system (\ref{ceq}) arising from initial data \begin{align}
		\begin{aligned}
			\left 	(\phi\big|_{\Sigma_0} , n_{_{\Sigma_0}}\phi\big|_{\Sigma_0}\right ) \equiv (\phi_0, \phi_1) \in H_{loc}^{k}(\Sigma_0) \times H_{loc}^{k-1}(\Sigma_0),  \\
			\left 	(\psi\big|_{\Sigma_0},n_{_{\Sigma_0}}\psi\big|_{\Sigma_0}\right ) \equiv (\psi_0, \psi_1) \in H_{loc}^{k}(\Sigma_0) \times H_{loc}^{k-1}(\Sigma_0)
		\end{aligned}
	\end{align}
	for any $ k\geq 2 $,
	with $ \Sigma_0 $ as in the foliation paragraph of Section \ref{Geometry}, and $ n_{\Sigma_0} $ its future directed unit normal.
	Then, the solution   $(\phi,\psi)$  is unique in $ \mathcal{R}(0,\tau) $, with 
	$ (\phi,\psi) \in H_{loc}^{k}(\Sigma_{\tau})\times H_{loc}^{k}(\Sigma_{\tau})  $ and $ (n_{_{\Sigma}}\phi,n_{_{\Sigma}}\psi ) \in H_{loc}^{k-1}(\Sigma_{\tau})\times H_{loc}^{k-1}(\Sigma_{\tau})$, for all $ \tau >0. $
	
	In order to obtain pointwise estimates, we also impose the extra assumptions  \begin{align}
		\lim_{p\to i^{0}}r\phi^{2}(p) \Big|_{\Sigma_0}=0, \hspace{2cm} 	\lim_{p\to i ^{0}}r\psi^{2}(p) \Big|_{\Sigma_0}=0.
	\end{align}
	For the rest of this paper, \textit{by \textbf{``generic initial data"} we implicitly refer to initial data satisfying the assumptions above.}
	
	We proceed with the following proposition, in which we show how to decouple the scalar system (\ref{ceq}).
	
	\begin{proposition} Consider the spherical harmonic decomposition of (\ref{ceq}), then the system of equations supported on the fixed frequency $ \ell \geq 2 $ decouple to
		\begin{align}
			\square_{g} \Psi_1^{^{(\ell)}}+  \left(\frac{5M}{r^3}- \frac{M(2\ell+1)}{r^3}-\frac{6M^2}{r^4}  \right)\Psi_1^{^{(\ell)}}=&\ 0 \label{deq1} \\
			\square_{g} \Psi_2^{^{(\ell)}}+  \left(\frac{5M}{r^3}+ \frac{M(2\ell+1)}{r^3}-\frac{6M^2}{r^4}  \right)\Psi_2^{^{(\ell)}}=&\ 0 \label{deq2}
		\end{align}
		where $ \mu^2= (\ell+2)(\ell-1)$ and \begin{align*}
			\Psi_1 ^{^{(\ell)}} : &= 2M\mu\  \tilde{\phi}_{\ell}+ 2M(\ell+2)\ \tilde{\psi}_{\ell} \\  \Psi_2 ^{^{(\ell)}} : &= 2M(\ell+2)\  \tilde{\phi}_{\ell}- 2M\mu\ \tilde{\psi}_{\ell}
		\end{align*}
		for $ \tilde{\psi}=\frac 1 4  \mu \psi \quad$ and $\quad \tilde{\phi} = M \phi $.
	\end{proposition}
	\begin{proof}
		We project system (\ref{ceq}) to their spherical harmonics $ \psi_{\ell},\phi_{\ell} $ and in view of $ \slashed{\Delta}\psi_{\ell}= -\frac{\ell(\ell+1)}{r^2}\psi_{\ell}, $ $ K=\frac{1}{r^2} $, they satisfy the system \begin{align*}
			\square_{g}\phi_{\ell}+\left(\dfrac{8M}{r^3}-\dfrac{6M^2}{r^4}\right)\phi_{\ell} &= \dfrac{\mu^2}{2r^3}\psi_{\ell}, \hspace{1cm} \mu^2= (\ell+2)(\ell-1) \\
			\square_g \psi_{\ell} + \left(\dfrac{2M}{r^3}- \dfrac{6M^2}{r^4}\right)\psi_{\ell} &= \dfrac{8M^2}{r^3} \phi_{\ell}.
		\end{align*}
		Now, consider the rescalings $$ \tilde{\psi}_{\ell} := \frac{\mu}{4}\psi_{\ell}, \hspace{1cm} \tilde{\phi}_{\ell} :=M\phi_{\ell}, $$
		and by writing $ \frac{8M}{r^3}= \frac{5M}{r^3}+\frac{3M}{r^3}, \ \ \frac{2M}{r^3} = \frac{5M}{r^3}-\frac{3M}{r^3}  $ we obtain the following system \begin{align*}
			\square_{g}\tilde{\phi}_{\ell}+\left(\dfrac{5M}{r^3}-\dfrac{6M^2}{r^4}\right)\tilde{\phi}_{\ell} &= \dfrac{2M\mu}{r^3}\tilde{\psi}_{\ell} - \dfrac{3M}{r^3}\tilde{\phi}_{\ell},  \\
			\square_g \tilde{\psi}_{\ell} + \left(\dfrac{5M}{r^3}- \dfrac{6M^2}{r^4}\right)\tilde{\psi}_{\ell} &= \dfrac{2M\mu}{r^3} \tilde{\phi}_{\ell}+\dfrac{3M}{r^3}\tilde{\psi}_{\ell}.
		\end{align*}
		Let us denote  by $ \mathcal{P} := \square_{g} + \left(\frac{5M}{r^3}-\frac{6M^2}{r^4}\right) $,  the operator of the  left-hand side above, and the above  system now  reads \begin{align}
			\left(r^3\mathcal{P}\right)\binom{\tilde{\phi}_{\ell}}{\tilde{\psi}_{\ell}} \ = \ \binom{-3M \quad 2M\mu}{2M\mu \quad 3M}\binom{\tilde{\phi}_{\ell}}{\tilde{\psi}_{\ell}} 
		\end{align}
		The decoupling of system (\ref{ceq}) will follow after we diagonalize the symmetric matrix $ \mathcal{C} = \binom{-3M \quad 2M\mu}{2M\mu \quad 3M}, $ with $ \det(\mathcal{C}) = - M^2 \cdot (2\ell+1)^2 <0 $. In order to find the eigenvalues of $ \mathcal{C}, $ we compute its characteristic polynomial \[ p_{_{\mathcal{C}}}(\lambda) = \det(\lambda I - \mathcal{C})=M^2(2\ell+1)^2-\lambda^2, \]
		and thus, we obtain the two  eigenvalues $ \lambda_1 =  M(2\ell+1) $, $ \lambda_2 = - M(2\ell+1).	 $ One can check directly that $ \binom{2M\mu}{3M+\lambda_1},\ \binom{3M-\lambda_2}{-2M\mu} $ are two distinct non trivial eigenvectors of $ \mathcal{C}, $ thus the following two scalars 
		\begin{align*}
			\Psi_1 ^{^{(\ell)}} : &= 2M\mu\  \tilde{\phi}_{\ell}+ (3M+\lambda_1)\ \tilde{\psi}_{\ell} \\  \Psi_2 ^{^{(\ell)}} : &= (3M-\lambda_2)\  \tilde{\phi}_{\ell}- 2M\mu\ \tilde{\psi}_{\ell}
		\end{align*}
		satisfy \begin{align*}
			(r^3\mathcal{P})\left (\p^{^{(\ell)}}\right )& = \lambda_i \p^{^{(\ell)}} \\
			\Leftrightarrow \quad \square_{g}\p^{^{(\ell)}} + \left(\frac{5M}{r^3}-\frac{6M^2}{r^4}\right)\p^{^{(\ell)}} &=\dfrac{\lambda_i}{r^3} \p^{^{(\ell)}} \\
			\Leftrightarrow \quad \square_{g}\p^{^{(\ell)}} + \left(\frac{5M}{r^3}-\dfrac{\lambda_i}{r^3}-\frac{6M^2}{r^4}\right)\p^{^{(\ell)}} &= 0,
		\end{align*}
		which concludes the proof. \\
	\end{proof}
	
	\begin{corollary} \label{inverse relation}
		Given the scalars $ \Psi_{i}^{^{(\ell)}} $, for $ \ell\geq 2 $ and $ i\in\br{1,2} $ as above, it is  immediate to check  that  \begin{align}
			2M^2\cdot \phi_{\ell} & = \dfrac{\mu}{(\ell+2 )(2\ell+1)}\cdot \Psi_1^{^{(\ell)}}+ \dfrac{1}{(2\ell+1)}\Psi_2^{^{(\ell)}} \\
			M\cdot \psi_{\ell} &= \dfrac{2}{\mu\cdot (2\ell+1 )}\cdot \Psi_1^{^{(\ell)}}- \dfrac{2}{(\ell+2)(2\ell+1)}\Psi_2^{^{(\ell)}}
		\end{align}
	\end{corollary}

	\section{Estimates for the decoupled scalar Regge--Wheeler equations.}  
\label{Regge estimates}	In this section, we prove estimates for each equation (\ref{deq1}), (\ref{deq2}) at the same time. Note, the decoupled system can be written as \begin{align}
		\Big (\square_{g} - V_i^{^{(\ell)}}\Big ) \Psi_i^{^{(\ell)}} =0, 
	\end{align}
	where  \begin{align}
		V_1^{^{(\ell)}} & =- \left (\frac{5M}{r^3}- \frac{M\sqrt{9+4\mu^2}}{r^3}-\frac{6M^2}{r^4}\right )= -\dfrac{2M}{r^3}\left (3\sqrt{D}-(\ell+1)\right ), \ \ \ \ \ell\geq 2 \label{potential1} \\ 
		V_2^{^{(\ell)}} & =-\left (\frac{5M}{r^3}+ \frac{M\sqrt{9+4\mu^2}}{r^3}-\frac{6M^2}{r^4} \right )  =- \dfrac{2M}{r^3}\left (3\sqrt{D}+\ell\right ).\ \ \ \ \ell\geq 2
	\end{align} 
	By projecting the system (\ref{ceq}) to the $ \ell=1 $ frequency we obtain only one equation, i.e \begin{align}
		\square_g (\psi_{\ell=1})-\dfrac{2M}{r^3}(2-3\sqrt{D})(\psi_{\ell=1})=0.
	\end{align}
	Note, the above wave equation can be included in the case $ \Psi_1^{(\ell=1)} $ with $ V_1^{(\ell=1)} $ as in (\ref{potential1}). 
	Thus, we will be studying solutions to the equation \begin{align}
		\begin{aligned}
			&\Big (\square_{g} - V_i^{^{(\ell)}}\Big ) \Psi_i^{^{(\ell)}} =0, \\
			\text{with}\quad 
			& V_i^{^{(\ell)}}= -\dfrac{2M}{r^3}\Big(3\sqrt{D}+(-1)^{i}(\ell+2-i)\Big), \ \ \ \ell\geq i, \label{waveeq}
		\end{aligned}
	\end{align}
	with $ \p^{^{(\ell)}} $ supported on the fixed frequency $ \ell\geq i, \ i\in\br{1,2}.  $ For brevity, the superscript $ (\ell) $ will be frequently dropped and  inferred through the equations.

	\subsection{Preliminaries} In this section, we briefly recall the vector field method.
	First, consider the energy-momentum tensor associated to the wave equation (\ref{waveeq})
	\begin{align}
		\begin{aligned}
			\mathcal{Q}_{\mu \nu }[\Psi_i]=  \nabla_{\mu}\Psi_i \cdot \nabla_{\nu} \Psi_i - \dfrac{1}{2}g_{\mu \nu}\Big(\nabla^{a}\Psi_i\cdot  \nabla_{a}\Psi_i + V_i\Psi_i^2 \Big).
		\end{aligned}
	\end{align}
	\begin{proposition}\label{GenCur}
		Consider a scalar $ \Psi_i $ verifying equation (\ref{waveeq}). Let $ X $ be a vectorfield, $ \omega $ a scalar function, and $ M $ a one form. Define the general current \begin{align}
			J_{\mu}^{X,\omega,M} [\Psi_i] = \mathcal{Q}_{\mu \nu}X^{\nu} + \dfrac{1}{2} \omega \Psi_i \cdot \nabla_{\mu}\Psi_i-\dfrac{1}{4}(\nabla_{\mu}\omega)\Psi_i^2 + \dfrac{1}{4}\Psi_i^2 M_{\mu}, 
		\end{align}
		then, \begin{equation}
			\begin{aligned}
				K^{X,\omega, M}[\Psi_i]:=Div \left(J_{\mu}^{X,\omega,M} [\Psi_i] \right)  =& \dfrac{1}{2}\mathcal{Q}\cdot {}^{(X)}\pi + \left (-\dfrac{1}{2}X(V_i)-\dfrac{1}{4}\square_g \omega\right )\Psi_i^2 \\ &+ \dfrac{1}{2}\omega \left( \nabla^{a}\Psi_i\cdot\nabla_{a}\Psi_i + V_i\Psi_i^2\right) +\dfrac{1}{4}\nabla^{\mu}\left (\Psi_i^2 M_{\mu}\right ).
			\end{aligned} 
		\end{equation}
		\begin{proof}
			We begin by computing $ Div\left(\mathcal{Q}_{\mu\nu}\right) = \nabla^{\mu} \mathcal{Q}_{\mu\nu} $  \begin{align*}
				\nabla^{\mu} \mathcal{Q}_{\mu\nu}   & = \left(\nabla^{\mu}\nabla_{\mu}\Psi_i\right)   \cdot \nabla_{\nu} \Psi_i + \nabla_{\mu}\Psi_i \cdot\nabla^{\mu} \nabla_{\nu} \Psi_i - \frac{1}{2}g _{\mu\nu} \left(\nabla^{\mu}\Big(\nabla^{a}\Psi_i\cdot  \nabla_{a}\Psi_i\Big) + \nabla^{\mu}\left(V_i\cdot \Psi_i^2 \right) \right).
			\end{align*}
			We first treat the following term \begin{align*}
				- g_{\mu\nu} \nabla^{\mu}\Big(\nabla^{a}\Psi_i\cdot  \nabla_{a}\Psi_i\Big) &= -\dfrac{1}{2} \nabla_{\nu} \Big(\nabla^{a}\Psi_i\cdot  \nabla_{a}\Psi_i\Big) = -\dfrac{1}{2}\left(\nabla^{a}\nabla_{\nu}\p\right) \nabla _{a}\p -\dfrac{1}{2}\nabla^{a}\p \cdot \nabla_{a}\nabla_{\nu}\p \\
				&= -\left(\nabla^{a}\nabla_{\nu}\p\right) \nabla _{a}\p = -\left(\nabla^{\mu}\nabla_{\nu}\p\right) \nabla _{\mu}\p. 
			\end{align*}
			Using the above relation and the fact that  $ \nabla^{\mu}\nabla_{\mu}\Psi_i  = \square \p = V_i \p$ we obtain \begin{align}
				\nabla^{\mu} \mathcal{Q}_{\mu\nu} = V_i \p \nabla_{\nu}\p - \dfrac{1}{2}\nabla_{\nu}(V_i) \p^2 - V_i \p \nabla_{\nu}\p = -\dfrac{1}{2}\nabla_{\nu}(V_i) \p^2.
			\end{align}
			On the other hand, we have \begin{align*}
				&\nabla^{\mu}\left(\dfrac{1}{2} \omega \Psi_i \cdot \nabla_{\mu}\Psi_i-\dfrac{1}{4}(\nabla_{\mu}\omega)\Psi_i^2\right)\\ &= \dfrac{1}{2}(\nabla^{\mu}\omega) \Psi_i \cdot \nabla_{\mu}\Psi_i + \dfrac{1}{2}\omega  \nabla^{a}\Psi_i\cdot\nabla_{a}\Psi_i +\dfrac{1}{2}\omega \p \square\p - \dfrac{1}{4}\square\omega \p^2 -\dfrac{1}{2}(\nabla^{\mu}\omega) \Psi_i \cdot \nabla_{\mu}\Psi_i \\
				& = \dfrac{1}{2}\omega\left( \nabla^{a}\Psi_i\cdot\nabla_{a}\Psi_i + V_i \p^2\right) - \dfrac{1}{4}\square\omega \cdot\p^2 
			\end{align*}
			Last, we recall that $ ^{(X)}\pi^{\mu\nu} := (\mathcal{L}_{X}g)^{\mu\nu} = (\nabla^{\mu}X)^{\nu} + (\nabla^{\nu}X)^{\mu}  $ is the deformation tensor, thus we have \begin{align*}
				2\mathcal{Q}_{\mu\nu} \nabla^{\mu}X^{\nu} =  \mathcal{Q}_{\mu\nu}\cdot ^{(X)}\pi^{\mu\nu} = \mathcal{Q}\cdot\  ^{(X)}\pi.
			\end{align*}
			Combining all the above, we conclude the formula for $ K^{X,\omega,M}[\p]. $

		\end{proof}
	\end{proposition}
	\begin{remark}
		If $ \p $ satisfies (\ref{waveeq}) with a non-homogeneous term on the right-hand side, i.e. $$ \Big (\square_{\bm{g}} - V_i^{^{(\ell)}}\Big ) \Psi_i^{^{(\ell)}} = F $$
		then we have \begin{align*}
			Div\left(J_{\mu}^{X,\omega,M}[\p]\right) =  K^{X,\omega, M}[\Psi_i] + X(\p)\cdot F,
		\end{align*}
		which can be easily checked by the calculations above when computing $ \nabla^{\mu}Q_{\mu\nu}. $
	\end{remark}
	
	\paragraph{The Vector Field Method.}
	The  vector field method is simply the application of Stokes' theorem  in appropriate regions for the current scalar $ J_{\mu}^{X,\omega,M} $. In particular, given a (0,1) current $ P_{\mu} $, then Stokes' theorem  yields in the region $ \mathcal{R}(0,\tau) $\begin{equation}\label{Stokes'}
		\int_{\Sigma_0}P_{\mu}n_{\Sigma_0}^{\mu} = \int_{\Sigma_{\tau}} P_{\mu}n_{\Sigma_{\tau}}^{\mu} + \int_{\mathcal{H}^+(0,\tau)} P_{\mu}n_{\mathcal{H}^+}^{\mu} + \int_{\mathcal{R}(0,\tau) } \nabla^{\mu}(P_{\mu}),
	\end{equation}
	where all the integrals are with respect to the induced volume form and the unit normals $ n_{S} $ are future-directed.
	
	\subsection{Uniform Boundedness of Degenerate Energy} Let us apply the vectorfield method for $ X=T=\partial_v $, $ \omega=M=0, $ where $ T $ is written with respect to the coordinate system $ (v,r,\vartheta,\varphi) $. Since $ T $ is killing, we have $ {}^{(T)}\pi=0 $ and $ T(V_i^{^{(\ell)}})=0 $ because $ V_i$ is a function of $ r $ alone, thus $ K^T[\p] =0$. Therefore, the divergence theorem in the region $ \mathcal{R}(0,\tau) $ yields \begin{equation}
		\int_{\Sigma_0}J_{\mu}^T[\Psi_i]n_{\Sigma_0}^{\mu} = \int_{\Sigma_{\tau}} J_{\mu}^T[\Psi_i]n_{S}^{\mu} + \int_{\mathcal{H}^+(0,\tau)} J_{\mu}^T[\Psi_i]n_{\mathcal{H}^+}^{\mu} 
	\end{equation}
	On the event horizon $ \mathcal{H}^+ $ we can take $ n_{\mathcal{H}^+}^{\mu}=T$, thus \begin{align*}
		J_{\mu}^T[\Psi_i]n_{\mathcal{H}^+}^{\mu} = \mathcal{Q}(T,T) = (\partial_v\Psi_i)^2\geq 0,
	\end{align*}
	which proves the following proposition. 
	\begin{proposition}  \label{T estimate}
		For all solutions $ \Psi_i $ to equation (\ref{waveeq}) we have \begin{align}
			\int_{\Sigma_{\tau}} J_{\mu}^T[\Psi_i]n_{\Sigma_{\tau}}^{\mu}	\leq \int_{\Sigma_0}J_{\mu}^T[\Psi_i]n_{\Sigma_0}^{\mu}.
		\end{align} 
	\end{proposition}
	\paragraph{The T-flux.} We will see below that $ J_{\mu}^T[\Psi_i]n_{\Sigma_{\tau}}^{\mu} $ is non-negative definite only after we integrate on the spheres $S^2_{v,r} $ and use Poincare's Inequality  due to the negative values of the potentials $ V_i^{^{(\ell)}}$, $ i=1,2. $ However, we also need to know how  $ J_{\mu}^T[\Psi_i]n_{\Sigma_{\tau}}^{\mu} $  depends on the 1-jet of $ \Psi_i $. In particular, write $ n_{{\Sigma_{\tau}}}= n_{{\Sigma_{\tau}}}^v \partial_v + n_{\Sigma_{\tau}}^r\partial_r $ and note that the normal to the spacelike hypersurface $ \Sigma_0 $, $ n_{\Sigma_0} $, was chosen such that \begin{align} \label{C1}
		&	\dfrac{1}{C_1}<-g(n_{\Sigma_0},n_{\Sigma_0})< C_1 \\
		&	\dfrac{1}{C_1}<\ -g(n_{\Sigma_0},T)\ \  <C_1, \label{C2}
	\end{align} for a positive constant $ C_1 $ depending only on $ M, \Sigma_0 $. Thus, the same holds for $ n_{\Sigma_{\tau}} $ since $ \Sigma_{\tau} = \phi_{\tau}(\Sigma_0)  $, where $ \phi_{\tau} $ is the flow associated to the killing vector field T. First, we recall the following 
	\begin{proposition}\textbf{(Poincare Inequality)} \label{poincare} 
		Let $ \psi \in L^2 (S^2_{v,r}) $ and $ \psi_{\ell}=0 $ for all $ \ell \leq L-1 $ for some finite natural number $ L $, then  \begin{equation}
			\dfrac{L(L+1)}{r^2}\int_{S^2_{v,r}} \psi^2 \leq 	\int_{S^2_{v,r}} |\slashed{\nabla}\psi|^2 ,
		\end{equation}  
		and equality holds if and only if $ \psi_{\ell}=0 $ for all $ \ell \neq L. $
	\end{proposition}

	\begin{proposition}\label{posT} Let $ \Psi_i^{^{(\ell)}} $ be a solution to equation (\ref{waveeq}), supported on the fixed frequency $ \ell\geq i,\ i\in\br{1,2} $, then there exists a positive constant $ C=C(M, \Sigma_0) $ such that  \begin{align}
			\int_{\Sigma_{\tau}}  \left(  \left(\partial_{v} \Psi_i\right)^{2} + D \left(\partial_{r} \Psi_i\right)^{2} + |\slashed{\nabla} \Psi_i|^{2} + \dfrac{\ell(\ell+1)}{r^2}\Psi_i^{2}\right) \leq 	C\int_{\Sigma_{0}} J_{\mu}^{T}[\Psi_i] n_{\Sigma_{0}}^{\mu} . 
		\end{align}
		
	\end{proposition}
	\begin{proof}
		Let $ n_{{\Sigma_{\tau}}}= n_{{\Sigma_{\tau}}}^v \partial_v + n_{\Sigma_{\tau}}^r\partial_r $, then direct computations yield \begin{align*}
			J_{\mu}^{T}[\Psi_i] n_{\Sigma_{\tau}}^{\mu} = n^{v}(\partial_v\p)^2 + \dfrac{1}{2}\left(Dn^v-2n^r\right) \dfrac{D}{4}(\partial_r\p)^2 + \dfrac{1}{2}\left(Dn^v-2n^r\right) \left(\abs{\slas{\nabla}\p}^2+ V_i \p^2\right) 
		\end{align*}
		We argue that relations (\ref{C1}, \ref{C2}) and Poincare inequality suffice to prove the proposition. 
		
		Away from the horizon $ \br{r\geq r_0},  $ for $ r_0  >  M $, we might as well choose $ n_{\Sigma_{0}} \equiv T $. However, near the horizon $ \h $, the relations (\ref{C1}, \ref{C2}) read \begin{align*}
			C_1 >& D (n^v)^2 - 2n^v n^r > \dfrac{1}{C_1} \\
			C_1 >& D n^v - n^r > \dfrac{1}{C_1}
		\end{align*}
		thus we must have $ n_{\Sigma_0}^r <0$ and $n_{\Sigma_0}^v >0 $. On the other hand, squaring the second relation  yields \begin{align*}
			&C_1^2 > D( D (n^v)^2 -2n^vn^r) +(n^r)^2 > \dfrac{1}{C_1^2}
		\end{align*}
		and since the first term is positive from (\ref{C1}), we obtain that  $ n^r $ is uniformly bounded. Hence, $ Dn^v-2n^r = (Dn^v-n^r) -n^r >\frac{1}{C_1} $ and bounded from above. Thus, using (\ref{C1}) we have \begin{align*}
			\dfrac{C_1}{\left(Dn^v-2n^r\right)}>n^v > \dfrac{1}{C_1\left(Dn^v-2n^r\right)} 
		\end{align*}
		and as we saw $ \left(Dn^v-2n^r\right) $ is bounded from above, which makes  $ n^v $ uniformly bounded from below by a positive constant. We put all the above together and we find a constant $ C $ depending on $ M, \Sigma_0 $ such that 
		\begin{align}
			(\partial_v\p)^2 + D(\partial_r\p)^2 +  \abs{\slas{\nabla}\p}^2+ V_i \p^2\leq C	J_{\mu}^{T}[\Psi_i] n_{\Sigma_{\tau}}^{\mu}. \label{T-flux dpsi}
		\end{align}
		However, the potential $ V_i^{^{(\ell)}}(r) $ takes negative values as well. 
		To show that the $ T- $flux is also coercive with respect to the zeroth order term we borrow from the angular derivative of $ \p $ using Poincare inequality. 
		Integrating (\ref{T-flux dpsi})  along $ \Sigma_{\tau} $ and using the uniform boundedness of degenerate energy yields \begin{align*}
			\int_{\Sigma_{\tau}}(\partial_v\p)^2 + D(\partial_r\p)^2 +  \abs{\slas{\nabla}\p}^2+ V_i \p^2 \leq C \int_{\Sigma_{0}}	J_{\mu}^{T}[\Psi_i] n_{\Sigma_{0}}^{\mu}.
		\end{align*}
		Now, we write $ \abs{\slas{\nabla}\p}^2 = (1-a)\abs{\slas{\nabla}\p}^2 + a\abs{\slas{\nabla}\p}^2  $, for some $ 0<a<1 $ to be determined later, and using  Poincare we examine the following term \begin{align}
			a \abs{\slas{\nabla}\p}^2+ V_i \p^2 = \left(a + \dfrac{V_i\cdot r^2}{\ell(\ell+1)}\right) \abs{\slas{\nabla}\p}^2. \label{finding a}
		\end{align}
		It suffices to study the above coefficient for each $ \ell\geq i, \ i\in\br{1,2} $ and for each $ r\geq M. $ \begin{itemize}
			\item If $ i=1, $ and for any $ \ell\geq 1 $ we have \begin{align*}
				a + \dfrac{V_1(r)\cdot r^2}{\ell(\ell+1)} = a - \dfrac{2M}{r}\dfrac{3\sqrt{D}}{\ell(
					\ell+1)}+\dfrac{2M}{\ell\cdot r}\geq a + \dfrac{M}{\ell\cdot r}\left(2-3\sqrt{D}\right) 
			\end{align*}
			The second term becomes negative only when $ r\geq 3M $, and it's easy to check that $ \frac{M}{r}(2-3\sqrt{D}) \geq -\frac{1}{3} $. Thus, it suffices to consider $ \frac{1}{3}<a<1, $ which works for all $ \ell \geq 1. $
			\item If $ i=2, $ for any $ \ell\geq 2  $ we have \begin{align*}
				a + \dfrac{V_2(r)\cdot r^2}{\ell(\ell+1)} = a- \dfrac{2M}{r}\dfrac{3\sqrt{D}}{\ell(\ell+1)} - \dfrac{2M}{(\ell+1)\cdot r} \geq a -\dfrac{1}{3} \dfrac{M}{r}\left(2+3\sqrt{D}\right) 
			\end{align*}
			However, if we set $ x:= \sqrt{D} \in [0,1) $ we write the above as \[ a + \left (x^2-\dfrac{1}{3}x -\frac{2}{3}\right ) \] 
			The quadratic polynomial above attains its minimum at $ x=\frac{1}{6} $ with value $ -\frac{25}{36} \sim 0.69 $. Thus, it suffices to consider $ 0.7 < a <1. $
		\end{itemize}
		We can see from the analysis above that $ a=0.7 $ is sufficient to get (\ref{finding a}) uniformly positive for all $ \ell\geq i, \ i\in\br{1,2} $ and $ r\geq M. $
	\end{proof}

	\subsection{Morawetz Estimates} We apply the vector field method for vector fields of the form  $ X= f(r^{\star})\frac{\partial}{\partial r^{\star}}, \  $ written with respect to the coordinate system $ (t,r^{\star}) $. We will often differentiate with respect to the $ r-$coordinate instead of $ r^{\star} $, and the two are related by $ \frac{\partial h}{\partial {r^{\star}}}=D\cdot \frac{\partial h}{\partial {r}} ,$ for any scalar function $ h. $ By $ h'= \frac{\partial h}{\partial r} $ we will denote the derivative with respect to the $ r-$coordinate in the $ (t,r,\vartheta,\varphi) $ coordinate system.

	\paragraph{The scalar current $ K^X $.} Let us consider the following current 
	\begin{align*}
		J^{X}_{\mu} [\p] = \mathcal{Q}_{\mu\nu}[\p]\cdot X^{\nu},
	\end{align*}
	then using Proposition \ref{GenCur} we arrive at \begin{align*}
		K^X  = \left(\dfrac{f'}{2}+ \dfrac{f}{r}\right) (\partial_t \p)^{2} + \left(\dfrac{f'}{2}-\dfrac{f}{r}\right) (\partial_{r^{\star}}\p)^{2} &+ \left(-\dfrac{(D\cdot f)'}{2}\right) \abs{\slas{\nabla}\p}^{2} \\ &+ \left(-\dfrac{V_i}{2}\left( (D\cdot f)' + \dfrac{2Df}{r}\right ) - \dfrac{1}{2} X(V_i) \right) \p^2
	\end{align*}
	However, there is no choice of scalar $ f $ such that all coefficients above are positive definite everywhere. Indeed, assume the coefficients of the first two terms of $ K^{X}  $ are positive, then  \begin{align*}
		f' > 2 \dfrac{f}{r}, \hspace{1cm} \text{and} \hspace{1cm} f' > -2 \dfrac{f}{r} \hspace{1cm}
		\Rightarrow \hspace{1cm} f' > \dfrac{2}{r} \abs{f} > 0
	\end{align*}
	On the other hand, if the coefficient of $ \abs{\slas{\nabla}\p} ^{2} $ is also non-negative  we must have \begin{align*}
		f \leq - f' \dfrac{D}{D'}
	\end{align*}
	Thus, going back to the coefficient of $ (\partial_t\p)^2 $ we obtain \begin{align*}
		\dfrac{f'}{2} + \dfrac{f}{r} \leq f' \left(\dfrac{1}{2}-\dfrac{D}{rD'}\right) = \dfrac{f'}{2} \left(1-2\sqrt{D}\right).
	\end{align*}
	Hence, in view of $ f'> 0 $ we have that $ \frac{f'}{2} + \frac{f}{r} $ becomes negative for $ r\geq r_{p}=2M. $
	
	Already, the above suggests that we modify the energy current by introducing more terms. The idea is to introduce a term with the effect of canceling $ (\partial_t\p)^2 $ when computing the scalar current. Nevertheless, we retrieve this term in the final Morawetz estimates of Proposition \ref{Morawetz final}.
	\paragraph{The scalar current $ K^{X,G}. $}
	Consider the following current \[ J_{\mu}^{X,G}[\Psi_i]:= Q_{\mu \nu }[\Psi_i]\cdot X^{\nu} +\dfrac{1}{2} G \Psi_i \nabla_{\mu}\Psi_i - \dfrac{1}{4}\left (\nabla_{\mu}G\right ) \Psi_i^2, \]
	where $ G= r^{-2}\partial_{r^{*}}\left (f\cdot r^2\right ) $. Consequently, this choice of $ G $ cancels the time derivative term when computing the scalar current. In particular, using Proposition \ref{GenCur} direct computations yield the expression \begin{align}
		\begin{aligned}
			K^{X,G}= &\  f'\left(\partial_{r^{*}} \Psi_i\right)^{2}+\frac{f \cdot P}{r}|\slashed{\nabla} \Psi_i|^{2}-\dfrac{1}{4}\left(\square_{g} G\right) \Psi_i^{2}  - \dfrac{1}{2}f\cdot \partial_r\Big (D\cdot V_i^{^{(\ell)}}\Big ) \Psi_i^2 , \label{K^X}
		\end{aligned}
	\end{align}
	where $ P(r):= \sqrt{D}(2\sqrt{D}-1) = \frac{1}{r^2} (r-M)(r-2M) $. While, in view of the factor $ P(r) $, the degeneracy of the angular derivative term at the photon sphere $\br{ r=r_p:= 2M }$ is inevitable, we shall find functions $ f, G $ such that the coefficients of the two first terms in $ K^{X,G} $ are non-negative definite. In particular, it is  imperative that we choose $ f $ that is increasing and changes sign from negative to positive at the photon sphere $ \br{r=r_p} $.  
	
	\paragraph{The choice of G, f.}
	In the extremal Reissner--Nordstr\"om spacetime, the wave operator with respect to the coordinate system $ (t,r^{\star},\vartheta,\varphi) $ reads \begin{equation}
		\square_g \Psi_i= \dfrac{1}{D}\Big (-\partial^2_t \Psi_i +r^{-2}\partial_{r^{*}}(r^2\partial_{r^{*}}\Psi_i)\Big ) + \slashed{\Delta}\Psi_i.
	\end{equation}
	Since $ G= r^{-2}\partial_{r^{*}}\left (f\cdot r^2\right ) $ is a function of $ r^{\star} $ alone, we have \begin{equation}
		\square_g G= \dfrac{1}{D\cdot r^2}\partial_{r^{*}}\Big (r^2 \partial_{r^{*}}G\Big )=\dfrac{1}{D\cdot r^2}\partial_{r^{*}}\Big (r^2 \partial_{r^{*}}\big(r^{-2}\partial_{r^{*}}(f\cdot r^2)\big)\Big ). \label{Box G}
	\end{equation} 
	Let us choose \begin{align}
		G(r)= \dfrac{2}{r}D(r), \ \ \ r\geq M.
	\end{align}
	Then, direct computation yields \begin{align*}
		r^2 \partial_{r^{*}}G=& 2D \sqrt{D}(2-3\sqrt{D}), \\
		\partial_{r^{*}}\left (r^2 \partial_{r^{*}}G\right )=& \dfrac{12 D^2}{r}(1-\sqrt{D})(1-2\sqrt{D}).
	\end{align*}
	Therefore, we have \begin{align}
		-\dfrac{1}{4}(\square_g G)=3 \dfrac{D}{r^3}(1-\sqrt{D})(2\sqrt{D}-1). \label{boxG}
	\end{align}
	In addition, now that we have $ G $ we can find $ f $  using the transport equation $  G= r^{-2}\partial_{r^{*}}\left (f\cdot r^2\right ) $ or equivalently $ \partial_{r}(f\cdot r^2)=\frac{r^2G}{D} $, thus integrating from $ r_p $ to $ r $ we obtain \begin{align}
		f(r)=& (2\sqrt{D}-1)(3-2\sqrt{D}) \label{fform}, \\
		f'(r)=& \dfrac{8}{r}(1-\sqrt{D})^2. \label{f'}
	\end{align}
	As we can see, $ f $ is increasing and changes sign at the horizon and thus satisfies the requirements we were looking for. However,
	for both $ i=1,2 $, the zeroth order coefficient in  $ K^{X,G} $ takes negative values as well, for all $ \ell \geq i $. 
	\paragraph{The zeroth order term of $K^{X,G}$.}
	Let's first study the expression $  - \dfrac{1}{2}f\cdot \partial_r\Big (D\cdot V_i^{^{(\ell)}}\Big ). $ Denote by $ x:=\sqrt{D}\in [0,1) $ then using (\ref{fform}), direct computations yield \begin{align}
		\begin{aligned} \label{z_i^l}
			z_1^{(\ell)}:=- \dfrac{1}{2}f\cdot \partial_r\Big (D\cdot V_1^{(\ell)}\Big )=&\dfrac{1}{r^3}(1-x)x(2x-1)(3-2x)\Big(-18x^2+14x-2+\ell(5x-2)\Big) \\ 
			z_2^{(\ell)}:=- \dfrac{1}{2}f\cdot \partial_r\Big (D\cdot V_2^{(\ell)}\Big )=&\dfrac{1}{r^3}(1-x)x(2x-1)(3-2x)\Big(-18x^2+9x-\ell(5x-2) \Big)  
		\end{aligned}
	\end{align}
	In this form, it is apparent that both $ z_i^{(\ell)} $ for all $ \ell \geq i, \ i\br{1,2} ,$ are not positive definite.  In addition, the term $ -\frac{1}{4} \square G = 3 \frac{x^2}{r^3}(1-x)(2x-1)$ comes to add extra "negativity" to the overall zeroth order term. 
	
	On the other hand,  $ f(r)\cdot P(r) $ degenerates at the photon sphere to second order, as opposed to the two aforementioned terms, thus even after using Poincare inequality to borrow from the angular derivative coefficient, we cannot obtain a positive definite zeroth order term of the current $ K^{X,G} .$  
	Nevertheless, in what follows we show that there is a modified energy current that ultimately provides us with the required positive bulk for all terms.

	\subsubsection{The scalar current \texorpdfstring{$ K^{X,G,h}. $}{PDFstring}}
	In view of the above discussion, we 
	consider the following modified current \begin{equation}
		J_{\mu}^{X,G,h}[\Psi_i] := J_{\mu}^{X,G}[\Psi_i] + \dfrac{h}{2}\Psi_i^2 \left (\dfrac{\partial}{\partial_{r^{*}}}\right )_{\mu}
	\end{equation}
	Then, applying Proposition \ref{GenCur}, the corresponding scalar current is given by \begin{align}
		\begin{aligned}
			K^{X,G,h} = &\  K^{X,G} + \left (\dfrac{h\sqrt{D}}{r}+ \dfrac{D}{2}h'\right )(\Psi_i)^2 + h \Psi_i\partial_{r^{*}}\Psi_i.  \\
			\Rightarrow K^{X,G,h}=&\ f'(\partial_{r^{*}}\Psi_i)^2 +\frac{f \cdot P}{r}|\slashed{\nabla} \Psi_i|^{2} -\dfrac{1}{4}\left(\square_{g} G\right) \Psi_i^{2} \\
			&-\dfrac{f}{2}(D\cdot V_1)' (\Psi_i)^2 + h \Psi_i\partial_{r^{*}}\Psi_i + \left (\dfrac{h\sqrt{D}}{r}+ \dfrac{D}{2}h'\right )(\Psi_i)^2. \label{K^{X,h}}
		\end{aligned}
	\end{align}
	Our goal is to find an appropriate function $ h(r) $ such that $ K^{X,G,h} $ is positive definite. However, we first need to treat the extra term $  h \Psi_i\partial_{r^{*}}\Psi_i $ such that only quadratic terms appear. We borrow from the coefficient of $ (\partial_{r^{*}} \Psi_i)^2 $  by writing \begin{align*}
		f'\abs{\partial_{r^{*}}\Psi_i}^2= \nu f'\abs{\partial_{r^{*}}\Psi_i}^2 + (1-\nu) f'\abs{\partial_{r^{*}}\Psi_i}^2,	
	\end{align*} for some $ \nu\in (0,1) $ to be determined in the end, and we complete the square as 
	\begin{align}
		\begin{aligned}
			& \nu f'(\partial_{r^{*}}\Psi_i)^2 +  h \Psi_i\partial_{r^{*}}\Psi_i =  \\
			& = \nu f'(\partial_{r^{*}}\Psi_i)^2 + 2 \dfrac{ h}{2\sqrt{\nu f'}}\sqrt{\nu f'} \Psi_i\partial_{r^{*}}\Psi_i  + \dfrac{h^2}{4\nu f'} \Psi_i^2- \dfrac{h^2}{4\nu f'} \Psi_i^2= \\
			& = \left (\sqrt{\nu f'}\partial_{r^{*}}\Psi_i + \dfrac{h}{2\sqrt{\nu f'}}\Psi_i\right )^2 - \dfrac{h^2}{4\nu f'} \Psi_i^2. \label{square}
		\end{aligned}
	\end{align}
	Therefore, using relation $ (\ref{square}) $, we can now rewrite  (\ref{K^{X,h}}) as \begin{align}
		\begin{aligned}
			K^{X,G,h}[\p] = & (1-\nu)f'(\partial_{r^{*}}\Psi_i)^2 + \left (\sqrt{\nu f'}\partial_{r^{*}}\Psi_i + \dfrac{h}{2\sqrt{\nu f'}}\Psi_i\right )^2 + \frac{f \cdot P}{r}|\slashed{\nabla} \Psi_i|^{2}  \\
			&
			+ \left (\left (\dfrac{h\sqrt{D}}{r}+ \dfrac{D}{2}h'- \dfrac{h^2}{4\nu f'}  \right )  -\dfrac{1}{4}\left(\square_{g} G\right) -\dfrac{f}{2}(D\cdot V_1)'\right ) \Psi_i^2.
		\end{aligned} \label{fK^{X,h}}
	\end{align}
	Denote by $ T(h) $ the expression below \begin{equation}
		T(h):= \dfrac{h\sqrt{D}}{r}+ \dfrac{D}{2}h'- \dfrac{h^2}{4\nu f'} -\dfrac{1}{4}\left(\square_{g} G\right)= \dfrac{1}{2r^2}\partial_{r}(Dr^2h)-\dfrac{h^2}{4\nu f'}-\dfrac{1}{4}\left(\square_{g} G\right). 
	\end{equation}
	We simply focus on finding a function $ h(r) $ such that $ T(h) $ beats the negative values of $ z_i^{(\ell)} $, $ \forall \ell \geq i, \ i\br{1,2}$, uniformly.
	
	\paragraph{The coefficient \texorpdfstring{$ T(h) $}{PDFstring}.} In order to better understand the expression $ T(h) $ we rewrite it in a more concise way using relation (\ref{Box G}), i.e. \begin{align*}
		\square_g G= \dfrac{1}{D\cdot r^2}\partial_{r^{*}}\Big (r^2 \partial_{r^{*}}G\Big ) = \dfrac{1}{r^2}\partial_r\left(Dr^2\partial_r G\right).
	\end{align*}
	Then we have \begin{align*}
		T(h) = \dfrac{1}{4r^2} \partial_r\Big(D(2r^2h -r^2\partial_rG)\Big)-\dfrac{h^2}{4\nu f'}.
	\end{align*}
	However, $ \partial_rG = 2r^{-2} \sqrt{D}(2-3\sqrt{D})$, and if we denote by $ x:=\sqrt{D} \in [0,1) $ and express $ T(h) = T(h(x(r))) $ we obtain \begin{align*}
		T(h) = \dfrac{(1-x)}{2r^3} \partial_x\Big(x^2\left(r^2\cdot h-x(2-3x)\right) \Big) -\dfrac{h^2}{4\nu f'}
	\end{align*}
	Now, consider the following choice \begin{align} \label{choice of h}
		h(r) : = \dfrac{1}{r^2} \left (x(2-3x) + x\right ) = \dfrac{1}{r^2} 3x (1-x),
	\end{align}
	and using the fact that $ f'(r) = \frac{8(1-x)^2}{r} $, then $ T(h) $ becomes \begin{align*}
		T(h) =  \dfrac{(1-x)}{2r^3} \partial_x\Big(x^3 \Big) - \dfrac{1}{r^4}\dfrac{9x^2(1-x)^2}{4\nu f'} = \dfrac{3(1-x)x^2}{2r^2} - \dfrac{1}{r^3} \dfrac{9x^2}{32\cdot \nu}  \label{T(h) formula}
	\end{align*}
	Let $ \nu = \frac{15}{16} < 1 $, then $ T(h) $ becomes \begin{align}
		T(h) = \dfrac{1}{r^3}\dfrac{3}{10} x^2(4-5x) \ = \ \dfrac{0.3}{r^3}\cdot D  (4-5\sqrt{D}).
	\end{align}	
	We can already see that $ T(h) $ is positive definite around the photon sphere and becomes negative only for $ r> 5M. $ Nevertheless, with the help of Poincare inequality we will show that the zeroth order term of $ K^{X,G,h}[\Psi_i^{^{(\ell)}}] $ is positive definite uniformly in $ \ell \geq i, \ i\in \br{1,2} $, with a degeneracy only at the horizon $ \h $.

	\begin{proposition}
		\label{zeroth order positivity} \label{Positivity of K^Xh}
		Let $ h(r) $ be as in (\ref{choice of h}) and choose $ \nu=\frac{15}{16} $ in (\ref{fK^{X,h}}), then there exists a positive constant depending only on $ M$, such that for all solutions $\Psi_i^{^{(\ell)}} $ to (\ref{waveeq}), supported on the fixed frequency  $ \ell\geq i, \ i\in\br{1,2} $, we have  \begin{equation}
			\int_{S^2(r)} K^{X,G,h}[\Psi_i]\geq C \int_{S^2(r)} \left(\dfrac{1}{r^3}(\partial_{r^{\star}}\Psi_i)^2+\dfrac{(r-M)(r-2M)^2}{r^4}\abs{\slashed{\nabla}\Psi_i}^2+\dfrac{\sqrt{D}}{r^3}\p^2\right). \label{morawetz}
		\end{equation}
		\begin{proof}
			After we integrate relation (\ref{fK^{X,h}}) on the spheres and apply Poincare inequality we obtain \begin{align*}
				\begin{aligned}
					\int_{S^2(r)}	K^{X,G,h}[\p] = & \int_{S^2(r)}  (1-\nu)f'(\partial_{r^{*}}\Psi_i)^2 + \left (\sqrt{\nu f'}\partial_{r^{*}}\Psi_i + \dfrac{h}{2\sqrt{\nu f'}}\Psi_i\right )^2 + \frac{f \cdot P}{r}|\slashed{\nabla} \Psi_i|^{2}  \\
					& -\dfrac{1}{4}\left(\square_{g} G\right) \Psi_i^{2}  
					+ \left (\left (\dfrac{h\sqrt{D}}{r}+ \dfrac{D}{2}h'- \dfrac{h^2}{4\nu f'}  \right ) -\dfrac{f}{2}(D\cdot V_i)'\right ) \Psi_i^2 \\
					& \geq   \int_{S^2(r)}  (1-\nu)f'(\partial_{r^{*}}\Psi_i)^2 + \frac{f \cdot P}{r}|\slashed{\nabla} \Psi_i|^{2} +\Big(T(h)+z_i^{(\ell)} \Big)\p^2.
				\end{aligned} 
			\end{align*}
			Recall that $ 	f(r)\cdot P(r) = (2\sqrt{D}-1)(3-2\sqrt{D})\cdot (2\sqrt{D}-1)\sqrt{D} $,  thus \begin{align*}
				\dfrac{f\cdot P}{r} &= \dfrac{(3-2\sqrt{D})(r-M)(r-2M)^2}{r^4} \\
				&\geq \dfrac{(r-M)(r-2M)^2}{r^4}.
			\end{align*}
			In addition, the coefficient of $ (\partial_{r^{\star}}\p)^2 $ becomes \begin{align*}
				(1-\nu)f'(r)= \dfrac{8 }{16} \cdot \dfrac{M^2}{r^3} = \dfrac{1}{2}\cdot \dfrac{M^2}{r^3}.
			\end{align*} 
			Finally, we need to treat the coefficient of the zeroth order term as well. For that, consider each case $ i=1,2 $ separately because they pose different difficulties. We remind our readers that for the calculations below we express all functions in terms of $ x := \sqrt{D} $, and we produce estimates for all $ x\in [0,1), $ which corresponds to $ r\geq M. $
			\paragraph{The zeroth order term of \texorpdfstring{$ K^{X,G,h}[\Psi_1^{^{(\ell)}}],\   \ell \geq 1.$}{PDFstring}} 
			According to relation (\ref{T(h) formula}), the zeroth order coefficient reads \begin{align*}
				T(h) + z_1^{(\ell)} = \dfrac{0.3}{r^3} x^2(4-5x) + z_1^{(\ell)},
			\end{align*}
			where \begin{align*}
				z_1^{(\ell)} = \dfrac{1}{r^3} x(1-x) f \left(-18x^2+14x-2+\ell(5x-2)\right).
			\end{align*}
			However, we may rewrite the quadratic polynomial as \begin{align*}
				-18x^2+14x-2+\ell(5x-2) &= -18x^2+9x+5x-2+ \ell(5x-2) \\ &= -9x(2x-1) + (l+1)(4x-2+x)\\ & = (2(\ell+1)-9x)(2x-1) + (\ell+1)x.		\end{align*}
			Hence, $ z_1^{(\ell)} $ reads \begin{align}
				z_1^{(\ell)} = \dfrac{1}{r^3} (1-x) f\cdot P \Big( 2(\ell+1) -9x \Big) + \dfrac{(\ell+1)}{r^3} x^2(1-x)f \label{z_1-good}
			\end{align}
			We treat each term separately below. 
			\begin{itemize}
				\item 	The first term in (\ref{z_1-good}) can be controlled using Poincare inequality by borrowing a fraction $ a \in (0,1) $ from the angular derivative coefficient. In particular, we have \begin{align*}
					\dfrac{1}{r^3} (1-x) f\cdot P \Big( 2(\ell+1) -9x \Big)& \left(\Psi_1^{^{(\ell)}}\right)^2 +  \dfrac{a}{r}f\cdot P \abs{\slas{\nabla}\Psi_1^{(\ell)}}^2 \\ & \geq \  \dfrac{1}{r^3}f\cdot P \Big((1-x)(2(\ell+1)-9x) + a \ell(\ell+1)\Big) \left(\Psi_1^{^{(\ell)}}\right)^2 \\
					& = \  \dfrac{1}{r^3}f\cdot P \Big(9x^2-(2\ell+11)x + (\ell+1)(2+a\ell)\Big)\left(\Psi_1^{^{(\ell)}}\right)^2. 
				\end{align*}	
				Since $ f\cdot P $ has the correct sign, we are only interested in choosing a sufficient $ a \in (0,1) $ such that the discriminant of the quadratic expression is negative uniformly in $ \ell\geq 1 $, i.e. \begin{align*}
					a > \dfrac{(2\ell-7)^2}{36\ell(\ell+1)}, \hspace{1cm} \forall \ \ell \geq 1.
				\end{align*}
				However, in view of the denominator growing quadratically in $ \ell $ with a higher rate than the numerator, we have \begin{align*}
					\dfrac{(2\ell-7)^2}{36\ell(\ell+1)} < \dfrac{(2\cdot 1-7)^2}{36(1+1)} < 0.35, \hspace{1cm} \forall \ell\geq 1, 
				\end{align*} 
				thus $ a=0.35 $ is sufficient to make the first term of (\ref{z_1-good}) non-negative definite for all $ \ell \geq 1. $
				\item We have a remaining fraction $ 0<b<0.65 $ of the angular derivative coefficient  that we can use to control the second term of (\ref{z_1-good}). Once again, using Poincare inequality we write \begin{align*}
					\dfrac{(\ell+1)}{r^3}x^2(1-x)f + &\dfrac{b}{r^3}\ell(\ell+1) f\cdot P \geq  	\dfrac{(\ell+1)}{r^3}x^2(1-x)f +  \dfrac{b}{r^3}\ell(\ell+1)(1-x) f\cdot P \\
					&=	\dfrac{(\ell+1)}{r^3}x^2(1-x)(2x-1)(3-2x) + \dfrac{b}{r^3}\ell(\ell+1) x(2x-1)^2(3-2x) (1-x) \\
					& = \dfrac{(\ell+1)}{r^3} x(1-x)(3-2x) (2x-1)\Big((1+2b\ell)x-b\ell\Big)
				\end{align*}
				Clearly for $ x\geq \frac{1}{2} $ the expression above is non-negative. Let us focus on $ x<\frac{1}{2} $ and notice that the only interval it takes negative values is $ N_{\ell}:= \left(x_{\ell}, \frac{1}{2}\right) $, where $ x_{\ell}:=\frac{b\ell}{1+2b\ell},\  \ell \geq 1, $ and $ x_{\ell} \xrightarrow{\ell\rightarrow \infty} \frac{1}{2} $. By studying the quadratic polynomial \begin{align*}
					p_{\ell}(x) := (2x-1)\Big((1+2b\ell)x-b\ell\Big) = (4b\ell+2)x^2 -(4b\ell+1)x+b\ell 
				\end{align*}
				we see that it attains its minimum at $ \underline{x}_{\ell} = \frac{(4b\ell+1)}{2(4b\ell+2)}  $ with the value \begin{align*}
					p_{\ell}(x) \geq p_{\ell}(\underline{x}_{\ell}) = - \dfrac{1}{4(4b\ell+2)}, \hspace{1cm} \forall x \in [0,1).
				\end{align*}
				$ \diamond $ For all $ \ell \geq 2 $ we may choose $ b=\frac{1}{2}\in (0,0.65) $ and thus we have the lower bound \begin{align*}
					\dfrac{(\ell+1)}{r^3} x(1-x)(3-2x) p_{\ell}(x) \geq -\dfrac{1}{8r^3}x(1-x)(3-2x), \hspace{1cm} \ \forall x\in N_{\ell},\ \forall \ \ell\geq 2
				\end{align*}
				On the other hand, note that $ N_{\ell=2} \supset N_{\ell},\ \forall \ell\geq 2$ and thus, for all $ x \in N_{\ell=2}=\left (\frac{1}{3},\frac{1}{2}\right ) $ we have \begin{align*}
					T(h)  -\dfrac{1}{8r^3}x(1-x)(3-2x) & = \dfrac{1}{r^3} x \cdot \left(\dfrac{3}{10}x(4-5x) - \dfrac{1}{8}(1-x)(3-2x) \right) \\
					& =\dfrac{1}{r^3}\frac{x}{4}\left(-7x^2 +7.3x-1.5\right),
				\end{align*}
				and its easy to check that the quadratic polynomial is uniformly positive in $ N_{\ell=2}. $ \\
				$ \diamond $ For $ \ell = 1$, we need a little more help from Poincare inequality so we choose $ b = \frac{5}{8} = 0.625 <0.65 $ and  now we have \begin{align*}
					\dfrac{2}{r^3} x(1-x)(3-2x) p_{_{\ell=1}}(x) \geq -\dfrac{1}{9r^3}x(1-x)(3-2x), \hspace{1cm} \ \forall x\in N_{_{\ell=1}}=\left (\frac{5}{18},\frac{1}{2}\right ).
				\end{align*}
				Once again, one can check that the expression below is uniformly positive for  $ x\in N_{_{\ell=1}} $, \[ 	T(h)  -\dfrac{1}{9r^3}x(1-x)(3-2x) \] 
			\end{itemize}
			Finally, regarding the region that $ T(h) $ becomes negative, i.e. $ r\geq 5M $, note that the Poincare inequality used earlier for $ b = \frac{1}{2} $ is sufficient to make it positive definite for all $ \ell \geq 1. $
			
			Combining all the above allows us to find a positive constant $ C $ that depends only on $ M $ such that \[ T(h)+ z_1^{(\ell)} \geq C \dfrac{\sqrt{D}}{r^3}. \]
			
			\paragraph{The zeroth order term of \texorpdfstring{$ K^{X,G,h}[\Psi_2^{^{(\ell)}}],\   \ell \geq 2.$}{PDFstring}} We approach this case similarly and now have \begin{align*}
				z_2^{(\ell)} & = \dfrac{1}{r^3}x(1-x)f \Big(-18x^2+9x-\ell(5x-2) \Big) \\
				& = \dfrac{1}{r^3}x(1-x)f \Big( -(2\ell+9x)(2x-1) - \ell x\Big) \\
				& = -\dfrac{(2\ell+9x)}{r^3}(1-x)f\cdot P -  \dfrac{\ell}{r^3}x^2(1-x)f
			\end{align*}
			\begin{itemize}
				\item Using Poincare inequality, we control the first term as \begin{align*}
					-\dfrac{(2\ell+9x)}{r^3}(1-x)f\cdot P + \dfrac{a}{r^3}\ell(\ell+1)f\cdot P &=   \dfrac{1}{r^3}f\cdot P \Big( -(1-x)(2\ell+9x) + a\ell(\ell+1) \Big) \\ &=
					\dfrac{1}{r^3}f\cdot P \Big( 9x^2+(2\ell-9)x + \ell(a(\ell+1)-2) \Big)
				\end{align*}
				and by checking the discriminant of the quadratic polynomial, it suffices to have $ a\in (0,1) $ such that \begin{align*}
					&(2\ell-9)^2-36\ell(a(\ell+1)-2)<0, \hspace{1cm} \forall \ \ell\geq 2 \\
					\Leftrightarrow \ \ \ 	& a>\dfrac{(2\ell+9)^2}{36\ell(\ell+1)},  \hspace{1cm} \forall \ \ell\geq 2
				\end{align*}
				However, the right-hand side term satisfies for all $ \ell \geq 2  $ \[  \dfrac{(2\ell+9)^2}{36\ell(\ell+1)} \leq \dfrac{(4+9)^2}{36\cdot 6} < 0.79,\]
				thus $ a=0.79 $ is sufficient to make the first term of $ z_2^{(\ell)} $ non-negative for all $ \ell \geq 2. $
				\item There is a remaining fraction $ 0<b<0.21 $ we can use to bound the second term of $ z_2^{(\ell)} $ and with the help of $ T(h) $ we can show uniform positivity for the zeroth order term. In particular, we have \begin{align*}
					-  \dfrac{\ell}{r^3}x^2(1-x)f\ +\ & b\ell(\ell+1) f\cdot P \geq -  \dfrac{\ell}{r^3}x^2(1-x)f +b\ell(\ell+1)x f\cdot P \\ & = \dfrac{\ell}{r^3}x^2(3x-2) (2x-1)\Big(b(\ell+1)(2x-1)-(1-x)\Big) \\
					& := \dfrac{\ell}{r^3}x^2(3x-2)\cdot  p_{\ell}(x)
				\end{align*}
				Clearly, for $ x\leq\frac{1}{2} $ the expression above is non-negative, and for $  x>\frac{1}{2} $ it is negative only in the interval $ N_{\ell}:=\left (\frac{1}{2},x_{\ell}\right ) $, where $ x_{\ell}:= \frac{1+b(\ell+1)}{2b(\ell+1)+1}>\frac{1}{2} $ and we have $ x_{\ell} \xrightarrow{\ell \to \infty} \frac{1}{2}. $ 
				It's straight forward calculations to check that $ p_{\ell}(x) $ attains its minimum at $ \underline{x}_{\ell} =  \frac{4b(\ell+1)+3}{2(4b(\ell+1)+2)}$ with the value \begin{align*}
					p_{\ell}(x)\geq p_{\ell}(\underline{x}_{\ell}) = - \dfrac{1}{8\Big(2b(\ell+1)+1 \Big)}, \hspace{1cm} \forall \ell \geq 2, \ \forall x \in N_{\ell}. 
				\end{align*}
				We choose $ b=0.2<0.21 $ and thus we have for all $ \ell \geq 2,\  x\in N_{\ell}, $ \begin{align*}
					\ell\cdot P_{\ell}(x) \geq - \dfrac{1}{8} \cdot \dfrac{10}{4\left ( 1+ \frac{7}{2\ell}\right )} &\geq - \frac{5}{16} \\
					\Rightarrow \ \ \ 	 \dfrac{\ell}{r^3}x^2(3x-2)\cdot  p_{\ell}(x) &\geq - \dfrac{5}{16r^3}x^2(3x-2).
				\end{align*}
				Hence, using $ T(h) $ we now have \begin{align*}
					\frac{3}{10r^3}x^2(4-5x)  - \dfrac{5}{16r^3}x^2(3x-2) = \dfrac{x^2}{80r^3}\left(146-195x\right) 
				\end{align*}
				which is uniformly positive for $ \frac{1}{2}< x  < x_{\ell=2} = \frac{8}{11}   $ for all $ \ell \geq 2. $
			\end{itemize}
			Once again, in the region where $ T(h) $ is negative, the Poincare inequality used earlier is sufficient to make it positive definite, and thus there exists a positive constant depending on $ M $ such that \[ T(h) + z_2^{(\ell)} > C \frac{\sqrt{D}}{r^3}, \hspace{1cm} \forall \ell\geq 2,\ \forall\ r\geq M. \]

		\end{proof}
	\end{proposition}
	
	\subsubsection{Retrieving the \texorpdfstring{$ (\partial_t\p)^2 $}{PDFstring} term.}
	Note that estimate (\ref{morawetz}) does not include any $ \partial_{t}\Psi_i $ term. To retrieve the $ \partial_{t}-$derivative  we introduce the current $ L_{\mu}^z = z(r) \Psi \nabla_{\mu}\Psi_i $ for an appropriate function $ z(r). $ In particular, we have the following proposition. 
	
	\begin{proposition}\label{Morawetz final}
		There exists a positive constant C depending only on $ M, $ such that for all solutions $\Psi_i^{^{(\ell)}} $ to (\ref{waveeq}), $ i\in \lbrace 1,2 \rbrace$, $ \ell\geq i$, we have \begin{align}
			\begin{aligned}
				& \int_{S^2(r)} \left(\dfrac{1}{r^3}\abs{\partial_{r^{*}}\Psi_i}^2  +\dfrac{(r-M)(r-2M)^2}{r^4}\left(\abs{\slashed{\nabla}\Psi_i}^2 + \dfrac{1}{r^2}(\partial_{t} \Psi_i)^2 \right)  +\dfrac{(r-M)}{r^4}\abs{\Psi_i}^2\right) \\ 
				& \hspace{4cm} \leq C 	\int_{S^2(r)} \left ( Div(L_{\mu}^z) + K^{X,G,h}[\Psi_i]\right )
			\end{aligned}
		\end{align}
	\end{proposition}
	\begin{proof}
		We compute \begin{align}
			\begin{aligned}
				Div(L_{\mu}^z) = - \dfrac{z(r)}{D}(\partial_{t}\Psi_i)^2 + \dfrac{z(r)}{D}(\partial_{r*}\Psi_i)^2 + z(r) \abs{\slashed{\nabla}\Psi_i}^2 + z(r)V_i\Psi_i^2 + z'(r) \Psi_i \partial_{r^{*}}\Psi_i. \label{divL}
			\end{aligned}
		\end{align}
		Consider \begin{align}
			z(r) :=& -\dfrac{c}{r}\cdot x^3(2x-1)^2(1-x)^2, \ \ \ \text{with} \\
			z'(r)= & \ \ c\dfrac{(2x-1)x^2(1-x)^2}{r^2}(16x^2-16x+3).
		\end{align}
		where again, $ x= \sqrt{D} $ and $ c>0 $ a scaling to be determined in the end. 
		By Cauchy–Schwarz inequality, (\ref{divL}) yields 
		\begin{align}\begin{aligned} \label{divLin}
				& Div(L_{\mu}^z) \geq \\ 
				&	\geq - \dfrac{z(r)}{D}(\partial_{t}\Psi_i)^2 + \left( \dfrac{z(r)}{D}- \dfrac{\abs{z'}}{2}\right)(\partial_{r*}\Psi_i)^2 + z(r) \abs{\slashed{\nabla}\Psi_i}^2 + \left(z(r)V_i\right ) \p^2-  \dfrac{\abs{z'}}{2} \Psi_i^2.
			\end{aligned}
		\end{align}
		Note that $ z(r),\ z'(r),\ \dfrac{z}{D} $ are bounded functions everywhere including the horizon and all terms in $ (\ref{divLin}) $ have the same degeneracy as $ K^{X,G,h} $ . The coefficient $ z(r)\cdot V_i(r) $ depends on $ \ell $ and thus can be controlled using Poincare inequality. Using  Proposition \ref{Positivity of K^Xh} and choosing  a small enough scaling $ c>0 $ of $ z(r) $  allows us to control the remaining negative terms. \\
	\end{proof}
	
	\paragraph{Non-degenerate Morawetz Estimate.} The estimate of Proposition \ref{Morawetz final} is degenerate with respect to the angular derivatives and $ \partial_t \Psi_i $ at the photon sphere. Below, we remove this degeneracy which will prove useful later on, however, at the cost of losing one derivative at the level of initial data.
	
	First, by commuting equation (\ref{waveeq}) with the killing vectorfield $ T $ and using Proposition \ref{Positivity of K^Xh} we obtain \begin{equation}
		\int_{S^2(r)} \left(\dfrac{\sqrt{D}}{r^3}(\partial_t\Psi_i)^2 + \dfrac{\sqrt{D}}{r^3}(\Psi_i)^2\right) \leq C \int_{S^2(r)}\sum_{k=0}^{1}K^{X,G,h}[T^k\Psi_i], \label{t0_non_deg}
	\end{equation}
	where C depends only on $ M. $ 
	
	The remaining derivatives are obtained using the current \begin{align*}
		L_{\mu}^{w}=w(r) \Psi_i\cdot \nabla_{\mu}\Psi_i,
	\end{align*}
	for $ w(r)= \dfrac{1}{r^3}D^{\frac{3}{2}} $. Indeed, taking the divergence of that current we obtain \begin{align}
		\begin{aligned}
			Div(L_{\mu}^{w}) =  &\dfrac{1}{r^3}\sqrt{D}(\partial_{r^{*}}\Psi_i)^2 + \dfrac{1}{r^3}D^{\frac{3}{2}}\left(\abs{\slashed{\nabla}\Psi_i}^2+V_i \Psi_i^2\right)\\ -& \dfrac{1}{r^3}\sqrt{D}(\partial_t\Psi_i)^2+\dfrac{3}{r^4}D(1-2\sqrt{D})\Psi_i\cdot \partial_{r^{*}}\Psi_i.
		\end{aligned}
	\end{align}
	By Cauchy Schwartz, \begin{align*}
		\dfrac{3}{r^4}D(1-2\sqrt{D})\Psi_i\cdot \partial_{r^{*}}\Psi_i \geq - \dfrac{3}{2 \varepsilon}\dfrac{D}{r^4}\Psi_i^2 -\dfrac{3\varepsilon}{2}\dfrac{D}{r^4}(\partial_{r^{*}}\Psi_i)^2,
	\end{align*}
	and choosing $ \varepsilon>0  $ small enough such that $ r/3\geq \varepsilon \sqrt{D}$, the last term above can be absorbed in the first term of $ Div(L_{\mu}^{w}) $. Thus, using the above and 
	relation (\ref{t0_non_deg}) we prove the following proposition \begin{proposition} \label{nd_mor} There exits a positive constant $ C $ depending only on $ M $ such that for all solutions to (\ref{waveeq}), $ i\in \lbrace 1,2 \rbrace$ we have\begin{align*}
			\int_{S^2(r)}\left(\dfrac{\sqrt{D}}{r^3}(\partial_t\Psi_i)^2 + \dfrac{\sqrt{D}}{2r^3}(\partial_{r^{*}}\Psi_i)^2+ \dfrac{\sqrt{D}}{r}\abs{\slashed{\nabla}\Psi_i}^2+ \dfrac{\sqrt{D}}{r^3}\Psi_i^2 \right) \leq \int_{S^2(r)}\left (Div(L_{\mu}^{w})+ C \sum_{k=0}^{1}K^{X,G,h}[T^k\Psi_i]\right ).
		\end{align*}
	\end{proposition}
	\begin{flushleft}
		Note, the decay rate for the angular derivative coefficient comes from Proposition \ref{Morawetz final}.
		
	\end{flushleft}

	\subsubsection{Degenerate and non-degenerate \texorpdfstring{$ X $}{PDFstring}--estimate.}
	We saw above that the modified scalar currents are positive definite, however, we also need to control the boundary terms that arise when applying the divergence identity in $ R(0,\tau) $. 
	
	\begin{proposition} \label{X,h_bountary} 
		Let $ X=f\partial_{r^{*}} $, where $ f(r) $ is bounded, then there exists a uniform  positive constant $C $ depending on $ M,\Sigma_0$ and $ \norm{f}_{L^{\infty}(\mathcal{R}(0,\tau))} $ such that for all solutions $ \Psi_i^{^{(\ell)}}, \ \ell\geq i, \ i\in\br{1,2} $ to (\ref{waveeq}) 
		\begin{align}
			\abs{\int_S J_{\mu}^{X,G,h}[\Psi_i] n_{S}^{\mu}} \leq C \int_S J_{\mu}^T[\Psi_i] n_{S}^{\mu},
		\end{align}
		for $ S $ being $ \Sigma_{\tau}  $ or $ \mathcal{H}^{+} $, for all $ \tau >0. $
	\end{proposition}
	\begin{proof}
		Since our estimates involve the horizon as well, let us express everything in terms of the ingoing coordinate system $ (v,r,\vartheta,\varphi) $. In particular, we have $ \frac{\partial}{\partial{r^{\star}}}(r^{\star},t)=\frac{\partial}{\partial v}(v,r)+D\frac{\partial}{\partial r}(v,r) $, then
		\begin{align}
			\begin{aligned}
				J_{\mu}^{X}n_S^{\mu}&= Q(X,n_S)=fQ(\partial_v,n_S)+fDQ(\partial_{r},n_S)=\\
				= & f n_S^{v} \left(\partial_v \Psi \right)^2+ f n_S^{v}D \left(\partial_v \Psi \right)(\partial_{r}\Psi) + \dfrac{1}{2}f n_S^{r} D \left(\partial_r \Psi \right)^2 \\ &+\left(-\dfrac{f}{2}n_S^{r} \right)\abs{\slashed{\nabla}\Psi}^2 + \left(-\dfrac{V_i}{2}f n_S^{r}\right) \Psi^2.  
			\end{aligned}
		\end{align}
		In addition, we have 
		\begin{align}
			\begin{aligned}
				J_{\mu}^{X,G,h}n_S^{\mu}=&J_{\mu}^{X}n_S^{\mu} + \dfrac{h^2}{2}\Psi^2 (\partial_{r^{*}})_{\mu}n_S^{\mu} \\
				=&   f n_S^{v} \left(\partial_v \Psi \right)^2+ f n_S^{v}D \left(\partial_v \Psi \right)(\partial_{r}\Psi) + \dfrac{1}{2}f n_S^{r} D \left(\partial_r \Psi \right)^2 \\ &+\left(-\dfrac{f}{2}n_S^{r} \right)\abs{\slashed{\nabla}\Psi}^2 + \left(  \dfrac{h}{2}n_S^r-\dfrac{V_i}{2}f n_S^{r}\right) \Psi^2 
			\end{aligned}.
		\end{align}
		Now, using Cauchy-Schwarz, the fact that $ f,n^u,n^r, h $ are bounded functions and that $ V_i $ is linear with respect to $ \ell $, Proposition \ref{posT} concludes the proof. \\
	\end{proof}
	\begin{flushleft}
		Now, we are in the position of proving the general Morawetz estimates of this main section. First, we proceed with the proof of a degenerate estimate that captures the \textit{trapping} effect on both the photon sphere $ \br{r=2M} $ and the horizon $ \h $. 
	\end{flushleft}
	\begin{theorem} \label{Morawetz-Spacetime}
		There exists a constant $ C>0 $ depending on  $ M,\Sigma_0 $ such that for all solutions $\Psi_i^{^{(\ell)}} $ to (\ref{waveeq}), $ i\in \lbrace 1,2 \rbrace$, $ \ell\geq i$ , we have \begin{align}
			\begin{aligned}
				\int_{\mathcal{R}(0,\tau) } \left(\dfrac{1}{r^3}\abs{\partial_{r^{*}}\Psi_i}^2  +\dfrac{(r-M)(r-2M)^2}{r^4}\left(\abs{\slashed{\nabla}\Psi_i}^2 + \dfrac{1}{r^2}(\partial_{t} \Psi_i)^2 \right)  +\dfrac{(r-M)}{r^4}\abs{\Psi_i}^2\right) \\ 
				\leq C \int_{\Sigma_0} J_{\mu}^T[\Psi_i]n_{\Sigma_0}^{\mu}
			\end{aligned}
		\end{align}
	\end{theorem}
	\begin{proof}
		We apply Stoke's Theorem in the region $ R(0,\tau) $ for the the currents \begin{align}
			J_{\mu}^{X,G,h}[\Psi_i] + L_{\mu}^{z} + A \cdot J_{\mu}^T[\Psi_i]
		\end{align}
		for a big enough  constant $ A>0 $. Note, there will be no boundary terms on the horizon $ \mathcal{H}^+ $ since all quantities vanish there, thus applying Propositions \ref{X,h_bountary},  \ref{Morawetz final} and  \ref{T estimate} we conclude the proof. 
		
	\end{proof}
	
	Similarly, using Proposition \ref{nd_mor} we prove an estimate that does not degenerate at the photon sphere, however, requires higher regularity on the initial data.
	\begin{theorem} \label{Spacetime non degenerate photon estimate}
		There exists a positive constant $ C $ depending on $ M,\Sigma_{0} $ alone such that for all solutions $ \Psi_i^{^{(\ell)}} $ to (\ref{waveeq}), $ i\in \lbrace 1,2 \rbrace$, $ \ell\geq i$, we have \begin{equation}
			\int_{\mathcal{R}(0,\tau)}\left(\dfrac{\sqrt{D}}{r^3}(\partial_t\Psi_i)^2 + \dfrac{\sqrt{D}}{2r^3}(\partial_{r^{*}}\Psi_i)^2+ \dfrac{\sqrt{D}}{r}\abs{\slashed{\nabla}\Psi_i}^2+ \dfrac{\sqrt{D}}{r^3} \Psi_i^2 \right)\leq C\left( \int_{\Sigma_0} J_{\mu}^T[\Psi_i]n_{\Sigma_0}^{\mu}+J_{\mu}^T[T\Psi_i]n_{\Sigma_0}^{\mu}\right)
		\end{equation}
	\end{theorem}

	\subsection{The Vector Field N} \label{The Vector Field N}
	In this section, we are looking for a timelike vectorfield $ N $ that captures the non-degenerate energy of a local observer near the horizon. However, in the extremal Reissner-Nordstr{\"o}m case, the absence of redshift effect poses a difficulty since there is no analogous redshift vectorfield to Dafermos--Rodnianski's; for any  timelike vectorfield $ N $, the corresponding scalar current won't be non-negative definite. Instead, we  go around this by modifying appropriately the current $ J^N_{\mu}[\p]$ which allows us to obtain a non-negative local integrated estimate, however, still degenerate in the transversal invariant direction due to trapping on the horizon $ \h $.
	First, let's take a closer look at why $ J^N_{\mu}[\p] $ alone won't work. In this subsection, we work with respect to the coordinate system $ (v,r,\vartheta,\varphi) $ which is regular on the event horizon $ \h. $
	
	\paragraph{Absence of redshift effect.}
	
	Let $ N= N^v(r)\partial_{v}+ N^r (r)\partial_r $ be a future-directed timelike $ \phi_{\tau}^T- $ invariant vector field. Then, if we consider $ J_{\mu}^N[\Psi_i] $ for a solution $ \Psi_i^{^{(\ell)}} $ to (\ref{waveeq}), $ i\in \lbrace 1,2 \rbrace$, $ \ell\geq i$, we obtain \begin{equation}\label{K^N}
		K^N[\Psi_i]= F_{vv}(\partial_v \Psi_i)^2+ F_{rr}(\partial_r \Psi_i)^2+ F_{\slashed{\nabla}}\abs{\slashed{\nabla} \Psi_i}^2+ F_{vr}(\partial_{v}\Psi_i)(\partial_r \Psi_i) +  F_{00}^{(i)} \Psi_i^2,
	\end{equation}
	where the coefficients are given by \begin{align}
		\begin{aligned} \label{F coefficients}
			&F_{vv}= N^v_r, \hspace{2cm} F_{rr}= D\left(\dfrac{N_r^r}{2}-\dfrac{N^r}{r}\right) - \dfrac{N^r D'}{2}, \\
			&F_{\slashed{\nabla}}= -\dfrac{1}{2}N^r_r, \hspace{1.5cm} F_{vr}= DN^v_r- \dfrac{2N^r}{r}, \\
			& F_{00}^{(i)}= -\dfrac{1}{2}N(V_i)-V_i\left(\dfrac{N_r^r}{2}+\dfrac{N^r}{r}\right). 
		\end{aligned}
	\end{align}
	Here we denote by $ N^i_r:= \frac{\partial N^i}{\partial_r}$ and $ D'= \frac{\partial D}{\partial r}. $ In hope of proving $ K^N $ is non-negative definite, we would like to control the term $ F_{vr}(\partial_{v}\Psi_i)(\partial_r \Psi_i) $ using the positivity of $ F_{rr}, F_{vv} $. However, note that $ F_{rr} $ vanishes on the horizon whereas $ F_{vr} $ doesn't. Indeed, recall the relations (\ref{C1}), (\ref{C2}), so if we are looking for $ N $ timelike everywhere, it's necessary that $ N^r(M) \neq 0 $ on the horizon $ \h $. In particular, $ K^N[\p] $ is linear with respect to $ \partial_{r} \Psi_i $ on the horizon.
	In addition, the coefficient $ F_{00} $ becomes negative for low frequencies for both $ i=1,2. $
	Therefore, no choice of timeline vectorfield $ N $ can make $ K^N[\Psi_i] $ non-negative definite.

	\paragraph{A Locally Non-Negative Spacetime Current.} 
	In order to remedy the situation above, we need to introduce extra terms to our initial current. Consider $ 	J_{\mu}^{X,\omega,M} [\Psi] $ as in Proposition \ref{GenCur} for the (0,1)-form  $ M_{\mu}:= h(\partial_r)_{\mu}, $ and  $ \omega=\omega(r) $, $ h=h(r) $  functions of $ r $ alone. In particular, if 
	\begin{align}
		J_{\mu}^{N,\omega,M} [\Psi_i] =  J_{\mu}^N[\Psi_i] + \dfrac{1}{2} \omega \Psi_i \cdot \partial_{\mu}\Psi_i-\dfrac{1}{4}(\partial_{\mu}\omega)\Psi_i^2 + \dfrac{h}{4}\Psi_i^2 (\partial_r)_{\mu}
	\end{align}
	Then, we have \begin{align}
		\begin{aligned}
			K^{N,\omega,h}[\Psi_i]:=Div(J_{\mu}^{X,\omega,M} [\Psi] )= K^N[\Psi_i]+ \dfrac{1}{2}\omega \left(\nabla^{a}\Psi_i\nabla_a\Psi_i + V\Psi_i^2\right) \\
			+ \dfrac{1}{4}\left (h'+\dfrac{2}{r}h\right )\Psi_i^2 +\dfrac{1}{2}h\Psi_i \partial_r\Psi_i.\\
			\Rightarrow K^{N,\omega,h}[\Psi_i]= F_{vv}(\partial_{v}\Psi_i)^2 + \left(F_{rr}+\dfrac{\omega}{2}D\right)(\partial_r \Psi_i)^2 +  \left(F_{\slashed{\nabla}}+ \dfrac{\omega}{2}\right) \abs{\slashed{\nabla} \Psi_i}^2 \\ 
			+ ( F_{vr}+ \omega)(\partial_{v}\Psi_i)(\partial_r \Psi_i) + \left (F_{00}^{(i)}-\dfrac{1}{2}V_i+ \dfrac{1}{4}h'+\dfrac{h}{2r}\right )\Psi_i^2 + 2\dfrac{h}{4}\Psi_i \partial_r\Psi_i. \label{KNwM}
		\end{aligned}
	\end{align}
	
	For simplicity, let us denote the above coefficients of $ (\partial_a\Psi_i \cdot \partial_b \Psi_i) $ by $ G_{ab} $ where $ a,b\in \br{ v,r, \slashed{\nabla}, 0}$ and define the vector field $ N $  in the region $ M\leq r \leq \frac{9M}{8} $ as 
	\begin{align}
		N^u(r)=16r, \  \ N^r(r)=- \dfrac{3}{2}r+M,  \label{N-choice}
	\end{align}
	which is timelike. In view of the discussion of the previous paragraph, we will define $ \omega $ such that $ G_{vr}\Big|_{r=M}=0, $     \ \ i.e. with the choice of $ N $ as above, we need $ \omega= -1   $
	constant.  \\ 
	Finally, we will see in the proposition below, there is an appropriate function $ h $ which makes the coefficient $ G_{00} $ positive definite.
	
	\begin{proposition} \label{N derivatives positivity} Let $  h(r) := 6 \frac{\sqrt{D}(1+3\sqrt{D})}{r}$ and $ N $ as in (\ref{N-choice}), then there exists a positive constant $ C $ depending only on $ M $ such that for all solutions $ \Psi_{i}^{^{(\ell)}} $ to (\ref{waveeq}), supported on the fixed frequency $ \ell \geq i, \ i\in\br{1,2} $, the current $  K^{N,-1,h}[\Psi_i]  $ is non-negative definite in $ \mathcal{A}_{N}:= \br{M\leq r\leq \frac{9M}{8}} $. In particular, we have \begin{align}
			\int_{S^2(r)}	K^{N,-1,h}[\Psi_i] \geq C \int_{S^2(r)} \left ((\partial_{v}\Psi_i)^2+ \sqrt{D}(\partial_{r}\Psi_i)^2  + \abs{\slashed{\nabla}\Psi_i}^2 + \dfrac{1}{r^2}\Psi_i^2\right )
		\end{align}
		\begin{proof}
			First, let's write down the coefficients $ G_{ab} $ adapted to the choice of $ N $ as in (\ref{N-choice}): \begin{align}
				\begin{aligned}
					&G_{vv}=16, \hspace{1cm} G_{vr}=\sqrt{D}(16\sqrt{D}+2), \hspace{1cm} G_{rr}= \dfrac{\sqrt{D}}{4}(2-\sqrt{D}),
					\\ & G_{\slashed{\nabla}}=\dfrac{1}{4}, \hspace{1cm} G_{0r}= 2 \dfrac{h}{4}, \hspace{1cm} G_{00}^{(i)}= F_{00}^{(i)}-\dfrac{1}{2}V_i+\left ( \dfrac{1}{4}h'+\dfrac{h}{2r}\right ). 
				\end{aligned}
			\end{align}
			Using Cauchy-Schwartz inequality we get \begin{align} \begin{aligned}
					G_{vr} (\partial_{v}\Psi_i)(\partial_{r}\Psi_i) = \left (\dfrac{\sqrt{D}}{\sqrt{2}}\cdot\partial_{r}\Psi_i\right )\cdot (\sqrt{2}(16\sqrt{D}+2)\cdot \partial_{v}\Psi_i)\\ \geq  - \dfrac{D}{4}(\partial_{r}\Psi_i)^2 - (16\sqrt{D}+2)^2(\partial_{v}\Psi_i)^2.
				\end{aligned}
			\end{align} 
			Similarly, we obtain \begin{align}
				\begin{aligned}
					G_{0r} (\Psi_i)\cdot (\partial_{r}\Psi_i)= 2 \dfrac{h}{4} \left (\dfrac{\sqrt{5}}{\sqrt{r}}\Psi_i\right )\cdot (\dfrac{\sqrt{r}}{\sqrt{5}}\cdot \partial_{r}\Psi_i) 
					\\ \geq  -\dfrac{5h}{4\cdot r} \Psi_i^2 - \dfrac{h\cdot r}{4\cdot 5} (\partial_{r} \Psi_i)^2.
				\end{aligned}
			\end{align}
			Therefore, going back to (\ref{KNwM}), we need to show that the new coefficients are non-negative definite in the region $ \mathcal{A}_N, $. For simplicity, we denote by $ x:=\sqrt{D} $ which in $ \mathcal{A}_{N} $ it satisfies $ 0\leq x \leq \frac{1}{9} $.\\ \\
			For the coefficient of $ (\partial_{v}\Psi_i)^2 $, which is $ 16-(16\sqrt{D}+2)^2 $, we have \begin{align*}
				16-(16x+2)^2= 4(1-8x)(3+8x) \geq \frac{4}{9}\cdot 3 = \frac{4}{3}, \hspace{1cm} \forall \ x\in \mathcal{A}_{N}.
			\end{align*}
			The coefficient of $  (\partial_{r}\Psi_i)^2 $ is \begin{align}
				\dfrac{x}{4}(2-x) - \dfrac{x^2}{4}- \dfrac{h\cdot r}{4\cdot 5}= \dfrac{1}{2}x(1-x)-\dfrac{6}{4\cdot 5}x(1+3x) = \dfrac{x}{5}(1-7x),
			\end{align}
			which again for $ x\leq \frac{1}{9} $, it is uniformly positive.
			
			We already know that $ G_{\slashed{\nabla}}=\dfrac{1}{4}>0 $, and it only remains to show positivity for the coefficient of $ \Psi_i^2 $. For that, we are going to need Poincare inequality but first, let's examine the zeroth order coefficient closer. We have, \begin{align}
				G_{00}^{(i)}-\dfrac{5h}{4r}= F_{00}^{(i)}- \dfrac{1}{2}V_i+\dfrac{1}{4}\left (h'-\dfrac{3h}{r}\right ).  \label{K^N-zeroth}
			\end{align}
			For $ i=1 $, we write \begin{align}
				F_{00}^{(1)}-\dfrac{1}{2}V_1= \dfrac{1}{r^2}\left( \ell (x^2-  x) -\left( \dfrac{3}{2} + x  - \dfrac{17 }{2}x^2 + 6 x^3  \right) \right)
			\end{align}
			The above expression is mostly negative, so borrowing from the coefficient of \\ $ G_{\slashed{\nabla}}=\frac{1}{4}=\frac{1}{16}+\frac{3}{16} $ by  using Poincare Inequality we obtain the extra term 
			\begin{align}
				F_{00}^{(1)}-\dfrac{1}{2}V_1 + \dfrac{3}{16}\dfrac{\ell(\ell+1)}{r^2}= \dfrac{1}{r^2} \left(\dfrac{3}{16} \ell^2+ \ell \left (x^2-  x+ \dfrac{3}{16}\right ) -\left( \dfrac{3}{2} + x  - \dfrac{17 }{2}x^2 + 6 x^3  \right) \right) 
			\end{align}
			Using the fact $ x\leq \dfrac{1}{9} $, it's easy to see that the above expression is uniformly positive definite for all $ \ell \geq 3  $ in $ \mathcal{A}_N $. However, for $ \ell=1,2 $ we need the help of the extra positive term in (\ref{K^N-zeroth}). In particular, 
			\begin{align}
				\dfrac{1}{4}\left (h'-\dfrac{3h}{r}\right )= \dfrac{3}{2}(1+x-18x^2)
			\end{align}
			Therefore, we obtain \begin{align}
				G_{00}^{(1)} + \dfrac{3}{16}\dfrac{\ell(\ell+1)}{r^2} =\dfrac{1}{r^2} \left(\dfrac{3}{16} \ell^2+ \ell \left (x^2-  x+ \dfrac{3}{16}\right ) +\left( \dfrac{1}{2} x  -\dfrac{37 }{2}x^2 - 6 x^3  \right) \right), 
			\end{align}
			which is positive for both $ \ell=1,2 $ in the region $ \mathcal{A}_N .$
			
			Similarly, we obtain the same result for the potential $ V_2^{(\ell)}(r) $ for every $ \ell\geq 2 $ in $ \mathcal{A}_N $, which concludes the proof. \\
		\end{proof}
	\end{proposition}

	\paragraph{The \texorpdfstring{$ J_{\mu}^{N,\delta,\tilde{h}}[\Psi_i] $}{PDFstring} current.}
	Outside the region $ \mathcal{A}_N $, the scalar current $  K^{N,-1,h}[\Psi_i]$  takes negative values as well. We extend the current $ J_{\mu}^{N,-1,h} $ introducing cut-off functions so that the corresponding scalar current will be non-negative definite away from the horizon except  maybe a compact region not including the photon sphere, in which we control it using  Morawetz estimates.
	
	In particular, we extend the vectorfield $ N $ outside $ \mathcal{A}_N $ as \begin{align}
		\begin{aligned}
			N^v(r) >0, \hspace{0.5cm} \forall\ r\geq M \hspace{1cm} \text{and} \hspace{1cm} N^v(r)=1, \ \ \ \ \forall\  r\geq \dfrac{8M}{7}, \\
			N^r(r)\leq 0, \hspace{0.5cm} \forall \ r\geq M \hspace{1cm} \text{and} \hspace{1cm} N^r(r)=0, \ \ \ \ \forall \ r\geq \dfrac{8M}{7},
		\end{aligned}
	\end{align}
	and $ N $ remains an $ \phi_{\tau}^{T} $- invariant timelike vectorfield.   \\
	Away from $ \mathcal{A}_N $, we extend both $ \omega(r)=-1 $ and $ h(r) $ by introducing the smooth cut-off functions $ \delta : [M, +\infty) \to \mathbb{R} $ 
	such that $ \delta=1 $ for $ r\in \left [M, \frac{9M}{8}\right ] $, and $ \delta(r)=0 $ for $ r \geq \frac{8M}{7} $, while also, $ \tilde{h}(r)= h(r)$ for $ r\in \left [M, \frac{9M}{8}\right ]$, and $  \tilde{h}(r)=0 $ for $ r\geq \frac{8M}{7}. $ Now, we consider  the current \begin{align}
		\begin{aligned}
			J_{\mu}^{N,\delta,\tilde{h}}[\Psi_i] : =  J_{\mu}^N[\Psi_i] - \dfrac{1}{2} \delta \Psi_i \cdot \partial_{\mu}\Psi_i + \dfrac{\tilde{h}}{4}\Psi_i^2 (\partial_r)_{\mu},
		\end{aligned}
	\end{align}
	and  we make the following observations, 
	\begin{enumerate}
		\item In the $ \mathcal{A}_N $ region, we have $ J_{\mu}^{N,\delta,\tilde{h}}[\Psi_i]  \equiv J_{\mu}^{N,-1,h}[\Psi_i] $,  and thus $ K^{N,\delta,\tilde{h}}[\Psi_i] =K^{N,-1,h}[\Psi_i] $. \\
		\item For $ r\geq \frac{8M}{7} $, we have $ K^{N,\delta,\tilde{h}}[\Psi_i]= K^{T}= 0$. \\ 
		\item For $ \frac{9M}{8}\leq r \leq \frac{8M}{7} < 2M, $ the scalar current $  K^{N,\delta,\tilde{h}}[\Psi_i]$ can be negative in general, however, it can be controlled by the Morawetz estimates which are non-degenerate away from the photon sphere and the event horizon.
	\end{enumerate}
	
	\subsubsection{N-multiplier boundary terms. }
	In this section we treat the boundary terms that appear when applying the multiplier method for  $ N. $ First, observe that the vectorfield $ N $ captures the non-degenerate energy of a local observer near the horizon. Indeed, since both $ N, n_{\Sigma_{\tau}}^{\mu} $ are timelike everywhere in $ R(0,\tau) $, following exactly the same proof as in proposition \ref{posT}, we obtain \begin{align} \begin{aligned}
			\label{non-deg_energy}
			& \int_{S^2(r)} J_{\mu}^{N}[\Psi_i] n_{\Sigma_{\tau}}^{\mu} \geq C\int_{S^2(r)}  \left(  \left(\partial_{v} \Psi_i\right)^{2} +  \left(\partial_{r} \Psi_i\right)^{2} + |\slashed{\nabla} \Psi_i|^{2} + \dfrac{\ell(\ell+1)}{r^2}\Psi_i^{2}\right), \\
			& \int_{S^2(r)} J_{\mu}^{N}[\Psi_i] n_{\mathcal{H}^+}^{\mu} \geq C\int_{S^2(r)}  \left(  \left(\partial_{v} \Psi_i\right)^{2}  + |\slashed{\nabla} \Psi_i|^{2} + \dfrac{\ell(\ell+1)}{r^2}\Psi_i^{2}\right).
		\end{aligned}
	\end{align}
	where the constant $ C $ depends only on $ M $ and $ \Sigma_0 $ since $ N $ is $ \phi_{\tau}^T -$ invariant. 
	Now, regarding the current $  J_{\mu}^{N,\delta,\tilde{h}}[\Psi_i]   $, we have the following propositions. 
	\begin{proposition} \label{Jnti;da<Jn<JNtilda}
		There exists a constant $ C>0  $ depending only on $ M,\Sigma_0 $ such that for all solutions $\Psi_i^{^{(\ell)}}  $ to (\ref{waveeq}), supported on the fixed frequency $ \ell\geq i, $ $ i\in \lbrace 1,2 \rbrace$ we have \begin{align} 
			\dfrac{1}{C} \int_{\Sigma_{\tau}}J_{\mu}^{N,\delta,\tilde{h}}[\Psi_i] n_{\Sigma_{\tau}}^{\mu}    \leq	\int_{\Sigma_{\tau}} J_{\mu}^{N}[\Psi_i]n_{\Sigma_{\tau}}^{\mu}  \leq 2 \int_{\Sigma_{\tau}}
			J_{\mu}^{N,\delta,\tilde{h}}[\Psi_i] n_{\Sigma_{\tau}}^{\mu} + C\int_{\Sigma_{\tau}} J_{\mu}^{T}[\Psi_i] n_{\Sigma_{\tau}}^{\mu}
		\end{align} 
		\begin{proof}
			To prove the left hand side inequality, notice that \begin{align} \label{J^N,d,ht-JN}
				\begin{aligned}
					J_{\mu}^{N,\delta,\tilde{h}}[\Psi_i] n^{\mu} &=  J_{\mu}^N[\Psi_i]n^{\mu} - \dfrac{1}{2} \delta \Psi_i \cdot \partial_{\mu}\Psi_in^{\mu} + \dfrac{\tilde{h}}{4}\Psi_i^2 (\partial_r)_{\mu}n^{\mu} \\
					&=  J_{\mu}^N[\Psi_i]n^{\mu} - \dfrac{1}{2}\delta \Psi_i \cdot \partial_{v}\Psi_i n^{v} - \dfrac{1}{2}\delta \Psi_i \cdot \partial_{r}\Psi_i n^{r}+ \dfrac{\tilde{h}}{4}\Psi_i^2 n^v 
				\end{aligned}
			\end{align}
			and we use Cauchy-Schwartz inequality along with relation (\ref{non-deg_energy}). 
			
			For the right-hand side inequality, using again Cauchy-Schwartz we can write \begin{align*}
				- \dfrac{1}{2}\delta \Psi_i \cdot \partial_{r}\Psi_i n^{r} = - \dfrac{1}{2}\delta \dfrac{\Psi_i}{\varepsilon} \cdot \varepsilon \partial_{r}\Psi_i n^{r} \geq \\
				\geq - \dfrac{1}{4} \delta \dfrac{(\Psi_i)^2}{\varepsilon}\abs{n^v} -  \dfrac{1}{4} \delta \varepsilon (\partial_{r}\Psi_i)^2\abs{n^v}.
			\end{align*}
			Doing the same for the term $  - \dfrac{1}{2}\delta \Psi_i \cdot \partial_{v}\Psi_i n^{v} $, and the fact that $ n^v,n^r, \tilde{h} $ are bounded, we choose $ \varepsilon>0 $ small enough such that relation (\ref{J^N,d,ht-JN}) yields 
			\begin{align}
				\begin{aligned}
					\int_{S^2(r)} J_{\mu}^{N,\delta,\tilde{h}}[\Psi_i] n^{\mu}  \geq \int_{S^2(r)} \dfrac{1}{2}J_{\mu}^N[\Psi_i]n^{\mu} - \dfrac{\tilde{c}}{\varepsilon} \Psi_i^2
				\end{aligned}
			\end{align}
			Therefore, using Proposition \ref{posT} we obtain \begin{align*}
				\int_{S^2(r)} J_{\mu}^N[\Psi_i]n^{\mu} \leq 2\int_{S^2(r)}  J_{\mu}^{N,\delta,\tilde{h}}[\Psi_i] n^{\mu}  + C J_{\mu}^T [\Psi_i]n^{\mu},
			\end{align*}
			for a big enough constant $ C $, depending only on $ M,\Sigma_0 $. \\
		\end{proof}
		
	\end{proposition}
	
	We now control the boundary term over $ \mathcal{H}^+ $. In particular, we have the following proposition. 
	\begin{proposition} \label{Jtilda<Jn horizon}
		For all solutions $\Psi_i^{^{(\ell)}} $ to (\ref{waveeq}), supported on the fixed frequency $ \ell\geq i,$ $ i\in \lbrace 1,2 \rbrace$ we have \begin{align}
			\int_{\mathcal{H}^+}J_{\mu}^{N,\delta,\tilde{h}}[\Psi_i] n^{\mu}_{\mathcal{H}^+} \geq \dfrac{1}{2}  \int_{\mathcal{H}^+}J_{\mu}^{N}[\Psi_i] n^{\mu}_{\mathcal{H}^+} - C \int_{\mathcal{H}^+}J_{\mu}^{T}[\Psi_i] n^{\mu}_{\mathcal{H}^+}
		\end{align}
		for a positive constant $ C $ depending only on $ M,\Sigma_0 $.
	\end{proposition} 
	\begin{proof}
		With the following convention in mind, $ n_{\mathcal{H}^+} \equiv T $, we write down equation (\ref{J^N,d,ht-JN}) on the horizon $ \mathcal{H}^+ $ and we obtain 
		\begin{equation}
			J_{\mu}^{N,\delta,\tilde{h}}[\Psi_i] n^{\mu}_{\mathcal{H}^+} = J_{\mu}^{N}[\Psi_i] n^{\mu}_{\mathcal{H}^+} - \dfrac{1}{2}\Psi_i \cdot \partial_{v} \Psi_i,
		\end{equation}
		since $ n^r_{\mathcal{H}^+}=\tilde{h}(M)=0 $ and $ \delta(M)=1. $
		Using Cauchy Schwartz we write 
		\begin{align*}
			\int_{S^2(r)}J_{\mu}^{N,\delta,\tilde{h}}[\Psi_i] n^{\mu}_{\mathcal{H}^+}&=	\int_{S^2(r)}J_{\mu}^{N}[\Psi_i] n^{\mu}_{\mathcal{H}^+}-\dfrac{1}{2}  \cdot \Psi_i \cdot \partial_{v} \Psi_i \geq \int_{S^2(r)} J_{\mu}^{N}[\Psi_i] n^{\mu}_{\mathcal{H}^+} - \dfrac{\varepsilon}{4} (\Psi_i)^2 - \dfrac{1}{4\varepsilon} (\partial_{v}\Psi_i)^2\\&  \geq \int_{S^2(r)} \dfrac{1}{2}J_{\mu}^{N}[\Psi_i] n^{\mu}_{\mathcal{H}^+} - \dfrac{1}{4\varepsilon} J_{\mu}^{T}[\Psi_i] n^{\mu}_{\mathcal{H}^+},
		\end{align*}
		for $ \varepsilon>0 $ sufficiently small, which concludes the proof. \\
	\end{proof}
	
	\paragraph{Uniform Boundedness of Local Observer's Energy and Integrated Local Energy.}
	We combine our results above to obtain an integrated local energy estimate which captures the trapping effect at the horizon $ \h $ due to the degeneracy of the redshift effect.
	
	\begin{theorem}\label{Nuniform }
		There exists a constant $ C>0 $ which depends on $ M,\Sigma_0 $ such that for all solutions $\Psi_i^{^{(\ell)}} $ to (\ref{waveeq}), supported on the fixed frequency $ \ell\geq i,$ $ i\in \lbrace 1,2 \rbrace$ \begin{align} 
			\begin{aligned}
				\int_{\Sigma_{\tau}} J_{\mu}^{N}[\Psi_i] &n^{\mu}_{\Sigma_{\tau}} +  \int_{\mathcal{H}^+}J_{\mu}^{N}[\Psi_i] n^{\mu}_{\mathcal{H}^+} 
				\\ + &\int_{\mathcal{A}_{N}} (\partial_{v}\Psi_i)^2+ \sqrt{D}(\partial_{r}\Psi_i)^2  + \abs{\slashed{\nabla}\Psi_i}^2 + \dfrac{1}{r^2}\Psi_i^2\
				\leq \  C 	\int_{\Sigma_{0}} J_{\mu}^{N}[\Psi_i] n^{\mu}_{\Sigma_{0}}
			\end{aligned}
		\end{align}
	\end{theorem}
	\begin{proof}
		We apply Stoke's theorem for the current $  J_{\mu}^{N,\delta,\tilde{h}}[\Psi_i]  $ in the region $ \mathcal{R}(0,\tau) $ to obtain
		\begin{align}
			\int_{\Sigma_{\tau}}J_{\mu}^{N,\delta,\tilde{h}}[\Psi_i] n^{\mu}_{\Sigma_{\tau}}+ \int_{\mathcal{R}(0,\tau) } K^{N,\delta,\tilde{h}} +  \int_{\mathcal{H}^+}J_{\mu}^{N,\delta,\tilde{h}}[\Psi_i] n^{\mu}_{\mathcal{H}^+}= \int_{\Sigma_{0}}J_{\mu}^{N,\delta,\tilde{h}}[\Psi_i] n^{\mu}_{\Sigma_{0}}
		\end{align}
		We have seen before that the spacetime term $ K^{N,\delta,\tilde{h}} $ is non-negative definite everywhere, except in the region $   \frac{9M}{8}\leq r \leq \frac{8M}{7} < 2M. $ 
		In that region, we control the spacetime term using Theorem \ref{Morawetz-Spacetime} and we make it non-negative definite in terms of the $ 1- $jet of $ \Psi_{i}^{^{(\ell)}},$ for all $ r\geq M $. 
		Therefore, using also the left-hand side inequality of Proposition \ref{Jnti;da<Jn<JNtilda}, we obtain \begin{align}
			\int_{\Sigma_{\tau}}J_{\mu}^{N,\delta,\tilde{h}}[\Psi_i] n^{\mu}_{\Sigma_{\tau}} +  \int_{\mathcal{H}^+}J_{\mu}^{N,\delta,\tilde{h}}[\Psi_i] n^{\mu}_{\mathcal{H}^+} \leq C \int_{\Sigma_{0}}J_{\mu}^{N}[\Psi_i] n^{\mu}_{\Sigma_{0}} 
		\end{align}
		On the other hand, using the right-hand side inequality of Proposition \ref{Jnti;da<Jn<JNtilda} and Proposition \ref{Jtilda<Jn horizon} the above relation yields \begin{align}
			\begin{aligned}
				&\dfrac{1}{2}\int_{\Sigma_{\tau}}J_{\mu}^{N}[\Psi_i] n^{\mu}_{\Sigma_{\tau}} + \dfrac{1}{2} \int_{\mathcal{H}^+}J_{\mu}^{N}[\Psi_i] n^{\mu}_{\mathcal{H}^+}\\ \leq & \ C \int_{\Sigma_{0}}J_{\mu}^{N}[\Psi_i] n^{\mu}_{\Sigma_{0}}  + C' \left(\int_{\Sigma_{\tau}} J_{\mu}^T[\Psi_i]n_{\Sigma_{\tau}}^{\mu} + \int_{\mathcal{H}^+(0,\tau)} J_{\mu}^T[\Psi_i]n_{\mathcal{H}^+}^{\mu} \right), 
			\end{aligned}
		\end{align}
		for a sufficiently big constant $ C'>0 $ depending only on $ M, \Sigma_{0} $. However, note \begin{align*}
			\int_{\Sigma_{\tau}} J_{\mu}^T[\Psi_i]n_{S}^{\mu} + \int_{\mathcal{H}^+(0,\tau)} J_{\mu}^T[\Psi_i]n_{\mathcal{H}^+}^{\mu} =\int_{\Sigma_{0}} J_{\mu}^T[\Psi_i]n_{\Sigma_{0}}^{\mu},
		\end{align*}
		which concludes the proof.

	\end{proof}

	\section{Energy Estimates for \texorpdfstring{$ \partial_r\Psi_i $}{PDFstring}. \label{EE first}} 
	In the previous sections, we derived integrated local energy estimates that are degenerate with respect to $ (\partial_{r}\Psi_i^{^{(\ell)}})^2 $ on the horizon $ \mathcal{H}^+. $ In this section, we remove the aforementioned degeneracy near the horizon $ \mathcal{H}^+ $, at the cost of losing one derivative at the level of initial data.  For that, we need to commute our equation (\ref{waveeq}) with the vectorfield $ \partial_r. $ 
	 In order to derive the instability that occurs along the horizon $ \mathcal{H}^+ $, we need to control $ \partial_r^k \Psi_i^{^{(\ell)}} $, for $ k\in \mathbb{N} $ depending on $ \ell. $ 
	
	\paragraph{The \texorpdfstring{$\partial_{r}- $}{PDFstring}Commutator.}
	For any scalar function $ \phi $ we have \begin{align}
		[\square_g,\partial_r]\phi= -D' \partial_r\partial_r\phi + \dfrac{2}{r^2}\partial_{v}\phi -R'\partial_r\phi + \dfrac{2}{r}\slashed{\Delta}\phi.
	\end{align}
	The above relation can be justified by direct computation and using the fact that \\ $ [\slashed{\Delta},\partial_r]\phi= \frac{2}{r}\slashed{\Delta}\phi. $ Therefore,  taking a $ \partial_r -$ derivative of equation (\ref{waveeq}) yields \begin{align}
		\begin{aligned}
			&\square_{g} (\partial_r\Psi_i) - [\square_g,\partial_r]\Psi_i- V_i(\partial_r\Psi_i)-V_i'(r)\Psi_i=0 \\
			\Rightarrow \hspace{0.5cm}&  \square_{g} (\partial_r\Psi_i) - V(\partial_r\Psi_i)=[\square_g,\partial_r]\Psi_i+V_i'(r)\Psi_i.
		\end{aligned}
	\end{align}
	Since $ \Psi_i^{^{(\ell)}} $ is supported on a fixed frequency $ \ell\geq i, \ i\in\br{1,2} $, we can write $$ V_i'\Psi_i=-\dfrac{V_i'\cdot r^2}{\ell(\ell+1)}\slashed{\Delta}\Psi_i. $$Now,  let's define 
	\begin{align*}
		M[\Psi_i]&:= [\square_g,\partial_r]\Psi_i+V_i'(r)\Psi_i = \\
		&= 	-D' \partial_r\partial_r\Psi_i+ \dfrac{2}{r^2}\partial_{v}\Psi_i-R'\partial_r\Psi_i+ \left(\dfrac{2}{r} -\dfrac{V_i'\cdot r^2}{\ell(\ell+1)}\right) \slashed{\Delta}\Psi_i
	\end{align*}
	Then, we obtain the following equations for $ \partial_r\Psi_i^{^{(\ell)}} $
	\begin{align}
		\left (\square_{g}-V_i^{^{(\ell)}}\right )\left (\partial_r\Psi_i^{^{(\ell)}}\right )= M[\Psi_i^{^{(\ell)}}], \hspace{1cm} \ell \geq i, \ i \in \br{1,2}. \label{non homogeneuos for derivative of psi}
	\end{align}
	For simplicity, we will drop the $ (\ell) $ superscript since  each $\Psi_i,V_i  $ is defined in terms of $ \ell\geq i $ in the first place.
	
	\paragraph{Control far from the horizon.}
	
	Away from the horizon $ \mathcal{H}^+ $, the spacetime term $ \int_{R\cap \br{r\geq r_0}} (\partial_r\Psi_i)^2 $, $ r_0> M $, is already controlled in the previous section.
	In addition, commuting our equation (\ref{waveeq}) with $ T $, using local elliptic estimates,  and $ L^2 $ bounds obtained earlier we control all second order derivatives away from $ \mathcal{H}^+ $, for all solutions $ \Psi_i^{^{(\ell)}} $ to (\ref{waveeq}), $ \ell\geq i,  $ $ i\in \lbrace 1,2 \rbrace$.  In particular, \begin{align} \label{H^2(Sigma)}
		\sum_{|\alpha|=2} \left( \int_{\Sigma_{\tau}\cap \lbrace r\geq r_0>M \rbrace} \abs{\partial^{\alpha}\Psi_i}^2\right) \leq C \left(\int_{\Sigma_{0}}J_{\mu}^T[\Psi_i] n_{\Sigma_0}^{\mu} + \int_{\Sigma_{0}}J_{\mu}^T[T\Psi_i] n_{\Sigma_0}^{\mu} \right), 
	\end{align}
	for a constant $ C>0 $ depending only on $ M,r_0, \Sigma_0. $  The same result holds for the bulk integrals where,  however, we must exclude the photon sphere $ \lbrace r = 2M \rbrace $ due to trapping. Specifically, since $ T $ is timelike for $ r\geq r_0 $, using local elliptic estimates and Theorem \ref{Morawetz-Spacetime} yield 
	\begin{align} \label{H^2(A)}
		\sum_{|\alpha|=2} \left( \int_{R(0,\tau)\cap \lbrace  r_0\leq r \leq r_1<  2M \rbrace} \abs{\partial^{\alpha}\Psi_i}^2\right) \leq C \left(\int_{\Sigma_{0}}J_{\mu}^T[\Psi_i] n_{\Sigma_0}^{\mu} + \int_{\Sigma_{0}}J_{\mu}^T[T\Psi_i] n_{\Sigma_0}^{\mu} \right). 
	\end{align}
	
	\subsection{Higher order control near the horizon.} 
	In view of the above, we now restrict our attention close to the horizon $ \mathcal{H}^+ $. In particular, we are looking for an $ \phi_{\tau}^T - $ invariant timelike vectorfield  $ \x=\x^{v} \partial_v + \x ^ r \partial_r $ that will act as a multiplier for the function $ \partial_{r}\Psi_i. $ 
	
	Applying Stokes Theorem in the region $ R(0,\tau) $ for the energy current $ J^{\x}_{\mu}[\partial_{r}\Psi_i]$ yields 
	\begin{align}
		\begin{aligned} \label{Div partial_rPsi}
			\int_{\Sigma_{\tau}}J_{\mu}^{\x}[\partial_r\Psi_i] n^{\mu}_{\Sigma_{\tau}}+ \int_{\mathcal{R}(0,\tau) } \nabla^{\mu} J_{\mu}^{\x}[\partial_r\Psi_i] +  \int_{\mathcal{H}^+}J_{\mu}^{\x}[\partial_r\Psi_i] n^{\mu}_{\mathcal{H}^+}= \int_{\Sigma_{0}}J_{\mu}^{\x}[\partial_r\Psi_i] n^{\mu}_{\Sigma_{0}}.
		\end{aligned}
	\end{align}
	Since $ \x $ is timelike in a region close to the horizon $\mathcal{H}^+ $, we have \begin{align}
		\begin{aligned}
			\int_{S^2(r)}J_{\mu}^{\x}[\partial_r\Psi_i] n^{\mu}_{\Sigma_{\tau}} \geq C \left(\int_{S^2(r)} (\partial_v\partial_r\Psi_i)^2 + (\partial_r\partial_r\Psi_i)^2 + \abs{\slashed{\nabla}\partial_r\Psi_i}^2 + \dfrac{1}{r^2}(\partial_r\Psi_i)^2 \right), 
		\end{aligned}
	\end{align}
	for a constant $ C>0 $ depending only on $ M,\Sigma_0,\x $. In addition, \textbf{on the horizon} we have \begin{align}
		\begin{aligned} \label{x horizon flux}
			\int_{S^2(r)}J_{\mu}^{\x}[\partial_r\Psi_i] n^{\mu}_{\mathcal{H}^+}= \int_{S^2(r)} \x^v(M)(\partial_v\partial_r\Psi_i)^2-\dfrac{\x^r(M)}{2}\left(1+ \dfrac{2\cdot (-1)^{i-1}}{(\ell+i-1)}\right)\abs{\slashed{\nabla}\partial_r\Psi_i}^2 
		\end{aligned}
	\end{align}
	where $ \ell $ is the fixed frequency support of $ \Psi_i^{^{(\ell)}}. $
	
	In what follows, we are studying the bulk term that appears in (\ref{Div partial_rPsi}) and we define accordingly an appropriate vectorfield $ \x $ that will allow for control of the second order derivatives of $ \Psi_i $, while also for $ (\partial_r\Psi_i)^2 $ near the horizon $ \mathcal{H}^+. $ Specifically, from Proposition \ref{GenCur} we obtain \begin{align}
		\begin{aligned}
			\nabla^{\mu} J_{\mu}^{\x}[\partial_r\Psi_i]&= \dfrac{1}{2}Q[\partial_r\Psi_i]\cdot ^{(\x)}\pi - \dfrac{1}{2}\x(V_i)(\partial_r\Psi_i)^2 + \x(\partial_r\Psi_i)\cdot M[\Psi_i]= \\
			&= K^{\x}[\partial_r\Psi_i] + \left (\x^v (\partial_v\partial_r\Psi_i) + \x^r (\partial_r\partial_r\Psi_i) \right ) M[\Psi_i].
		\end{aligned}
	\end{align}
	Note above, we have the extra third term since $ \partial_r\Psi_i $ satisfies the non-homogeneous wave equation (\ref{non homogeneuos for derivative of psi}).
	Now, similarly to Section \ref{The Vector Field N}, we have \begin{align} \begin{aligned}
			K^{\x}[\partial_r\Psi_i] =& F_{vv}(\partial_v \partial_r\Psi_i)^2+ F_{rr}(\partial_r \partial_r\Psi_i)^2+ F_{\slashed{\nabla}}\abs{\slashed{\nabla} \partial_r\Psi_i}^2\\ &+ F_{vr}(\partial_{v}\partial_r\Psi_i)(\partial_r \partial_r\Psi_i) + F_{00}(\partial_r\Psi_i)^2,
		\end{aligned}
	\end{align}
	where the coefficients $ F_{ij} $ are defined as in relation (\ref{F coefficients}) for the vectorfield $ \x. $
	Expanding also $ M_i[\Psi_i] $ we obtain altogether \begin{equation}
		\begin{aligned} \label{E_i coefficients}
			\nabla^{\mu} J_{\mu}^{\x}[\partial_r\Psi_i] & = 
			E_{1}\left(\partial_{v} \partial_{r} \Psi_i\right)^{2}+E_{2}\left(\partial_{r} \partial_{r} \Psi_i\right)^{2}+E_{3}\left|\slashed{\nabla} \partial_{r} \Psi_i\right|^{2}+E_{4}\left(\partial_{v} \partial_{r} \Psi_i\right)\left(\partial_{v} \Psi_i\right)\\ &+ E_{5}\left(\partial_{v} \partial_{r} \Psi_i\right)\left(\partial_{r} \Psi_i\right) 
			+E_{6}\left(\partial_{r} \partial_{r} \Psi_i\right)\left(\partial_{v} \Psi_i\right)+E_{7}\left(\partial_{v} \partial_{r} \Psi_i\right) \slashed{\Delta} \Psi_i\\ & +E_{8}\left(\partial_{r} \partial_{r} \Psi_i\right) \slashed{\Delta} \Psi_i+ E_{9}\left(\partial_{v} \partial_{r} \Psi_i\right)\left(\partial_{r} \partial_{r} \Psi_i\right) 
			+E_{10}\left(\partial_{r} \partial_{r} \Psi_i\right)\left(\partial_{r} \Psi_i\right) \\
			& + E_{11}(\partial_r\Psi_i)^2,
		\end{aligned}
	\end{equation}
	with coefficients $ E_j $, $ j \in \br{1,..,11} $ 
	\begin{align}
		\begin{aligned}
			&E_{1}=\x^{v}_r,\hspace{1cm} E_{2}=D\left[\frac{\x^r_r}{2}-\frac{\x^{r}}{r}\right]-\frac{3 D^{\prime}}{2} \x^{r}, \hspace{1cm} E_{3}=-\frac{1}{2}\x_r^r,& \\
			&E_{4}=2 \frac{\x^{v}}{r^{2}},\hspace{1cm} E_{5}=-\x^{v} R^{\prime},\hspace{1cm} E_{6}=2 \frac{\x^{r}}{r^{2}},\hspace{1cm} E_{7}= \x^v\left(\dfrac{2}{r} -\dfrac{V_i'\cdot r^2}{\ell(\ell+1)}\right),\hspace{1cm}& \\ &E_{8}=\x^r\left(\dfrac{2}{r} -\dfrac{V_i'\cdot r^2}{\ell(\ell+1)}\right), \hspace{1cm} 
			E_{9}=D\x^v_r-D^{\prime} \x^{v}-2 \frac{\x^{r}}{r}, \hspace{1cm}  E_{10}=-\x^{r} R^{\prime}, &\\
			& E_{11} =-\dfrac{1}{2}\x(V_i)-V_i\left (\dfrac{\x_r^r}{2}+\dfrac{\x^r}{r}\right ).&	
		\end{aligned}
	\end{align}
	
	First, let's make some observations concerning the choice of $ \x $ and its consequence on the coefficients $ E_j. $ If $ \x $ is timelike everywhere, then $ \x^v(M)>0 $, $ \x^r(M)<0 $ and thus the same holds in a neighborhood of the horizon $ \mathcal{H}^+. $ This already ensures that $ E_2\geq 0 $, near the horizon, and it vanishes to first order on it. Indeed, 
	\begin{align}
		E_2=D\left[\frac{\x^r_r}{2}-\frac{\x^{r}}{r}\right]-\frac{3 D^{\prime}}{2} \x^{r}= \sqrt{D}\left(-\dfrac{\x^r}{r}(3-\sqrt{D})+\dfrac{\x_r^r}{2}\sqrt{D}\right),
	\end{align}
while the first term in the parenthesis is the dominant one and positive near $ \h. $

	Moreover, we have the freedom of choosing $ \x_r^v  $ and $ -\x_r^r $ positive and sufficiently large such that $ E_1>0, \ E_3>0 $ near $ \mathcal{H}^+ $. \\When it comes to controlling the coefficients $ E_j $ not corresponding to quadratic terms, part of the idea is to absorb them  in the first three terms $ E_1,E_2,E_3, $ and for that we need $ -\x_r^r(M) \gg -\frac{\x_r(M)}{M} $, as well as $ \x^v_r(M) \gg \frac{\x^v(M)}{M}$.
	
	For the reasons above, we restrict our attention to a region close to $ \mathcal{H}^+ $, $ \mathcal{A}_c := \br{M \leq r \leq r_c} $ (Figure \ref{Near horizon regions}) for an appropriate $ r_c\in (M,2M) $, to be determined at the end, and choose $ \x $ with the above requirements. Outside $ \mathcal{A}_c $ we extend $ \x $ as  \begin{align}
		\x^v >0, \ -\x^r>0 \ \ \ \text{for} \ r\leq r_d \ \ \text{and} \ \ \x\equiv 0, \ \ \text{for} \ \ r \geq r_d,
	\end{align}
	for an $ r_d $ satisfying $ M<r_c<r_d<2M. $  
	
	\begin{figure}[hbt!]
		\centering
		\def\svgscale{0.7}
			\includegraphics[scale=0.7]{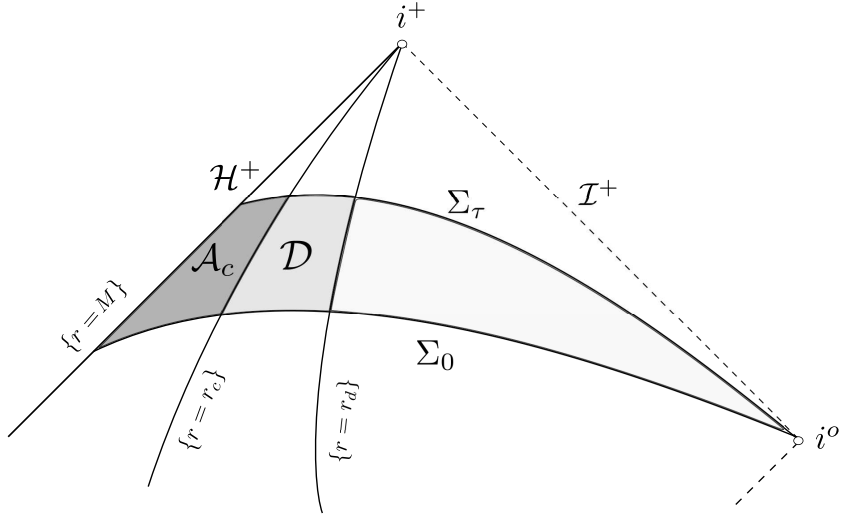}
		\caption{Regions near the event horizon $ \h $.}
		\label{Near horizon regions}
	\end{figure}
	
	\subsection{Hardy Inequalities and Lemmas}
	Before we treat the bulk terms in (\ref{E_i coefficients}), we present a few lemmas that will prove useful later on when controlling boundary terms.

	\begin{lemma}(\textbf{First Hardy Inequality})\label{hardy}
		Consider $ r_c\in (M,2M) $, then for any scalar function $ \psi $ and $ \varepsilon>0 $ we have \begin{align}
			\int_{\mathcal{H}^+\cap \Sigma_{\tau}} \psi^2 \leq \varepsilon \int_{\Sigma_{\tau}\cap \br{r\leq r_c}}(\partial_v\psi)^2 + (\partial_r\psi)^2 + C_{\varepsilon}\int_{\Sigma_{\tau}\cap \br{r\leq r_c}}\psi^2,
		\end{align}
		for a constant $ C_{\varepsilon} $ depending only on $ M,\varepsilon,r_c $ and $ \Sigma_{0} $.
	\end{lemma}
	\begin{proof}
		Note that for any scalar $ \phi $ we have \begin{align} \label{Hardy identity}
			\int_M^{r_c}(\partial_{\rho}\phi)\psi^2 d\rho = \phi\cdot \psi^2\Big|_M^{r_c} - 2 \int _M^{r_c}\phi\cdot  \psi(\partial_{\rho}\psi)d\rho 
		\end{align} 
		In particular, if we choose $ \phi= r-r_c $ we obtain \begin{align*} 
			\begin{aligned}
				(r_c-M)\psi^2(M)= & \int_M^{r_c}\psi^2 + 2 (r-r_c) \psi\cdot (\partial_{\rho}\psi)d\rho\\ \leq &\   \varepsilon \int _M^{r_c}(\partial_{\rho}\psi)^2d\rho + \int_M^{r_c}\left (1+\dfrac{(r-r_c)^2}{\varepsilon}\right )\psi^2 d\rho
			\end{aligned}
		\end{align*}
		Note $ \rho $ is bounded in the region $ \br{M\leq r \leq r_c} $, thus integrating over $ \mathbb{S}^2 $ we obtain \begin{align*}
			\begin{aligned}
				&\dfrac{(r_c-M)}{M^2}\int_{\mathbb{S}^2}\psi^2(M)M^2 d\omega\leq \varepsilon \int_{M}^{r_c}\int_{\mathbb{S}^2} (g_1\partial_v\psi + \partial_r\psi)^2\dfrac{1}{\rho^2}\cdot  \rho^2 d\rho d\omega  \\  &\hspace{4cm}  + \int_M^{r_c} \int_{\mathbb{S}^2}\left (1+\dfrac{(r-r_c)^2}{\varepsilon}\right )\psi^2 \dfrac{1}{\rho^2}\cdot  \rho^2  d\rho d\omega 
				\\ \Rightarrow \hspace{1cm} &	\int_{\mathcal{H}^+\cap \Sigma_{\tau}} \psi^2 \leq \varepsilon \left (\int_{\Sigma_{\tau}\cap \br{r\leq r_c}}(\partial_v\psi)^2 + (\partial_r\psi)^2 \right )+ C_{\varepsilon}\int_{\Sigma_{\tau}\cap \br{r\leq r_c}}\psi^2.
			\end{aligned}
		\end{align*}
	\end{proof}
	
	\begin{lemma}\label{second hardy} (\textbf{Second Hardy Inequality})
		Let $ M< r_c<r_d $, and consider the regions $\mathcal{C}:=R(0,\tau)\cap \br{M\leq r\leq r_c}   $ and  $ \mathcal{D}=R(0,\tau)\cap \br{r_c\leq r\leq r_d}    $, then for any scalar function $ \psi $ we have \begin{align}
			\int_{\mathcal{C}}\psi^2 \leq C \int_{\mathcal{D}} \psi^2 + C \int_{\mathcal{C}\cup\mathcal{D}}D(r)\left[(\partial_v\psi)^2+(\partial_r\psi)^2\right]   
		\end{align}
		for a positive constant $ C $ depending on $ M, r_c, r_d, $ and $ \Sigma_0. $
		\begin{proof}
			We apply relation (\ref{Hardy identity}) for $ \phi := \sqrt{D(r)} = \left(1-\frac{M}{r}\right) $ and we obtain \begin{align*}
				\int_{M}^{r_c} \left(\dfrac{M}{r^2}\right) \psi^{2} d\rho = \left(1-\dfrac{M}{r_c}\right) \psi^{2}(r_c) - 2 \int_{M}^{r_c} \sqrt{D} \psi \cdot \partial_{\rho} \psi d\rho
			\end{align*}
			However, by repeating the steps of the proof of the first Hardy Inequality, we can also show \begin{align*}
				(r_d-r_c) \psi^{2}(r_c)  \leq \epsilon' \int_{r_c}^{r_d} (\partial_{\rho}\psi)^{2} d\rho + C_{\epsilon'} \int_{r_c}^{r_d} \psi^{2}d\rho
			\end{align*}
			for any $ \epsilon'>0 $ and a positive constant $ C_{\epsilon'} $ depending on $ \epsilon', r_c, r_d. $ Let $ \epsilon'=1 $, then going back to the first relation and applying Cauchy-Schwarz we get \begin{align*}
				\int_{M}^{r_c} \left(\dfrac{M}{r^2}-\epsilon\right)\psi^{2}d\rho \leq C_1 \int_{r_c}^{r_d} \psi^{2} d\rho + C_1 \int_{r_c}^{r_d} (\partial_{\rho}\psi)^{2} d\rho + \dfrac{1}{\epsilon} \int_{M}^{r_c} D (\partial_{\rho}\psi)^{2} d\rho
			\end{align*}
			We choose $ \epsilon >0  $ small enough so that the coefficient of the left-hand side above is uniformly positive in $ \br{r_c\leq r\leq r_d} $ and we obtain \begin{align*}
				\int_{M}^{r_c} \psi^{2} d\rho \leq C \int_{r_c}^{r_d} \psi^{2}d\rho +  C \int_{M}^{r_d} D(r) \left[(\partial_v\psi)^{2} + (\partial_r\psi)^{2}\right],
			\end{align*}
			for a positive constant depending on $ M,r_c, r_d. $ Integrating on the spheres $ \mathbb{S}^{2} $ and in time $ \tilde{\tau} \in [0,\tau] $  while also using coarea formula yields the estimate of the assumption.
		\end{proof}
	\end{lemma}

	\begin{remark}
		Note, $ \psi $ need not be a solution to a wave equation, i.e. the result above holds for any scalar in $H^1_{loc}(\mathcal{R}(0,\tau)).  $
	\end{remark}

	\begin{lemma}\label{psi_r Control}
		For any solution $ \Psi_i^{^{(\ell)}},\ \ell\geq i, \ i\in\br{1,2} $ to $ (\ref{waveeq}) $, we have \begin{align}
			\int_{\mathcal{A}_c}(\partial_r\Psi_i)^2\leq \varepsilon\left(\int_{\mathcal{A}_c}E_1(\partial_v\partial_r\Psi_i)^2+E_2(\partial_r\partial_r\Psi_i)^2 \right) + C_{\varepsilon} \left(\int_{\Sigma_{0}}J_{\mu}^{N}[\Psi_i]n_{\Sigma_{0}}^{\mu}+\int_{\Sigma_{0}}J_{\mu}^{N}[T\Psi_i]n_{\Sigma_{0}}^{\mu} \right) 
		\end{align} 
		where the constant $ C_{\varepsilon}>0 $ depends on $ M, r_c, r_d , \Sigma_0,\varepsilon. $
	\end{lemma}
	\begin{proof}
		We apply Lemma \ref{second hardy} for $ \partial_r\Psi $ and regions $ \mathcal{C}\equiv \mathcal{A}_c, \ \mathcal{D} $ and we get \begin{align}
			\int_{\mathcal{A}_c}(\partial_r\Psi_i)^2\leq C \int_{D}(\partial_r\Psi_i)^2+ C\int_{\mathcal{A}_c\cup \mathcal{D}}D\left[(\partial_v\partial_r\Psi_i)^2+(\partial_r\partial_r\Psi_i)^2 \right] 
		\end{align}
		Now, note that $ E_1>0 $ and $ E_2\geq 0 $ close to the horizon, with $ E_2 $  vanishing to first order on $ \mathcal{H}^+. $ On the other hand, $ C\cdot D $ vanish to second order on $ \mathcal{H}^+ $, thus for any $ \varepsilon>0 $ there exists a region close to the horizon $ A_{\varepsilon}:=\br{M\leq r\leq r_{\varepsilon}\leq r_c} $ such that \begin{align}
			\begin{aligned}
				&C\int_{\mathcal{A}_c\cup \mathcal{D}}D\left[(\partial_v\partial_r\Psi_i)^2+(\partial_r\partial_r\Psi_i)^2 \right]
				\\ \leq&  \varepsilon \left (\int_{A_{\varepsilon}\cap A_c} E_1(\partial_v\partial_r\Psi_i)^2+E_2(\partial_r\partial_r\Psi_i)^2\right ) 
				+ C\int_{\mathcal{A}_c\cup\mathcal{D}\cap\br{r\geq r_{\varepsilon}}}D\left[(\partial_v\partial_r\Psi_i)^2+(\partial_r\partial_r\Psi_i)^2 \right] \\
				\leq &\varepsilon \left (\int_{A_{\varepsilon}\cap A_c} E_1(\partial_v\partial_r\Psi_i)^2+E_2(\partial_r\partial_r\Psi_i)^2\right ) 
				+ C_{\varepsilon}\int _{\Sigma_{0}}J_{\mu}^N[\Psi_i]+ C_{\varepsilon} \int _{\Sigma_{0}}J_{\mu}^N[T\Psi_i],
			\end{aligned}
		\end{align} 
		since all second-order derivatives are controlled away from the horizon in terms of the fluxes $ J^N_{\mu}[\Psi], J^N_{\mu}[T\Psi]  $. On the other hand,
		since the region $ \mathcal{D} $ is away from the horizon, Theorem \ref{Nuniform } yields  $$  C \int_{D}(\partial_r\Psi_i)^2 \leq \ \tilde{C} \int_{\Sigma_{0}}J^N_{\mu}[\Psi] n_{\Sigma_{0}}^{\mu},$$
		which concludes the proof.
	\end{proof}

	We now state two more lemmas that allow us to treat boundary terms at the horizon.

	\begin{lemma} \label{Lemma-psi_vpsi_r-control}
		For all solutions  $ \Psi_i^{^{(\ell)}},\ \ell\geq i, \ i\in\br{1,2} $ to $ (\ref{waveeq}) $, we have \begin{align}
			\abs{\int_{\mathcal{H}^+}(\partial_v\Psi_i)(\partial_r\Psi_i)}\leq C_{\varepsilon} \int_{\Sigma_{0}}J^{N}_{\mu}[\Psi_i]n_{\Sigma_{0}}^{\mu} + \varepsilon \int_{\Sigma_{0}\cup \Sigma_{\tau}} J_{\mu}^{\x}[\partial_r\Psi_i] n_{\Sigma}^{\mu},
		\end{align}
		for any $ \varepsilon>0 $ where the constant $ C_{\varepsilon}>0 $ depends only on $ M,\Sigma_0,\varepsilon. $
	\end{lemma}
	\begin{proof}
		Using integration by parts we write \begin{align}
			\int_{\mathcal{H}^{+} (0,\tau)} (\partial_v\Psi)(\partial_r\Psi_r) = - \int_{\mathcal{H}^{+} (0,\tau)} \Psi_i(\partial_v\partial_r\Psi_i)+ \int_{\mathcal{H}^{+}\cap \Sigma_{\tau}}\Psi_i(\partial_r\Psi_i) - \int_{\mathcal{H}^{+}\cap \Sigma_{0}}\Psi_i(\partial_r\Psi_i)  
		\end{align}
		Now, since $ \square \Psi_i -V_i \Psi_i=0 $, on the horizon  $ \br{r=M} $ we have \begin{align}
			2\partial_v\partial_r\Psi_i+ \dfrac{2}{M}\partial_v\Psi_i+ \slashed{\Delta}\Psi_i- V_i(M)\Psi_i=0 
		\end{align}
		Thus, we can write \begin{align}
			\begin{aligned}
				- \int_{\mathcal{H}^{+} (0,\tau)} \Psi_i(\partial_v\partial_r\Psi_i) =& - \int_{\mathcal{H}^{+} (0,\tau)} \Psi_i \left(-\dfrac{1}{M}\partial_v\Psi_i+\dfrac{V_i}{2}\Psi_i - \dfrac{\slashed{\Delta}\Psi_i}{2}\right)\\ =
				& \int_{\mathcal{H}^{+} (0,\tau)}\dfrac{1}{M}\Psi_i\cdot (\partial_v\Psi_i)-\dfrac{V_i}{2}\Psi_i^2 - \dfrac{1}{2}\abs{\slashed{\nabla}\Psi_i}^2
			\end{aligned}
		\end{align} 
		However, using Cauchy Schwartz and the fact that $ V_i $ is linear with respect to $ \ell $, relation (\ref{non-deg_energy})  and Theorem \ref{Nuniform } allow us to control the integral above in terms of the flux $ \int_{\Sigma_{0}}J_{\mu}^{N}[\Psi_i]n^{\mu}_{\Sigma_{0}} $.
		
		Finally, we have \begin{align}
			\int_{\mathcal{H}^+\cap \Sigma}\Psi_i \cdot (\partial_r\Psi_i)\leq \int_{\mathcal{H}^+\cap \Sigma} \Psi_i^2 + \int_{\mathcal{H}^+\cap \Sigma}(\partial_r\Psi_i)^2, 
		\end{align}
		and using Lemma \ref{hardy} for $ \psi\equiv \Psi_i, \partial_r\Psi_i $ we obtain \begin{align}
			\int_{\mathcal{H}^+\cap \Sigma}\Psi_i \cdot (\partial_r\Psi_i) \leq  \varepsilon \int_{\Sigma}J_{\mu}^{\x}[\partial_r\Psi_i]n_{\Sigma}^{\mu} + C_{\varepsilon}\int_{\Sigma}J_{\mu}^{N}[\Psi_i]n_{\Sigma}^{\mu}
		\end{align}
		which concludes the proof.
	\end{proof}
	
	\begin{lemma} \label{Lemma-Control Delta,r Psi on H}
		For all solutions  $ \Psi_i^{^{(\ell)}},\ \ell\geq i, \ i\in\br{1,2} $ to $ (\ref{waveeq}) $, we have 
		\begin{align}
			\abs{\int_{\mathcal{H}^+}	(\slashed{\Delta}\Psi_i)(\partial_r\Psi_i)}\leq C_{\varepsilon} \int_{\Sigma_{0}}J^{N}_{\mu}[\Psi_i]n_{\Sigma_{0}}^{\mu} + \varepsilon \int_{\Sigma_{0}\cup \Sigma_{\tau}} J_{\mu}^{\x}[\partial_r\Psi_i] n_{\Sigma}^{\mu}
		\end{align}
		for any $ \varepsilon>0 $ where the constant $ C_{\varepsilon}>0 $ depends only on $ M,\Sigma_0,\varepsilon. $
	\end{lemma}
	\begin{proof}
		Once again, we use the wave equation (\ref{waveeq}) evaluated on the horizon,\\  $  	2\partial_v\partial_r\Psi_i+ \dfrac{2}{M}\partial_v\Psi_i+ \slashed{\Delta}\Psi_i- V_i(M)\Psi_i=0  $, and we write \begin{align}
			\int_{\mathcal{H}^+}	(\slashed{\Delta}\Psi_i)(\partial_r\Psi_i)= -\dfrac{2}{M}\int_{\mathcal{H}^+}(\partial_v\Psi_i)(\partial_r\Psi_i)-2\int_{\mathcal{H}^+}(\partial_v\partial_r\Psi_i)(\partial_r\Psi_i) + \int_{\mathcal{H}^+}V_i\cdot (\partial_r\Psi_i) \Psi_i \label{R.Lemma5.3}
		\end{align}
		Since $ \Psi_i $ is supported on the $ \ell $-frequency we have $ \slashed{\Delta}\Psi_i=- \frac{\ell(\ell+1)}{r^2}\Psi_i $, so \begin{align}
			\int_{\mathcal{H}^+}	(\slashed{\Delta}\Psi_i)(\partial_r\Psi_i)-V_i(\partial_r\Psi_i)\Psi_i= \int_{\mathcal{H}^+}	\left (1+V_i \dfrac{M^2}{\ell(\ell+1)}\right )(\slashed{\Delta}\Psi_i)(\partial_r\Psi_i) 
		\end{align}
		For $ i=1 $, on the horizon $ \mathcal{H}^+$ we get
		\begin{align}
			1+V_i \cdot \dfrac{M^2}{\ell(\ell+1)}=1-2 \dfrac{(-\ell-1)}{\ell(\ell+1)}=1+\dfrac{2}{\ell}\geq 1, \ \ \forall \ell\geq 1.
		\end{align}
		Similarly, for $ i=2 $, 
		\begin{align}
			1+V_i \cdot \dfrac{M^2}{\ell(\ell+1)}=\dfrac{\ell}{\ell(\ell+1)}=1-\dfrac{2}{\ell+1}\geq 1-\dfrac{2}{3}=\dfrac{1}{3}, \hspace{1cm} \forall \ell\geq 2,
		\end{align}
		Thus, going back to (\ref{R.Lemma5.3}) we have \begin{align}\begin{aligned}
				\dfrac{1}{3}\abs{\int_{\mathcal{H}^+(0,\tau)}	(\slashed{\Delta}\Psi_i)(\partial_r\Psi_i)}\leq& \abs{
					-\dfrac{2}{M}\int_{\mathcal{H}^+(0,\tau)}(\partial_v\Psi_i)(\partial_r\Psi_i) -\int_{\mathcal{H}^+(0,\tau)}\partial_v\left [(\partial_r\Psi_i)^2\right ] }  \\
				=&  \abs{-\dfrac{2}{M}\int_{\mathcal{H}^+(0,\tau)}(\partial_v\Psi_i)(\partial_r\Psi_i) + \int_{\mathcal{H}^+\cap \Sigma_{\tau}}(\partial_r\Psi_i)^2 - \int_{\mathcal{H}^+\cap \Sigma_{0}}(\partial_r\Psi_i)^2}, 
			\end{aligned}
		\end{align}
		so using  Lemma \ref{Lemma-psi_vpsi_r-control} and the first Hardy inequality (Lemma \ref{hardy}) we  conclude the proof.\\
	\end{proof}
	\subsection{Estimates for the Spacetime Terms.}
	We are now ready to control the bulk terms that appear in  (\ref{E_i coefficients}) in terms of the fluxes $ J^{N}_{\mu}[\Psi_i], J^{N}_{\mu}[T\Psi_i]\ $, and $J^{\x}_{\mu}[\partial_r\Psi_i] $ or absorb them in known positive definite terms. Note, since the vectorfield $ \x(r)=0 $ for $ r\geq r_d $, then $ \nabla^{\mu} J_{\mu}^{\x}[\partial_r\Psi_i]=0 $ for $ r\geq r_d $, so we focus our estimates in the region $ \mathcal{R}(0,\tau)\cap \br{r\leq r_d} $ only.
	
	\paragraph{  $$\boxed{\textbf{The term}\hspace{0.5cm} \int_{R} E_{4}\left(\partial_{v} \partial_{r} \Psi_i\right)\left(\partial_{v} \Psi_i\right)}$$}
	
	For any $ \varepsilon>0, $ Cauchy-Schwartz and Theorem \ref{Nuniform } yield
	\begin{align*}
		\begin{aligned}
			\abs{\int_{\mathcal{A}_c}E_4\left(\partial_{v} \partial_{r} \Psi_i\right)\left(\partial_{v} \Psi_i\right)}\leq \varepsilon \int_{\mathcal{A}_c}(\partial_v\partial_r\Psi_i)^2 + \dfrac{(E_4)^2}{\varepsilon}\int_{\mathcal{A}_c}(\partial_v\Psi_i)^2\\
			\leq \varepsilon \int_{\mathcal{A}_c}(\partial_v\partial_r\Psi_i)^2 + C_{\varepsilon}\int_{\Sigma_{0}}J_{\mu}^{N}[\Psi_i]n_{\Sigma_{0}}^{\mu},
		\end{aligned}
	\end{align*}
	for a constant $ C_{\varepsilon} $ depending on $ \varepsilon,M,\Sigma_{0}. $
	
	\paragraph{  $$\boxed{\textbf{The term}\hspace{0.5cm} \int_{R} E_{5}\left(\partial_{v} \partial_{r} \Psi_i\right)\left(\partial_{r} \Psi_i\right)}$$}
	Note,  $ \abs{E_5}= \abs{ \x^v(r)R'(r) } $ is uniformly positive and bounded in the region $ \mathcal{A}_c $, thus we have 
	\begin{align*}
		\abs{\int_{\mathcal{A}_c} E_{5}\left(\partial_{v} \partial_{r} \Psi_i\right)\left(\partial_{r} \Psi_i\right)} \leq \varepsilon \int_{\mathcal{A}_c}(\partial_{v} \partial_{r} \Psi_i)^2 + \dfrac{\max_{\mathcal{A}_c} (E_5)^2}{\varepsilon}\int_{\mathcal{A}_c} (\partial_r\Psi_i)^2
	\end{align*}
	for any $ \varepsilon>0 $. So, applying Lemma \ref{psi_r Control} for the second term above, for $ \tilde{\varepsilon}= \dfrac{\varepsilon^2}{\max_{\mathcal{A}_c} (E_5)^2}>0 $
	we obtain \begin{align}
		\begin{aligned} \label{E_5 spacetime control}
			\abs{\int_{\mathcal{A}_c} E_{5}\left(\partial_{v} \partial_{r} \Psi_i\right)\left(\partial_{r} \Psi_i\right)} \leq & \varepsilon \int_{\mathcal{A}_c}(\partial_{v} \partial_{r} \Psi_i)^2 + \varepsilon\left(\int_{\mathcal{A}_c}E_1(\partial_v\partial_r\Psi_i)^2+E_2(\partial_r\partial_r\Psi_i)^2 \right) \\
			& +  C_{\varepsilon} \left(\int_{\Sigma_{0}}J_{\mu}^{N}[\Psi_i]n_{\Sigma_{0}}^{\mu}+\int_{\Sigma_{0}}J_{\mu}^{N}[T\Psi_i]n_{\Sigma_{0}}^{\mu} \right),
		\end{aligned}
	\end{align}
	for $ C_{\varepsilon}>0 $ depending  on $  M, r_c, r_d , \Sigma_0,\varepsilon. $

	\paragraph{  $$\boxed{\textbf{The term}\hspace{0.5cm} \int_{R} E_{6}\left(\partial_{r} \partial_{r} \Psi_i\right)\left(\partial_{v} \Psi_i\right)}$$}
	Here, we apply divergence theorem for $ P_{\mu}= E_6 (\partial_r\Psi_i)(\partial_v\Psi)(\partial_r)_{\mu} $, in the region $ R(0,\tau). $ 
	\begin{align*}
		\begin{aligned}
			\int_{R}\nabla^{\mu} P_{\mu} = \int_{\Sigma_{0}}P_{\mu}n_{\Sigma_{0}}^{\mu}- \int_{\Sigma_{\tau}}P_{\mu}n_{\Sigma_{\tau}}^{\mu}- \int_{\mathcal{H}^+}P_{\mu}n_{\mathcal{H}^+}^{\mu}.
		\end{aligned}
	\end{align*}
	First, we expand the left-hand side term and using the fact that $ Div(\partial_r)=\frac{2}{r} $ we get \begin{align*}
		\begin{aligned}
			\int_{R}\nabla^{\mu} P_{\mu}= \int_{R} E_6 (\partial_v\Psi_i)(\partial_r\partial_r \Psi_i) + \int_{R}\left (\partial_rE_6+ \dfrac{2E_6}{r}\right )(\partial_v\Psi_i)(\partial_r\Psi_i)\\ + \int_{R}E_6(\partial_v\partial_r\Psi_i)(\partial_r\Psi_i)
		\end{aligned}
	\end{align*}
	Thus, we write \begin{align*}
		\begin{aligned}
			\int_{R} E_6 (\partial_r\partial_r \Psi_i)(\partial_v\Psi_i)=  \int_{\Sigma_{0}}P_{\mu}n_{\Sigma_{0}}^{\mu}- \int_{\Sigma_{\tau}}P_{\mu}n_{\Sigma_{\tau}}^{\mu}- \int_{\mathcal{H}^+}P_{\mu}n_{\mathcal{H}^+}^{\mu} \\ -
			\int_{R}\left (\partial_rE_6+ \dfrac{2E_6}{r}\right )(\partial_v\Psi_i)(\partial_r\Psi_i) -\int_{R}E_6(\partial_v\partial_r\Psi_i)(\partial_r\Psi_i).
		\end{aligned}
	\end{align*}
	The boundary term over the horizon $\mathcal{H}^+  $ can be estimated using Lemma \ref{Lemma-psi_vpsi_r-control}, and the term over $ \Sigma_{\tau} $, using Cauchy Schwartz and Theorem \ref{Nuniform } to estimate it in terms of the N-flux on $ \Sigma_{0}. $
	
	Now, since $ \abs{E_6}  $ is uniformly positive in $ \mathcal{A}_c $, we estimate the spacetime term $ \int_{R}E_6(\partial_v\partial_r\Psi_i)(\partial_r\Psi_i) $ via exactly the same relation as in  (\ref{E_5 spacetime control}).
	
	Finally, we control the remaining spacetime term $ \int_{R}\left (\partial_rE_6+ \frac{2E_6}{r}\right )(\partial_v\Psi_i)(\partial_r\Psi_i) $ by using Cauchy Schwartz, Lemma \ref{psi_r Control} and Theorem \ref{Nuniform }.

	\paragraph{  $$\boxed{\textbf{The term}\hspace{0.5cm} \int_{R} E_{7}\left(\partial_{v} \partial_{r} \Psi_i\right)\slashed{\Delta}\Psi_i }$$}
	Divergence theorem for $ P_{\mu}=E_7(\partial_r\Psi_i)\slashed{\Delta}\Psi_i \cdot (\partial_v)_{\mu} $ in the region $ \mathcal{R}(0,\tau) $ yields 
	\begin{align*}
		\begin{aligned}
			\int_R \nabla^{\mu}P_{\mu}=\int_{\Sigma_{0}}P_{\mu}n_{\Sigma_{0}}^{\mu}- \int_{\Sigma_{\tau}}P_{\mu}n_{\Sigma_{\tau}}^{\mu},
		\end{aligned}
	\end{align*}
	and note there is no boundary term on the horizon $ \mathcal{H}^+ $, since $ (\partial_v)_{\mu}\cdot n_{\mathcal{H}^+}^{\mu}= 0 $
	In particular, we obtain \begin{align*}
		\begin{aligned}
			\int_{R} E_{7}\left(\partial_{v} \partial_{r} \Psi_i\right)\slashed{\Delta}\Psi_i= \int_{\Sigma_{0}}E_7(\partial_r\Psi_i)\slashed{\Delta}\Psi_i \cdot (\partial_v)_{\mu}n_{\Sigma_{0}}^{\mu}- \int_{\Sigma_{\tau}}E_7(\partial_r\Psi_i)\slashed{\Delta}\Psi_i \cdot (\partial_v)_{\mu}n_{\Sigma_{\tau}}^{\mu}\\
			- \int_{R} E_{7}\left( \partial_{r} \Psi_i\right)\slashed{\Delta}(\partial_v\Psi_i)
		\end{aligned}
	\end{align*}
	First, note that since $ V_i'(r) $ is linear with respect to $ \ell $,$\  E_7 $ is uniformly bounded with respect to $ \ell $ in the region $ \mathcal{A}_{c}. $
	
	Away from the horizon, the boundary terms are controlled by the fluxes $ J^N[\Psi_i], J^T[T\Psi_i] $. On the other hand, near the horizon, we have \begin{align*}
		\begin{aligned}
			\abs{\int_{\Sigma\cap \mathcal{A}_c}E_7(\partial_r\Psi_i)\slashed{\Delta}\Psi_i} = \abs{\int_{\Sigma\cap \mathcal{A}_c}E_7\slashed{\nabla}(\partial_r\Psi_i)\slashed{\nabla}\Psi_i} 	 \leq \varepsilon \int_{\Sigma\cap \mathcal{A}_c}J_{\mu}^{\x}[\partial_r\Psi_i]n^{\mu}_{\Sigma}+ C_{\varepsilon}\int_{\Sigma\cap \mathcal{A}_c}J_{\mu}^{N}[\Psi_i]n^{\mu}_{\Sigma}
		\end{aligned},
	\end{align*}
	with $ C_{\varepsilon}>0 $ depending on $ \Sigma_0,M,\varepsilon. $ As far as the last spacetime term, we estimate it as \begin{align*}
		\begin{aligned}
			&\abs{\int_{\mathcal{A}_c} E_{7}\left( \partial_{r} \Psi_i\right)\slashed{\Delta}(\partial_v\Psi_i)} = \abs{\int_{\mathcal{A}_c} E_{7}\slashed{\nabla}\left( \partial_{r} \Psi_i\right)\slashed{\nabla}(\partial_v\Psi_i)} \\
			& \leq \varepsilon \cdot C\int_{\mathcal{A}_c}\left(\slashed{\nabla} \partial_{r} \Psi_i\right)^2 + C_{\varepsilon}\int_{\mathcal{A}_c} \abs{\slashed{\nabla} \partial_{v} \Psi_i}^2 \\ &\leq  \varepsilon \cdot C\int_{\mathcal{A}_c}\left(\slashed{\nabla} \partial_{r} \Psi_i\right)^2 + C_{\varepsilon}	\int_{\Sigma_{0}}J_{\mu}^{N}[T\Psi_i]n_{\Sigma_{0}}^{\mu}.
		\end{aligned}
	\end{align*}

	\paragraph{  $$\boxed{\textbf{The term}\hspace{0.5cm} \int_{R} E_{8}\left(\partial_{r} \partial_{r} \Psi_i\right)\slashed{\Delta}\Psi_i} $$ }
	
	Applying divergence theorem for $ P_{\mu} = E_8\cdot \slashed{\Delta}\Psi_i (\partial_r\Psi_i)(\partial_r)_{\mu} $ in the region $ \mathcal{R}(0,\tau) $ yields 
	\begin{align*}
		\begin{aligned}
			&\int_{R}E_{8}\left(\partial_{r} \partial_{r} \Psi_i\right)\slashed{\Delta}\Psi_i = \int_{\Sigma_{0}} E_8\cdot \slashed{\Delta}\Psi_i (\partial_r\Psi_i)(\partial_r)_{\mu}n_{\Sigma_{0}}^{\mu}- \int_{\Sigma_{\tau}} E_8\cdot \slashed{\Delta}\Psi_i (\partial_r\Psi_i)(\partial_r)_{\mu} n_{\Sigma_{\tau}}^{\mu} \\ -& \int_{\mathcal{H}^+} E_8\cdot \slashed{\Delta}\Psi_i (\partial_r\Psi_i)(\partial_r)_{\mu}n_{\mathcal{H}^+}^{\mu} - \int_{R}\partial_r(E_8\slashed{\Delta}\Psi_i)\partial_r\Psi_i- \int_{R}\dfrac{2E_8}{r}\slashed{\Delta}\Psi_i \partial_r\Psi_i
		\end{aligned}
	\end{align*}
	Again, we first show that $ E_8 $ is uniformly bounded in $ \ell $, since $ V_i'\cdot r^3 $ is linear in $ \ell $ and  \begin{align*}
		\begin{aligned}
			\abs{E_{8}}=\abs{\dfrac{2\x^r}{r}\left(1-\dfrac{V_i'\cdot r^3}{2\ell(\ell+1)}\right)}\leq 2 C\dfrac{\abs{\x^r}}{r},
		\end{aligned}
	\end{align*}
	for a positive constant $ C $.
	
	The boundary  term over $ \mathcal{H}^+ $ is then controlled by Lemma \ref{Lemma-Control Delta,r Psi on H}  and using Cauchy Schwartz we also have in the region $ \mathcal{A}_c $\begin{align*}
		\begin{aligned}
			\int_{\Sigma\cap \mathcal{A}_c}E_8\cdot \slashed{\Delta}\Psi_i (\partial_r\Psi_i) = -\int_{\Sigma\cap \mathcal{A}_c}E_8\cdot \slashed{\nabla}\Psi_i \slashed{\nabla}(\partial_r\Psi_i) \leq \varepsilon \int_{\Sigma\cap \mathcal{A}_c} J_{\mu}^{\x}[\partial_r\Psi_i]n_{\Sigma}^{\mu} + C_{\varepsilon} \int_{\Sigma\cap \mathcal{A}_c} J_{\mu}^{N}[\Psi_i]n_{\Sigma}^{\mu} 
		\end{aligned}
	\end{align*}
	Finally, for the remaining two spacetime integrals, using the fact that $ [\slashed{\Delta},\partial_r]\Psi_i=\frac{2}{r}\Psi_i $, we obtain \begin{align*}
		\begin{aligned}
			-\left (\int_{\mathcal{A}_c}\partial_r(E_8\slashed{\Delta}\Psi_i)\partial_r\Psi_i+ \int_{\mathcal{A}_c}\dfrac{2E_8}{r}\slashed{\Delta}\Psi_i \partial_r\Psi_i\right )=\int_{\mathcal{A}_c}E_8 (\slashed{\nabla}\partial_r\Psi_i)^2+ \int_{\mathcal{A}_c}(\partial_rE_8) \slashed{\nabla}\Psi_i \cdot \slashed{\nabla}\partial_r\Psi_i.
		\end{aligned}
	\end{align*}
	The first integral in the right-hand side, is estimated as long as $ \abs{E_8}\leq E_3,  $ however $ \x $ can be chosen such that \begin{align*}
		\abs{E_8}\leq 2 C \dfrac{\abs{\x^r}}{r}\leq -\dfrac{\x^r_r}{100}<E_3.
	\end{align*}

	For the second integral on the right-hand side, Cauchy Schwartz yields \begin{align*}
		\begin{aligned}
			\int_{\mathcal{A}_c}(\partial_rE_8) \slashed{\nabla}\Psi_i \cdot \slashed{\nabla}\partial_r\Psi_i\leq \varepsilon \int_{\mathcal{A}_c}\abs{\partial_r\slashed{\nabla}\Psi_i}^2 + C_{\varepsilon} \int_{\mathcal{A}_c}(\partial_rE_8)^2  \abs{\slashed{\nabla}\Psi_i}^2,
		\end{aligned}
	\end{align*}
	and $ \partial_rE_8 $ is uniformly bounded in $ \ell $ for the same reason $ E_8 $ was. So, $ J^N[\Psi_i] $ flux controls the latter integral.
	
	\paragraph{  $$\boxed{\textbf{The term}\hspace{0.5cm} \int_{R} E_{9}\left(\partial_{v} \partial_{r} \Psi_i\right)\left(\partial_{r} \partial_{r} \Psi_i\right)} $$ }
	
	Using the wave equations for $ \partial_v\partial_r\Psi_i $ we obtain \begin{align*}
		\begin{aligned}
			\int_{R} E_{9}\left(\partial_{v} \partial_{r} \Psi_i\right)\left(\partial_{r} \partial_{r} \Psi_i\right)= \int_R E_9(\partial_r\partial_r\Psi_i)\left(-\dfrac{1}{2}D(\partial_r\partial_r\Psi_i) -\dfrac{1}{r}\partial_v\Psi_i- \dfrac{1}{2}R(r)\partial_r\Psi_i-\dfrac{1}{2}\slashed{\Delta}\Psi_i +\dfrac{1}{2}V_i\Psi_i\right) 
		\end{aligned}
	\end{align*}
	
	We use  $ \Psi_i= -\frac{r^2}{\ell(\ell+1)}\slashed{\Delta}\Psi_i $ and we rewrite the equation	 above as 
	
	\begin{align} \label{misc E_9 }
		\begin{aligned}
			\int_{R} E_{9}\left(\partial_{v} \partial_{r} \Psi_i\right)\left(\partial_{r} \partial_{r} \Psi_i\right) =& - \int_{R}E_9\dfrac{D}{2}(\partial_r\partial_r\Psi_i)^2-\int_{R}\dfrac{E_9}{r}(\partial_r\partial_r\Psi_i)(\partial_v\Psi_i) \\ &- \int_{R} E_9\dfrac{R(r)}{2}(\partial_r\partial_r\Psi_i)(\partial_r\Psi_i)- \int_{R}\dfrac{E_9}{2}\left(1+\dfrac{V_i\cdot r^2}{\ell(\ell+1)}\right)(\partial_r\partial_r\Psi_i)\slashed{\Delta}\Psi_i 
		\end{aligned}
	\end{align}
	
	The first integral on the right-hand is estimated in $ \mathcal{A}_c $ because of the $ D $ factor that degenerates to second order on the horizon, whereas $ E_2 $ degenerates to first order.
	Similarly, by Cauchy Schwartz, we have\begin{align*}
		\int_{\mathcal{A}_c}E_9 \dfrac{R}{2}(\partial_r\partial_r\Psi_i)(\partial_r\Psi_i)\leq \int_{\mathcal{A}_c}(\partial_r\Psi_i)^2 + \int_{\mathcal{A}_c}(E_9 R(r))^2(\partial_r\partial_r\Psi_i)^2
	\end{align*}
	and since $ R(r)^2 \sim D(r) $ the second integral is estimated, whereas the first integral is estimated by Lemma \ref{psi_r Control}.
	
	Next, note that $ -\dfrac{E_9}{r} \sim D + E_6 $, therefore the second integral at (\ref{misc E_9 }) is already estimated above.
	
	Finally, in $ \mathcal{A}_c $ we have \begin{align*}
		\left(1+\dfrac{V_i\cdot r^2}{\ell(\ell+1)}\right) \leq 2, \ \ \forall \ell\geq i, \ i\in\br{1,2} \\
	\end{align*}
	and $ \abs{E_9} \ll E_3 $ , so  the last term at (\ref{misc E_9 }) can be estimated as in the $ E_8 $ case above.

	\paragraph{  $$\boxed{\textbf{The term}\hspace{0.5cm} \int_{R} E_{10}\left(\partial_{r} \partial_{r} \Psi_i\right)\left(\partial_{r} \Psi_i\right)} $$ }
	Applying the divergence theorem for the current $ P_{\mu}: =E_{10} (\partial_r\Psi_i)^2(\partial_{r})_{\mu}$ in the region $ \mathcal{R}(0,\tau), $ we obtain
	\begin{align}
		\begin{aligned}\label{misc E_10}
			2\int_{R}E_{10}(\partial_r\partial_r\Psi_i)(\partial_r\Psi_i) =& \int_{\Sigma_{0}} E_{10} (\partial_r\Psi_i)^2(\partial_{r})_{\mu}n_{\Sigma_{0}}^{\mu} - \int_{\Sigma_{\tau}} E_{10} (\partial_r\Psi_i)^2(\partial_{r})_{\mu}n_{\Sigma_{\tau}}^{\mu} \\
			- &\int_{R}\left((\partial_rE_{10})+ \dfrac{2}{r}E_{10}\right)(\partial_r\Psi_i)^2 - \int_{\mathcal{H}^+(0,\tau)}E_{10}(\partial_r\Psi_i)^2. 
		\end{aligned}
	\end{align}
	The boundary terms over $ \Sigma_{\tau},\Sigma_{0} $ are estimated by the $ J^N[\Psi_i] $ flux, and the spacetime integral on the right is estimated by Lemma \ref{psi_r Control}. Last, we need to deal with the boundary term on the horizon $ \mathcal{H}^+. $
	In view of $ E_{10}(M)=-\x^r(M)\cdot R'(M)= -\x^r(M)\dfrac{2}{M^2}>0 $, the only way to control this term is to borrow from the horizon boundary term of (\ref{Div partial_rPsi}). \\
	In particular, we use Poincare Inequality \begin{align}
		\abs{\int_{\mathcal{H}^+}\dfrac{E_{10}}{2}(\partial_r\Psi_i^2)}\leq \int_{\mathcal{H}^+}\dfrac{E_{10}}{2}\cdot \dfrac{M^2}{\ell(\ell+1)}\abs{\slashed{\nabla}\partial_r\Psi_i}^2=\int_{\mathcal{H}^+}- \dfrac{\x^r(M)}{\ell(\ell+1)}\abs{\slashed{\nabla}\partial_r\Psi_i}^2
	\end{align}
	Thus, using the relation (\ref{x horizon flux}), it suffices to have \begin{align*}
		\begin{aligned}
			-\dfrac{\x^r(M)}{2}\left(1+ \dfrac{2\cdot (-1)^{i-1}}{(\ell+i-1)}\right)\geq - \dfrac{\x^r(M)}{\ell(\ell+1)} \hspace{0.5cm} \Leftrightarrow \hspace{0.5cm} \ell(\ell+1)\left(1+ \dfrac{2\cdot (-1)^{i-1}}{(\ell+i-1)}\right)\geq 2.
		\end{aligned}
	\end{align*} 
	For $ i=1, $ we have $ \ell(\ell+1)\left(1+\dfrac{2}{\ell}\right)=(\ell+1)(\ell+2)\geq 6 >2  $ for all $ \ell\geq 1. $ \\
	For $ i=2 $, we get $ \ell(\ell+1)\left(1-\dfrac{2}{\ell+1}\right)=\ell(\ell-1)\geq 2  $ for all $ \ell\geq 2. $ 
	
	Note, in the case $ i=2 $, $ \ell=2, $ we need to use the whole quantity (\ref{x horizon flux}) to control the last term of (\ref{misc E_10}). \textbf{Thus, there will be no $ \abs{\slashed{\nabla}\left (\partial_r\Psi_2^{(2)}\right )}^2 $ term in the boundary integral of $ \mathcal{H}^+ $ at (\ref{Div partial_rPsi}).}
	This is a manifestation of a conservation law that holds on the horizon $ \mathcal{H}^+ $ for $ \Psi_2^{(2)} $, which is supported on the $ \ell=2 $ frequency; see Section \ref{Conservation Section}. We see later on that an analogous situation occurs when controlling higher-order estimates.

	\paragraph{  $$\boxed{\textbf{The term}\hspace{0.5cm} \int_{R} E_{11}\left(\partial_{r}  \Psi_i\right)^2 }$$ }
	Let us first examine the coefficient $ E_{11} $. It's straightforward computations to check that \begin{align*}
		V_i'(r)=-\dfrac{3}{r}V_i(r)- \dfrac{6M^2}{r^5}, \hspace{1cm} \text{for both }\ i=1,2.
	\end{align*} 
	Thus, we have \begin{align}
		\begin{aligned}
			E_{11}(r)=-\dfrac{1}{2}\x(V_i)-V_i\left (\dfrac{\x_r^r}{2}+\dfrac{\x^r}{r}\right)=&-\dfrac{1}{2}\x^r\cdot V_i'(r)-V_i\left (\dfrac{\x_r^r}{2}+\dfrac{\x^r}{r}\right) \\
			= &\ \ \dfrac{1}{2}V_i\left(\dfrac{\x^r}{r}-\x^r_{r}\right) + 3\dfrac{M^2}{r^5}\x^r. 
		\end{aligned}
	\end{align}
	Note $ V_i $ depends on $ \ell, $ thus we cannot use Lemma \ref{psi_r Control} for an estimate since the constants will depend on $ \ell $ as well. Instead, we will show that  $ \int_R E_{11}(\partial_r\Psi_i)^2 $ is positive definite for $ i=1, $ and while non-positive for $ i=2 $, Poincare inequality suffices to estimate it.
	In particular, \begin{itemize}
		\item For $ i=1 $, we have \begin{align*}\begin{aligned}
				E_{11} &= \dfrac{M}{r^3}\left (\ell+1-3\sqrt{D}\right )\left (\dfrac{\x^r}{r}-\x^r_{r}\right )+ \dfrac{3M^2}{r^5}\x^r \\ 
				& = \dfrac{M}{r^3}\left(-\x_r^r\big(\ell+1-3\sqrt{D}\big)+ \dfrac{\x^r}{r}\big (4-6\sqrt{D}+\ell\big )\right) 
			\end{aligned}
		\end{align*}
		In $ \mathcal{A}_c, $ we chose $ -\x_r^r(r)\gg -\dfrac{\x^r(r)}{r} $, thus  \begin{align*}
			\begin{aligned}
				E_{11}>& -\dfrac{\x^r}{r}\dfrac{M}{r^3}\left(25(\ell+1-3\sqrt{D})-4+6\sqrt{D}-\ell\right)\\
				&=- \dfrac{\x^r}{r}\dfrac{M}{r^3}\big( 24 \ell +21 - 69\sqrt{D} \big)   
			\end{aligned}
		\end{align*}
		which is positive definite for all $ \ell\geq 1 $ in $\mathcal{A}_c.$
		\item For $ i=2, $ \begin{align}
			\begin{aligned}
				E_{11} =&  -\dfrac{M}{r^3}(3\sqrt{D}+\ell)\left(\dfrac{\x^r}{r}-\x_r^r\right)+ 3\dfrac{M^2}{r^4}\dfrac{\x^r}{r} \\
				=& \dfrac{M}{r^3}\left((3\sqrt{D}+\ell)\x_r^r + \dfrac{\x^r}{r}(3-6\sqrt{D}-\ell)\right). 
			\end{aligned}
		\end{align}
		The dominant term $ \x_r^r $ has the wrong sign, so using Poincare inequality and borrowing a fraction $ \dfrac{2}{3}<\alpha<1 $ from $ E_3 \abs{\slashed{\nabla}(\partial_r\Psi_2)}^2 $ we obtain \begin{align}
			\begin{aligned}
				\int_{S^2(r)} E_{11} (\partial_r\Psi_2)^2 + \alpha E_3 \abs{\slashed{\nabla}(\partial_r\Psi_2)}^2 \geq  \left(E_{11} + \dfrac{\alpha\cdot \ell(\ell+1)}{r^2}E_3 \right) \int_{S^2(r)} (\partial_r\Psi_2)^2	\end{aligned}
		\end{align}
		We will show that the coefficient of the right-hand side above, is positive definite uniformly in $ \ell $ in the region $\mathcal{A}_c.  $ Indeed, for $ \ell\geq L $ where $ L $ is a fixed large enough integer, 
		\begin{align} \label{misc E11}
			\begin{aligned}
				&E_{11} + \dfrac{\alpha\cdot \ell(\ell+1)}{r^2}E_3= \\ &=  \dfrac{1}{r^2}\left((1-x)\left((3\sqrt{D}+\ell)\x_r^r + \dfrac{\x^r}{r}(3-6\sqrt{D}-\ell)\right) -\alpha \dfrac{\ell(\ell+1)}{2}\x^r_r \right) 
			\end{aligned}
		\end{align}
		and the $ \x_r^r $ term's coefficient has the right sign everywhere in $ R(0,\tau) $, due to the large value of $ -\ell(\ell+1). $ Thus, for $ \ell\geq L $, the term above is positive definite everywhere. For, $ 2\leq \ell \leq L-1 $, we evaluate the expression (\ref{misc E11}) on the horizon $ \mathcal{H}^+ $ and we obtain
		\begin{align}
			\begin{aligned}
				E_{11} + \dfrac{\alpha\cdot \ell(\ell+1)}{r^2}E_3 \Big|_{\br{r=M}}=\dfrac{1}{M^2}\left( \dfrac{\ell}{2}\left( 2 -\alpha(\ell+1)\right) \x_r^r(M)+\dfrac{\x^r(M)}{M}(3-\ell)\right), 
			\end{aligned}
		\end{align}
		and for any $ 2/3 < \alpha<1 $ the above expression is positive for all $ 2\leq \ell\leq L-1 $ in view of our choice  $ -\x_r^r(M)\gg -\frac{\x^r(M)}{M} . $ 
		Thus, for each $ 2\leq \ell\leq L-1 $ there exists a neighborhood $ A_\ell $ of the horizon where (\ref{misc E11}) is positive definite. By, choosing the intersection of all these neighborhoods we obtain a region including $ \mathcal{A}_c $ where (\ref{misc E11}) is positive definite for every $ \ell \geq 2 $ in the case $ i=2.  $
	\end{itemize}
	
	\paragraph*{Conclusion.} As we can see from the analysis above, it suffices to consider $ -\x^r_r(M) $, $ \x^v_r(M) $ sufficiently large, $ r_d<2M $ and $ r_c $  close enough to the horizon $ \mathcal{H}^+ \equiv\lbrace r=M\rbrace $ in order to obtain the following  estimate \begin{theorem} \label{remove_degeneracy}
		There exists a positive constant $ C >0 $ depending only on $ \Sigma_0, M $ such that for all solutions $ \Psi_i^{^{(\ell)}}$ to (\ref{waveeq}), supported on the fixed frequency $ \ell\geq i, \ i\in \lbrace 1,2 \rbrace $,  the following estimate hold \begin{align}
			\begin{aligned}\label{H^2-control}
				& \int_{\Sigma_{\tau} \cap \mathcal{A}_c}\left(\partial_{v} \partial_{r} \Psi_i\right)^{2}+\left(\partial_{r} \partial_{r} \Psi_i\right)^{2}+\left| \slashed{\nabla} \partial_{r} \Psi_i\right|^{2} +(\partial_r\p)^2 + \int_{\mathcal{H}^{+}}\left(\partial_{v} \partial_{r} \Psi_i\right)^{2}+ \bm{\chi}_{i,2}\left|\slashed{\nabla}  \partial_{r} \Psi_i\right|^{2} \\
				&+\int_{\mathcal{A}_c}\left(\partial_{v} \partial_{r} \Psi_i\right)^{2}+\sqrt{D}\left(\partial_{r} \partial_{r} \Psi_i\right)^{2}+\left|\slashed{\nabla} \partial_{r} \Psi_i\right|^{2} + (\partial_r\Psi_i)^2 \\
				\leq & C \left(  \int_{\Sigma_{0}} J_{\mu}^{n_{\Sigma_0}}[\Psi_i] n_{\Sigma_{0}}^{\mu}+ \int_{\Sigma_{0}} J_{\mu}^{n_{\Sigma_0}}[T \Psi_i] n_{\Sigma_{0}}^{\mu}+ \int_{\Sigma_{0}\cap \mathcal{A}_c}  J_{\mu}^{n_{\Sigma_0}}\left[\partial_{r} \Psi_i\right] n_{\Sigma_0}^{\mu}\right),
			\end{aligned}
		\end{align}
		where $ \bm{\chi}_{i,2}  = \begin{cases}
			0, \ \text{if}\ \ i=\ell =2 \\
			1,\ \text{otherwise}.
		\end{cases}. $ 
	\end{theorem}
	\begin{proof}
		With $ \x  $ chosen as described in the conclusion
		, we apply the energy identity (\ref{Div partial_rPsi}) and using estimates (\ref{H^2(Sigma)}, \ref{H^2(A)}) while also all bulk estimates controlling $ E_{4-11} $  terms for an $ \varepsilon >0 $ small enough, we obtain a constant $ C>0 $ depending on $ M,\Sigma_0,r_c, r_{d}, \x $ satisfying (\ref{H^2-control}). However, $ \x,r_c $ and $ r_d $ depend solely on $ M $ and fixing specific values satisfying the restrictions of the conclusion we obtain a constant $ C $ that depends only on $ M, \Sigma_0. $
		
		Note, the right-hand side of (\ref{H^2-control}) is justified since $ \x \sim n_{\Sigma_{\tau}} $ in the region $ \mathcal{A}_c $. In addition, we observe the degeneracy $ \bm{\chi}_{2,2} $ at the horizon for the angular term of $ \Psi_2^{^{(\ell=2)}} $ only. We will see a similar degeneracy holds for higher order transversal derivatives of $ \Psi_i^{^{(\ell)}} $, $ \ell \geq i, \ i\in \br{1,2} , $ close to $ \mathcal{H}^+ $ . \\
	\end{proof}

	\section{Conservation Laws Along the Event Horizon.} \label{Conservation Section} In this section we show that solutions $ \Psi_i^{^{(\ell)}} $ to (\ref{waveeq}) satisfy conservation laws along the event horizon $ \mathcal{H}^+. $ The results developed here are used later on not only to produce higher order $ L^2 $ estimates but also to prove non-decay and growth of solutions to (\ref{waveeq}) asymptotically along the event horizon. 
	
	We start with a low-frequency example that will motivate the idea for the proof of the general case. Let us consider the solution to (\ref{waveeq}), $\psi:= \Psi_2^{(2)} $, i.e. $ i=2, \ell =2 ,$ and with respect to the ingoing  coordinates $ (v,r,\vartheta,\varphi) $ the wave equation for $ \psi $ reads
	\[ D \partial_r\partial_r\psi + 2 \partial_v\partial_r\psi + \dfrac{2}{r}\partial_v\psi +R(r)\partial_r\psi + \slashed{\Delta}\psi - V_2 \psi =0,\]
	where $  R = D' + \frac{2D}{r} $ and $ V_2 (r)= -\frac{2M}{r^3}(3\sqrt{D}-2) $.
	Since $ \psi $ is supported on the fixed frequency $ \ell=2 $ we have $ \slashed{\Delta}\psi= -\frac{6}{r^2}\psi $, and in view of $ D(M)=R(M)=0 $, the wave
	equation on the horizon $ \mathcal{H}^+ $ yields \begin{align}
		T\left (M^2\partial_r\psi+M\psi \right ) = \psi.	\label{con-misc1}
	\end{align}
	Now, we take a $ \partial_r-$derivative of the wave equation (\ref{waveeq}), i.e. $ \partial_{r}\left(\square_{g} \psi-V_2\psi\right)=0  $ and we obtain \begin{align*}
		\begin{aligned}
			D \partial_{r} \partial_{r} \partial_{r} \psi+2 T \partial_{r} \partial_{r} \psi+\frac{2}{r} \partial_{r} T \psi+R \partial_{r} \partial_{r} \psi+\partial_{r} \slashed{\Delta} \psi+D^{\prime} \partial_{r} \partial_{r} \psi-\frac{2}{r^{2}} T \psi\\ +R^{\prime} \partial_{r} \psi - (\partial_rV_2)\psi - V_2\partial_r\psi\  =\  0. 
		\end{aligned}
	\end{align*}
	Evaluating the equation above on the horizon $ \mathcal{H}^+ $ yields 
	\begin{align}
		T\left(\partial_r\partial_r\psi + \dfrac{1}{M}\partial_r\psi-\dfrac{1}{M^2}\psi\right)  + \dfrac{3}{M^3}\psi=0.    \label{con_misc2}
	\end{align} 
	In particular, note that the coefficient of $ \partial_r\psi $ is zero since \[ \left (R'(M)-V_2(M) - \dfrac{6}{M^2}\right )\partial_r\psi = \left(\dfrac{2}{M^2}+ \dfrac{4}{M^2}-\dfrac{6}{M^2}\right) \partial_r\psi =0.\]
	Now, using both (\ref{con-misc1},\ref{con_misc2}) we obtain \begin{align}
		\begin{aligned}
			T\left(\partial_r\partial_r\psi + \dfrac{1}{M}\partial_r\psi-\dfrac{1}{M^2}\psi\right)  + \dfrac{3}{M^3}T\left (M^2\partial_r\psi+M\psi \right ) &= 0 \\
			\Rightarrow\ \  T\left(\partial_r\partial_r\psi + \dfrac{4}{M}\partial_r\psi+\dfrac{2}{M^2}\psi  \right) \ &= \ 0.
		\end{aligned}
	\end{align}
	Since the vectorfield $ T $ is tangent to the null generators of $ \mathcal{H}^+ $, we deduce from above that the expression \[ H_{2}[\Psi_2]:= \partial_r\partial_r\Psi_2^{(2)} + \dfrac{4}{M}\partial_r\Psi_2^{(2)}+\dfrac{2}{M^2}\Psi_2^{(2)} \]
	is \textbf{conserved} along the null geodesics of $ \mathcal{H}^+. $
	
	Similarly, one can show that  for $ i=1, \ell=1 $ the corresponding conserved quantity along the horizon $ \h $ is \begin{align*}
		H_1[\Psi_1]:=\partial_{r}^{3} \Psi_1^{(1)}+\frac{11}{M} \partial_{r}^{2} \Psi_1^{(1)}+\frac{29}{3 M^{2}} \partial_{r} \Psi_1^{(1)}-\frac{19}{3 M^{3}} \Psi_1^{(1)},
	\end{align*}  \\
	In fact, we show below that for any $ \ell\geq i, \ i \in\br{1,2} $, there is a similar expression in terms of $ \Psi_i^{^{(\ell)}} $ denoted by $ H_{\ell}[\Psi_i] $ which is conserved along the null geodesics of $ \mathcal{H}^+. $  
	\begin{remark}
		The notation $ H_{\ell}[\Psi_i] $ comes from \cite{aretakis2011stability} where such conserved quantities were discovered by Aretakis in the case of homogeneous wave equations $ \square_g\psi=0 $, and they usually go by the name of Aretakis constants.
	\end{remark}
	Before we proceed to the general case, we first need the following proposition.
	\begin{proposition} \label{Span-prop}
		For any solution $ \Psi_i^{^{(\ell)}}$ to (\ref{waveeq}), supported on the fixed frequency $ \ell\geq i, \ i\in\br{1,2}  $ and  for any $ 0\leq s \leq (\ell-1) +(-1)^{i+1} $ we have \begin{align}
			\partial_r^{s}\Psi_i  \in span \Big\lbrace  T \Big(M^{j+1}\cdot\partial_r^j\Psi_i\Big) \ \ \big| \ \ 0\leq j \leq s + 1 \Big\rbrace
		\end{align} 
	\end{proposition}
	\begin{proof}
		Let us fix a frequency $ \ell\geq i. $ Taking $ s- $many $ \partial_r $ derivatives of (\ref{waveeq}), i.e. $ \partial_r^{s}(\square_g\Psi_i-V_i\Psi_i)=0 $,  yields the equation \begin{align}
			\begin{aligned}
				&D\left(\partial_{r}^{s+2} \Psi_i\right)+2 \partial_{r}^{s+1} T \Psi_i+\frac{2}{r} \partial_{r}^{s} T \Psi_i+R \partial_{r}^{s+1} \Psi_i+\partial_{r}^{s} (\slashed{\Delta}\Psi_i) +\sum_{j=1}^{s}\binom{s}{j} \partial_{r}^{j} D \cdot \partial_{r}^{s-j+2} \Psi_i  \\
				&+\sum_{j=1}^{s}\binom{s}{j} \partial_{r}^{j} \left( \frac{2}{r} \right)  \cdot \partial_{r}^{s-j} T \Psi_i+\sum_{j=1}^{s}\binom{s}{j} \partial_{r}^{j} R \cdot \partial_{r}^{s-j+1} \Psi_i -  \sum_{j=0}^{s}\binom{s}{j} \partial_{r}^{j} V_i \cdot \partial_{r}^{s-j} \Psi_i \ = \ 0. \label{k-wave}
			\end{aligned}
		\end{align}
		In view of $ D(M)=R(M)=0 $ and $  \slashed{\Delta}\Psi_i =- \frac{\ell(\ell+1)}{r^2}\Psi_i\  $, evaluating the above equation on the horizon $ \mathcal{H}^+ $ allows us to express the top order $ \partial_r- $term, $ \partial_r^s \Psi_i$, in terms of lower order $ \partial_r-$ derivatives and $ T(\partial_r^m\Psi_i) $, for $ 0\leq m\leq s+1, $ as long as the coefficient of $ \partial_r^s\Psi_i $ does not vanish. The coefficient of $ \partial_r^s\Psi_i $ on the horizon $ \mathcal{H}^+ $ is given by \begin{align}
			\begin{aligned}
				\binom{s}{2}D''(r)+\binom{s}{1}R'(r)- \dfrac{\ell(\ell+1)}{M^2} + \dfrac{2}{M^2}(-1)^{i}(\ell+2-i)= \\ = \begin{cases}
					\dfrac{1}{M^2}\Big(s(s+1)-(\ell+1)(\ell+2) \Big), \hspace{1cm} i=1 \\
					\\
					\dfrac{1}{M^2}\Big(s(s+1)-\ell(\ell-1) \Big), \hspace{1.9cm} i=2 \label{top-order coefficient}
				\end{cases}
			\end{aligned}
		\end{align}
		By assumption we have $ s\leq (\ell-1)+(-1)^{i+1} $, thus we see that the coefficient of $ \partial_r^s\Psi_i $ is not zero. Repeating the same procedure consecutively
		for all lower order terms $ \partial_r^k\Psi_i $, $ 0\leq k< s, $ we finally express $ \partial_r^s\Psi_i $ only in terms of $ T(\partial_r^j\Psi_i) $, $ 0\leq j \leq s+1 $.

	\end{proof}
	
	\begin{theorem}\label{conservation horizon}
		Let $ \ell \in \mathbb{N} $, with $ \ell\geq i, \ i\in\br{1,2} $, then there exist constants $ c_i^j$,  $\ j=0,1,\dots,\ell+(-1)^{i+1} $ depending on $ M,\ell, $ such that for all solution $ \Psi_i^{^{(\ell)}} $ to (\ref{waveeq}), supported on the fixed frequency $ \ell, $  the quantities
		\begin{align}
			H_{\ell}[\Psi_1]&\ = \ \partial_r^{\ell+2}\Psi_1^{^{(\ell)}} + \sum_{j=0}^{\ell+1} c_1^j \cdot \partial_r^j\Psi_1^{^{(\ell)}}  \\
			H_{\ell}[\Psi_2]&\ =\  \partial_r^{\ell}\Psi_2^{^{(\ell)}} + \sum_{j=0}^{\ell-1} c_2^j \cdot \partial_r^j\Psi_2^{^{(\ell)}}  
		\end{align}
		are conserved along the null generators of $ \mathcal{H}^+. $
	\end{theorem}
	\begin{proof}
		We consider (\ref{k-wave}) for $ s=\ell+1 $ when $ i=1 $, and $ s=\ell-1 $ when $ i=2 $ and we evaluate it on the horizon  $ \mathcal{H}^+. $  In view of (\ref{top-order coefficient}) and $ D(M) =R(M)=0$, we have that the terms $ \partial_r^{\ell+3}\Psi_1,\ \partial_r^{\ell+2}\Psi_1, \ \partial_r^{\ell+1}\Psi_1 $, while also  $ \partial_r^{\ell+1}\Psi_2,\ \partial_r^{\ell}\Psi_2, \ \partial_r^{\ell-1}\Psi_2 $ in the respected wave equation, vanish on $ \mathcal{H}^+. $ 
		
		In addition, using Proposition \ref{Span-prop} we can express each $ \partial_r^p\Psi_1, \partial_r^q\Psi_2 $, for $ p\leq \ell, \ q\leq \ell-2 $ as a linear combination of $ T(\partial_r^m\Psi_i) $, $\ m\leq \ell+(-1)^{i+1} $ with coefficients depending on $ M,\ell $. Thus, we obtain altogether \begin{align}
			\begin{aligned}
				T(\partial_r^{\ell+2}\Psi_1 ) + \sum_{j=0}^{\ell+1}c_1^j\cdot T(\partial_r^j\Psi_1) &\ =\ 0, \\
				T(\partial_r^{\ell}\Psi_2 ) + \sum_{j=0}^{\ell-1}c_2^j\cdot T(\partial_r^j\Psi_2) &\ =\ 0,
			\end{aligned}
		\end{align}
		for $ c_i^j $ depending on $ M,\ell $ alone, which concludes the proof. \\ \\
	\end{proof}
	
	\section{Higher order \texorpdfstring{$ L^2 $}{PDFstring} Estimates.}\label{EE higher} Similarly to Section \ref{EE first},  we derive energy estimates for $ \partial^k_r\Psi_i^{^{(\ell)}} $, for all $ k\leq \ell + (-1)^{i+1}$, $\ \ell \geq i, \ i\in \br{1,2}. $ We do so, by repeatedly commuting our wave equation (\ref{waveeq}) with the vector field $ \partial_r $ and using induction. 
	By taking $ \partial_r^k $-derivatives of (\ref{waveeq}) we obtain \begin{align*}
		\partial_r^k(\square \Psi_i - V_i\Psi_i )&=0  \\
		\Rightarrow \ \ \square\left(\partial_r^k \Psi_i\right) + [\partial_r^k,\square]\Psi_i - \partial_r^k \left(V_i\cdot \Psi_i\right)&=0 \\
		\Rightarrow \ \   \square \left(\partial_r^k\Psi_i\right) -V_i \left (\partial_r^k\Psi_i\right ) &= [\square,\partial_r^k]\Psi_i + \sum_{j=1}^{k} \binom{k}{j} \partial_r^{j}V_i \cdot \partial_r^{k-j}\Psi_i =: M^k[\Psi_i].
	\end{align*} 
	where the Laplace Beltrami operator commutes with $ \partial_r^k $ according to \[ \begin{aligned}
		\left[\square_{g}, \partial_{r}^{k}\right] \Psi_i=&-\sum_{j=1}^{k}\binom{k}{j} \partial_{r}^{j} D \cdot \partial_{r}^{k-j+2} \Psi_i-\sum_{j=1}^{k}\binom{k}{j} \partial_{r}^{j} \frac{2}{r} \cdot T \partial_{r}^{k-j} \Psi_i \\
		&-\sum_{j=1}^{k}\binom{k}{j} \partial_{r}^{j} R \cdot \partial_{r}^{k-j+1} \Psi_i-\sum_{j=1}^{k}\binom{k}{j} r^{2} \partial_{r}^{j} r^{-2} \cdot \slashed{\Delta} \partial_{r}^{k-j} \Psi_i
	\end{aligned} \]
	with $ R = \partial_rD + \frac{D}{2r}. $ \\ Since $ \Psi_i^{^{(\ell)}} $ is supported on the fixed frequency $ \ell, $ so is $ \partial_r^s\Psi_i $ for any $ s\geq 0, $ and using $$ \slashed{\Delta}(\partial_r^s\Psi_i)=-\frac{\ell(\ell+1)}{r^2}\partial_r^s\Psi_i, $$ we can absorb the potential term of $ M^k[\Psi_i] $ in the last term of $ 	\left[\square_{g}, \partial_{r}^{k}\right] \Psi_i $. In particular, it's a direct computation to check by induction that for any $ j\geq 0 $ we have \begin{align}
		\begin{aligned}
			\partial_r^j V_i(r) & =\ (-1)^j \dfrac{(j+2)!}{r^j}\left(\dfrac{1}{2}V_i(r)+j\dfrac{M^2}{r^4}\right)\\
			& =\ (j+2)\ r^2\cdot \partial_r^j(r^{-2})\left(\dfrac{1}{2}V_i(r)+j\dfrac{M^2}{r^4}\right),
		\end{aligned}
	\end{align}
	thus, we can write
	\begin{align}
		\begin{aligned}
			M^k[\Psi_i] = &-\sum_{j=1}^{k}\binom{k}{j} \partial_{r}^{j} D \cdot \partial_{r}^{k-j+2} \Psi_i-\sum_{j=1}^{k}\binom{k}{j} \partial_{r}^{j} \frac{2}{r} \cdot T \partial_{r}^{k-j} \Psi_i \\
			&-\sum_{j=1}^{k}\binom{k}{j} \partial_{r}^{j} R \cdot \partial_{r}^{k-j+1} \Psi_i-\sum_{j=1}^{k}\binom{k}{j} r^{2} \partial_{r}^{j} r^{-2} A_{j,\ell}(r)\cdot   \slashed{\Delta} \partial_{r}^{k-j} \Psi_i
		\end{aligned}
	\end{align}
	where \[ A_{j,\ell}(r) := 1+\dfrac{(j+2)}{2\ell(\ell+1)}r^2V_i+ \dfrac{j(j+2)}{\ell(\ell+1)}\left (\dfrac{M}{r}\right )^2 \]
	
	\paragraph{Energy Identity.} Now that we have the wave equation for $ \partial_r^k\Psi_i $,  i.e.  $ (\square - V_i)\partial_r^k\Psi_i = M^k[\Psi_i] $, 
	we can proceed with the energy identity of the current $ J_{\mu}^{\x_k}[\partial_r^k\Psi_i] $ for an appropriate vector field $ \x_k. $
	
	In this section, we are only interested in obtaining higher-order estimates in a region close to the horizon $ \mathcal{H}^+. $ However, if we were to stay away from the horizon and the photon sphere,  we can easily control the $ L^2 $ norm of all k-derivatives by using local elliptic estimates, commuting the wave equation (\ref{waveeq}) with $ T^j, \ j \in \br{1,\dots ,k-1}$ and using the Morawetz estimates from Theorem \ref{Morawetz-Spacetime}, to obtain \begin{align}
		\int_{R(0,\tau)\cap\br{r_c\leq r\leq r_d}} \abs{\partial^{a}\Psi_i}^2 \leq C\int_{\Sigma_0}\left( \sum_{j=0}^{k-1}J_{\mu}^T[T^j\Psi_i]n_{\Sigma_0}^{\mu} \right) 
	\end{align}
	for any $ r_c,r_d $ satisfying  $ M<r_c<r_d<2M, $ and $ \abs{a}= k, $ where the constant $ C $ depends on $ M, \Sigma_0  $ and $ k. $ Thus, we restrict our attention to a region $ R(0,\tau)\cap\br{M\leq r\leq r_c} $ for $ M<r_c<2M. $  In particular, we prove the following theorem 
	\begin{theorem} \label{high-order horizon estimates} There exists $ r_c\in (M,2M) $ and  a positive constant $ C >0 $ depending only on $  M,\ \Sigma_{0}\ \text{and} \ k$ such that in the region $ \mathcal{A}_c = R(0,\tau)\cap \br{M\leq r\leq r_c} $, for all solutions $ \Psi_i^{^{(\ell)}}$  of  (\ref{waveeq}) supported on the fixed frequency  $\ \ell\geq i ,\  i\in \lbrace 1,2 \rbrace$,  we have \begin{align}
			\begin{aligned}\label{H^2- higher control}
				& \int_{\Sigma_{\tau} \cap \mathcal{A}_c}\left(T \partial_{r}^k \Psi_i\right)^{2}+\left( \partial_{r}^{k+1} \Psi_i\right)^{2}+\left| \slashed{\nabla} \partial_{r}^k \Psi_i\right|^{2} + (\partial_r^k\p)^2\ +\ \int_{\mathcal{H}^{+}(0,\tau)}\left(T \partial_{r}^k \Psi_i\right)^{2}+ \bm{\chi}_{i,k}\left|\slashed{\nabla}  \partial_{r}^k \Psi_i\right|^{2} \\
				&\hspace{4cm}\quad +\int_{\mathcal{A}_c}\left(T \partial_{r}^k \Psi_i\right)^{2}+\sqrt{D}\left(\partial_{r}^{k+1}  \Psi_i\right)^{2}+\left|\slashed{\nabla} \partial_{r}^k \Psi_i\right|^{2} + (\partial_r^k\Psi_i)^2 \\
				& \hspace{6cm}\leq  \ C \left(  \sum_{j=0}^{k} \int_{\Sigma_{0}} J_{\mu}^{n_{\Sigma_{0}}}\left[T^{j} \Psi_i\right] n_{\Sigma_{0}}^{\mu}+ \sum_{j=1}^{k} \int_{\Sigma_{0} \cap \mathcal{A}_c} J_{\mu}^{n_{\Sigma_{0}}}\left[\partial_{r}^{j} \Psi_i\right] n_{\Sigma_{0}}^{\mu}\right),
			\end{aligned}
		\end{align}
		for any $ k \leq \ell + (-1)^{i+1} $, where here we define  $ \bm{\chi}_{i,k}  := \begin{cases}
			0, \ \  \text{if}\ \  k=\ell + (-1)^{i+1} \\
			1,\ \ \text{otherwise}.
		\end{cases}. $ 
	\end{theorem}
	\begin{proof}
		We prove this by induction on the number of derivatives $ k\leq \ell + (-1)^{i+1} $. In particular, assume that for any  $ 0\leq s\leq (\ell-1)+(-1)^{i+1}   $  the estimate
		(\ref{H^2- higher control}) holds for all $ k\leq s $, then we will prove that it also holds for $ s= \ell + (-1)^{i+1}. $ 
		
		Note,  both base cases $ s=0,\ s= 1 $ correspond to the results of Theorem  \ref{Nuniform } and Theorem  \ref{remove_degeneracy} respectively.
		
		Let us work with an $ \phi_{\tau}^T $-invariant timelike vector field $ \x_k = \x_k^v \partial_v + \x_k^r \partial_r$ with constraints that will be apparent in the end, such that it vanishes in  $ r\geq r_d   $, for a choice of $ r_d,r_c $ with $ 2M>r_d>r_c>M $ determined later. Then, the energy identity for the current $ J_{\mu}^{\x_k}[\partial_r^k\Psi_i] $ is \begin{align}
			\int_{\Sigma_{\tau}}J_{\mu}^{\x_k}[\partial_r^k\Psi_i] n^{\mu}_{\Sigma_{\tau}}+ \int_{\mathcal{R}(0,\tau) } \nabla^{\mu} J_{\mu}^{\x_k}[\partial_r^k\Psi_i] +  \int_{\mathcal{H}^+}J_{\mu}^{\x_k}[\partial_r^k\Psi_i] n^{\mu}_{\mathcal{H}^+}= \int_{\Sigma_{0}}J_{\mu}^{\x_k}[\partial_r^k\Psi_i] n^{\mu}_{\Sigma_{0}}, \label{k-order identity}
		\end{align}
		where we use the abbreviated notation $ \mathcal{H}^+ \equiv \mathcal{H}^+(0,\tau) .\  $
		Regarding the energy fluxes we have
		\begin{align}
			\begin{aligned}
				\int_{\Sigma}J_{\mu}^{\x_k}[\partial_r^k\Psi_i] n^{\mu}_{\Sigma_{\tau}} \geq C \left(\int_{\Sigma} (\partial_v\partial_r^k\Psi_i)^2 + (\partial_r\partial_r^k\Psi_i)^2 + \abs{\slashed{\nabla}\partial_r^k\Psi_i}^2 + \dfrac{1}{r^2}(\partial_r^k\Psi_i)^2 \right), 
			\end{aligned}
		\end{align}
		for a constant $ C>0 $ depending only on $ M,\Sigma_0,\x_k $. In addition, \textbf{on the horizon} $ \mathcal{H}^+ $ we have \begin{align}
			\begin{aligned} \label{x horizon flux k order}
				\int_{\mathcal{H}^+}J_{\mu}^{\x_k}[\partial_r^k\Psi_i] n^{\mu}_{\mathcal{H}^+}= \int_{\mathcal{H}^+} \x_k^v(M)(\partial_v\partial_r^k\Psi_i)^2-\dfrac{\x_k^r(M)}{2}\left(1+ \dfrac{2\cdot (-1)^{i-1}}{(\ell+i-1)}\right)\abs{\slashed{\nabla}\partial_r^k\Psi_i}^2 
			\end{aligned}
		\end{align}
		where $ \ell $ is the frequency support of $ \Psi_i^{^{(\ell)}}. $ Note the expression above is positive definite since $  \x_k $ is timelike, i.e. $ \x_k^r<0. $
		
		We now focus on the bulk term of (\ref{k-order identity}) and similarly to Section \ref{EE first} we have
		\begin{align}
			\begin{aligned}
				\nabla^{\mu} J_{\mu}^{\x_k}[\partial_{r}^k\Psi_i]&= \dfrac{1}{2}Q[\partial_{r}^k \Psi_i]\cdot ^{(\x_k)}\pi - \dfrac{1}{2}\x_k(V_i)(\partial_{r}^k\Psi_i)^2 + \x_k(\partial_{r}^k\Psi_i)\cdot M^k[\Psi_i]= \\
				&= K^{\x_k}[\partial_{r}^k\Psi_i] + \left (\x_k^v \cdot T\partial_{r}^k\Psi_i + \x_k^r\cdot \partial_{r}^{k+1}\Psi_i \right ) M^k[\Psi_i].
			\end{aligned}
		\end{align}
		Again, the first term is given by  \begin{align} \begin{aligned}
				K^{\x_k}[\partial_{r}^k\Psi_i] =& F^k_{vv}(T \partial_{r}^k\Psi_i)^2+ F^k_{rr}(\partial_{r}^{k+1}\Psi_i)^2+ F^k_{\slashed{\nabla}}\abs{\slashed{\nabla} \partial_{r}^k\Psi_i}^2\\ &+ F^k_{vr}(T\partial_{r}^k\Psi_i)( \partial_{r}^{k+1}\Psi_i) + F^k_{00}(V_i)(\partial_{r}^k\Psi_i)^2, \label{K^k-terms}
			\end{aligned}
		\end{align}
		where the coefficients $ F^k_{ab}, \ a,b \in \br{0,v,r, \slas{\nabla}} $ are defined as in (\ref{F coefficients}) for the vectorfield $ \x_k, $ i.e. \begin{align}
			\begin{aligned}
				&F^k_{vv}= \partial_r\x_k^v, \hspace{2cm} F^k_{rr}= D\left(\frac{\partial_r\x_k^r}{2}-\frac{\x_k^r}{r}\right) - \frac{\x_k^r D'}{2}, \\
				&F^k_{\slashed{\nabla}}= -\frac{1}{2}\partial_r\x_k^r, \hspace{1.5cm} F^k_{vr}= D\partial_r\x_k^v- \frac{2\x_k^r}{r}, \\
				& F^k_{00}(V_i)= -\frac{1}{2}\x_k(V_i)-V_i\left(\frac{\partial_r\x_k^r}{2}+\frac{\x_k^r}{r}\right).
			\end{aligned}
		\end{align}
		Expanding also $ M^k[\Psi_i] $ yields \begin{align}
			\begin{aligned}
				&\left (\x_k^v \cdot T\partial_{r}^k\Psi_i + \x_k^r\cdot \partial_{r}^{k+1}\Psi_i \right ) M^k[\Psi_i]\ \   = \\ 
				&-\sum_{j=1}^{k}\binom{k}{j} \x_k^v\ \partial_{r}^{j} D \cdot (\partial_{r}^{k-j+2} \Psi_i) (T\partial_r^k\Psi_i)\ - \ \sum_{j=1}^{k}\binom{k}{j} \x_k^r\ \partial_{r}^{j} D \cdot (\partial_{r}^{k-j+2} \Psi_i) (\partial_r^{k+1}\Psi_i) \\ &-\sum_{j = 1}^{k}\binom{k}{j}\x_k^v\  \partial_{r}^{j} \frac{2}{r} \cdot (T \partial_{r}^{k-j} \Psi_i)(T\partial_r^k\Psi_i)\ - \ \sum_{j = 1}^{k}\binom{k}{j}\x_k^r\  \partial_{r}^{j} \frac{2}{r} \cdot (T \partial_{r}^{k-j} \Psi_i)(\partial_r^{k+1}\Psi_i)\\
				&-\sum_{j=1}^{k}\binom{k}{j}\x_k^v\ \partial_{r}^{j} R \cdot (\partial_{r}^{k-j+1} \Psi_i)(T\partial_r^k\Psi_i)\ - \ \sum_{j=1}^{k}\binom{k}{j}\x_k^r\ \partial_{r}^{j} R \cdot (\partial_{r}^{k-j+1} \Psi_i)(\partial_r^{k+1}\Psi_i)\\ &-\sum_{j=1}^{k}\binom{k}{j}\x_k^v\ r^{2} \partial_{r}^{j} r^{-2} A_{j,\ell}(r)\cdot   (\slashed{\Delta} \partial_{r}^{k-j} \Psi_i)(T\partial_r^k\Psi_i)\\ &-\ \sum_{j=1}^{k}\binom{k}{j}\x_k^r\ r^{2} \partial_{r}^{j} r^{-2} A_{j,\ell}(r)\cdot   (\slashed{\Delta} \partial_{r}^{k-j} \Psi_i)(\partial_r^{k+1}\Psi_i) \label{k-bulk}
			\end{aligned}
		\end{align} 
		In what follows, we show that the bulk terms of (\ref{K^k-terms}, \ref{k-bulk}) are controlled by terms of the right-hand side of (\ref{H^2- higher control}) for $ k\leq (\ell-1)+(-1)^{i+1} $ of the inductive hypothesis, or they can be absorbed as a small $ \varepsilon>0 $ fraction in the positive definite terms of the left-hand side of (\ref{H^2- higher control}).
		
		First, we are going to prove the following lemma which will be frequently used in the estimates below. 
		\begin{lemma} \label{horizon-control-r's}
			For solutions $ \Psi_i^{^{(\ell)}},\ \ell\geq i, \ i\in\br{1,2} $ to (\ref{waveeq}), and  for any $ 0\leq j \leq (l-1) - (-1)^{i} $, we have \begin{align}\begin{aligned}
					\abs{\int_{\mathcal{H}^+}(\partial_r^j\Psi_i)(\partial_r^{^{\ell - (-1)^i}}\Psi_i)} \leq C_{\varepsilon} & \left(\sum_{j=0}^{(\ell-1) - (-1)^i} \int_{\Sigma_{0}} J_{\mu}^{n_{\Sigma_{0}}}\left[T^{j} \Psi_i\right] n_{\Sigma_{0}}^{\mu}+ \sum_{j=1}^{\ell - (-1)^i} \int_{\Sigma_{0} } J_{\mu}^{n_{\Sigma_{0}}}\left[\partial_{r}^{j} \Psi_i\right] n_{\Sigma_{0}}^{\mu}\right) \\
					& \ \ +\varepsilon \cdot \int_{\Sigma_{\tau}\cap A_c} J_{\mu}^{\x_k}[\partial_r^{^{\ell-(-1)^i}}\Psi_i]\cdot n _{\Sigma}^{\mu}\ +\  \varepsilon \cdot\int_{\mathcal{H}^+}\left(T\partial_r^{^{\ell-(-1)^i}}\Psi_i\right)^2,  \label{misc higher 1}
				\end{aligned}
			\end{align}
			
			for a positive constant $ C_{\varepsilon} $ depending on $ \varepsilon, M,\Sigma_0  $ and $ j. $
			\begin{flushleft}
				\textit{Proof.} In view of proposition \ref{Span-prop}, we can work instead with the integrals
			\end{flushleft}
			 \begin{align*}
					\int_{\mathcal{H}^+}(T\partial_r^s\Psi_i)(\partial_r^{^{\ell - (-1)^i}}\Psi_i),
				\end{align*}
				for $ 0\leq s\leq \ell - (-1)^i. $ \\ \\
				$ \bullet\  $ For $ s=\ell-(-1)^i $ we have \begin{align*}
					\int_{\mathcal{H}^+} (T\partial_r^{^{\ell - (-1)^i}}\Psi_i)(\partial_r^{^{\ell - (-1)^i}}\Psi_i) &= \dfrac{1}{2}\int_{\mathcal{H}^+} T\left(\partial_r^{^{\ell - (-1)^i}}\Psi_i\right)^2  \\
					&= \dfrac{1}{2}\int_{\Sigma_{\tau}\cap\mathcal{H}^+} \left (\partial_r^{^{\ell - (-1)^i}}\Psi_i\right )^2  - \dfrac{1}{2}\int_{\Sigma_{0}\cap\mathcal{H}^+}\left (\partial_r^{^{\ell - (-1)^i}}\Psi_i\right )^2
				\end{align*}
				Thus, using Lemma \ref{hardy} for each term above we obtain \begin{align*}
					\int_{\Sigma\cap\mathcal{H}^+}\left (\partial_r^{^{\ell - (-1)^i}}\Psi_i\right )^2 \leq &\  \varepsilon \int_{\Sigma\cap A_c} \left[ \left (T\partial_r^{^{\ell - (-1)^i}}\Psi_i\right  )^2 + \left (\partial_r^{^{\ell+1 - (-1)^i}}\Psi_i\right )^2\right] + C_{\varepsilon}\int_{\Sigma\cap A_c}\left (\partial_r^{^{\ell - (-1)^i}}\Psi_i\right )^2
					\\ \leq &\   \varepsilon \int_{\Sigma\cap A_c} J_{\mu}^{\x_k}[\partial_r^{^{\ell-(-1)^i}}\Psi_i]\cdot n _{\Sigma}^{\mu} + C_{\varepsilon} \int_{\Sigma\cap A_c} J_{\mu}^{\x_k}[\partial_r^{^{(\ell-1)-(-1)^i}}\Psi_i]\cdot n _{\Sigma}^{\mu},
				\end{align*}
				where the last term above is estimated from the first term of the right-hand side of  (\ref{misc higher 1}), due to the inductive hypothesis.
				\\
				\\ $ \bullet $ Now let $ 0\leq s < \ell - (-1)^i, $ then integration by parts yields \begin{align*}
					\int_{\mathcal{H}^+}(T\partial_r^s\Psi_i)(\partial_r^{^{\ell - (-1)^i}}\Psi_i) =&  \int_{\h\cap\Sigma_{\tau}}(\partial_r^s\Psi_i)(\partial_r^{^{\ell - (-1)^i}}\Psi_i)  - \int_{\h\cap \Sigma_0}	(\partial_r^s\Psi_i)(\partial_r^{^{\ell - (-1)^i}}\Psi_i)   \\
					-& \int_{\h}\partial_r^s\Psi_i \left(T\partial_r^{^{\ell - (-1)^i}}\Psi_i  \right) 
				\end{align*} 
				After applying a Cauchy-Schwartz, the first two boundary terms are treated similarly to the $ s=\ell-(-1)^i $ case, using Lemma \ref{hardy} and the inductive hypothesis. 
				
				For the last term, we have \begin{align*}
					\abs{\int_{\h}\partial_r^s\Psi_i\left(T\partial_r^{^{\ell - (-1)^i}}\Psi_i  \right)} \leq \dfrac{1}{2  \varepsilon}\int_{\h} (\partial_r^s\Psi_i)^2 +  \dfrac{\varepsilon}{2}  \int_{\h}\left(T\partial_r^{^{\ell - (-1)^i}}\Psi_i  \right)^2 
				\end{align*}
				and the first term on the right is estimated by the estimate
				(\ref{H^2- higher control}), for $ 0\leq s\leq (\ell-1)+(-1)^{i+1}   $. \\
			\begin{flushright}
			$ 	\blacksquare $
			\end{flushright}
		\end{lemma}
		
		\paragraph{Controlling the bulk terms.} Now we are in the position of controlling the spacetime terms that appear in the energy identity. 
		\paragraph{\begin{align}
				\boxed{\text{Estimate for the terms \quad  }  E_j^1 (\partial_r^{k-j+1}\Psi_i)(\partial_r^{k+1}\Psi_i), \ j\geq 0.  \label{higher misc 1}}
		\end{align}}

		$ \bullet\ $ For $ j=0, $ the coefficient  of the bulk term $ (\partial_r^{k+1}\Psi_i)^2 $ is $ E_0^1 = F_{rr}^k - k\x_k^rD' $ and since $ \x_k $ is timelike, $ E_0^1 \geq 0 $ in $ A_c. $ 
		\\ $ \bullet\  $ For $ j\geq 1, $ using integration by parts we obtain 
		\begin{align*}
			\begin{aligned}
				& \int_{\mathcal{R}} E_{j}^{1}\left(\partial_{r}^{k-j+1} \Psi_i\right)\left(\partial_{r}^{k+1} \Psi_i\right) \ = \ \\
				&= \int_{\Sigma_{0}} E_{j}^{1}\left(\partial_{r}^{k-j+1} \Psi_i\right)\left(\partial_{r}^{k} \Psi_i\right) \partial_{r} \cdot n_{\Sigma_{0}}-\int_{\Sigma_{\tau}} E_{j}^{1}\left(\partial_{r}^{k-j+1} \Psi_i\right)\left(\partial_{r}^{k} \Psi_i\right) \partial_{r} \cdot n_{\Sigma_{\tau}} \\ &-\int_{\mathcal{H}^{+}} E_{j}^{1}\left(\partial_{r}^{k-j+1} \Psi_i\right)\left(\partial_{r}^{k} \Psi_i\right)  
				-\int_{\mathcal{R}}\left(\partial_{r} E_{j}^{1}+\frac{2}{r} E_{j}^{1}\right)\left(\partial_{r}^{k-j+1} \Psi_i\right)\left(\partial_{r}^{k} \Psi_i\right)-\int_{\mathcal{R}} E_{j}^{1}\left(\partial_{r}^{k-j+2} \Psi_i\right)\left(\partial_{r}^{k} \Psi_i\right) 
			\end{aligned}
		\end{align*}
		Lemma \ref{psi_r Control} holds also for $ \partial_r^k\Psi_i $ with the corresponding right-hand side estimate, and the proof follows similarly. Thus, for any $ j\geq 1 $ the boundary terms over $ \Sigma_0,\Sigma_{\tau} $ and the last two spacetime terms are estimated using this Lemma, Cauchy-Schwartz, and the inductive hypothesis.
		
		The boundary term over the horizon $ \h $ is estimated for $ j\geq 2 $ using Lemma \ref{horizon-control-r's}   and it only remains to estimate it for $ j=1, $ in which case the coefficient is \[ E_1^1=- \x_k^r\left(\dfrac{k(k+1)}{2}D''+k R' \right)  \]
		Using Poincare inequality for the angular term of the event horizon $ \h $ integral in the energy identity (\ref{k-order identity}), it suffices to have \[ -\dfrac{\x_k^r(M)}{2M^2} k(k+1) \leq  -\dfrac{\x_k^r(M)}{2M^2} \left(\ell(\ell+1) + V_i(M)M^2\right).  \]
		For $ i=1 $ the above is equivalent to $ k(k+1)\leq (\ell+1)(\ell+2) \ \Leftrightarrow \  k\leq\ell+1 $ and for $ i=2 $, $ k(k+1)\leq \ell(\ell-1)\ \Leftrightarrow \  k\leq \ell-1 $.  
		
		Thus, we see that when $ k=\ell+(-1)^{i+1} $ we have to use the \textbf{entire} coefficient of the angular term $ \abs{\slashed{\nabla}\partial_r^k\Psi_i}^2 $ on horizon $ \h $ in order to close the above estimate, hence and the appearance of the factor $ \bm{\chi}_{i,k} $ in (\ref{H^2- higher control}).
		\paragraph{\begin{align}
				\boxed{\text{Estimates for the terms \quad}  E_j^2 (T\partial_r^{k-j+1}\Psi_i)(\partial_r^{k+1}\Psi_i), \ j\geq 1. \label{higher misc 2}}
		\end{align}}
		Let us first deal with the $ j\geq 2 $ case. Using integration by parts we obtain 
		
		\begin{align*}
			\begin{aligned}
				& \int_{\mathcal{R}} E_{j}^{2}\left(T \partial_{r}^{k-j+1} \Psi_i\right)\left(\partial_{r}^{k+1} \Psi_i\right)=  \\
				& =\int_{\Sigma_{0}} E_{j}^{2}\left(T \partial_{r}^{k-j+1} \Psi_i\right)\left(\partial_{r}^{k} \Psi_i\right) \partial_{r} \cdot n_{\Sigma_{0}} -\int_{\Sigma_{\tau}} E_{j}^{2}\left(T \partial_{r}^{k-j+1} \Psi_i\right)\left(\partial_{r}^{k} \Psi_i\right) \partial_{r} \cdot n_{\Sigma_{\tau}}\\ &-\int_{\mathcal{H}^{+}} E_{j}^{2}\left(T \partial_{r}^{k-j+1} \Psi_i\right)\left(\partial_{r}^{k} \Psi_i\right) 
				-\int_{\mathcal{R}}\left(\partial_{r} E_{j}^{2}+\frac{2}{r} E_{j}^{2}\right)\left(T \partial_{r}^{k-j+1} \Psi_i\right)\left(\partial_{r}^{k} \Psi_i\right)-\int_{\mathcal{R}} E_{j}^{2}\left(T \partial_{r}^{k-j+2} \Psi_i\right)\left(\partial_{r}^{k} \Psi_i\right)
			\end{aligned}
		\end{align*}
		Now, the boundary terms over $ \Sigma_0,\Sigma_{\tau} $ are estimated using Cauchy-Schwartz and the inductive hypothesis while the horizon $ \h $ term is treated using Lemma \ref{horizon-control-r's}. The second last spacetime term is estimated by Cauchy-Schwartz, Lemma \ref{psi_r Control} (generalized) and the inductive hypothesis while the remaining last spacetime term is treated similarly where we use instead the $ \varepsilon $ variation of Cauchy-Schwartz and Lemma \ref{psi_r Control} for $ \varepsilon' = \varepsilon^2. $
		
		For the $ j=1 $ case, using the wave equation (\ref{waveeq}) for $ \Psi_i $ and $ \slashed{\Delta}\p = -\frac{\ell(\ell+1)}{r^2}\p $ we have \[ T\partial_r\Psi_i = -\frac{D}{2}\partial_r^2\Psi_i -\dfrac{1}{r}T\Psi_i-\dfrac{R(r)}{2}\partial_r\Psi_i - \dfrac{1}{2} \left (1+\dfrac{r^2\cdot V_i(r)}{\ell(\ell+1)}\right ) \slashed{\Delta}\p.\]
		Thus, we write \begin{align*}
			\int_{\mathcal{R}}&E_1^2(T\partial_r^k\p)(\partial_r^{k+1}\p)= \int_{\mathcal{R}}E_1^2(\partial_r^{k-1}T\partial_r\p)(\partial_r^{k+1}\p) = \\
			=  &\int_{\mathcal{R}}\dfrac{E_1^2}{2}\ \partial_r^{k-1}\left(-D\partial_r^2\Psi_i -\dfrac{2}{r}T\Psi_i-R(r)\partial_r\Psi_i -  \left (1+\dfrac{r^2\cdot V_i(r)}{\ell(\ell+1)}\right ) \slashed{\Delta}\p \right) \partial_r^{k+1}\p \ = \\
			= &\hspace{0.5cm}  \int_{\mathcal{R}}- \dfrac{E_1^2}{2}\  \sum_{s=0}^{k-1}\binom{k-1}{s}\partial_r^{s}D \cdot \partial_r^{k-s+1}\p \cdot \partial_r^{k+1}\p \\
			&+  \int_{\mathcal{R}}(- E_1^2)\  \sum_{s=0}^{k-1}\binom{k-1}{s}\partial_r^{s}r^{-1} \cdot T\partial_r^{k-s-1}\p\cdot \partial_r^{k+1}\p \\
			&+  \int_{\mathcal{R}}\dfrac{ -E_1^2}{2}\  \sum_{s=0}^{k-1}\binom{k-1}{s}\partial_r^{s}R(r)\cdot  \partial_r^{k-s}\p \cdot \partial_r^{k+1}\p \\
			&+ \int_{\mathcal{R}}\dfrac{- E_1^2}{2}\  \sum_{s=0}^{k-1}\binom{k-1}{s}\partial_r^{s} B_{\ell}(r) \cdot \partial_r^{k-s-1}(\slashed{\Delta}\p) \cdot \partial_r^{k+1}\p,
		\end{align*}
		where $ B_{\ell}(r) : =  1+\frac{r^2\cdot V_i(r)}{\ell(\ell+1)}. $ \\
		$ \bullet\  $The terms of first-line sum are estimated for $ s=0,1 $ in view of the degenerate coefficient $ D(r),D'(r) $ on the horizon $ \h $. The terms for $ s\geq 2 $ are treated similarly to estimates \ref{higher misc 1}. \\
		$ \bullet\  $The terms of second-line are estimated as the $ j\geq 2 $ case studied above.
		\\ $ \bullet\  $ On the third line, the $ s=0 $ case is estimated in view of the degenerate coefficient $ R(r) $ on the horizon $ \h, $ and for $ s\geq 1 $ we work similarly to \ref{higher misc 1} estimates.\\
		$ \bullet\  $ For the last line terms, it's direct a computation to check the following commutation for any  $ m\in\mathbb{N} $ \[ \partial_r^{m}(\slashed{\Delta}\p)= \slashed{\Delta}\partial_r^m\p + \left [\partial_r^m,\slashed{\Delta}\right ]\p = \slashed{\Delta}(\partial_r^m\p) + \sum_{j=1}^m\binom{m}{j}r^2\partial_r^j(r^{-2})\cdot \slashed{\Delta}(\partial_r^{m-j}\p) \]
		Thus, terms of the fourth line can be expressed in the form of $ (\slashed{\Delta}\partial_r^{k-j}\p)(\partial_r^{k+1}\p), \ j\geq 1, $ and such type of terms are estimated in the \ref{higher misc 6} paragraph below, as long as the coefficients of the corresponding terms are comparable. However, note that $ E_1^2=-k\x_k^vD'+D\partial_r(\x_k^v)- \dfrac{2\x_k^r}{r} $  and $ \partial_r^sB_{\ell}(r) $ is uniformly bounded in $ A_c $ with respect to $ \ell, $ while $ E_j^6 \sim_k \left (- \dfrac{2\x_k^r}{r} \right ) $ and uniformly bounded with respect to $ \ell $ for every $ j\geq 1 $. Thus, the estimates of paragraph \ref{higher misc 6} yield estimates for the terms of the fourth line with uniform coefficients with respect to $ \ell. $

		\paragraph{\begin{align}
				\boxed{\text{Estimates for the terms \quad}  E_j^3 (\partial_r^{k-j+2}\p)(T\partial_r^k\p), \ j\geq 1. }
		\end{align}}
		The $ j=1 $ case is treated already in the \ref{higher misc 2}-estimates. For $ j\geq 2 ,$ we use Cauchy-Schwartz with $ \varepsilon>0 $, Lemma \ref{psi_r Control} (generalized)  for $ \varepsilon'= \epsilon^2>0 $ and the inductive hypothesis. 
		
		\paragraph{\begin{align}
				\boxed{\text{Estimates for the terms \quad}  E_j^4 (T\partial_r^{k-j}\p)(T\partial_r^k\p), \ j\geq 1.  }
		\end{align}}
		We simply use Cauchy-Schwartz with $ \varepsilon>0 $ and the inductive hypothesis.

		\paragraph{\begin{align}
				\boxed{\text{Estimates for the terms \quad}  E_j^5 (\slashed{\Delta}\partial_r^j\p)(T\partial_r^k\p), \ 0\leq j\leq k-1.   }
		\end{align}}
		We use induction on $ j\leq k-1 $. For $ j=0 $, we use the wave equation (\ref{waveeq}) for $ \p $ and we express \[ \left (1+\dfrac{r^2\cdot V_i(r)}{\ell(\ell+1)}\right ) \slashed{\Delta}\p = - 2T\partial_r\Psi_i  +D(r)\partial_r^2\Psi_i +\dfrac{2}{r}T\Psi_i+R(r)\partial_r\Psi_i  \]
		Note, the coefficient of $ \slashed{\Delta}\p $ satisfies \[ C_1\leq \abs{1+\dfrac{r^2\cdot V_i(r)}{\ell(\ell+1)}}\leq C_2 \] for $ C_1,C_2 $ positive constants depending only on $ M, $
		for all $ \ell\geq i,\ i\in\br{1,2}, $ thus we can solve for $ \slashed{\Delta}\p $.  Then, using Cauchy-Schwartz and  the result obtained above we estimate each term. 
		
		Now, assume we can estimate all terms of the form $ (\slashed{\Delta}\partial_r^j\p)(T\partial_r^k\p), $ for any  $ s\leq j-1 $, then we have \begin{align*}
			(\slashed{\Delta}\partial_r^j\p) \cdot T\partial_r^k\p =& (\partial_r^j\slashed{\Delta}\p )\cdot T\partial_r^k\p + \left [ \slashed{\Delta}, \partial_r^j \right ]\p  \cdot T\partial_r^k\p. 
		\end{align*}
		However, after solving for $ \slashed{\Delta} \p$ in the wave equation we take $ \partial_r^j $-many derivatives of the equation and we use the results obtained previously to estimate the first term of the right-hand side above. Next, we have \[ \left [ \slashed{\Delta}, \partial_r^j \right ]\p  \cdot T\partial_r^k\p =  \sum_{s=1}^j\binom{j}{s}r^2\partial_r^s(r^{-2})\cdot \slashed{\Delta}(\partial_r^{j-s}\p) \cdot T\partial_r^k\p  \]
		and we use the inductive hypothesis to estimate it, since $ j-s\leq j-1 $ for all $ s\geq 1.   $
		
		\paragraph{\begin{align}
				\boxed{\text{Estimates for the terms \quad} E_j^6 (\slashed{\Delta}\partial_r^{k-j}\p)(\partial_r^{k+1}\p),\ j\geq 1. \label{higher misc 6}}
		\end{align}}
		Once again, using integration by parts we obtain 
		\begin{align*}
			\begin{aligned}
				& \int_{\mathcal{R}} E_j^6\left(\slashed{\Delta}\partial_r^{k-j}\p\right)\left(\partial_{r}^{k+1} \Psi_i\right)=  \\
				& =\int_{\Sigma_{0}} E_j^6\left(\slashed{\Delta}\partial_r^{k-j}\p\right)\left(\partial_{r}^{k} \Psi_i\right) \partial_{r} \cdot n_{\Sigma_{0}} -\int_{\Sigma_{\tau}} E_j^6\left(\slashed{\Delta}\partial_r^{k-j}\p\right)\left(\partial_{r}^{k} \Psi_i\right) \partial_{r} \cdot n_{\Sigma_{\tau}}\\ &-\int_{\mathcal{H}^{+}} E_j^6\left(\slashed{\Delta}\partial_r^{k-j}\p\right)\left(\partial_{r}^{k} \Psi_i\right) 
				-\int_{\mathcal{R}}\left(\partial_{r} E_j^6+\frac{2}{r} E_j^6\right)\left(\slashed{\Delta}\partial_r^{k-j}\p\right)\left(\partial_{r}^{k} \Psi_i\right)-\int_{\mathcal{R}} E_j^6\left(\partial_r\slashed{\Delta}\partial_r^{k-j}\right)\left(\partial_{r}^{k} \Psi_i\right)
			\end{aligned}
		\end{align*}
		Now, the first two boundary terms $ \Sigma_0,\Sigma_{\tau} $ and the second last spacetime term are estimated using Stokes' theorem for the angular derivative on the sphere $ S^2(r) $, applying Cauchy-Schwartz and using the  inductive hypothesis. For the last spacetime term, we write \[ \left (\partial_r\slashed{\Delta}\partial_r^{k-j}\p\right )\left(\partial_r^k\p\right) = \left(\slashed{\Delta}\partial_r^{k-j+1}\p  \right) \left(\partial_r^k\p\right)+ \sum_{s=0}^{k-j}\binom{k-j}{s}r^2\partial_r^s(r^{-2}) \left (\slashed{\Delta}\partial_r^{k-j-s}\p\right ) \left(\partial_r^k\p\right)  \]
		so all the sum terms are estimated using Cauchy-Schwartz and the inductive hypothesis, for any $ j\geq 1 $. The first term on the left-hand side is estimated the same for $ j\geq 2 $, while for $ j=1 $ we use Stokes's theorem and we obtain \[ \abs{\int_{\mathcal{R}}E_1^6\ \slashed{\Delta}\partial_r^k\p \cdot \partial_r^k\p} \leq \int_{\mathcal{R}}\abs{E_1^6}\cdot \abs{\slashed{\nabla}\partial_r^k\p}^2, \]
		which can be absorbed in the  $ F^k_{\slashed{\nabla}}\abs{\slashed{\nabla}\partial_r^k\p}^2 $ term of  (\ref{K^k-terms}) as long as $ \abs{E_1^6}< -\frac{1}{2}\partial_r\x_k^r. $ However, in view of  $ A_{j,\ell}(r)  $ being uniformly bounded in $ A_c $ with respect to $ \ell $, it suffices to choose $ -\x_k^r \ll -\frac{1}{2}\partial_r\x_k^r. $
		
		Finally, we will estimate the horizon $ \h $ boundary term by induction. In particular, it suffices to estimate \[ \int_{\h}\left (\slashed{\Delta}\partial_r^m\p\right )\left(\partial_r^k\p\right),  \] 
		for any $ 0\leq m \leq k-1. $ For $ m=0, $ using the wave equation (\ref{waveeq}) for $ \p $ we have \[  \left (1+\dfrac{r^2\cdot V_i(r)}{\ell(\ell+1)}\right ) \slashed{\Delta}\p= -2T\partial_r\Psi_i  -D(r)\partial_r^2\Psi_i -\dfrac{2}{r}T\Psi_i-R(r)\partial_r\Psi_i  \]
		where $ \frac{1}{3}\leq \left (1+\frac{M^2\cdot V_i(M)}{\ell(\ell+1)}\right )\leq 3 $ for all $ \ell\geq i. $ Thus, using Lemma \ref{horizon-control-r's} we obtain an estimate for the $ m=0 $ case. Now, assume we also have an estimate for any $ m\leq k-2 $, then we write \begin{align*}
			\int_{\h}\left (\slashed{\Delta}\partial_r^{k-1}\p\right )\left(\partial_r^k\p\right) = \int_{\h}(\partial_r^{k-1} \slashed{\Delta}\p)\partial_r^k\p - \int_{\h} \sum_{s=1}^{k-1}\binom{k-1}{s}r^2\partial_r^s(r^{-2})\slashed{\Delta}\partial_r^{k-1-s}\p \cdot \partial_r^k\p
		\end{align*}
		The terms of the sum are all estimated by the inductive hypothesis for $ m\leq k-2. $
		For the first term, using the wave equation and taking $ \partial_r^{k-1} $-derivatives we obtain \begin{align}
			\begin{aligned}
				B_{\ell}(r)\cdot  & \partial_r^{k-1} \slashed{\Delta}\p \cdot  \partial_r^k\p=\\
				=\  & -   \sum_{s=0}^{k-1}\binom{k-1}{s}\partial_r^{s}D \cdot \partial_r^{k-s+1}\p \cdot   \partial_r^k\p\\
				&- 2\sum_{s=0}^{k-1}\binom{k-1}{s}\partial_r^{s}r^{-1} \cdot T\partial_r^{k-s-1}\p \cdot  \partial_r^k\p\\
				&-   \sum_{s=0}^{k-1}\binom{k-1}{s}\partial_r^{s}R(r)\cdot  \partial_r^{k-s}\p \cdot   \partial_r^k\p\\
				&-  \sum_{s=1}^{k-1}\binom{k-1}{s}\partial_r^{s} B_{\ell}(r) \cdot \partial_r^{k-s-1}(\slashed{\Delta}\p)  \cdot  \partial_r^k\p \\
				& - 2T\partial_r^{k}\p \cdot  \partial_r^k\p ,
			\end{aligned}
		\end{align}
		In view of $ D(M)=D'(M)=R(M)=0,$ the non-vanishing terms on the horizon $ \h $ of the above sums are estimated using Lemma \ref{horizon-control-r's} and the inductive hypothesis. All coefficients in the estimates are independent of $ \ell $ since $ \partial_r^sB_{\ell}(r) $ is uniformly bounded with respect to $ \ell $ for all $ s\geq 0. $
		
		\paragraph{\begin{align}
				\boxed{\text{Estimates for the terms \quad} F_{00}^k(V_i)(\partial_r^k\p) ^2. }
		\end{align}}
		This term is estimated precisely in the same way as $ E_{11}(\partial_r\p) $ in Section \ref{EE first}.
		\\
		\\
		\paragraph{Conclusion.}
		Thus, in order to conclude the proof of Theorem \ref{high-order horizon estimates}, it suffices to consider  a timelike vectorfield $ \x_k $, i.e. $ \x_k^v(M)>0 , \x_k^r(M)<0,$ such that  $ -\partial_r\x_k^r(M)\gg -\x_k^r(M) $ and $ \partial_r\x_k^v(M) $ to be sufficiently large.  Then, choosing
		$ r_c $ close enough to the horizon $ \h $ and $ \varepsilon>0 $ sufficiently small, allows us to obtain the estimate (\ref{H^2- higher control}) for $ k=\ell+(-1)^{i+1} $, and thus closing the induction argument. \\
	\end{proof}

	\newpage
	
	\section{The \texorpdfstring{$r^p$}{PDFstring}--hierarchy estimates}\label{r-p} In this section, we aim to derive decay for both the degenerate and non-degenerate energy flux leading to pointwise decay estimates for $ \Psi_{i}^{^{(\ell)}}, \ \forall \ \ell \geq i, \ i\in\br{1,2
	} $. The key part is to produce the so-called $ r^p- $hierarchy estimates by applying the vector field method for vectorfields of the form $ Z=r^{q}  \partial_v $, written with respect to the double null coordinate system $ (u,v) $. All results in this section are expressed in the double null coordinate system and we work with the foliation $ \check{\Sigma}_{\tau} $ introduced in Section \ref{Geometry}.
	
	Let us make some notation remarks regarding regions where these estimates take place. For any $ \tau_1< \tau_2 $ we define the  following regions, as seen in the Penrose diagram below, \begin{align*}
		\check{M}_R(\tau_1,\tau_2) := \cup_{\tau\in[\tau_1,\tau_2]}\check{\Sigma}_{\tau}, \ \ \check{I}_{\tau_1}^{\tau_2}:= \check{M}_R(\tau_1,\tau_2) \cap \lbrace r\geq R\rbrace, \ \ \check{O}_{\tau}:= \check{\Sigma}_{\tau} \cap\lbrace r\geq R\rbrace.
	\end{align*}

	\begin{figure}[hbt!]
		\centering
		\def\svgscale{0.7}
			\includegraphics[scale=0.7]{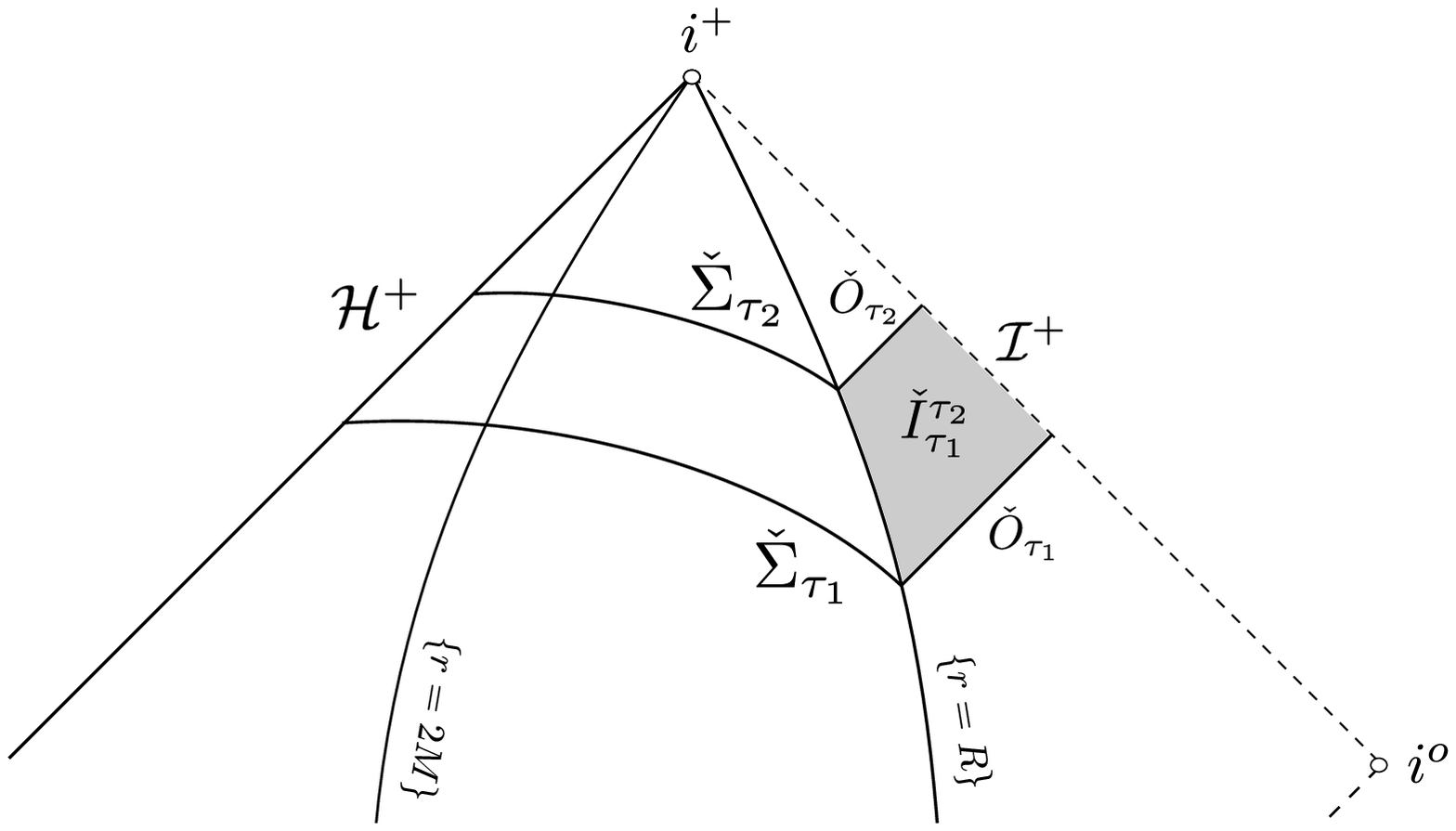}
		\caption{Regions near future null infinity $ \mathcal{I}^{+} $.}
		\label{Foliation}
	\end{figure}
	
	\begin{flushleft}	
		We proceed with the following proposition. 
	\end{flushleft}
	\begin{proposition} \label{r_p}
		Let $ 0 < p \leq 2 $, then there exists a positive constant $ C>0 $ depending only on $ M,\check{\Sigma}_0 $, such that for any solution $ \Psi_{i}^{^{(\ell)}} $ to (\ref{waveeq}), supported on the fixed frequency $ \ell \geq i, \ i\in \br{1,2} $, the rescaled quantity $ \Phi_i := r \Psi_i $ satisfies the following estimate
		\begin{align}
			\begin{aligned}
				\int_{\check{O}_{\tau_2}}r^p \dfrac{(\partial_v \Phi_i)^2}{r^2} +  \int_{\check{I}_{\tau_1}^{\tau_2}} p r^{p-1} \dfrac{(\partial_v\Phi_i)^2}{r^2} +  \int_{\check{I}_{\tau_1}^{\tau_2}}  \dfrac{r^{p-1}}{4}\sqrt{D}&\left((2-p)  + (p-4)\dfrac{M}{r}\right)\left(\abs{\slashed{\nabla}\Psi_i}^2 + V_i \Psi_i^2\right) \\
				+  \int_{\check{I}_{\tau_1}^{\tau_2}} \dfrac{r^{p-1}}{4}DV_i  \Psi_i^2\ \leq \ & C \int_{\check{\Sigma}_{\tau_1}}J_{\mu}^{T}[\Psi_i]n_{\check{\Sigma}_{\tau_1}}^{\mu}+
				\int_{\check{O}_{\tau_1}}r^p \dfrac{(\partial_v \Phi_i)^2}{r^2}.
			\end{aligned}
		\end{align}
	\end{proposition}
	\begin{proof}
		As described in the introduction, we apply the vectorfield method in the $ \check{I}_{\tau_1}^{\tau_2} $ region, for the multiplier $ Z = r^q \partial_v $, where $ q=p-2. $ In order to avoid timelike boundary terms we also fix  a smooth cut-off function of $ r $ alone, such that $ \zeta_R(r) \equiv 0  $ for $ r\leq R $ and $ \zeta_R(r)\equiv1 $ for $ r\geq R+1 $.  Then, the energy identity associated to $ \zeta_R\cdot \Phi_i $ reads \begin{align} \label{Misc r_p 1}
			\int_{\check{I}_{\tau_1}^{\tau_2}} Div\left(J_{\mu}^{Z}[\zeta_R\cdot\Phi_i]\right)=		\int_{\partial\check{I}_{\tau_1}^{\tau_2}} J_{\mu}^{Z}[\zeta_R\cdot\Phi_i].
		\end{align}
		However, for $ r\geq R+1 $ we have $ \zeta_R\cdot \Phi_i = \Phi_i $ and for solutions $ \Psi_i $ to equation (\ref{waveeq}), it's immediate computations to check that $ \Phi_i = r\Psi_i$ satisfies the following non-homogeneous equation \begin{align}
			\square\Phi_i-V_i\Phi_i = \dfrac{D'}{r}\Phi_i - \dfrac{2}{r}\Big(\partial_u \Phi_i-\partial_v\Phi_i\Big)=: M[\Phi_i] . 
		\end{align}
		Therefore, from proposition \ref{GenCur} we have \begin{align}
			\begin{aligned}
				Div\left(J_{\mu}^{Z}[\Phi_i]\right) =& \ \ \dfrac{1}{2}Q[\Phi_i]\cdot ^{(Z)}\pi - \dfrac{1}{2}Z(V_i)\Phi_i^2 + Z(\Phi_i)M[\Phi_i]= \\
				=& \ \ p r^{q-1} (\partial_v\Phi_i)^2 + r^{q-1}D' \cdot \Phi_i\cdot (\partial_v\Phi_i) - \dfrac{r^{q-1}}{4}\left(D'r+Dq\right) \left(\abs{\slashed{\nabla}\Phi_i}^2+V_i \Phi_i^2\right) 
				\\ & + \dfrac{r^{q-1}}{4}D \left[V_i + \dfrac{6}{r^2}\left(\dfrac{M}{r}\right)^2 \right]\Phi_i^2.
			\end{aligned}
		\end{align}
		To treat the non-quadratic term we use Stoke's theorem and we write 
		\begin{align}
			\begin{aligned} \label{Misc r_p 2}
				\int_{\check{I}_{\tau_1}^{\tau_2}} r^{q-1}D'(\zeta_R\cdot\Phi_i) \partial_v(\zeta_R\cdot \Phi_i) = & \  \int_{\mathcal{I}^+} \dfrac{D'\cdot D}{4}r^{q-1}(\zeta_R\cdot\Phi_i)^2 - \int_{\check{I}_{\tau_1}^{\tau_2}\cap\lbrace r=R \rbrace} \dfrac{r^{q-1}}{4}D'(\zeta_R\cdot\Phi_i)^2 \\
				&\ + 	\int_{\check{I}_{\tau_1}^{\tau_2}} \dfrac{r^{q-1}}{2}D\left(\dfrac{M}{r}\right)\left[\sqrt{D}(6-p)-1\right]  (\zeta_R\cdot\Phi_i)^2 
			\end{aligned}
		\end{align}
		Note, the boundary integral over $ \check{I}_{\tau_1}^{\tau_2}\cap\lbrace r=R \rbrace $ vanishes since $ \zeta_R \equiv 0 $ on $ \lbrace r=R\rbrace $, whereas  the remaining terms are positive definite in $ \check{I}_{\tau_1}^{\tau_2} $ for $ p  \leq 2 $ and $ R $ large enough.
		
		Now, when it comes to the error terms in equations (\ref{Misc r_p 1}, \ref{Misc r_p 2}) coming from the region $ R\leq r\leq R+1 $, as quadratic terms of the 1st-order jet of $ \Psi_i $,  we control them in terms of the T-flux by Theorem \ref{Morawetz-Spacetime}, as long as $ R> 2M $.
		
		Finally, for the boundary terms on the right-hand side of (\ref{Misc r_p 1}) we have \begin{align}
			\begin{aligned}
				\int_{\partial\check{I}_{\tau_1}^{\tau_2}} J_{\mu}^{Z}[\zeta_R\cdot\Phi_i] = \int_{\check{O}_{\tau_1}}r^p  \dfrac{\left [\partial_v (\zeta_R\Phi_i)\right ]^2}{r^2} - \int_{\check{O}_{\tau_2}}r^p  \dfrac{\left [\partial_v (\zeta_R\Phi_i)\right ]^2}{r^2} - \int_{\mathcal{I}^+} \dfrac{D}{4}\left(\abs{\slashed{\nabla}\Psi_i}^2+V_i\Psi_i^2\right) 
			\end{aligned}
		\end{align}
		and seeing that the last two terms have the right sign, we conclude the proof.
		
	\end{proof}
	
	\paragraph{Estimates for the non-degenerate energy near null infinity $ \mathcal{I}^+ $.} Using the proposition above, we obtain local integrated energy decay estimates near future null infinity $ \mathcal{I}^{+}.$ 
	\begin{proposition} \label{control N flux at null infinity}
		There exists a positive constant $ C $ depending on $ \check{\Sigma}_{\tau_0}, \ M $, such that for all solution $ \Psi_i^{^{(\ell)}} $ to (\ref{waveeq}), supported on the fixed frequency $\ell\geq i, \  i\in \lbrace 1,2 \rbrace$, we have \begin{align}
			\int_{\tau_1}^{\tau_2} \left(\int_{\check{O}_{\tau}}J_{\mu}^{N}[\Psi_i]n^{\mu}_{\check{O}_{\tau}}\right) d\tau\ \leq\ &C \left(  \int_{\check{\Sigma}_{\tau_1}}J_{\mu}^{T}[\Psi_i]n_{\check{\Sigma}_{\tau_1}}^{\mu} + \int_{\check{O}_{\tau_1}} \dfrac{1}{r}(\partial_v\Phi_i)^2 \right), \label{misc r_p 3} \\
			\int_{\tau_1}^{\tau_2} \left(\int_{\check{O}_{\tau}} \dfrac{1}{r}(\partial_v\Phi_i)^2\right) d\tau\ \leq\ &C \left(  \int_{\check{\Sigma}_{\tau_1}}J_{\mu}^{T}[\Psi_i]n_{\check{\Sigma}_{\tau_1}}^{\mu} + \int_{\check{O}_{\tau_1}} (\partial_v\Phi_i)^2 \right).		\label{misc r_p 4}	
		\end{align}
	\end{proposition}
	\begin{proof}
		Using proposition \ref{r_p} for $ p=1 $ we obtain \begin{align}
			\begin{aligned}
				\int_{\check{I}_{\tau_1}^{\tau_2}}  \dfrac{(\partial_v\Phi_i)^2}{r^2} +  \int_{\check{I}_{\tau_1}^{\tau_2}}  \dfrac{1}{4}\sqrt{D}\left(1  - 3\dfrac{M}{r}\right)&\left(\abs{\slashed{\nabla}\Psi_i}^2 + V_i \Psi_i^2\right) 
				+  \int_{\check{I}_{\tau_1}^{\tau_2}} \dfrac{1}{4}DV_i  \Psi_i^2\\ \leq \ & C \int_{\check{\Sigma}_{\tau_1}}J_{\mu}^{T}[\Psi_i]n_{\check{\Sigma}_{\tau_1}}^{\mu}+
				\int_{\check{O}_{\tau_1}} \dfrac{1}{r^2}(\partial_v \Phi_i)^2.
			\end{aligned}
		\end{align}
		However, in view of $ V_i = \mathcal{O}\left (\frac{1}{r^3}\right ) $, choosing $ R  $ large enough and using Poincare inequality we can absorb the last term of the left-hand side in the angular derivative term. In addition, the first term expands as \begin{align}
			\dfrac{1}{r^2} (\partial_v\Phi_i)^2 = \dfrac{1}{r^2}\left [\partial_v(r\Psi_i)\right ]^2 =     (\partial_v\Psi_i)^2 + \dfrac{D^2}{4r^2}\Psi_i^2 + \dfrac{D}{r}\Psi_i \partial_v\Psi_i  
		\end{align}
		and in order to treat the last term above, using the cut-off smooth function introduced earlier, we write \begin{align}
			\begin{aligned}
				\int_{\check{I}_{\tau_1}^{\tau_2}} \dfrac{D}{r}(\zeta_R \Psi_i) \partial_v (\zeta_R \Psi_i) = \dfrac{1}{2}	\int_{\check{I}_{\tau_1}^{\tau_2}} \dfrac{D}{r}\partial_v \left[(\zeta_R \Psi_i)^2\right]  &= \\
				= - \dfrac{1}{2}\int_{\check{I}_{\tau_1}^{\tau_2}} \partial_v\left(\dfrac{D}{r}\right) (\zeta_R \Psi_i)^2 - \int_{\check{I}_{\tau_1}^{\tau_2} \cap \lbrace r= R\rbrace} & \dfrac{\sqrt{D}D}{4r}(\zeta_R \Psi_i)^2 + \int_{\mathcal{I}^+} \dfrac{D^2}{4r}(\zeta_R \Psi_i)^2 
			\end{aligned}
		\end{align} 
		The middle term in the last line vanishes, and since $ \partial_v (D/r) = D\sqrt{D}(2-3\sqrt{D})/2r^2 <0 $ for large values of $ R $, the remaining two terms above are positive definite. 
		
		All error terms that occur in the region $ \br{R\leq r\leq R+1} $ are controlled in terms of the $ T-$flux along $ \check{\Sigma}_{\tau_{1}} $ using theorem \ref{Morawetz-Spacetime}, and combining all estimates above we show 
		\begin{align}
			\begin{aligned} \label{misc r_p 2}
				\int_{\check{I}_{\tau_1}^{\tau_2}} (\partial_v\Psi_i)^2 +  \dfrac{1}{2}&\left(\abs{\slashed{\nabla}\Psi_i}^2 + V_i \Psi_i^2\right) 
				\leq & C\left(  \int_{\check{\Sigma}_{\tau_1}}J_{\mu}^{T}[\Psi_i]n_{\check{\Sigma}_{\tau_1}}^{\mu}+
				\int_{\check{O}_{\tau_1}} \dfrac{1}{r^2}(\partial_v \Phi_i)^2 \right). 
			\end{aligned}
		\end{align}
		On the other hand, since $ n
		_{\check{O}_{\tau}} \equiv \frac{\partial}{\partial v},  $ we have  
		\[ J_{\mu}^{N}[\Psi_i]n^{\mu}_{\check{O}_{\tau}} =(\partial_v \Psi_i)^2 + \dfrac{D}{2}\left(\abs{\slashed{\nabla}\Psi_i}^2 + V_i \Psi_i^2\right),  \]
		thus using the coarea formula for the spacetime term  of (\ref{misc r_p 2}) we conclude the proof of (\ref{misc r_p 3}). 
		\begin{flushright}
			$ \blacksquare$
		\end{flushright}
		In order to prove (\ref{misc r_p 4}), we use proposition \ref{r_p} for $ p=2 $ and we obtain \begin{align}
			\begin{aligned} \label{misc r_p 11}
				\int_{\check{I}_{\tau_1}^{\tau_2}}  2\dfrac{(\partial_v\Phi_i)^2}{r} \leq C  \int_{\check{\Sigma}_{\tau_1}}J_{\mu}^{T}[\Psi_i]n_{\check{\Sigma}_{\tau_1}}^{\mu} + \int_{\check{O}_{\tau_1}} (\partial_v\Phi_i)^2 + \int_{\check{I}_{\tau_1}^{\tau_2}} \sqrt{D}\dfrac{M}{2}\left(\abs{\slashed{\nabla}\Psi_i}^2 + V_i \Psi_i^2\right)  \\+  \int_{\check{I}_{\tau_1}^{\tau_2}}  \dfrac{D}{4} (- r\cdot V_i)\Psi_i^2.	
			\end{aligned}
		\end{align}
		By Poincare's inequality, we write \begin{align}
			\int_{\check{I}_{\tau_1}^{\tau_2}}  \dfrac{D}{4}\abs{r\cdot V_i}\Psi_i^2 =  \int_{\check{I}_{\tau_1}^{\tau_2}}  D\dfrac{M}{2}\dfrac{\Big|3\sqrt{D}+(-1)^i(\ell+2-i)\Big|}{\ell(\ell+1)}\abs{\slashed{\nabla}\Psi_i}^2 \leq 10\int_{\check{I}_{\tau_1}^{\tau_2}} D\dfrac{M}{2} \abs{\slashed{\nabla}\Psi_i}^2,
		\end{align} 
		for all $ \ell\geq i $. Thus, the last two terms of (\ref{misc r_p 11}) are bounded by the right-hand side of  (\ref{misc r_p 2}) which concludes the proof.

	\end{proof}
	\subsection{Decay of Energy} With the above estimates in hand, we are ready to show decay for the non-degenerate energy of $ \Psi_i^{^{(\ell)}} $, for all $ \ell \geq i. $ 
	First,  Theorem \ref{remove_degeneracy} and coarea formula yield near the event horizon $ \mathcal{H}^+ $ \begin{align*}
		\int_{\tau_1}^{\tau_2} \left(\int_{\mathcal{A}_c\cap \check{\Sigma}_{\tau}}J_{\mu}^N[\Psi_i]n_{\check{\Sigma}_{\tau}}^{\mu}\right)  d\tau \leq  C \left(  \int_{\check{\Sigma}_{\tau_1}} J_{\mu}^{N}[\Psi_i] n_{\check{\Sigma}_{\tau_1}}^{\mu}+ \int_{\check{\Sigma}_{\tau_1}} J_{\mu}^{N}[T \Psi_i] n_{\check{\Sigma}_{\tau_1}}^{\mu}+ \int_{\check{\Sigma}_{\tau_1}\cap \mathcal{A}_c}  J_{\mu}^{N}\left[\partial_{r} \Psi_i\right] n_{\check{\Sigma}_{\tau_1}}^{\mu}\right)
	\end{align*} 
	On the other hand, close to future null infinity $ \mathcal{I}^+ $ we have  estimate (\ref{misc r_p 3}) and using 
	Theorem \ref{Spacetime non degenerate photon estimate} to control the energy in the  intermediate region $ \lbrace r_c\leq r \leq R \rbrace \cap \check{\Sigma}_{\tau} $, we obtain 
	\begin{align} \label{NonDEB}
		\int_{\tau_1}^{\tau_2} \left(\int_{\check{\Sigma}_{\tau}}J_{\mu}^N[\Psi_i^{^{(\ell)}}]n_{\check{\Sigma}_{\tau}}^{\mu}\right)  d\tau \leq  C  \cdot \mathcal{E}[\Psi_i^{^{(\ell)}}](\tau_1), \hspace{1cm} \ell\geq i,  
	\end{align}
	for a positive constant $ C = C(M,\check{\Sigma}_0)$, where 
	\begin{align*}
		\mathcal{E}[\Psi_i](\tau) :=  \int_{\check{\Sigma}_{\tau}} J_{\mu}^{N}[\Psi_i] n_{\check{\Sigma}_{\tau}}^{\mu}+ \int_{\check{\Sigma}_{\tau}} J_{\mu}^{N}[T \Psi_i] n_{\check{\Sigma}_{\tau}}^{\mu}+ \int_{\check{\Sigma}_{\tau}\cap \mathcal{A}_c}  J_{\mu}^{N}\left[\partial_{r} \Psi_i\right] n_{\check{\Sigma}_{\tau}}^{\mu}+  \int_{\check{O}_{\tau}}  r^{-1}(\partial_v\Phi_i)^2.
	\end{align*} 
	Last, in order to improve the decay of the non-degenerate energy we also need an estimate for
	\[ \int_{\tau_1}^{\tau_2} 	\mathcal{E}[\Psi_i](\tau) d\tau.\]
	Theorem \ref{high-order horizon estimates} for $ k=2 $ and coarea formula yield \begin{align*}
		\int_{\tau_1}^{\tau_2} \int_{\check{\Sigma}_{\tau}\cap \mathcal{A}_c}  J_{\mu}^{N}\left[\partial_{r} \Psi_i\right] n_{\check{\Sigma}_{\tau}}^{\mu} \leq 	C \left(  \sum_{j=0}^{2} \int_{\check{\Sigma}_{\tau_1}} J_{\mu}^{N}\left[T^{j} \Psi_i\right] n_{\check{\Sigma}_{\tau_1}}^{\mu}+ \sum_{j=1}^{2} \int_{\check{\Sigma}_{\tau_1} \cap \mathcal{A}_c} J_{\mu}^{N}\left[\partial_{r}^{j} \Psi_i\right] n_{\check{\Sigma}_{\tau_1}}^{\mu}\right),
	\end{align*}
	for all $ \ell \geq 2+(-1)^{i} $. 
	Thus, applying (\ref{NonDEB}) for both $ \Psi_i, T \Psi_i $ (note $ T\Psi_i $ also satisfies (\ref{waveeq})), using  (\ref{misc r_p 4})   and the estimate above, we obtain \begin{align} \label{decay misc_1}
		\int_{\tau_1}^{\tau_2} 	\mathcal{E}[\Psi_i](\tau) d\tau \leq C\left(\mathcal{E}[\Psi_i](\tau_1) + \mathcal{E}[T\Psi_i](\tau_1) + \int_{\check{\Sigma}_{\tau}\cap \mathcal{A}_c}  J_{\mu}^{N}\left[\partial_r\partial_{r} \Psi_i\right] n_{\check{\Sigma}_{\tau}}^{\mu} +\int_{\check{O}_{\tau_1}} (\partial_v(r\Psi_i))^2 \right), 
	\end{align}
	for all $ \ell\geq 2+(-1)^{i} $. Note, the above estimate holds for all frequencies $ \ell\geq 1 $ in the $ i=1 $ case, however, in the $ i=2 $ case  it holds only for $ \ell \geq 3.$  Thus, using an effective method introduced by Dafermos-Rodnianski in \cite{dafermos2010new}, we obtain the following inverse polynomial energy decay estimates.
	
	\begin{proposition} \label{non-degen decay}
		There exists a positive constant $ C $ depending only on $\check{\Sigma}_0, M $ such that for all solutions $ \Psi_i^{^{(\ell)}} $ to (\ref{waveeq}), supported on the fixed frequency $ \ell \geq i, \ i \in \br{1,2} $, the following inverse linear decay holds \begin{align} \label{t^-1 decay}
			\int_{\check{\Sigma}_{\tau}}J_{\mu}^N[\Psi_i]n_{\check{\Sigma}_{\tau}}^{\mu} \leq C \mathcal{E}[\Psi_i](0) \cdot \dfrac{1}{\tau}
		\end{align}
		Moreover, if $ \ell \geq 2 + (-1)^i $,  we have the improved inverse polynomial decay 
		\begin{align} \label{t^-2 decay}
			\int_{\check{\Sigma}_{\tau}}J_{\mu}^N[\Psi_i]n_{\check{\Sigma}_{\tau}}^{\mu} \leq C \bar{\mathcal{E}}[\Psi_i](0) \cdot \dfrac{1}{\tau^2},
		\end{align}
		where  $ \bar{\mathcal{E}}[\Psi_i](0)  $ is equal to the right-hand side of (\ref{decay misc_1}) for $ \tau_1=0. $ \\ 
	\end{proposition}

	\begin{proof}
		Applying (\ref{NonDEB})  for $ \tau_1 = 0 $ and $ \tau_2=\tau $ for any $ \tau\geq 0, $ we obtain \begin{align}
			\int_0^{\tau}\left(\int_{\check{\Sigma}_{\tau}}J_{\mu}^N[\Psi_i]n_{\check{\Sigma}_{\tau}}^{\mu}\right) \leq C \bar{\mathcal{E}}[\Psi_i](0)
		\end{align}
		Using the fundamental theorem of calculus and the uniform boundedness of the non-degenerate energy of Theorem \ref{Nuniform }, i.e. \[ 
		\int_{\check{\Sigma}_{s_2}}J_{\mu}^N[\Psi_i]n_{\check{\Sigma}_{s_2}}^{\mu} \leq C 	\int_{\check{\Sigma}_{s_1}}J_{\mu}^N[\Psi_i]n_{\check{\Sigma}_{s_1}}^{\mu}, \ \ \forall s_1\leq s_2, \]
		we obtain \begin{align}
			(\tau-0)\int_{\check{\Sigma}_{\tau}}J_{\mu}^N[\Psi_i]n_{\check{\Sigma}_{\tau}}^{\mu} \leq C 	\mathcal{E}[\Psi_i](0),  
		\end{align}
		which readily proves (\ref{t^-1 decay}).
		
		Now,  let $ \ell \geq 2 +(-1)^i $ and considering (\ref{decay misc_1}) for $ \tau_1=0 $ while taking $ \tau_2 \rightarrow +\infty $ and we obtain \begin{align}
			\int_0^T \mathcal{E}[\Psi_i](\tau)d\tau \leq C \bar{\mathcal{E}}[\Psi_i](0), \hspace{1cm} \forall \ T\geq 0.
		\end{align}
		This allows us to find a sequence $\lbrace \tau_n \rbrace_{n\in \mathbb{N}}$ with $ \tau_0\geq1 $ and $ 2\tau_n\leq \tau_{n+1}\leq 3\tau_n  $ such that \begin{align} \label{dyadic}
			\mathcal{E}[\Psi_i](\tau_n)\leq \dfrac{1}{\tau_n} C \cdot \bar{\mathcal{E}}[\Psi_i](0).
		\end{align}
		Note, the above sequence has the  property that $ \tau_n \sim \tau_{n+1}\sim (\tau_{n+1}-\tau_n )$, uniformly for any $ n\in\mathbb{N}, $ which also makes it unbounded, i.e. $ \tau_{n} \xrightarrow{n\to \infty} \infty$.
		\\ So, applying (\ref{NonDEB}) for the interval $ [\tau_n,\tau_{n+1}] $  and using relation (\ref{dyadic})
		we get
		\begin{align}
			\int_{\tau_n}^{\tau_{n+1}} \left(\int_{\check{\Sigma}_{\tau}}J_{\mu}^N[\Psi_i]n_{\check{\Sigma}_{\tau}}^{\mu}\right)  d\tau \leq  C	\mathcal{E}[\Psi_i](\tau_n) \leq \dfrac{1}{\tau_n}\bar{C} \bar{\mathcal{E}}[\Psi_i](0).
		\end{align}
		Thus, using the fundamental theorem of calculus and uniform energy boundedness we obtain
		\begin{align}
			\begin{aligned}
				(\tau_{n+1}-\tau_n)	\int_{\check{\Sigma}_{\tau_{n+1}}}J_{\mu}^N[\Psi_i]n_{\check{\Sigma}_{\tau_{n+1}}}^{\mu}   d\tau \leq   \dfrac{1}{\tau_n}\bar{C} \bar{\mathcal{E}}[\Psi_i](0)\\
				\Rightarrow \ \ \int_{\check{\Sigma}_{\tau_{n+1}}}J_{\mu}^N[\Psi_i]n_{\check{\Sigma}_{\tau_{n+1}}}^{\mu}   d\tau \leq \dfrac{1}{\tau_{n+1}^2}C \bar{\mathcal{E}}[\Psi_i](0), \ \ \ \forall \in \mathbb{N},
			\end{aligned}
		\end{align}
		where we used the dyadic property of the sequence above. Now, note that for any $ \tau\geq 1 $, there exists $ n\in\mathbb{N} $ such that $ \tau_n\leq \tau \leq \tau_{n+1} $ and using the above estimate we finally get \begin{align}
			\begin{aligned}
				\int_{\check{\Sigma}_{\tau}}J_{\mu}^N[\Psi_i]n_{\check{\Sigma}_{\tau}}^{\mu}   d\tau  \leq C' \int_{\check{\Sigma}_{\tau_{n}}}J_{\mu}^N[\Psi_i]n_{\check{\Sigma}_{\tau_n}}^{\mu}   d\tau \leq \dfrac{1}{\tau_n^2}\tilde{C} \bar{\mathcal{E}}[\Psi_i](0) \sim \dfrac{1}{\tau_{n+1}^2} \tilde{C} \bar{\mathcal{E}}[\Psi_i](0) \leq \dfrac{1}{\tau^2}C\bar{\mathcal{E}}[\Psi_i](0).
			\end{aligned}
		\end{align}

	\end{proof}
	In order to prove pointwise decay estimates that are $ L^2(\tau_0,\infty) $ integrable in time, we need the non-degenerate energy to decay like $ 1/\tau^2 $. However, we see above that we don't have such decay for the energy of $ \Psi_2^{(2)} $, i.e. $  i=2,\ \ell =2. $ We remedy this situation by working with the \textbf{degenerate} energy, $ J_{\mu}^T[\Psi_i] $, which we show decays like $ \frac{1}{\tau^{2}} $, and using certain interpolation inequalities we manage to show adequate pointwise decay for all frequencies $ \ell\geq i,\ i\in \br{1,2} $.

	\begin{proposition}
		There exists a positive constant $ C $ depending only on $\check{\Sigma}_0, M $ such that for all solutions $ \Psi_i^{^{(\ell)}} $ to (\ref{waveeq}), supported on the fixed frequency $ \ell\geq i $, $ i\in \lbrace 1,2 \rbrace$, the following decay estimate holds \begin{align} \label{T-t^-1 decay}
			\int_{\check{\Sigma}_{\tau}}J_{\mu}^T[\Psi_i]n_{\check{\Sigma}_{\tau}}^{\mu} \leq C  \bar{\mathcal{E}}^T[\Psi_i](0) \cdot \dfrac{1}{\tau^2},
		\end{align}
		where $\bar{\mathcal{E}}^T[\Psi_i](\tau)  $ is defined in the proof below.
	\end{proposition}
	\begin{proof}
		Using Proposition \ref{N derivatives positivity}, Theorem \ref{Nuniform } and coarea formula we obtain for any $ \tau_1<\tau_2 $\begin{align*}
			\int_{\tau_1}^{\tau_2}\left(\int_{\check{\Sigma}_{\tau}\cap A_N}(\partial_v\Psi_i^2)+ \sqrt{D}(\partial_r\Psi_i^2)+\abs{\slashed{\nabla}\Psi_i}^2+V_i\Psi_i^2\right)d\tau\leq C \int_{\check{\Sigma}_{\tau_1}}J_{\mu}^N[\Psi_i] n ^{\mu}_{\check{\Sigma}}
		\end{align*}
		However, since $ \sqrt{D}\geq D $ for all $ r\geq M $, we have \begin{align*}
			\int_{\tau_1}^{\tau_2}\left( \int_{\check{\Sigma}_{\tau}\cap A_N}J_{\mu}^T[\Psi_i]n_{\check{\Sigma}_{\tau}}^{\mu}  \right) \leq C \int_{\check{\Sigma}_{\tau_1}}J_{\mu}^N[\Psi_i] n ^{\mu}_{\check{\Sigma}}. 
		\end{align*}
		Thus, using the above estimate, 
		Theorem \ref{Spacetime non degenerate photon estimate} and (\ref{misc r_p 3}) for $ T\sim N $ readily yields \begin{align}
			\int_{\tau_1}^{\tau_2}\left( \int_{\check{\Sigma}_{\tau}}J_{\mu}^T[\Psi_i]n_{\check{\Sigma}_{\tau}}^{\mu}  \right) \leq C \mathcal{E}^T[\Psi_i](\tau_1),
		\end{align}
		with \[ 	\mathcal{E}^T[\Psi_i](\tau) :=  \int_{\check{\Sigma}_{\tau}} J_{\mu}^{N}[\Psi_i] n_{\check{\Sigma}_{\tau}}^{\mu}+ \int_{\check{\Sigma}_{\tau}} J_{\mu}^{N}[T \Psi_i] n_{\check{\Sigma}_{\tau}}^{\mu}+  \int_{\check{O}_{\tau}}  r^{-1}(\partial_v\Phi_i)^2. \]
		In addition, note \begin{align}
			\int_{\tau_1}^{\tau_2} 	\mathcal{E}^T[\Psi_i](\tau) d\tau \leq C \bar{\mathcal{E}}^T[\Psi_i](\tau_1), \ \ \forall \ell\geq i, \ i \in \br{1,2}
		\end{align}
		where, \[ \bar{\mathcal{E}}^T[\Psi_i](\tau):= \left(\mathcal{E}[\Psi_i](\tau) + \mathcal{E}[T\Psi_i](\tau) +\int_{\check{O}_{\tau}} (\partial_v\Phi_i)^2 \right) \]
		
		Thus, using the uniform boundedness of degenerate energy and following the same idea for a dyadic sequence $ \lbrace\tau_n\rbrace_{n\in\mathbb{N}} $ as the proposition above we prove the corresponding decay result.
		
	\end{proof}

	\section{Decay, Non-Decay and Blow-up for Regge--Wheeler solutions.} \label{Section - Regge Wheeler estimates}
	We have reached the point where we can finally state and prove pointwise estimates for solutions to the Regge--Wheeler system. First, we show how to obtain pointwise decay estimates for solutions $ \Psi_i^{^{(\ell)}}, \ \ell \geq i, \ i \in \br{1,2} $ to (\ref{waveeq}), on the exterior up to and including the event horizon $ \h. $
	We then generalize the energy results of the previous section to obtain higher-order pointwise estimates accordingly.
	Such findings are essential later on when we prove non-decay and growth estimates for translation invariant derivatives of $ \Psi_i^{^{(\ell)}} $,  asymptotically along the event horizon $ \h. $ 
	Finally, using elliptic identities we pass these estimates to the corresponding gauge invariant quantities $ \bm{q}^{F}, \bm{p} $ (and $ \underline{\bm{q}}^{F}, \underline{\bm{p}} $) satisfying the Regge--Wheeler system (\ref{csystem}).

	\subsection{Pointwise estimates} With the energy decay results obtained in the previous section, we can now prove pointwise estimates for $ \Psi_i^{^{(\ell)}}, $  $ i\in \lbrace 1,2\rbrace $, for all $ \ell \geq i. $ The results are based on the Sobolev inequality and the spherical symmetry of our spacetime that allow us to use the angular momentum operators $ \Omega_j $, $ j\in \br{1,2,3} $ as commutators.
	
	Working on the foliation $ \Sigma_{\tau} $ (or $ \check{\Sigma}_{\tau} $) with the associated coordinate system $ (\rho,\omega) $, for any $ r\geq M $ we have \begin{align*}
		\Psi_i^2(r,\omega) = -\int_{r}^{\infty}\partial_{\rho}(\Psi_i^2)d\rho = -2\int_{r}^{\infty}\Psi_i\cdot \partial_{\rho}\Psi_i d\rho = -2\int_{r}^{\infty}k \dfrac{\Psi_i}{\sqrt{\rho}} (\partial_v\Psi_i \cdot \sqrt{\rho})d\rho  - 2\int_{r}^{\infty}\Psi_i\cdot \partial_r\Psi_i d\rho
	\end{align*}
	where in the last equality we used that $ \partial_{\rho} = k(r)\partial_v + \partial_{r} $, for $ k(r) $ bounded; see Section \ref{Geometry}. Then, Cauchy-Schwartz and H\"{o}lder inequality  (for both $ p=\infty, q=1 $ and $ p=q=2 $ cases) yield \begin{align*}
		\Psi_i^2(r,\omega) \leq & C  \left (\int_{r}^{\infty} \dfrac{\Psi_i^2 }{\rho} d\rho + \int_{r}^{\infty} (\partial_v\Psi_i)^2 \cdot \rho d\rho \right ) +  C\left (\int_{r}^{\infty} \Psi_i^2  d\rho \right )^{1/2}\left (\int_{r}^{\infty} (\partial_r\Psi_i)^2  d\rho \right )^{1/2} \\
		\leq & \dfrac{C}{r} \int_{r}^{\infty} \dfrac{\Psi_i^2}{\rho^2} \cdot \rho^2  d\rho  + \dfrac{C}{r}\int_{r}^{\infty} (\partial_v\Psi_i)^2 \cdot \rho^2  d\rho + 
		\dfrac{C}{r}\left (\int_{r}^{\infty} \dfrac{\Psi_i^2}{\rho^2} \cdot \rho^2 d\rho \right )^{1/2}\left (\int_{r}^{\infty} (\partial_r\Psi_i)^2 \cdot \rho^2 d\rho \right )^{1/2}
	\end{align*}
	for a constant C depending only on $ M,\Sigma_0. $
	Thus, integrating upon the sphere $ \mathbb{S}^2 $ and using again H\"{o}lder inequality ($ p=2,q=2 $) for the last term above, we obtain 
	\begin{align*}
		\int_{\mathbb{S}^2}	\Psi_i^2(r,\omega) d\omega\ \leq \  &  \dfrac{C }{r} \int_{\mathbb{S}^2}\int_{r}^{\infty} \dfrac{\Psi_i^2}{\rho^2} \rho^2 d\rho d\omega + \dfrac{C}{r}\int_{\mathbb{S}^2}\int_{r}^{\infty} \partial_v\Psi_i^2 \rho^2 d\rho d\omega  \\ + &\dfrac{C}{r} \left (\int_{\mathbb{S}^2}\int_{r}^{\infty} \dfrac{\Psi_i^2}{\rho^2} \rho^2 d\rho d\omega \right )^{1/2}\left (\int_{\mathbb{S}^2}\int_{r}^{\infty} (\partial_r\Psi_i)^2 \rho^2 d\rho d\omega \right )^{1/2} \\
		\leq & \ \dfrac{C}{r}\int_{\check{\Sigma}_{\tau}} J_{\mu}^T[\Psi_i]n_{\check{\Sigma}}^{\mu} + \dfrac{C}{r}\left(\int_{\check{\Sigma}_{\tau}} J_{\mu}^T[\Psi_i]n_{\check{\Sigma}}^{\mu}\right)^{1/2} \left(\int_{\check{\Sigma}_{\tau}} J_{\mu}^N[\Psi_i]n_{\check{\Sigma}}^{\mu}\right)^{1/2} 
	\end{align*}
	Note, we are using the non-degenerate flux $ J^{N}_{\mu}[\Psi_i] $ in the last term above since $ \partial_r\Psi_i $ appears without a $ D(r) $ factor in front of it.
	Thus, using proposition \ref{non-degen decay} and relation (\ref{T-t^-1 decay}) we show that for all $ r\geq M $,  \begin{align*}
		\int_{\mathbb{S}^2}	\Psi_i^2(r,\omega) d\omega\ \leq \dfrac{C}{r} \bar{\mathcal{E}}^T[\Psi_i](0) \dfrac{1}{\tau^2} +  \dfrac{C}{r}\Big(  \bar{\mathcal{E}}^T[\Psi_i](0) \mathcal{E}[\Psi_i](0) \Big)^{1/2}\dfrac{1}{\tau^{3/2} } 
	\end{align*}
	\begin{align}
		\Rightarrow \ \ \ 	\int_{\mathbb{S}^2}	\Psi_i^2(r,\omega) d\omega\ \leq C \mathcal{E}_1[\Psi_i](0) \dfrac{1}{ r\cdot\tau^{\frac{3}{2}}} \label{L^2 decay 1}
	\end{align}
	where $ \mathcal{E}_1[\Psi_i] $ is an expression involving initial data of $ \Psi_i^{(\ell)}. $ 
	Note the above decay  estimate holds for all $ \ell \geq i, \ i\in \br{1,2}.  $ 
	We now present our first pointwise estimate.
	\begin{theorem}\label{allfreqdecay}
		There exists a positive constant $ C $ depending only on $M,\check{\Sigma}_0 $ such that for all solutions $ \Psi_i^{^{(\ell)}} $ to (\ref{waveeq}), supported on the fixed frequency $ \ell \geq i, \ i \in \br{1,2} $, there are expressions $ \mathcal{E}_1, \mathcal{E}_2 $ in terms of norms of initial data of $ \Psi_i^{^{(\ell)}} $ such that \begin{align}
			\abs{\Psi_i^{^{(\ell)}}}(\tau,r) \leq C \sqrt{\mathcal{E}_1}\dfrac{1}{\sqrt{r}\cdot \tau^{\frac{3}{4}}}, \ \ \ \forall r\geq M, \ \tau\geq 1.
		\end{align}
		Moreover, for any $ \ell \geq 2 + (-1)^{i} $ we have instead \begin{align}
			\abs{\Psi_i^{^{(\ell)}}} (\tau, r) \leq C \sqrt{\mathcal{E}}_2 \dfrac{1}{\sqrt{r}\cdot \tau}, \hspace{1cm} \forall  r\geq M, \ \tau \geq 1.
		\end{align}
	\end{theorem}
	\begin{proof}
		Commuting the wave equation (\ref{waveeq}) with the Killing angular momentum operators $ \Omega_j $ we obtain the corresponding estimate (\ref{L^2 decay 1}) and using the Sobolev embedding theorem on the sphere $ \mathbb{S}^2 $, we get \begin{align*}
			\abs{\Psi_i}^2 \leq C \mathcal{E}_1 \dfrac{1}{ r\cdot\tau^{\frac{3}{2}}}, \ \ \ \ \forall r\geq M, \tau\geq 1,
		\end{align*}
		where $ \mathcal{E}_1 \:= \sum_{\abs{a}\leq 2}\sum_{j}\mathcal{E}_1[\Omega_j^a \Psi_i] $. 
		
		Now, consider $ \ell \geq 2 + (-1)^i $ and using H\"{o}lder inequality we write for any $ r\geq M $
		
		\begin{align*}
			\Psi_i^2(r,\omega) = \left(\int_r^{\infty}(\partial_\rho\Psi_i)d\rho\right)^2 = \left(\int_r^{\infty}(\partial_\rho\Psi_i) \dfrac{1}{\rho} \cdot \rho   d\rho\right)^2 \leq \left(\int_r^{\infty}\dfrac{1}{\rho^2}d\rho\right)\left(\int_r^{\infty}(\partial_{\rho}\Psi_i)^2 \rho^2 d\rho\right).  
		\end{align*}
		Integrating upon the sphere $ \mathbb{S}^2 $ we obtain  \begin{align} 
			\int_{\mathbb{S}^2}\Psi_i^2(r,\omega) d\omega \leq \dfrac{1}{r} 	\int_{\mathbb{S}^2} \int_r^{\infty}(\partial_{\rho}\Psi_i)^2 \rho^2 d\rho d\omega \leq \dfrac{1}{r} \int_{\check{\Sigma}_{\tau}\cap\lbrace r'\geq r \rbrace } J_{\mu}^{N}[\Psi_i]n_{\check{\Sigma}}^{\mu} \label{N-decay t}
		\end{align}
		Thus,  using (\ref{t^-2 decay}) and (\ref{N-decay t}) we obtain \begin{align*}
			\int_{\mathbb{S}^2}\Psi_i^2(r,\omega) d\omega \leq  C \bar{\mathcal{E}}[\Psi_i](0) \dfrac{1}{\tau^2 \cdot r}.
		\end{align*}
		Once again, from Sobolev embedding theorem we get \begin{align*}
			\abs{\Psi_i^{^{(\ell)}}} \leq C \mathcal{E}_2 \dfrac{1}{\sqrt{r}\cdot \tau}, \ \ \ \forall\  r\geq M, \ \tau\geq 1, \quad  \ell \geq 2 + (-1)^i
		\end{align*}
		for $  \mathcal{E}_2 := \sum_{\abs{a}\leq 2}\sum_{j}\bar{\mathcal{E}}[\Omega_j^a \Psi_i](0).  $
	\end{proof}

	\subsection{Higher order pointwise estimates.} Using the ideas introduced earlier we show pointwise decay results for the derivatives $ \partial_r^k\Psi_i^{^{(\ell)}} $, for all $ \ell \geq i, \ i\in\br{1,2} $ and  $ k\leq \ell+(-1)^{i+1}, $ near and including the event horizon $ \h. $ Again, pointwise decay will follow from showing energy decay first.  
	In particular, we have the following proposition 
	\begin{proposition} \label{HED Pro}
		Let us fix $ \bar{R}>M, \ \ell\in \mathbb{N},\ \ell\geq i,\ i\in\br{1,2} $ and consider $ \tau\geq 1, $ then there exists a positive constant $ C $ depending on $ M,k,\bar{R} $ and $ \check{\Sigma}_0 $ such that for all solutions $ \p^{^{(\ell)}}$  to (\ref{waveeq}), supported on the fixed frequency $ \ell\geq i, $ there exist norms $ \mathcal{E}^{k}_{i,\ell} $ of the initial data such that \begin{align}
			\int_{\check{\Sigma}_{\tau}\cap \br{M\leq r\leq \bar{R}}}J_{\mu}^{N}[\partial_r^{k}\p]n_{\check{\Sigma}_{\tau}}^{\mu} &\leq C \ \mathcal{E}^{k}_{i,\ell}\ \dfrac{1}{\tau}, \hspace{1cm} \ k={\ell-1-(-1)^{i}} \label{HED l-1} \\
			\int_{\check{\Sigma}_{\tau}\cap \br{M\leq r\leq \bar{R}}}J_{\mu}^{N}[\partial_r^{k}\p]n_{\check{\Sigma}_{\tau}}^{\mu} &\leq C \ \mathcal{E}^{k}_{i,\ell}\ \dfrac{1}{\tau^2}, \hspace{1cm} k\leq (\ell-2)-(-1)^{i} \label{HED <l-1}
		\end{align}
	\end{proposition}
	\begin{proof}
		For any $ r_c>M $, in the region $ \check{\Sigma}_{\tau}\cap \br{r_c\leq r\leq\bar{R}} $ we can produce the above decay results by commuting the wave equation (\ref{waveeq}) with the killing vector field $ T $, applying local elliptic estimates and using the decay results obtained in proposition \ref{non-degen decay}.
		Thus, we focus on showing decay in $ \check{\Sigma}_0\cap \mathcal{A} $, for $ \mathcal{A} $ an  $ \phi^T_{\tau}- $invariant neighborhood of $ \h. $  
		
		Using estimate  (\ref{H^2- higher control}) of theorem \ref{high-order horizon estimates} for all $ s \leq \ell-(-1)^{i}$ and coarea formula yields the estimate   \begin{align}
			\sum_{m=0}^{^{\ell-1-(-1)^{i}}}	 \int_{0}^{T}\left(\int_{\check{\Sigma}_{\tilde{\tau}}\cap \mathcal{A}_c} J_{\mu}^{N}[\partial_r^{m}\p] \ n_{\check{\Sigma}_{\tilde{\tau}}}^{\mu}\right) d\tilde{\tau} \leq C \ \mathcal{E}^{^{\ell-(-1)^{i}}}_{_{energy}}[\p](0), 
		\end{align}
		for any $ T\geq0 $,  where \[ \mathcal{E}_{_{energy}}^{q}[\p](\tau) := \sum_{j=0}^{q} \int_{\check{\Sigma}_{\tau}} J_{\mu}^{N}\left[T^{j} \Psi_i\right] n_{\check{\Sigma}_{\tau}}^{\mu}+ \sum_{j=1}^{q} \int_{\check{\Sigma}_{\tau} \cap \mathcal{A}_c} J_{\mu}^{N}\left[\partial_{r}^{j} \Psi_i\right] n_{\check{\Sigma}_{\tau}}^{\mu}.  \]
		This allows us to find a sequence $\lbrace \tau_n \rbrace_{n\in \mathbb{N}}$ with $ \tau_0\geq1 $ and $ 2\tau_n\leq \tau_{n+1}\leq 3\tau_n , \ \forall n\in\mathbb{N}, $ such that \begin{align}
			\sum_{m=0}^{^{\ell-1-(-1)^{i}}}	\int_{\check{\Sigma}_{\tau_n}\cap \mathcal{A}_c} J_{\mu}^{N}[\partial_r^{m}\p] \ n_{\check{\Sigma}_{\tau_n}}^{\mu} \leq  C \ \mathcal{E}^{^{\ell-(-1)^{i}}}_{_{energy}}[\p](0) \cdot \dfrac{1}{\tau_n}, \ \ \  \ \forall n\in\mathbb{N}.\label{Misc 10.1}
		\end{align}
		On the other hand, commuting the wave equation (\ref{waveeq}) with $ T^{s} $ for all $ s < \ell-(-1)^{i} $ and using the corresponding estimate of relation (\ref{t^-1 decay}) we obtain 
		\begin{align}
			\int_{\check{\Sigma}_{\tau}}J_{\mu}^{N}\left[T^{s}\Psi_i\right]n_{\check{\Sigma}_{\tau}}^{\mu} \leq C \ \mathcal{E}[T^{s}\Psi_i](0) \cdot \dfrac{1}{\tau}, \hspace{1cm} \forall \ s < \ell-(-1)^{i}. \label{Misc 10.1.1}
		\end{align}
		Then, for any $ \tau\geq 1 $ we can find $ n\in\mathbb{N} $ such that $ \tau_n\leq \tau\leq \tau_{n+1} $ and applying estimate (\ref{H^2- higher control}) in $ \mathcal{R}(\tau_n,\tau) $ for $k=(\ell-1)-(-1)^i  $ while using (\ref{Misc 10.1}, \ref{Misc 10.1.1})  yields \begin{align*}
			\int_{\check{\Sigma}_{\tau}\cap \mathcal{A}_c} J_{\mu}^{N}[\partial_r^{^{\ell-1-(-1)^i}}\p] \ n_{\check{\Sigma}_{\tau}}^{\mu} \leq C \ \mathcal{E}_{_{energy}}^{^{\ell-1-(-1)^i}}[\p](\tau_n) \leq & C \ \mathcal{E}_{i}^{^{\ell}}\cdot \dfrac{1}{\tau_n} \sim C\ \mathcal{E}_{i}^{^{\ell}}\dfrac{1}{\tau_{n+1}} 
			\leq  C \mathcal{E}_{i}^{^{\ell}} \cdot \dfrac{1}{\tau}
		\end{align*}
		where $ \mathcal{E}_{i}^{^{\ell}} $ involves norms of initial data of $ \Psi_i^{^{(\ell)}}  $ only, and C is a positive constant that depends on $ M, \bar{R}, \check{\Sigma}_{0} $ and $ k=\ell-1-(-1)^{i} $.
		\begin{flushleft}
			\hrulefill
		\end{flushleft}
		Now, fix $ k\leq (\ell-2)-(-1)^i $  and apply (\ref{H^2- higher control}) in $ \mathcal{R}(a_n,b_{n}) $, where $ a_n :=\tau_{n}+\frac{\tau_{n+1}-\tau_n}{3}, \ b_{n}:=\tau_{n+1}- \frac{\tau_{n+1}-\tau_n}{3} $, along with coarea formula to obtain \begin{align*}
			\int_{a_n}^{b_n}\left(	\int_{\check{\Sigma}_{\tilde{\tau}}\cap \mathcal{A}_c} J_{\mu}^{N}[\partial_r^{k}\p] \ n_{\check{\Sigma}_{\tilde{\tau}}}^{\mu} \right) d\tilde{\tau} \leq C  \mathcal{E}_{_{energy}}^{k+1}[\p](\tau_n)\leq C \ \mathcal{E}_{i,\ell}^{k+1} \dfrac{1}{\tau_n}  ,
		\end{align*}
		where the latter estimate comes from the analysis above.
		By fundamental theorem of calculus, there exists $ \xi_n \in (a_n,b_n) $ such that 
		\begin{align}
			(b_n-a_n)	\int_{\check{\Sigma}_{\xi_{n}}\cap \mathcal{A}_c} J_{\mu}^{N}[\partial_r^{k}\p] \ n_{\check{\Sigma}_{\xi_{n}}}^{\mu}  \leq C\  \mathcal{E}_{i,\ell}^{k+1} \dfrac{1}{\tau_n} \ \label{Misc 10.2}  
		\end{align}
		However, $ b_n-a_n= \frac{\tau_{n+1}-\tau_n}{3}, \ \forall n \in \mathbb{N}, $ and using the property of $ \br{\tau_n}_{n\in\mathbb{N}} $, i.e. $ \tau_n \sim \tau_{n+1} \sim (\tau_{n+1}-\tau_{n}) $,  relation 
		(\ref{Misc 10.2}) yields 
		\begin{align}
			\int_{\check{\Sigma}_{\xi_{n}}\cap \mathcal{A}_c} J_{\mu}^{N}[\partial_r^{k}\p] \ n_{\check{\Sigma}_{\xi_{n}}}^{\mu}  \leq 3C \ \mathcal{E}_{i,\ell}^{k+1} \dfrac{1}{\tau_n^2}\leq 27C \ \mathcal{E}_{i,\ell}^{k+1} \dfrac{1}{\tau_{n+1}^2}\leq C' \ \mathcal{E}_{i,\ell}^{k+1} \dfrac{1}{\xi_n ^{2}} \label{Misc 10.3}
		\end{align}
		for all $ n\in\mathbb{N}. $ It's a direct computation to check that the sequence $ \br{\xi_n}_{n\in\mathbb{N}}  $ satisfies $$ \frac{4}{3}\ \xi_n\leq \xi_{n+1}\leq 7\  \xi_{n}, \hspace{1cm} \forall\ n\in\mathbb{N} $$
		which also corresponds to an unbounded dyadic sequence with $ \xi_n \sim \xi_{n+1}\sim (\xi_{n+1}-\xi_{n}), $ uniformly $ \forall n \in \mathbb{N} $. Thus, for any $ \tau\geq 1 $ we can find $ n\in\mathbb{N} $ such that $ \xi_n\leq \tau \leq \xi_{n+1} $ and  applying estimate (\ref{H^2- higher control}) in $ \mathcal{R}(\xi_n,\tau) $ and using estimate (\ref{Misc 10.3}) with the results of proposition \ref{non-degen decay}, we conclude the proof of (\ref{HED <l-1}). \\
	\end{proof}
	
	\begin{theorem} \label{PWD-estimates}
		Let $ \bar{R}>M, \ \tau \geq 1 $ and fix $ \ell\geq i, \ i\in\br{1,2}, $ then for any $ k\leq \ell- (-1)^{i} $ there exists a positive constant $ C $ depending on $ M,\ k,\ \bar{R} $ and $ \check{\Sigma}_0 $ such that for all solution $ \p^{(\ell)} $ to (\ref{waveeq}), supported on the fixed frequency $ \ell, $ the following pointwise decay estimates hold \begin{align}
			\abs{\partial_r^k\p^{^{(\ell)}}}(\tau)\ \leq &\ C \ \mathcal{I}_{k,\ell}[\p]\ \dfrac{1}{\tau}, \hspace{1cm} \text{for}\ \ k \leq (\ell-2) -(-1)^{i} \\
			\abs{\partial_r^{^{\ell-1 -(-1)^i}}\p^{^{(\ell)}}}(\tau)\ \leq &\ C \ \mathcal{I}_{\ell-1,\ell}[\p]\ \dfrac{1}{\tau^{\frac{3}{4}}} , \\
			\abs{\partial_r^{^{\ell -(-1)^i}}\p^{^{(\ell)}}}(\tau)\ \leq &\ C \ \mathcal{I}_{\ell,\ell}[\p]\ \dfrac{1}{\tau^{\frac{1}{4}}},
		\end{align}
		in $ \br{M\leq r \leq \bar{R}}, $ where $ \mathcal{I}_{k,\ell}[\p] $ are norms of initial data of $ \p^{^{(\ell)}}. $
	\end{theorem}
	\begin{proof}
		Consider a  smooth-cutoff function  $ \chi_{_{\bar{R}+1}} : [M,	\bar{R}+1]\to[0,1] $ of r alone, such that $ \chi_{_{\bar{R}+1}}(r)=1 $ for $ M\leq r\leq \bar{R} +\frac{1}{4} $ and $ \chi_{_{\bar{R}+1}}(r)=0 $ for $ r\geq \bar{R}+\frac{1}{2}. $ Let $ r \in [M,\bar{R}] $, then we write \begin{align*}
			(\partial_r^k\p)^2(r,\omega) =& - \int_{r}^{\bar{R}+1}\partial_{\rho}\left(\chi_{_{\bar{R}+1}}\cdot \partial_r^k\p\right)^2 d\rho\   \\ \ =& \  -\int_{r}^{\bar{R}+1}\partial_{\rho}(\chi_{_{\bar{R}+1}}^2) (\partial_r^k\p)^2 d\rho - \int_{r}^{\bar{R}+1}2\chi_{_{\bar{R}+1}}\left(\partial_r^k\p\right)\partial_{\rho}\left(\partial_r^k\p\right)  d\rho \\
			\leq  &\  C\int_{r}^{\bar{R}+1}\abs{\partial_{\rho}(\chi_{_{\bar{R}+1}}^2)} (\partial_r^k\p)^2 \rho^2 d\rho \\ & + \ C \left (\int_{r}^{\bar{R}+\frac{1}{2}}\left(\partial_r^k\p\right)^2\rho^2 d\rho\right )^{\frac{1}{2}}\left (\int_{r}^{\bar{R}+\frac{1}{2}}\left [\partial_{\rho}\left(\partial_r^k\p\right)\right ]^2 \rho^2 d\rho \right)^{\frac{1}{2}}.
		\end{align*}
		Now, integrating on the sphere $ \mathbb{S}^2 $ and since $ \chi_{_{\bar{R}+1}} $ is smooth, we obtain \begin{align*}
			\int_{\mathbb{S}^2}(\partial_r^k\p)^2(r,\omega) d\omega\  \leq &\ C \int_{\check{\Sigma}_{\tau}\cap\br{M\leq r\leq \bar{R}+1}}J_{\mu}^N[\partial_r^{k-1}\p]n_{\check{\Sigma}_{\tau}}^{\mu} + \\ + & \ C 
			\left(\int_{\check{\Sigma}_{\tau}\cap\br{M\leq r\leq \bar{R}+\frac{1}{2}}}J_{\mu}^N[\partial_r^{k-1}\p]n_{\check{\Sigma}_{\tau}}^{\mu}\right)^{\frac{1}{2}} \cdot \left(\int_{\check{\Sigma}_{\tau}\cap\br{M\leq r\leq \bar{R}+\frac{1}{2}}}J_{\mu}^N[\partial_r^{k}\p]n_{\check{\Sigma}_{\tau}}^{\mu}\right)^{\frac{1}{2}}
		\end{align*}
		with $ C $ depending on $ M, \check{\Sigma}_0. $ Thus, using the results of Proposition \ref{HED Pro} we have that for $ k\leq (\ell-2)-(-1)^i $  all  integrals of the right-hand side decay at the rate of $ \frac{1}{\tau^2} $. For $ k=\ell-1 -(-1)^i $  the first integral of the right-hand side decays like $ \frac{1}{\tau^2} $, whereas the last integral decays like $ \frac{1}{\tau}. $ Finally, for $ k=\ell-(-1)^i $ the first integral decays like $ \frac{1}{\tau} $ and from Theorem \ref{high-order horizon estimates}, the last integral is bounded.
		
		Thus, repeating the same estimates after we commute with the angular momentum operators $ \Omega $ and using the Sobolev inequality on the sphere $ \mathbb{S}^2 $ we prove the decay rates of the assumption. \\
	\end{proof}

	\paragraph{Non-Decay and Blow-up Estimates.}
	
	\begin{proposition} \label{non-decay}
		For all solutions  $ \p^{^{(\ell)}} $ to (\ref{waveeq}), supported on the fixed frequency $ \ell\geq i, \ i \in \br{1,2},$ we have the following non decay estimates \begin{align}
			\partial_r^{^{\ell+2}}\Psi_1^{^{(\ell)}} (\tau,\vartheta,\varphi) &\xrightarrow{\tau \to \infty} H_{\ell}[\Psi_1](\vartheta,\varphi), \\
			\partial_r^{\hspace{0.05cm}\ell}\Psi_2^{^{(\ell)}}(\tau,\vartheta,\varphi) &\xrightarrow{\tau \to \infty} H_{\ell}[\Psi_2](\vartheta,\varphi) ,		
		\end{align} 
		asymptotically along the event horizon $ \h$, where $ H_{\ell}[\p] $ are the conserved quantities of theorem \ref{conservation horizon} . In addition, for generic initial data of $ \Psi_{i}^{^{(\ell)}} $, $ H_{\ell}[\p](\vartheta,\varphi) $ is almost everywhere non-zero on $ \check{\Sigma}_0\cap \h. $
	\end{proposition}
	\begin{proof}
		Using theorem \ref{conservation horizon}, we know the expressions  below are conserved along the null generators of $ \h, $
		\begin{align*}
			H_{\ell}[\Psi_1](\vartheta,\varphi)&\ = \ \partial_r^{\ell+2}\Psi_1^{^{(\ell)}} (\tau,\vartheta,\varphi)  + \sum_{j=0}^{\ell+1} c_1^j \cdot \partial_r^j\Psi_1^{^{(\ell)}}  (\tau,\vartheta,\varphi)  \\
			H_{\ell}[\Psi_2] (\vartheta,\varphi) &\ =\  \partial_r^{\ell}\Psi_2^{^{(\ell)}}  (\tau,\vartheta,\varphi) + \sum_{j=0}^{\ell-1} c_2^j \cdot \partial_r^j\Psi_2^{^{(\ell)}} (\tau,\vartheta,\varphi)  . 
		\end{align*} 
		Thus, using the pointwise decay estimates of Theorem \ref{PWD-estimates} along the event horizon $ \h $, we deduce that all terms of the sum above converge to zero asymptotically as $ \tau \to \infty, $ and so the limits of the assumption follow.
		
		In addition, as in \cite{aretakis2011stability}, we may introduce a new hypersurface $ \tilde{\Sigma}_0 $ formed by ingoing null hypersurfaces in an $ \phi^{T}- $invariant neighborhood of the horizon $ \h $, and prescribed initial data for $ \p $ on $ \tilde{\Sigma}_0 $ can be compared to initial data prescribed on $ \check{\Sigma}_0 $ using energy boundedness results obtained already. In view of $ H_{\ell}[\Psi_i] $ involving only transversal derivatives to $ \h $,  we see that it is completely determined on $ \check{\Sigma}_{0}\cap H $ by the initial data prescribed on $ \tilde{\Sigma}_0 $. Hence, generic initial data for $ \p $ on $ \tilde{\Sigma}_0 $ produces generic eigenfunctions $ H_{\ell}[\p] $ of order $ \ell $ of $ \slas{\Delta} $ on $ \check{\Sigma}_0 $, and thus almost everywhere non-zero.\\
	\end{proof}

	\begin{theorem} \label{Scalar blow up} 
		Let $ k,\ell \in \mathbb{N} $, then for all solutions $ \p^{^{(\ell)}} $ to (\ref{waveeq}), supported on the fixed frequency $ \ell\geq i $, $ i\in\br{1,2} $, we have \begin{align}
			\partial_{r}^{k}\left(\partial_{r}^{^{\ell+1-(-1)^{i}}}\p^{^{(\ell)}}\right)(\tau, \vartheta,\varphi) = (-1) ^{k} a_{i,k}^{^{(\ell)}} H_{\ell}[\p](\vartheta,\varphi) \cdot \tau^{k} + \mathcal{O}(\tau^{k-\frac{1}{4}}), 
		\end{align}
		as $ \tau \rightarrow \infty $ along the event horizon $ \h $, where \begin{align*}
			a_{i,k}^{^{(\ell)}}:= \dfrac{1}{(2M^2)^{k}}\dfrac{\left(k+1+2(\ell-(-1)^{i})\right)!}{\left(1+2(\ell-(-1)^{i})\right)!}.
		\end{align*}
	\end{theorem}
	\begin{proof}
		We use induction on $ k\in \mathbb{N}.$ To show the $ k=1  $ case, we consider (\ref{k-wave}) for $ s=\ell+1 - (-1) ^{i}$ and we evaluate it on the horizon $ \h $ to obtain
		\begin{align}
			\begin{aligned}
				&T\left( 2\partial_r^{^{\ell+2-(-1)^i}}\p+\frac{2}{M}\partial_r^{^{\ell+1-(-1)^i}}\p+L_{_{\leq\ell-(-1)^i}}[\p]\right) \\ &\quad + \dfrac{2}{M^2}(\ell+1-(-1)^i) \partial_r^{^{\ell+1-(-1)^i}}\p + \tilde{L}_{_{\leq\ell-(-1)^i}}[\p] \ = \ 0, \label{Misc 10.4}
			\end{aligned}
		\end{align}
		where $ L_{\leq q}[\p],\ \tilde{L}_{\leq q}[\p]$ is a linear combination of $ \partial_r-$derivatives of $ \p $ up to order $ q, $ with coefficients depending on $ \ell,q,M. $ However, in view of the pointwise estimates of Theorem \ref{PWD-estimates} we have \begin{align*}
			L_{_{\leq\ell-(-1)^i}}[\p](\tau) &= \mathcal{O}(\tau^{-\frac{1}{4}}), \\
			\int_{\tau_0}^{\tau} \tilde{L}_{_{\leq\ell-(-1)^i}}[\p]( \tilde{\tau})  d\tilde{\tau} &= \mathcal{O}(\tau^{\frac{3}{4}}), \hspace{1cm} \text{as} \quad \tau \xrightarrow{\h} \infty. 
		\end{align*}
		Thus, using the conserved expressions of theorem \ref{conservation horizon} and integrating (\ref{Misc 10.4}) along the null geodesics of $ \h $ we obtain \begin{align}
			\partial_r^{^{\ell+2-(-1)^i}}\p = - \textstyle \frac{\left (\ell+1-(-1)^i\right )}{M^2} H_{\ell}[\p]\cdot \tau + \mathcal{O}\left (\tau^{\frac{3}{4}}\right ), \hspace{1cm} \text{as}\quad \tau \xrightarrow{\h}\infty,
		\end{align}
		which proves the $ k=1 $  case. Now, assume the relation of the assumption holds for all $q \leq k-1$, then evaluating (\ref{k-wave}) for $ s=\ell+k-(-1) ^{i}$ on the horizon $ \h $ yields accordingly 
		\begin{align*}
			\begin{aligned}
				&T\left( 2\partial_r^{^{\ell+k+1-(-1)^i}}\p+\frac{2}{M}\partial_r^{^{\ell+k-(-1)^i}}\p+L_{_{\leq\ell+k-1-(-1)^i}}[\p]\right) \\ &\quad + \dfrac{k}{M^2}\Big(2(\ell-(-1)^i)+k+1\Big) \partial_r^{^{\ell+k-(-1)^i}}\p + \tilde{L}_{_{\leq\ell+k-1-(-1)^i}}[\p] \ = \ 0. \label{Misc 10.4'}
			\end{aligned}
		\end{align*}
		However, using the inductive hypothesis we  obtain \begin{align*}
			\begin{aligned}
				&T\left( 2\partial_r^{^{\ell+k+1-(-1)^i}}\p+\frac{2}{M}\partial_r^{^{\ell+k-(-1)^i}}\p+L_{_{\leq\ell+k-1-(-1)^i}}[\p]\right) \\&  =  - \dfrac{k}{M^2}\Big(2(\ell-(-1)^i)+k+1\Big) (-1) ^{k-1} a_{i,k-1} H_{\ell}[\p](\vartheta,\varphi) \cdot \tau^{k-1} + \mathcal{O}(\tau^{(k-1)-\frac{1}{4}})
			\end{aligned} 
		\end{align*}
		Thus, integrating the above relation along the null generators of $ \h $ and using the inductive hypothesis we conclude the proof. \\
	\end{proof}
	
	\begin{remark}
		The proposition above holds also for $ k=0 $ which corresponds to the non-decaying  result of proposition \ref{non-decay}.
	\end{remark}

	\begin{corollary} \label{psi-phi l study }
		Let $ (\phi,\psi) $ be a solution to the coupled system (\ref{ceq}), then for generic initial data, the following decay estimates hold on the exterior \footnote{If $ f(x)\lesssim_{p} g(x) $ then there exists a constant $ C>0 $ depending on $ p $, such that $ f(x) \leq C \cdot g(x)$.} \begin{align}
			\abs{\phi}(r,\tau) \lesssim_{M} \frac{1}{\sqrt{r}\cdot \tau^{\frac{3}{4}}}, \hspace{1cm} \abs{\psi}(r,\tau)  \lesssim_{M} \frac{1}{\sqrt{r}\cdot \tau^{\frac{3}{4}}},
		\end{align}
		for all $ r\geq M $ and $ \tau\geq 1. $ Moreover, 
		the following decay, non-decay and blow up estimates hold asymptotically on $ \h $, as $ \tau \to \infty $.
		For any $ \ell \geq 2, $ we have
		\begin{itemize}
			\item Decay: \begin{align}
				\abs{\partial_r^{\ell-1}\phi_{\ell}} &\lesssim_{_{\ell,M}} \dfrac{1}{\tau^{\frac{1}{4}}}, \hspace{2cm} \abs{\partial_r^{\ell-1}\psi_{\ell}} \lesssim_{_{\ell,M}} \dfrac{1}{\tau^{\frac{1}{4}}} , \\
				\abs{\partial_r^{\ell-2}\phi_{\ell}} &\lesssim_{_{\ell,M}} \dfrac{1}{\tau^{\frac{3}{4}}}, \hspace{2cm} \abs{\partial_r^{\ell-2}\psi_{\ell}} \lesssim_{_{\ell,M}} \dfrac{1}{\tau^{\frac{3}{4}}},\\ 
				\abs{\partial_r^k \phi_{\ell}} &\lesssim_{_{k,M}} \dfrac{1}{\tau}, \hspace{2cm}  \abs{\partial_r^k \psi_{\ell}} \lesssim_{_{k,M}} \dfrac{1}{\tau} \hspace{1cm}  \forall \ k\leq \ell-3
			\end{align}
			\item Non-decay 
			\begin{align}
				\partial_r^{\ell}\phi_{_{\ell}}(\tau,\omega) &\xrightarrow{\tau \to \infty} \dfrac{1}{2M^2(2\ell+1)}H_{\ell}[\Psi_2](\omega), \\	\partial_r^{\ell}\psi_{_{\ell}}(\tau,\omega) &\xrightarrow{\tau \to \infty} -\dfrac{2}{M(2\ell+1)(\ell+2)}H_{\ell}[\Psi_2](\omega), 
			\end{align}

			\item Blow up: 
			\begin{align}
				&\partial_r^{k+\ell}\phi_{\ell}(\tau,\omega) = (-1)^{k} \dfrac{a_{2,k}^{^{(\ell)}}}{2M^{2}(2\ell+1)}H_{\ell}[\Psi_2](\omega) \cdot\tau ^{k}+ \mathcal{O}(\tau^{k-\frac{1}{4}}), \hspace{1cm} k\geq 0, \\ 	&\partial_r^{k+\ell}\psi_{\ell}(\tau,\omega) = -(-1)^{k}\frac{2a_{2,k}^{^{(\ell)}}}{M(\ell+2)(2\ell+1)} H_{\ell}[\Psi_2](\omega) \cdot\tau ^{k}+ \mathcal{O}(\tau^{k-\frac{1}{4}}), \hspace{1cm} k\geq 0, 
			\end{align}
			where the constants $ a_{i,k}^{^{(\ell)}} $ are given in Theorem \ref{Scalar blow up}.
		\end{itemize}
	\end{corollary}

	\begin{proof}
		We recall from Corollary \ref{inverse relation} the following expressions for any $ \ell \geq 2 $ \begin{align*}
			\phi_{\ell} & = \dfrac{\mu}{2M^2(\ell+2 )(2\ell+1)}\cdot \Psi_1^{^{(\ell)}}+ \dfrac{1}{2M^2(2\ell+1)}\Psi_2^{^{(\ell)}} \\
			\psi_{\ell} &= \dfrac{2}{M\mu\cdot (2\ell+1 )}\cdot \Psi_1^{^{(\ell)}}- \dfrac{2}{M(\ell+2)(2\ell+1)}\Psi_2^{^{(\ell)}}
		\end{align*}
		Thus, the proof follows immediately from Theorem \ref{PWD-estimates}, \ref{non-decay}, and \ref{Scalar blow up}. Note, $ \psi_{\ell=1} = \Psi_1^{^{(\ell=1)}}$, thus we already have estimates for $ \psi_{\ell=1} $.
	\end{proof}
	
	\subsection{Estimates for \texorpdfstring{$ \mathbf{q^F, p} $}{PDFstring}.}
	Using elliptic identities of Section \ref{Geometry}  we can show that a similar hierarchy of estimates also holds for the gauge--invariant quantities $ \bm{q^F}$ and $\bm{p} $ in $ L^2(S^2_{\tau,M}). $  First, let us prove the following lemma that allows us to do so.
	
	\begin{lemma} 
		Let $ \xi_{A}, \theta_{AB} $ be a one tensor and a symmetric traceless 2-tensor on $ S^2_{v,r} $, respectively, and consider the scalars $ f_{\xi} : = r\dd \xi,\  f_{\theta }:= r^2\dd\DD \theta $. Then, \begin{itemize}
			\item For any $ k\in\mathbb{N} $, we have \begin{align}
				\partial_r^k f_{\xi} = r\dd \left(\slas{\nabla}_{\partial_r}^k\xi\right), \hspace{1cm}  \partial_r^k f_{\theta} = r^2\dd\DD \left(\slas{\nabla}_{\partial_r}^k\theta\right) \label{misc12-a}
			\end{align}
			\item The following elliptic estimates hold for any $ k\in\mathbb{N} $\begin{align}
				\int_{S^2_{v,r}} \abs{\slas{\nabla}_{\partial_r}^k\xi}^2 \leq \int_{S^2_{v,r}}\abs{\partial_r^kf_{\xi}}^2, \hspace{1cm} \int_{S^2_{v,r}} \abs{\slas{\nabla}_{\partial_r}^k\theta}^2 \leq \int_{S^2_{v,r}}\abs{\partial_r^kf_{\theta}}^2\label{misc12-b}
			\end{align}
		\end{itemize}  
		\begin{proof}
			From the definition of $ \dd,\DD $ and the commutations (\ref{cummutation formula}) we have $ [\slas{\nabla}_{\partial_r},r\dd] =  [\slas{\nabla}_{\partial_r},r\DD] = 0 $ which shows (\ref{misc12-a}). For example,   \[ \partial_r^kf_{\theta} = \partial_r^k\left (r\dd(r\DD\theta)\right ) = r\dd \left( \slas{\nabla}_{\partial_r}^k(r\DD\theta)\right) = r^2\dd\DD (\slas{\nabla}_{\partial_r}^k\theta).  \]
			To show (\ref{misc12-b}), we use identities (\ref{elliptic 2.5}),(\ref{elliptic 3}) and for motivation, we only look at the second estimate, i.e. \begin{align*}
				\int_{S^2_{v,r}} \abs{\slas{\nabla}_{\partial_r}^k\theta}^2 =&		\int_{S^2_{v,r}}r^2 \cdot \dfrac{1}{r^2} \abs{\slas{\nabla}_{\partial_r}^k\theta}^2 \leq 	\int_{S^2_{v,r}} r^2\abs{\DD\left (\slas{\nabla}_{\partial_r}^k\theta\right )}^2 \\ =& \int_{S^2_{v,r}} r^4 \dfrac{1}{r^2}\abs{\DD\left (\slas{\nabla}_{\partial_r}^k\theta\right )}^2 
				\leq \int_{S^2_{v,r}} \abs{r^2\dd\DD\left (\slas{\nabla}_{\partial_r}^k\theta\right )}^2  =  \int_{S^2_{v,r}}\abs{\partial_r^kf_{\theta}}^2,
			\end{align*}
			where in the last equation we used relation (\ref{misc12-a}). \\
		\end{proof}
	\end{lemma}
	
	In the propositions that follow, we define the $ L^2(S^{2}_{v,r}) $-norm  of any $ S^2_{v,r}- $tensor $ \xi $ as \begin{align*}
		\norm{\xi}_{S^2_{v,r}}^{2} :=  \int_{S^2_{v,r}}r^2 \sin\theta d\theta d\phi \abs{\xi}^2.	
	\end{align*}
	
	\begin{proposition} \label{finished estimates of Regge}
		Let $ (\bm{q^{F}}, \bm{p}) $ and  $ (\underline{\bm{q}}^{F}, \underline{\bm{p}}) $ be a solutions to the coupled  system (\ref{csystem}), then for generic initial data, we have 
		\begin{align}
			\norm{\bm{p}}_{S^2_{\tau,r}} \lesssim_{_{M}} \dfrac{1}{\sqrt{r}\cdot \tau^{\frac{3}{4}}}, \hspace{2cm}           \norm{\bm{q^F}}_{S^2_{\tau,r}} \lesssim_{_{M}} \dfrac{1}{\sqrt{r}\cdot \tau^{\frac{3}{4}}}
		\end{align}	
		Also, the following decay, non-decay and blow-up results hold  asymptotically along the event horizon $ \h $
		\begin{itemize}
			\item Decay: \begin{align}
				&\norm{\bm{p}}_{S^2_{\tau,M}} \lesssim_{_{M}} \tau ^{-\frac{3}{4}}, \hspace{2cm}           \norm{\bm{q^F}}_{S^2_{\tau,M}} \lesssim_{_{M}} \tau ^{-\frac{3}{4}} \\
				&\norm{\slas{\nabla}_{\partial_r}\bm{p}}_{S^2_{\tau,M}} \lesssim_{_{M}} \tau ^{-\frac{1}{4}}, \hspace{1.3cm}           \norm{\slas{\nabla}_{\partial_r}\bm{q^F}}_{S^2_{\tau,M}} \lesssim_{_{M}} \tau ^{-\frac{1}{4}} 
			\end{align}
			\item Non-Decay: \begin{align}
				&\norm{\slas{\nabla}_{\partial_r}^2\bm{p}}_{S^2_{\tau,M}}\quad \xrightarrow{\tau \to \infty} \dfrac{1}{60M}\mathcal{H}_{\ell=2}[\Psi_2] , \\           &\norm{\slas{\nabla}_{\partial_r}^2\bm{q^F}}_{S^2_{\tau,M}}\ \xrightarrow{\tau \to \infty} \dfrac{1}{120 M^2}\mathcal{H}_{\ell=2}[\Psi_2], 
			\end{align}
			\item Blow-up: \begin{align}
				\norm{\slashed{\nabla}_{\partial_{r }}^k \bm{ p}}_{S^2_{\tau,M}}\ =\ & \frac{a_{2,k}^{^{(2)}}}{60 M} \mathcal{H}_{\ell=2}[\Psi_2] \cdot \tau ^{k-2} + \mathcal{O}\left(\tau ^{k-2-\frac{1}{4}}\right), \quad
				\quad k\geq 2. \label{Misc10.6}
				\\ 	\norm{\slashed{\nabla}_{\partial_r}^k \bm{q^F}}_{S^2_{\tau,M}}
				\	=\ & \frac{a_{2,k}^{^{(2)}}}{120 M^2} \mathcal{H}_{\ell=2}[\Psi_2] \cdot \tau ^{k-2} + \mathcal{O}\left(\tau ^{k-2-\frac{1}{4}}\right), \quad
				\quad k\geq 2. \label{Misc10.7} 
			\end{align}
		\end{itemize}
		where $ \mathcal{H}_{\ell=2}[\Psi_2] :=\norm{H_{2}[\Psi_2]}_{S^{2}_{\tau,M}}, \ \forall \ \tau \geq 1, $ and $ a_{2,k}^{^{(2)}} $ is given in Theorem \ref{Scalar blow up}. The same estimates hold for $ (\underline{\bm{q}}^{F}, \underline{\bm{p}}).   $
	\end{proposition}
	\begin{proof} We recall the definitions $ \psi:= r\dd\bm{p} $ and $ \phi := r^2 \dd\DD \bm{q^{F}} $. Let us motivate the idea of showing decay by working with $ \bm{q^{F}}$ only, and let $ k\leq 1 $, then using the elliptic identity (\ref{elliptic 6}) and relation 
		(\ref{misc12-a}) we write \begin{align} \label{L-2 formula scalar to tensor}
			\int_{S^{2}_{\tau,M}} \abs{\slas{\nabla}_{\partial_r}^{k}\bm{q^{F}}}^{2} = \sum_{\ell\geq 2}\left ( \int_{S^{2}_{\tau,M}}\abs{\slas{\nabla}_{\partial_r}^{k}\bm{q^{F}}_{\ell}}^{2}\right ) = \sum_{\ell\geq 2} \left[\dfrac{2}{\ell(\ell+1)}\cdot \dfrac{1}{\ell(\ell+1)-2}\left ( \int_{S^{2}_{\tau,M}}\abs{\partial_r^{k}\phi_{\ell}}^{2}\right )\right]
		\end{align}
		However, in view of the decay estimates of Corollary \ref{psi-phi l study } we obtain \begin{align*}
			\int_{S^{2}_{\tau,M}} \abs{\slas{\nabla}_{\partial_r}^{k}\bm{q^{F}}}^{2} \leq \sum_{\ell\geq 2} \left (\dfrac{2}{\ell(\ell+1)}\cdot \dfrac{1}{\ell(\ell+1)-2}\right ) (4\pi M^{2})  \norm{\partial_r^{k}\phi_{\ell}}^{2}_{L^{\infty}(S^{2}_{\tau,M})} \leq \begin{cases}
				C \tau ^{-\frac{3}{2}},  \hspace{0.5cm} k=0
				\\
				C \tau^{-\frac{1}{2}}, \hspace{0.5cm} k=1
			\end{cases} 
		\end{align*}
		for a constant $ C $ that depends only on $ M, \check{\Sigma}_{0} $.
		
		Now, fix $ k\geq 2 $ and working with equation (\ref{L-2 formula scalar to tensor}) we have that for any $ \ell > k $ the corresponding term decays in view of Corollary \ref{psi-phi l study }, i.e. \begin{align*}
			\sum_{\ell>k} \left[\dfrac{2}{\ell(\ell+1)}\cdot \dfrac{1}{\ell(\ell+1)-2}\left ( \int_{S^{2}_{\tau,M}}\abs{\partial_r^{k}\phi_{\ell}}^{2}\right )\right] \leq C \tau ^{-\frac{1}{2}},
		\end{align*}
		for a positive constant depending on $ k, M, \check{\Sigma}_0. $
		On the other hand, for the remaining $ k-1 $ first terms of the sum we see from Corollary \ref{psi-phi l study } that the $  \ell=2 $ term is the dominant one and we write \begin{align*}
			\sum_{2\leq \ell \leq k} \left[\dfrac{2}{\ell(\ell+1)}\cdot \dfrac{1}{\ell(\ell+1)-2}\left ( \int_{S^{2}_{\tau,M}}\abs{\partial_r^{k}\phi_{\ell}}^{2}\right )\right] = \dfrac{b_k}{12}\norm{H_{2}[\Psi_2]}_{S^{2}_{\tau,M}}^{2} \left (\tau^{k-2}\right )^{2} + \mathcal{O}\left (\tau^{2k-4-\frac{1}{2}}\right )
		\end{align*}
		for \[ b_k := \left (\dfrac{a_{2,k}^{^{(2)}}}{10M^2}\right )^{2}\]
		Note for $ k=2 $, we simply obtain a non-decaying term while the others decay in the expression above. We work similarly to obtain the decay, non-decay, and blow-up estimates of $ \bm{p}. $ 
		\\
	\end{proof}

	\section{The positive spin Teukolsky equations instability.}\label{Teukolsky positive section} With the results of the previous section in hand, we can now show estimates for the Teukolsky solutions along the event horizon $ \h. $ In particular, in this section we are interested in the positive spin gauge invariant quantities $ \alpha, \mathfrak{f}, \tilde{\beta}  $ which satisfy the relating equation \begin{align}
		^{(F)}\rho \dfrac{1}{\underline{\kappa}_{\star}} \slas{\nabla}_{3^{\star}}\left(r^3 \underline{\kappa}_{\star}^2 {\alpha}_{\star}\right) = - \left(^{(F)}\rho^{2} + 3 \rho\right) r^3 \underline{\kappa}_{\star}{\mathfrak{f}}_{\star} - r^3 \underline{\kappa}_{\star}  \DD^{\star}\left(\tilde{{\beta}}_{\star}\right). \label{relating a,f,b}
	\end{align}
	Note that  $ \rho$ and ${}^{(F)} \rho$ are regular quantities on the horizon $ \h $, while $\alpha, \ \tilde{\beta}, \  \mathfrak{f}, \ \underline{\kappa} $ are not. However, the rescaled quantities $\alpha_{\star}= D^2 \cdot \alpha,\  \mathfrak{f}_{\star} = D \cdot  \mathfrak{f}, \  \tilde{\beta}_{\star} = D \cdot \tilde{\beta}, \  \underline{\kappa}_{\star} = D^{-1}\underline{\kappa} $ extend regularly on $ \h $. With respect to the Ingoing Eddington--Finkelstein coordinates $\br{ v,r }$, the quantities involved in the relating equation (\ref{relating a,f,b}) are given in the ERN spacetime by \begin{align*}
		^{(F)}\rho = \dfrac{M}{r^2}, \hspace{1cm} \rho = -\dfrac{2M}{r^3}\sqrt{D}, \hspace{1cm} \underline{\kappa}_{\star}= -\dfrac{2}{r}, \hspace{1cm} e_{3}^{\star} = - \dfrac{\partial}{\partial r }.
	\end{align*}
	In addition, we recall the transformation identities relating $ \mathfrak{f},\tilde{\beta} $ to the Regge--Wheeler solutions $ \bm{q}^{F}, \bm{p}  $ respectively \begin{align}
		\begin{aligned} \label{Trans + spin}
			\bm{q}^{F} &= \dfrac{1}{\underline{k}_{\star}} \slas{\nabla}_{3^{\star}}\left(r^3 \underline{k}_{\star}\cdot \mathfrak{f}_{\star} \right) = - r \slas{\nabla}_{\partial_r}(r^2\mathfrak{f}_{\star}) \\
			\bm{p} &= \dfrac{1}{\underline{k}_{\star}} \slas{\nabla}_{3^{\star}}\left(r^3 \underline{k}_{\star}\cdot \tilde{\beta}_{\star} \right) = - r \slas{\nabla}_{\partial_r}(r^4\tilde{\beta}_{\star}).	\end{aligned}
	\end{align}
	We first write the induced scalars we will be working with and after we obtain estimates for those, we pass them to the respected tensors using standard elliptic identities. Let \begin{align}
		h_{\alpha}:= r^2 \dd\DD \alpha_{\star}, \hspace{1cm} h_{f}:= r^2\dd\DD \mathfrak{f}_{\star},  \hspace{1cm} h_{\tilde{\beta}}:= r^2\dd\DD \tilde{\beta}_{\star},
	\end{align}
	and after we project to the $ E_{\ell} $ eigenspace and using the commutation formulae (\ref{cummutation formula}), i.e. $ [\slas{\nabla}_{\partial_r},r\dd] =  [\slas{\nabla}_{\partial_r},r\DD] = 0$, the relating equation (\ref{relating a,f,b}) now reads
	\begin{align}
		- \partial_r\left(r \cdot h_{\alpha_{\ell}}\right) = r^2 \dfrac{\mu^2}{2M}\cdot h_{\tilde{\beta}_{\ell}} + 2(1-4\sqrt{D}) \cdot h_{\mathfrak{f}_{\ell}}, \hspace{1.5cm} \mu^2 = (\ell-1)(\ell+2) \label{scalar relating a,f,b}
	\end{align}
	Moreover, the transport equations (\ref{Trans + spin}) yield \begin{align}
		\begin{aligned}
			- \dfrac{\phi}{r} &= \partial_r\left (r^2\cdot h_{\mathfrak{f}}\right ) \\
			- \dfrac{\psi}{r} & = \partial_r\left (r^4\cdot h_{\tilde{\beta}}\right ) 
		\end{aligned} \label{transport + spin scalar}
	\end{align}

	\subsection{Energy and decay estimates for \texorpdfstring{$ \mathfrak{f}_{\star}, \ \tilde{\beta}_{\star}, \ \alpha_{\star} $}{PDFstring} on the exterior. } \label{f,beta estimates}
	In this section, we use the transport equations to produce energy decay estimates that ultimately lead to pointwise control on the exterior up to and including the event horizon $ \h. $ First, we present the following lemma. 
	
	\begin{lemma} \label{transport-3}
		Let $ f,F $ be two scalar functions satisfying \begin{align}
			\partial_{u}\left(r^m f\right) = - \dfrac{D(r)}{2r} F, \hspace{1cm} m\in\mathbb{N}, \label{f-F misc}
		\end{align} 
		with respect to the double null coordinate system $ (u,v,\vartheta,\varphi) $ introduced in Section \ref{Geometry}, then the following estimates hold near future null infinity \begin{align}
			\int_{\check{O}_{\tau_2}} r^{2m+k-2}f^2 + \int_{\I}r^{2m+k-3} f^2 \leq C \int_{\check{O}_{\tau_1}} r^{2m+k-2}f^2 + C \int_{\I} r^{k-3}F^2.
		\end{align}
		for any $ \tau_1<\tau_2 $ and $ k>0 $, where $ C>0 $  depends on $ M,k $ and $ R $ of Section \ref{r-p}.
	\end{lemma}
	\begin{proof}
		Multiplying (\ref{f-F misc}) with $ 2r^m\cdot f $ we obtain 
		\begin{align*}
			2 r^m f \partial_u(r^m f) &= -\dfrac{D}{r} r^m f \cdot F\\
			\Rightarrow \quad \partial_u \left(r^{2m} f^2\right) &=  -\dfrac{D}{r} r^m f \cdot F.
		\end{align*}
		To obtain the best possible decay in r we multiply the above relation with a factor of $ r^{q} $  which yields \begin{align*}
			\partial_u \left(r^{2m+q}f^2 \right) + q\dfrac{D}{2}r^{2m+q-1} f^2&= -\dfrac{D}{r} r^{m+q}f\cdot
			F \\
			&= - D r^{q-1} \cdot (r^{m}f) \cdot F 
			\\ &\leq  \varepsilon \dfrac{D}{2} r^{q-1} r^{2m}f^2 + \dfrac{D}{2\varepsilon} r^{q-1} F^2.
		\end{align*}
		Then, for $ \varepsilon = \frac{q+2}{2} >0 $ we obtain \begin{align*}
			\partial_u \left(r^{2m+q}f^2 \right) + (q-2)\dfrac{D}{4}r^{2m+q-1} f^2 \leq \dfrac{r^{q-1}}{q+2}D F^2.
		\end{align*}
		Now, in view of $ div\left (\frac{\partial}{\partial_u}\right ) = - \frac{\sqrt{D}}{r} $ we may rewrite the above estimate as \begin{align*}
			div \left( r^{2m+q} f^2  \frac{\partial}{\partial_u} \right) + r^{2m+q}\dfrac{\sqrt{D}}{r}f^2 +  (q-2)\dfrac{D}{4}r^{2m+q-1} f^2 \leq   \dfrac{r^{q-1}}{q+2}D F^2
		\end{align*}
		\begin{align}
			\Rightarrow \quad 	div \left( r^{2m+q} f^2 \frac{\partial}{\partial_u} \right) + r^{2m+q-1}\sqrt{D}\left(1+\frac{(q-2)\sqrt{D}}{4}\right)f^2 \leq   \dfrac{r^{q-1}}{q+2}D F^2. \label{div-misc}
		\end{align}
		Thus, integrating in $ \I $ and applying the divergence theorem readily yields \begin{align*}
			\int_{\check{O}_{\tau_2}} r^{2m+q}f^2 + \int_{\I}r^{2m+q-1}(4+(q-2)\sqrt{D}) f^2 \leq C \int_{\check{O}_{\tau_1}} r^{2m+q}f^2 + C \int_{\I} \dfrac{r^{q-1}}{q+2}F^2
		\end{align*} 
		for a positive constant $ C  $ depending on $ R $. In particular, in view of $ q+2>0 $, we can consider $ q=k-2 $ for $ k > 0 $ and conclude the estimate.

	\end{proof}
	\paragraph{Control for the remaining first-order derivatives.} It is clear that using the transport equation (\ref{f-F misc}), we control    $ \partial_u f $ in terms of $ F $ and $ f $. In addition, 
	using the above Lemma we may also produce energy estimates for $ \partial_{v}f $ as well. Indeed, by taking a $ \partial_v $-derivative of (\ref{f-F misc}) we obtain the relating equation \begin{align*}
		\partial_{u}\left( r^{m}\partial_vf\right)= -m\cdot\partial_u \left(\dfrac{D}{2r}\right)r^{m}f+\left(m\cdot \left(\dfrac{D}{2r}\right)^2-\partial_v\left(\dfrac{D}{2r}\right)  \right)F   - \dfrac{D}{2r}\partial_v(F)   
	\end{align*}
	Then, using Lemma \ref{transport-3} for $ \partial_vf $ with $ k+2 $ instead of $ k $, and using Cauchy Schwartz yields \begin{align*}
		\begin{aligned}
			\int_{\check{O}_{\tau_2}} r^{2m+k}(\partial_vf)^2 + \int_{\I}r^{2m+k-1}(\partial_vf)^2\ \leq& \  C \int_{\check{O}_{\tau_1}} r^{2m+k}(\partial_vf)^2 \\ &+ C \int_{\I} r^{k+2m-3}f^2+ r^{k-3}F^2 + r^{k-1} (\partial_vF)^2
		\end{aligned}
	\end{align*}
	\begin{align}
		\begin{aligned}
			\Rightarrow \quad 		\int_{\check{O}_{\tau_2}} r^{2m+k}(\partial_vf)^2 + \int_{\I}r^{2m+k-1}(\partial_vf)^2\ \leq& \ C \int_{\check{O}_{\tau_1}} r^{2m+k}(\partial_vf)^2 + r^{2m+k-2} f^2 \\
			& +  C \int_{\I}  r^{k-3}F^2 + r^{k-1} (\partial_vF)^2.
		\end{aligned}
	\end{align}
	Finally, in view of $ \abs{\slas{\nabla} f} = \slas{g}^{AB}\slas{\nabla}_{A}f\slas{\nabla}_{B}f = \frac{1}{r^2}(\partial_{\vartheta}f)^2 + \frac{1}{r^2\sin^2\vartheta}(\partial_{\varphi}f)^2 $, we can obtain estimates for the angular derivatives of $ f $ in terms of  $ F $ using (\ref{f-F misc}) as \begin{align*}
		\partial_u\left (r^{m+1} \dfrac{\partial_{\vartheta}f}{r}\right ) &= -\dfrac{D}{2r}\partial_{\vartheta}F \\
		\partial_u\left (r^{m+1} \dfrac{\partial_{\varphi}f}{r\sin\vartheta}\right ) &= -\dfrac{D}{2r\sin\vartheta}\partial_{\varphi}F 
	\end{align*}
	Using the above relations, we apply lemma \ref{transport-3} twice to obtain \begin{align}
		\int_{\check{O}_{\tau_2}} r^{2m+k} \abs{\slas{\nabla}f}^2 + \int_{\I}r^{2m+k-1} \abs{\slas{\nabla}f}^2 \leq C \int_{\check{O}_{\tau_1}} r^{2m+k}\abs{\slas{\nabla}f}^2 + C \int_{\I} r^{k-1}\abs{\slas{\nabla}F}^2.
	\end{align}
	\paragraph{Energy estimates near null infinity $ \I $.}
	Using the above lemma we obtain estimates for $ h_{\mathfrak{f}},\   {h_{\tilde{\beta}_{\ell}}}$ which satisfy (\ref{transport + spin scalar}) and  with respect to the double null coordinate system they read \begin{align}
		\partial_u(r^2 h_{\mathfrak{f}}) &= \dfrac{D(r)}{2r} \phi \label{trans-misc1} \\ 
		\partial_u(r^4 {h_{\tilde{\beta}_{\ell}}}) &= \dfrac{D(r)}{2r} \psi \label{trans-misc2}
	\end{align}
	In order to obtain bounds with respect to initial data of $ \p^{^{(\ell)}} $, we will be working with the spherical harmonics decomposition of the scalars above. In particular, we have the following estimates
	\begin{corollary} \label{h estimates away}
		For any $ k>0 $ the following estimates hold for a constant $ C $ depending on $ k, M, R.$\begin{align}
			\begin{aligned}
				\int_{\check{O}_{\tau_2}}& r^{k+2}({h_{\mathfrak{f}_{\ell}}})^{2}+ r^{k+4}(\partial_v{h_{\mathfrak{f}_{\ell}}})^2 + r^{k+4}\abs{\slas{\nabla}{h_{\mathfrak{f}_{\ell}}}}^2 + \int_{\I}r^{k+1}{h_{\mathfrak{f}_{\ell}}}^{2}+ r^{k+3}(\partial_v{h_{\mathfrak{f}_{\ell}}})^2 + r^{k+3}\abs{\slas{\nabla}{h_{\mathfrak{f}_{\ell}}}}^2 \\ &\leq  C \int_{\check{O}_{\tau_1}} r^{k+2}{h_{\mathfrak{f}_{\ell}}}^{2}+ r^{k+4}(\partial_v{h_{\mathfrak{f}_{\ell}}})^2 + r^{k+4}\abs{\slas{\nabla}{h_{\mathfrak{f}_{\ell}}}}^2 + C \int_{\I}r^{k-3}\phi_{\ell}^{2}+ r^{k-1}(\partial_v\phi_{\ell})^2 + r^{k-1}\abs{\slas{\nabla}\phi_{\ell}}^2 
			\end{aligned} \\
			\begin{aligned}
				\int_{\check{O}_{\tau_2}}& r^{k+6}h_{\tilde{\beta}_{\ell}}^{2}+ r^{k+8}(\partial_vh_{\tilde{\beta}_{\ell}})^2 + r^{k+8}\abs{\slas{\nabla}h_{\tilde{\beta}_{\ell}}}^2 + \int_{\I}r^{k+5}h_{\tilde{\beta}_{\ell}}^{2}+ r^{k+7}(\partial_vh_{\tilde{\beta}_{\ell}})^2 + r^{k+7}\abs{\slas{\nabla}h_{\tilde{\beta}_{\ell}}}^2. \\ &\leq  C \int_{\check{O}_{\tau_1}} r^{k+6}h_{\tilde{\beta}_{\ell}}^{2}+ r^{k+8}(\partial_vh_{\tilde{\beta}_{\ell}})^2 + r^{k+8}\abs{\slas{\nabla}h_{\tilde{\beta}_{\ell}}}^2 + C \int_{\I}r^{k-3}{\psi_{\ell}}^{2}+ r^{k-1}(\partial_v{\psi_{\ell}})^2 + r^{k-1}\abs{\slas{\nabla}{\psi_{\ell}}}^2. 
			\end{aligned}
		\end{align}
	\end{corollary}
	
	\paragraph{Energy estimates in the intermediate region $ \br{r_c\leq r\leq R} $.} Given $ r_c\in(M,2M) $ and $ R>2M $ introduced in Section \ref{r-p}, we can apply the divergence theorem for (\ref{div-misc}) to obtain bounds for $ {h_{\mathfrak{f}_{\ell}}}, {h_{\tilde{\beta}_{\ell}}} $ in the intermediate region $\mathcal{D}(0,\tau):=R(0,\tau)\cap \br{r_c\leq r\leq R}.$ \\
	We no longer need to pay attention in the weights in $ r $ and thus obtain the following estimates 
	\begin{align}
		\begin{aligned}
			\int_{\check{\Sigma}_{\tau_2}\cap\mathcal{D}}& \frac{1}{r^2} {h_{\mathfrak{f}_{\ell}}}^2 + (\partial_v {h_{\mathfrak{f}_{\ell}}})^2 + (\partial_u {h_{\mathfrak{f}_{\ell}}})^2 + \abs{\slas{\nabla}{h_{\mathfrak{f}_{\ell}}}}^2 + 		\int_{\mathcal{D}(0,\tau)} \frac{1}{r^2}{h_{\mathfrak{f}_{\ell}}}^2 + (\partial_v {h_{\mathfrak{f}_{\ell}}})^2 + (\partial_u {h_{\mathfrak{f}_{\ell}}})^2 + \abs{\slas{\nabla}{h_{\mathfrak{f}_{\ell}}}}^2 \\ & \leq C 		\int_{\check{\Sigma}_{\tau_1}\cap\mathcal{D}} \left( \frac{1}{r^2}{h_{\mathfrak{f}_{\ell}}}^2 + (\partial_v {h_{\mathfrak{f}_{\ell}}})^2 + (\partial_u {h_{\mathfrak{f}_{\ell}}})^2 + \abs{\slas{\nabla}{h_{\mathfrak{f}_{\ell}}}}^2 + \dfrac{1}{r^2}\phi_{\ell}^2\right) + 		C\int_{\mathcal{D}(0,\tau)} \frac{1}{r^2}\phi_{\ell}^2 + (\partial_v\phi_{\ell})^2 +\abs{\slas{\nabla}\phi_{\ell}}^2 \\
			&  \hspace{1cm}+ C\int_{\mathcal{D}\cap \br{r=R}}\frac{1}{r^2}{h_{\mathfrak{f}_{\ell}}}^2 +  (\partial_v {h_{\mathfrak{f}_{\ell}}})^2 + \abs{\slas{\nabla}{h_{\mathfrak{f}_{\ell}}}}^2. 
		\end{aligned}
	\end{align} 
	Note, using the relation (\ref{trans-misc1}) we were able to  obtain  bounds for $ (\partial_u{h_{\mathfrak{f}_{\ell}}})^2 $ as well, explaining where the extra term $ \frac{1}{r^2}\phi_{\ell}^2 $ on the $ \check{\Sigma}_{\tau_1}\cap \mathcal{D} $ integral comes from. In addition, by applying the divergence theorem in $ \mathcal{D}(0,\tau) $ we also obtain on the right-hand side a timelike boundary term which is controlled altogether by the right-hand side of the estimates in Corollary \ref{h estimates away}.
	
	Thus, using the Morawetz estimate of Theorem \ref{Spacetime non degenerate photon estimate} and Proposition \ref{posT} we may rewrite the above estimate with respect to the $ (t,r^{\star}) $ coordinate system as \begin{align}
		\begin{aligned}
			\int_{\check{\Sigma}_{\tau_2}\cap\mathcal{D}}&\frac{1}{r^2} {h_{\mathfrak{f}_{\ell}}}^2 + (\partial_t {h_{\mathfrak{f}_{\ell}}})^2 + (\partial_{r^{\star}} {h_{\mathfrak{f}_{\ell}}})^2 + \abs{\slas{\nabla}{h_{\mathfrak{f}_{\ell}}}}^2 + \int_{\mathcal{D}(0,\tau)} \frac{1}{r^2}{h_{\mathfrak{f}_{\ell}}}^2 + (\partial_t {h_{\mathfrak{f}_{\ell}}})^2 + (\partial_{r^{\star}} {h_{\mathfrak{f}_{\ell}}})^2 + \abs{\slas{\nabla}{h_{\mathfrak{f}_{\ell}}}}^2 \\ \leq
			&\  C\ \sum_{i=1}^{2}\left( \int_{\check{\Sigma}_1} J_{\mu}^T[\Psi_i]n_{\Sigma_{\tau_1}}^{\mu}+J_{\mu}^T[T\Psi_i]n_{\check{\Sigma}_1}^{\mu}\right) +  C	\int_{\check{\Sigma}_{\tau_1}\cap\mathcal{D}}\frac{1}{r^2} {h_{\mathfrak{f}_{\ell}}}^2 + (\partial_t {h_{\mathfrak{f}_{\ell}}})^2 + (\partial_{r^{\star}} {h_{\mathfrak{f}_{\ell}}})^2 + \abs{\slas{\nabla}{h_{\mathfrak{f}_{\ell}}}}^2 \\
			&  \hspace{8cm}+ C\int_{\mathcal{D}\cap \br{r=R}}\frac{1}{r^2}{h_{\mathfrak{f}_{\ell}}}^2 +  (\partial_v {h_{\mathfrak{f}_{\ell}}})^2 + \abs{\slas{\nabla}{h_{\mathfrak{f}_{\ell}}}}^2. \label{h-D-estimates}
		\end{aligned}
	\end{align}
	The same estimate also holds for $ {h_{\tilde{\beta}_{\ell}}} $ in the region $ \mathcal{D}(0,\tau). $
	
	\paragraph{Energy estimates near and including the event horizon $ \h $.} In this paragraph, we will be using the coordinate system $ (v,r,\vartheta,\varphi) $ which is regular on  the horizon $\h.  $ Using the transport equations (\ref{transport + spin scalar}) we obtain the corresponding Lemma  \ref{transport-3}, where all scalars involved are regular on the horizon $ \h. $ Simply by repeating the ideas above we obtain the following estimates in the region $ \mathcal{A}_c := \mathcal{R}(0,\tau) \cap \br{M\leq r\leq r_c} $  \begin{align*}
		\begin{aligned}
			\int_{\check{\Sigma}_{\tau_2}\cap \mathcal{A}_c}& \frac{1}{r^2} {h_{\mathfrak{f}_{\ell}}}^2 + (\partial_v{h_{\mathfrak{f}_{\ell}}})^2 + (\partial_r{h_{\mathfrak{f}_{\ell}}})^2 +\abs{\slas{\nabla}{h_{\mathfrak{f}_{\ell}}}}^2 +  \int_{ \mathcal{A}_c}\frac{1}{r^2} {h_{\mathfrak{f}_{\ell}}}^2 + (\partial_v{h_{\mathfrak{f}_{\ell}}})^2 + (\partial_r{h_{\mathfrak{f}_{\ell}}})^2 +\abs{\slas{\nabla}{h_{\mathfrak{f}_{\ell}}}}^2   \\ 
			& \leq C 		\int_{\check{\Sigma}_{\tau_1}\cap \mathcal{A}_{c}} \left( {h_{\mathfrak{f}_{\ell}}}^2 + (\partial_v {h_{\mathfrak{f}_{\ell}}})^2  + \abs{\slas{\nabla}{h_{\mathfrak{f}_{\ell}}}}^2 + \dfrac{1}{r^2}\phi_{\ell}^2\right) + 		C\int_{\mathcal{A}_c}\frac{1}{r^2} \phi_{\ell}^2 + (\partial_v\phi_{\ell})^2 +\abs{\slas{\nabla}\phi_{\ell}}^2 \\
			&  \hspace{1cm}+ C\int_{\mathcal{A}_c\cap \br{r=R}}\frac{1}{r^2}{h_{\mathfrak{f}_{\ell}}}^2 +  (\partial_v {h_{\mathfrak{f}_{\ell}}})^2 + \abs{\slas{\nabla}{h_{\mathfrak{f}_{\ell}}}}^2.
		\end{aligned}
	\end{align*} 
	Then, using Proposition \ref{posT} and Theorem \ref{Nuniform } we obtain \begin{align}
		\begin{aligned}
			\int_{\check{\Sigma}_{\tau_2}\cap \mathcal{A}_c}&\frac{1}{r^2} {h_{\mathfrak{f}_{\ell}}}^2 + (\partial_v{h_{\mathfrak{f}_{\ell}}})^2 + (\partial_r{h_{\mathfrak{f}_{\ell}}})^2 +\abs{\slas{\nabla}{h_{\mathfrak{f}_{\ell}}}}^2 +  \int_{ \mathcal{A}_c} \frac{1}{r^2} {h_{\mathfrak{f}_{\ell}}}^2 + (\partial_v{h_{\mathfrak{f}_{\ell}}})^2 + (\partial_r{h_{\mathfrak{f}_{\ell}}})^2 +\abs{\slas{\nabla}{h_{\mathfrak{f}_{\ell}}}}^2   \\ 
			& \leq C 		\int_{\check{\Sigma}_{\tau_1}\cap \mathcal{A}_{c}} \left(\frac{1}{r^2} {h_{\mathfrak{f}_{\ell}}}^2 + (\partial_v {h_{\mathfrak{f}_{\ell}}})^2  + \abs{\slas{\nabla}{h_{\mathfrak{f}_{\ell}}}}^2\right) + C \sum_{i=1}^{2}\int_{\check{\Sigma}_{\tau_1}}J^{N}_{\mu}[\p]n_{\check{\Sigma}}^{\mu}\\
			&  \hspace{1cm}+ C\int_{\mathcal{A}_c\cap \br{r=R}}\frac{1}{r^2} {h_{\mathfrak{f}_{\ell}}}^2 +  (\partial_v {h_{\mathfrak{f}_{\ell}}})^2 + \abs{\slas{\nabla}{h_{\mathfrak{f}_{\ell}}}}^2		 \label{h-A-estimates}
		\end{aligned}
	\end{align} 
	
	\paragraph{Combining all estimates above.}
	Now, in order to write all the above estimates more concisely we introduce the following energy quantities. For any scalar $ h $ regular up to and including the horizon  $  \h$, we consider 
	\begin{align}
		E_{\geq R}^{p}[h](\tau) :=& 	\int_{\check{O}_{\tau}} r^{p+1}{h}^{2}+ r^{p+3}(\partial_v{h})^2 + r^{p+3}\abs{\slas{\nabla}{h}}^2, \hspace{1cm}  p>0 \\
		E_{c,R}[h](\tau) : =& 	\int_{\check{\Sigma}_{\tau}\cap\mathcal{D}}\frac{1}{r^2} {h}^2 + (\partial_t h)^2 + (\partial_{r^{\star}} h)^2 + \abs{\slas{\nabla}{h}}^2 \\
		E_{\mathcal{A}_c}[h](\tau) :=& \int_{\check{\Sigma}_{\tau}\cap \mathcal{A}_c}\frac{1}{r^2} {h}^2 + (\partial_v{h})^2 + (\partial_r{h})^2 +\abs{\slas{\nabla}{h}}^2
	\end{align}
where in each region we use the associated coordinate system as done earlier.
	In addition, we also denote by \begin{align}
		E^{p}[h](\tau):= 	E_{\geq R}^{p}[h](\tau) + 	E_{c,R}[h](\tau) + 	E_{\mathcal{A}_c}[h](\tau).
	\end{align}
We are now ready to show the following proposition.
		\begin{proposition} Let $ \mathfrak{f}, \tilde{\beta} $ be  solutions to the Teukolsky system and consider the induced scalars
			$ h_{\mathfrak{f}} = r^2\dd\DD \mathfrak{f}_{\star} $, $  {h_{\tilde{\beta}_{\ell}}}= r^2\dd\DD \tilde{\beta}_{\star},  $   then \begin{align}
				E^{1-\delta}[h_{\mathfrak{f}_{\ell}}](\tau) \leq C \tilde{E}^{3-\delta}[h_{\mathfrak{f}_{\ell}}](0)\dfrac{1}{\tau^2},  \hspace{1cm} \forall \ell \geq 2,
			\end{align}
			where $ \tilde{E}^{3-\delta}[h_{\mathfrak{f}_{\ell}}](0)$ is the right-hand side of (\ref{h-misc-12}) at $ \check{\Sigma}_{0}. $ Similarly, we have \begin{align}
				E^{5-\delta}[h_{\tilde{\beta}_{\ell}}](\tau) \leq C \tilde{E}^{7-\delta}[h_{\tilde{\beta}_{\ell}}](0)\dfrac{1}{\tau^2}, \hspace{1cm} \forall \ell \geq 1.
			\end{align}
		\end{proposition}
	\begin{proof}
			Let $ 0<\delta\ll1 $, then using Corollary \ref{h estimates away} for $ k=1-\delta $,  estimates (\ref{h-D-estimates}, \ref{h-A-estimates}) and coarea formula we obtain \begin{align}	
			\begin{aligned}
				\int_{\tau_1}^{\tau_2} E^{1-\delta} [h_{\mathfrak{f}_{\ell}}](\tilde{\tau})d\tilde{\tau}\ \leq\ C &  E^{2-\delta} [h_{\mathfrak{f}_{\ell}}](\tau_1)  + C\sum_{i=1}^{2}\int_{\check{\Sigma}_{\tau_1}}J_{\mu}^{N}[\p] n^{\mu}_{\check{\Sigma}} \\  &+ C\sum_{i=1}^{2}\int_{\check{\Sigma}_{\tau_1}}J_{\mu}^{N}[T\p] n^{\mu}_{\check{\Sigma}} + C \int_{\check{O}_{\tau_1}}r^{1-\delta} \dfrac{(\partial_v\Phi_i)^2}{r^2}.
			\end{aligned}
		\end{align}
		Let us denote by $ \tilde{E}^{2-\delta}[h_{\mathfrak{f}_{\ell}}](\tau) $ the right-hand side of the above relation with $ \tau $ instead $ \tau_1. $ Then, we also need an estimate for \begin{align} \label{QQQ}
			\int_{\tau_1}^{\tau_2}  \tilde{E}^{2-\delta}[h_{\mathfrak{f}_{\ell}}](\tilde{\tau})d\tilde{\tau}.
		\end{align}
		However, applying the same steps as above along with the estimates of Proposition \ref{r_p} for $ k=2-\delta $, and using (\ref{NonDEB}) we obtain \begin{align}
			\begin{aligned}
				\int_{\tau_1}^{\tau_2}  \tilde{E}^{2-\delta}[h_{\mathfrak{f}_{\ell}}](\tilde{\tau})d\tilde{\tau}\ \leq \ C E^{3-\delta}[h_{\mathfrak{f}_{\ell}}]&(\tau_{1}) +  C\sum_{i=1}^{2}\int_{\check{\Sigma}_{\tau_1}}J_{\mu}^{N}[\p] n^{\mu}_{\check{\Sigma}}+ \\ +   &C\sum_{i=1}^{2}\Big(\mathcal{E}[\p](\tau_1)+ \mathcal{E}[T\p](\tau_1)\Big) + C \int_{\check{O}_{\tau_1}}\dfrac{1}{r^{\delta}} (\partial_v\Phi_i)^2. \label{h-misc-12}
			\end{aligned}
		\end{align}
		Therefore, repeating the argument of a dyadic sequence of Proposition \ref{non-degen decay} we show decay for the energy fluxes of the assumption. \\
	\end{proof}
\begin{remark}
	Note, we had to give some  $ \delta>0 $ ``space" when obtaining estimates for $ E^{1-\delta}[h_{\mathfrak{f}_{\ell}}]$ and $ E^{5-\delta}[h_{\tilde{\beta}_{\ell}}] $. That is because if we apply Corollary \ref{h estimates away} for $ k=2 $, instead of $ k=2-\delta $, then we cannot control the right-hand side $ \phi_{\ell},\psi_{\ell} $ terms using Proposition \ref{r_p}.
	
\end{remark}

		Once again, we repeat the $ \mathbb{S}^{2} $ estimates of Section \ref{Section - Regge Wheeler estimates} for the scalars $ \left(r^{\frac{4-\delta}{2}}h_{\mathfrak{f}}\right), \left(r^{\frac{8-\delta}{2}}h_{\tilde{\beta}}\right) $, and we commute with the angular momentum operators to obtain the pointwise estimates below, using Sobolev inequalities. In addition, repeating the ideas of this section for $ h_{\alpha_{\ell}} $ satisfying the transport equation (\ref{scalar relating a,f,b}) we conclude 
		\begin{corollary} \label{Decay teukolsky +}
			For any $ 0<\delta \ll 1 $, there exists $ C>0 $ depending on $ M, \check{\Sigma}_0, \delta $ and norms of initial data such that \begin{align}
			  	\abs{r^{\frac{5-\delta}{2}} h_{\mathfrak{f}_{\ell}}}\leq C \dfrac{1}{\tau}, \hspace{1cm} 	\abs{r^{\frac{9-\delta}{2}}h_{\tilde{\beta}_{\ell}}}\leq C \dfrac{1}{\tau}, \hspace{1cm} 	\abs{r^{\frac{4-\delta}{2}} h_{\alpha_{\ell}}} \leq C \dfrac{1}{\tau},  \label{a,f,b decay}
			\end{align}
		\end{corollary}
		for all $ r\geq M, \ \tau \geq 1,$ and all admissible frequencies $ \ell \in \mathbb{N}. $ Using standard elliptic identities of Section \ref{Geometry}, and Sobolev inequalities we obtain the following pointwise estimates 
		\begin{align}
			\norm{\mathfrak{f}_{\star}}_{L^{\infty}(S^{2}_{v,r})} \leq C \dfrac{1}{v \cdot r^{\frac{5-\delta}{2}}}, \hspace{0.5cm} \norm{\tilde{\beta}_{\star}}_{L^{\infty}(S^{2}_{v,r})} \leq C \dfrac{1}{v \cdot r^{\frac{9-\delta}{2}}}, \hspace{0.5cm} \norm{\alpha_{\star}}_{L^{\infty}(S^{2}_{v,r})} \leq C \dfrac{1}{v \cdot r^{\frac{5-\delta}{2}}}
		\end{align}
		for all $ r\geq M $ and $ v>0. $
		
		\subsection{Decay, Non-decay and Blow-up for \texorpdfstring{$ \mathfrak{f}_{\star}, \tilde{\beta}_{\star}$}{PDFstring}, and \texorpdfstring{$ \alpha_{\star} $}{PDFstring} along the horizon \texorpdfstring{$ \h $}{PDFstring}.} 
		
		We derive estimates for both $ \mathfrak{f}_{\star}, \tilde{\beta}_{\star} $ and their transversal derivatives along the event horizon $ \h. $ In addition, using the relating equation (\ref{scalar relating a,f,b}) we also prove estimates for the extreme curvature component $ \alpha $ along $\h. $
		\begin{theorem} \label{f,b big Theorem}
			Let $ \mathfrak{f}_{\star}, \tilde{\beta}_{\star} $ be solutions to the generalized Teukolsky equations of positive spin, then for generic initial data the following estimates hold asymptotically along $ \h $ for all $ \tau \geq 1 $ 
			\begin{itemize}
				\item Decay
				\begin{align*}
					\norm{\slashed{\nabla}_{\partial_r}^k\mathfrak{f}_{\star}}_{L^{\infty}(S^2_{\tau,M})} &\lesssim_{_{M}} \dfrac{1}{\tau^{\left(\frac{4-k}{4}\right)^{k}}}\ , \hspace{1cm} 0\leq k\leq 2, \\
					\norm{\slashed{\nabla}_{\partial_r}^k\tilde{\beta}_{\star}}_{L^{\infty}(S^2_{\tau,M})} &\lesssim_{_{M}} \dfrac{1}{\tau^{\left(\frac{4-k}{4}\right)^{k}}}\ , \hspace{1cm} 0\leq k\leq 2
				\end{align*}
				
				\item Non-decay
				\begin{align*}
					\norm{\slashed{\nabla}_{\partial_r}^3\mathfrak{f}_{\star}}_{S^2_{\tau,M}}\quad \quad   &\xrightarrow{\tau \to \infty} \quad \dfrac{1}{10\sqrt{12}M^5}\mathcal{H}_{\ell=2}[\Psi_2], \\
					\norm{\slashed{\nabla}_{\partial_r}^3\tilde{\beta}_{\star}}_{S^2_{\tau,M}}\quad \quad   &\xrightarrow{\tau \to \infty} \quad \dfrac{1}{10\sqrt{6}M^6}\mathcal{H}_{\ell=2}[\Psi_2]
				\end{align*}
				
				\item Blow-up
				\begin{align*}
					\norm{\slashed{\nabla}_{\partial_r}^{k+3}\mathfrak{f}_{\star}}_{S^2_{\tau,M}} \ &= \ \frac{a_{2,k}^{^{(2)}}}{10\sqrt{12}M^5} \mathcal{H}_{\ell=2}[\Psi_2]\cdot \tau^{k} + \mathcal{O}\left(\tau^{k-\frac{1}{4}}\right), \hspace{1cm} \forall\ k\geq0, \\
					\norm{\slashed{\nabla}_{\partial_r}^{k+3}\tilde{\beta}_{\star}}_{S^2_{\tau,M}} \ &= \ \frac{a_{2,k}^{^{(2)}}}{10\sqrt{6}M^6} \mathcal{H}_{\ell=2}[\Psi_2] \cdot \tau^{k} + \mathcal{O}\left(\tau^{k-\frac{1}{4}}\right), \hspace{1cm} \forall\ k\geq0,
				\end{align*}
			\end{itemize}
			where $ \mathcal{H}_{\ell=2}[\Psi_2] :=\norm{H_{2}[\Psi_2]}_{S^{2}_{\tau,M}}, \ \forall \ \tau \geq 1, $ and $ a_{2,k}^{^{(2)}} $ is given in Theorem \ref{Scalar blow up}.
		\end{theorem}
		\begin{proof}
			Using the decay estimates we have acquired for  $ h_{\mathfrak{f}} $, we  produce estimates for $ \partial_r^k(h_{\mathfrak{f}}) $.
			Going back to (\ref{transport + spin scalar}), we can write it as  \begin{align*}
				\partial_r(h_{\mathfrak{f}}) &= -\dfrac{\phi}{r^3}-\dfrac{2}{r}h_{\mathfrak{f}}
			\end{align*} 
			By induction, it's immediate to check for all $  k\in\mathbb{N} $ \begin{align}
				\partial_r^{k+1}(h_{\mathfrak{f}}) & = -\dfrac{\partial_r^{k}\phi}{r^3} + \mathfrak{D}^{<k}(\phi) + A_k(r)\cdot h_{\mathfrak{f}},
			\end{align}
			where $ \mathfrak{D}^{k}(\phi) $ is an expression involving $ \partial_r- $derivatives of $ \phi $ of order less than $ k $ and $ A_k(r) $ is  a function of $ r $ alone, for any $ k $. 
			Thus, Corollary \ref{psi-phi l study } and relation  (\ref{a,f,b decay})  yield the following decay estimates  asymptotically on $ \h $ \begin{align}
				\norm{\partial_r^k h_{\mathfrak{f}_{\ell}}}_{L^{\infty}(S^2_{\tau,M})} \lesssim_{M} \begin{cases}
					\frac{1}{\tau}, \ \	\quad k=0, \\
					\frac{1}{\tau^{\frac{3}{4}}}, \quad k=1, \\
					\frac{1}{\tau^{\frac{1}{4}}}, \quad k=2, \label{misc13-decay}
				\end{cases} 
			\end{align}
			for any $ \ell \geq 2. $ In addition, we deduce the following estimates for higher order derivatives along $ \h $ \begin{align}
				\abs{\partial_r^{k} h_{\mathfrak{f}_{\ell}}} &
				\lesssim_{_{k,M}} \dfrac{1}{\tau^{\frac{1}{4}}}, \hspace{1cm} \forall \ k\leq \ell \label{misc12-4} \\
				\partial_r^{k+\ell+1}(h_{\mathfrak{f}_{\ell}})(\tau,\omega)\ &\ =\  (-1)^{k+1}\dfrac{a_{2,k}^{^{(\ell)}}}{2M^{5}(2\ell+1)} H_{\ell}[\Psi_2](\omega)\cdot \tau ^{k} + \mathcal{O}\left(\tau^{k-\frac{1}{4}}\right), \hspace{1cm} k\geq 0.
			\end{align}
			Now, using the elliptic identity (\ref{elliptic 6})  and  relation (\ref{misc12-a}) we write \begin{align*}
				\int_{S^{2}_{\tau,M}} \abs{\slas{\nabla}_{\partial_r}^{k}\mathfrak{f}_{\star}}^{2} = \sum_{\ell\geq 2}\left ( \int_{S^{2}_{\tau,M}}\abs{\slas{\nabla}_{\partial_r}^{k}\mathfrak{f}_{\star_{\ell}}}^{2}\right ) = \sum_{\ell\geq 2} \left[\dfrac{2}{\ell(\ell+1)}\cdot \dfrac{1}{\ell(\ell+1)-2}\left ( \int_{S^{2}_{\tau,M}}\abs{\partial_r^{k}h_{\mathfrak{f}_{\ell}}}^{2}\right )\right]
			\end{align*}
			In view of $ h_{\mathfrak{f}_{\ell=2}} $ having the dominant behavior along $ \h, $ we obtain the decay results for $ k\leq 2 $, and for $ k>3 $ we see that the infinite sum for $ \ell\geq k $ decays uniformly in  $ \ell $ while the first remaining terms of the sum inherit the dominant asymptotic of $ \partial_r^{k}h_{\mathfrak{f}_{\ell=2}}. $ 
			
			To produce the estimates of the assumption for $ \tilde{\beta}_{\star} $, we work similarly as above, however, we use the elliptic identity  (\ref{elliptic 5 }) instead.
			\\ 
		\end{proof}
		We conclude the study of the positive spin gauge invariant components by
		proving estimates for the extreme curvature tensor $ \alpha_{\star} $, using the estimates obtained for $ \mathfrak{f}_{\star} $ and $ \tilde{\beta}_{\star}. $
		\begin{theorem} \label{a big Theorem}
			Let $ \alpha_{\star}$ be a solution to the generalized Teukolsky equation of $ +2 $ spin, then for generic initial data the following estimates hold asymptotically along $ \h $ for all $ \tau \geq 1 $ \begin{align*}
				&\norm{\slashed{\nabla}_{\partial_r}^k\alpha_{\star}}_{L^{\infty}(S^2_{\tau,M})} \lesssim_{_{M}} \dfrac{1}{\tau}\ , \hspace{1cm} 0\leq k\leq 2. \\
				&\norm{\slashed{\nabla}_{\partial_r}^{k+2}\alpha_{\star}}_{L^{\infty}(S^2_{\tau,M})} \lesssim_{_{M}} \dfrac{1}{\tau^{\left(\frac{4-k}{4}\right)^k }}\ , \hspace{1cm} 1\leq k\leq 2. \\
				&\norm{\slashed{\nabla}_{\partial_r}^5\alpha_{\star}}_{S^2_{\tau,M}}\quad \quad   \xrightarrow{\tau \to \infty} \quad  \left ( \dfrac{1}{12}\left(\dfrac{16}{35M^7}\right) ^{2}\mathcal{H}_{\ell=2}[\Psi_2]^{2} +
				\dfrac{1}{60}\left(\dfrac{6}{35M^6}\right)^{2}\mathcal{H}_{\ell=3}[\Psi_2]^{2} \right )^{\frac{1}{2}} \\
				&	\norm{\slas{\nabla}_{\partial_r}^{k+5}\alpha_{\star}}_{S^{2}_{\tau,M}} = c_k \cdot \tau^{k} + \mathcal{O}\left (\tau^{k-\frac{1}{4}}\right ), \hspace{2cm} \forall\ k\geq 0\\
			\end{align*}
			where $ \mathcal{H}_{\ell}[\Psi_2] :=\norm{H_{\ell}[\Psi_2]}_{S^{2}_{\tau,M}}, \ \forall \ \tau \geq 1, $ and the constant $ c_k $ are given by \[ c_k = \left(\dfrac{1}{12}\left(\dfrac{4(k+4)a_{2,k}^{^{(2)}}}{5M^{7}}\right)^{2}\mathcal{H}_{2}^{2}[\Psi_2] + \dfrac{1}{60}\left(\dfrac{6a_{2,k}^{^{(3)}}}{35M^{6}}\right)^{2}\mathcal{H}_{3}^{2}[\Psi_2] \right) ^{\frac{1}{2}}, \hspace{1cm}\forall\ k\geq 0. \]
		\end{theorem} 
		\begin{proof}
			Let us denote by \begin{align*}
				\tilde{h}_{\mathfrak{f}} : = r^2 \cdot h_{\mathfrak{f}}, \hspace{1cm} \tilde{h}_{\tilde{\beta}}: = r^4 \cdot h_{\tilde{\beta}}
			\end{align*}
			then, we rewrite the relating equation (\ref{scalar relating a,f,b}) as \begin{align}
				\partial_r h_{\alpha_{\ell}} = - \dfrac{\mu^{2}}{2M} \dfrac{1}{r^3} \tilde{h}_{\tilde{\beta}_{\ell}} - \dfrac{2}{r^3} (1-4\sqrt{D}) \tilde{h}_{\mathfrak{f}_{\ell}} - \dfrac{1}{r} h_{\alpha_{\ell}}
			\end{align}
			By taking k many $ \partial_r-$derivatives of the above relation and by writing explicitly the first two top order terms only,  we obtain the following expression \begin{align}\begin{aligned}
					\partial_r^{k+1}h_{\alpha_{\ell}}\ =\  & -\dfrac{\mu^2}{2M}\dfrac{1}{r^3} \partial_r^{k} \tilde{h}_{\tilde{\beta}_{\ell}} - \dfrac{2}{r^3}(1-4\sqrt{D})\partial_r^{k}\tilde{h}_{\mathfrak{f}_{\ell}}\\
					& + \dfrac{(3k+1)}{2M} \dfrac{\mu^2}{r^4} \partial_r^{k-1}\tilde{h}_{\tilde{\beta}_{\ell}} + \dfrac{2}{r^4}\left(7k+1 -(4+16k)\sqrt{D}\right) \partial_r^{k-1}\tilde{h}_{\mathfrak{f}_{\ell}}\\ & + \mathcal{D}^{\leq k-2}\left[\tilde{h}_{\tilde{\beta}}, \tilde{h}_{\mathfrak{f}}, h_{\alpha}\right] \label{k+1-of a}
				\end{aligned}
			\end{align}
			where $ \mathcal{D}^{\leq s}[\ \cdot \ ] $ is a linear expression involving up to $ s-$many $ \partial_r-$derivatives of its arguments. However, in view of the transport equations that $ \tilde{h}_{\mathfrak{f}}, \tilde{h}_{\tilde{\beta}}$ satisfy, i.e. \begin{align*}
				- \dfrac{\phi}{r} = \partial_r \tilde{h}_{\mathfrak{f}}, \hspace{1cm} 	- \dfrac{\psi}{r} = \partial_r \tilde{h}_{\tilde{\beta}},
			\end{align*}
			relation (\ref{k+1-of a}) reads \begin{align}
				\begin{aligned} 
					\partial_r^{k+1}h_{\alpha_{\ell}}\ = \ &\  \dfrac{\mu^2}{2M}\dfrac{1}{r^4} \partial_r ^{k-1}\psi_{\ell} + \dfrac{2}{r^4}(1-4\sqrt{D})\partial_r^{k-1} \phi_{\ell} \\
					& - \dfrac{\mu^{2}}{2M}\dfrac{(4k+1)}{r^5} \partial_r^{k-2}\psi_{\ell}  - \dfrac{2}{r^5} \left(8k+1 - (4+20k)\sqrt{D}\right) \partial_r^{k-2} \phi_{\ell}\\ & + \mathcal{D}^{k-3}[\phi,\psi] + A_1(r)\cdot h_{\alpha_{\ell}} + A_2(r)\cdot \tilde{h}_{\mathfrak{f}_{\ell}} + A_3(r) \cdot \tilde{h}_{\tilde{\beta}_{\ell}}, \label{k+1- of a in terms of phi,psi}
				\end{aligned}
			\end{align}
			where $ A_i(r),\  i\in \br{1,2,3} $ are scalar functions of $ r $ alone.
		
			In view of the estimates obtained for $ \Psi_i^{^{(\ell)}}, \ i \in \br{1,2} $ in Section \ref{Section - Regge Wheeler estimates}, it suffices to study the low frequencies $ h_{\alpha_{\ell}} $, $ \ell = 2,3 $, which give the dominant behavior of $ \alpha_{\star} $ along the event horizon $ \h. $	
			Using the estimates of Corollary \ref{psi-phi l study } and relation (\ref{k+1- of a in terms of phi,psi}), we obtain the following decay estimates  \begin{align} \label{a decay 1}
				\abs{\partial_r^{k} h_{\alpha_{\ell}}}(\tau, \omega) &\lesssim_{_{M}} 
				\dfrac{1}{\tau^{\frac{1}{4}}}, \hspace{2cm}  \forall \ \tau \geq 1,\ \ell \geq 3, \ k\leq 4
			\end{align}
			and in particular, for $ \ell=2 $ we have
			\begin{align}
				\abs{\partial_r^{k} h_{\alpha_{_{\ell=2}}}}(\tau, \omega) &\lesssim_{_{M}} 
				\dfrac{1}{\tau}, \hspace{2cm}  \forall \ \tau \geq 1, \ k\leq 2  \\ 	\abs{\partial_r^{k+2} h_{\alpha_{_{\ell=2}}}}(\tau, \omega) &\lesssim_{_{M}}  \dfrac{1}{
					\tau^{\left(\frac{4-k}{4}\right)^k}}, \hspace{1cm} \forall \tau \geq 1, \ 1\leq k \leq 2 \label{a decay 2}
			\end{align}	
			\begin{flushleft}
				\hrulefill
			\end{flushleft}
			
			Next, the first non-decay estimate occurs at the level of five transversal invariant derivatives for both the $ \ell=2, \ \ell=3  $ frequencies.
			For the $ \ell=2$ frequency, relation (\ref{k+1- of a in terms of phi,psi}) yields along the event horizon $ \h $ \begin{align*}
				\partial_r^{k+1} h_{\alpha_{_{\ell=2}}}=\ & \dfrac{2}{M^4} \left(M  \partial_r^{k-1}\psi_{_{\ell=2}} +  \partial_r^{k-1} \phi_{_{\ell=2}} \right)  - \dfrac{2(4k+1)}{M^5} \left(M \partial_r^{k-2}\psi_{_{\ell=2}}+ \partial_r^{k-2} \phi_{_{\ell=2}}\right)   - \dfrac{8k}{M^5} \partial_r^{k-2}\phi_{\ell=2} \\ & + \mathcal{D}^{k-3}[\phi,\psi]\Big|_{r=M} + A_1(M)\cdot h_{\alpha_{_{\ell=2}}} + A_2(M)\cdot \tilde{h}_{\mathfrak{f}_{_{\ell=2}}} + A_3(M) \cdot \tilde{h}_{\tilde{\beta}_{_{\ell=2}}} 
			\end{align*}
			However, $( M\psi_{\ell=2} + \phi_{\ell=2} )$ is a scalar multiple of $ \Psi_1^{^{(\ell=2)}} $ and thus the dominant term in the expression above is $ -\frac{8k}{M^5}  \partial_r^{k-2} \phi_{_{\ell=2}}.$ 
			Hence, for $ k=4 $  we obtain \begin{align}
				\partial_r^{5} h_{\alpha_{_{\ell=2}}}(\tau,\omega) = - \dfrac{32}{10}\dfrac{1}{M^7} \partial_r^{2} \Psi_2^{^{(\ell=2)}} + \mathcal{O}(\tau^{-\frac{1}{4}}) \xrightarrow{\tau \to \infty} -\dfrac{32}{10M^7} H_2[\Psi_2](\omega).
			\end{align}
			Similarly, for $ \ell=3 $ and $ k=4 $ relation  (\ref{k+1- of a in terms of phi,psi})  yields \begin{align}
				\partial_r^{5}h_{\alpha_{_{\ell=3}}}(\tau,\omega) =\dfrac{3}{M^5}\partial_r^{3}\psi_{_{\ell=3}} +  \left (\dfrac{2}{M^5}\partial_r^{3}\psi_{_{\ell=3}} + \dfrac{2}{M^4} \partial_r^{3} \phi_{_{\ell=3}}\right ) + \mathcal{O}(\tau^{-\frac{1}{4}}) \xrightarrow{\tau \to \infty} - \dfrac{6}{35 M^6} H_{3}[\Psi_2](\omega)
			\end{align}
			\begin{flushleft}
				\hrulefill
			\end{flushleft}
			
			For the blow-up estimates, we can see from the non-decaying results above that both frequencies $ \ell=2$  and $\ell=3 $ ought to contribute to the leading term of the asymptotics. In particular, for $ 
			\ell=2 $ we have established above that \begin{align*}
				\partial_r^{k+1} h_{\alpha_{_{\ell=2}}} =\  & - \dfrac{4k}{5M^7}\partial_r^{k-2}\Psi_2^{^{(\ell=2)}} + \mathcal{D}^{k-3}[\Psi_2^{^{(\ell=2)}}] + \mathcal{D}^{k-2}[\Psi_1^{^{(\ell=2)}}] \\ & +  A_1(M)\cdot h_{\alpha_{_{\ell=2}}} + A_2(M)\cdot \tilde{h}_{\mathfrak{f}_{_{\ell=2}}} + A_3(M) \cdot \tilde{h}_{\tilde{\beta}_{_{\ell=2}}} 
			\end{align*}
			Thus, using Theorem \ref{Scalar blow up} we obtain the following asymptotic along $ \h $ \begin{align}
				\partial_r^{k+1}h_{\alpha_{\ell=2}}(\tau,\omega) = - \dfrac{4k}{5M^{7}}(-1)^{k-4} a_{2,k-4}^{^{(2)}} H_{2}[\Psi_2] (\omega) \cdot \tau^{k-4} + \mathcal{O}\left(\tau^{k-4-\frac{1}{4}}\right)
			\end{align}
			On the other hand, we consider relation (\ref{k+1- of a in terms of phi,psi}) for $ \ell=3 $ and we obtain quite similarly as above \begin{align}
				\partial_r^{k+1}h_{\alpha_{\ell=3}}(\tau,\omega) = (-1)^{k+1}\dfrac{6}{35}\frac{a_{2,k-4}^{^{(3)}}}{M^6} H_{3}[\Psi_2] (\omega)\cdot \tau^{k-4} + \mathcal{O}\left(\tau^{k-4-\frac{1}{4}}\right)
			\end{align}
			
			Now, we pass the above estimates for the scalar $ h_{\alpha} $ to the corresponding tensor $ \alpha_{\star} $ using once again the following elliptic identity \begin{align}
				\int_{S^{2}_{\tau,M}} \abs{\slas{\nabla}_{\partial_r}^{k}\alpha_{\star}}^{2} = \sum_{\ell\geq 2}\left ( \int_{S^{2}_{\tau,M}}\abs{\slas{\nabla}_{\partial_r}^{k}\alpha_{\star_{\ell}}}^{2}\right ) = \sum_{\ell\geq 2} \left[\dfrac{2}{\ell(\ell+1)}\cdot \dfrac{1}{\ell(\ell+1)-2}\left ( \int_{S^{2}_{\tau,M}}\abs{\partial_r^{k}h_{\alpha_{\ell}}}^{2}\right )\right] \label{a - sum equation}
			\end{align}
			For $ k\leq 4 $	 we use the decay estimates (\ref{a decay 1}-\ref{a decay 2})
			and the equation above to obtain $ L^2- $decay for $ \slas{\nabla}_{\partial_r}^{k} \alpha_{\star}$. The $ L^{\infty}(S^{2}_{\tau,M}) $ estimate comes from commuting with the angular momentum operators and using the Sobolev inequality  on the sphere.
			
			Now let $ k \geq 5 $, then note that the tail of the sum (\ref{a - sum equation}) decays uniformly in $ \ell $ for frequencies  $ \ell \geq k-1 $ and we only have to deal with the first $ k-2 $ terms of the sum. However, the dominant behavior comes from the frequencies $ \ell=2 $ and $ \ell=3 $ along $ \h $ and thus we obtain \begin{align*}
				\int_{S^{2}_{\tau,M}} \abs{\slas{\nabla}_{\partial_r}^{k+1}\alpha_{\star}}^{2} = \dfrac{1}{12}\left(\dfrac{4ka_{2,k-4}^{^{(2)}}}{5M^{7}}\right)^{2}\mathcal{H}_{2}^{2}[\Psi_2] \tau^{2k-8} + \dfrac{1}{60}\left(\dfrac{6a_{2,k-4}^{^{(3)}}}{35M^{6}}\right)^{2}\mathcal{H}_{3}^{2}[\Psi_2] \tau^{2k-8} + \mathcal{O}\left(\tau^{2k-8-\frac{1}{4}}\right)
			\end{align*}
			Thus, by changing the parameter $ k $ we can write the above estimate as \begin{align*}
				\norm{\slas{\nabla}_{\partial_r}^{k+5}\alpha_{\star}}_{S^{2}_{\tau,M}} = c_k \cdot \tau^{k} + \mathcal{O}\left (\tau^{k-\frac{1}{4}}\right ), \hspace{2cm} \forall\ k\geq 0, \ \tau \geq 1
			\end{align*}	
			where \[ c_k = \left(\dfrac{1}{12}\left(\dfrac{4(k+4)a_{2,k}^{^{(2)}}}{5M^{7}}\right)^{2}\mathcal{H}_{2}^{2}[\Psi_2] + \dfrac{1}{60}\left(\dfrac{6a_{2,k}^{^{(3)}}}{35M^{6}}\right)^{2}\mathcal{H}_{3}^{2}[\Psi_2] \right) ^{\frac{1}{2}} \]
			
		\end{proof}
		
		\begin{remark}
			As opposed to the gauge invariant quantities $ \mathfrak{f}_{\star},\tilde{\beta}_{\star} $, which their behavior is dominated by their $ \ell=2 $ frequency asymptotically along $ \h $, the extreme curvature component $ \alpha_{\star} $  is dominated by both $ \ell=2 $ and $ \ell=3  $ frequencies. 
		\end{remark}

		\section{The negative spin Teukolsky equations instability.} \label{Teukolsky negative section}
		We show that solutions to the negative spin Teukolsky system exhibit an even more unstable behavior compared to the positive spin case. In particular, the $ L^2(S^{2}_{v,M})- $norm of the extreme curvature component $ \underline{\alpha}_{\star} $ does \textbf{not} decay  asymptotically along the even horizon $ \h. $ In other words, no transversal derivative of the aforementioned quantity is required for the instability to manifest. 
		
		First, we obtain decay estimates outside the event horizon $ \h $ using the transport equations, and then, we derive decay, non-decay, and blow-up  estimates asymptotically along $ \h. $ For the latter, as opposed to the $ + $ spin case where we obtained estimates using only the Regge--Wheeler system and their transport equations, for the corresponding negative spin quantities we need to use the coupled Teukolsky system they satisfy. 
				
		\subsection{The negative spin components.} 
		We are interested in obtaining estimates for the gauge invariant quantities $ \underline{\alpha}_{\star}, \underline{\mathfrak{f}}_{\star} $ and $ \underline{\beta}_{\star} $ on the exterior up to and including the event horizon $ \h $, which satisfy the following relations with respect to the Regge--Wheeler solutions 	\begin{align}
			\begin{aligned} 
				\underline{\bm{q}}^F &= \dfrac{1}{\kappa_{\star}} \slas{\nabla}_{4^{\star}}\left (r^3 \kappa_{\star}\cdot \underline{\mathfrak{f}}_{\star}\right ), \\
				\underline{\bm{p}} &= \dfrac{1}{\kappa_{\star}} \slas{\nabla}_{4^{\star}}\left (r^5 \kappa_{\star}\cdot \underline{\tilde{\beta}}_{\star}\right ). \label{trans-underline}
			\end{aligned}
		\end{align}
		They also satisfy the following relating equation, which we  use later on to obtain estimates for $ \underline{\alpha}_{\star} $ \begin{align}
			^{(F)}\rho \dfrac{1}{\kappa_{\star}} \slas{\nabla}_{4^{\star}}\left(r^3 \kappa_{\star}^2 \underline{\alpha}_{\star}\right) =  \left(^{(F)}\rho^{2} + 3 \rho\right) r^3 \kappa_{\star} \underline{\mathfrak{f}}_{\star} + r^3 \kappa_{\star} \DD^{\star}\left(\tilde{\underline{\beta}}_{\star}\right)  \label{underline relating tensor}
		\end{align}
		Here, all quantities with star subscript are expressed in the rescaled frame $ \mathcal{N}_{\star} $, and we recall that   \begin{align} \label{star- values}
			\begin{aligned}
				e_{4}^{\star} = 2 \dfrac{\partial}{\partial_v} + D \dfrac{\partial}{\partial_r}, \hspace{1cm} \kappa_{\star} =  \dfrac{2D}{r}, \hspace{1cm} ^{(F)} \rho = \dfrac{M}{r^2}, \hspace{1cm} \rho = -\dfrac{2M}{r^3}\sqrt{D}.
			\end{aligned}
		\end{align}
		Once again, it will be more convenient to work with the derived scalar quantities. In particular, using standard elliptic identities from Section \ref{Geometry}, we obtain the following equations for the scalars $ h_{\underline{\mathfrak{f}}} := r^2\dd\DD(\underline{\mathfrak{f}}_{\star}), \ h_{\underline{\tilde{\beta}}}:= r\dd(\underline{\tilde{\beta}}_{\star}),$ $ \underline{\phi} := r^2 \dd\DD\left(\underline{\bm{q}}^{F}\right), \ \underline{\psi} := r^2\dd\DD\left(\underline{\bm{p}}\right) $
		\begin{align} \label{scalar transport underline}
			\begin{aligned}
				\underline{\phi} &= 2 r^3\partial_{v}\left(h_{\underline{\mathfrak{f}}}\right) +  r \partial_{r}\left (r^2 D h_{\underline{\mathfrak{f}}}\right ) \\
				\underline{\psi} &= 2 r^5\partial_{v}\left(h_{\underline{\tilde{\beta}}}\right) +  r \partial_{r}\left (r^4 D h_{\underline{\tilde{\beta}}}\right )
			\end{aligned}
		\end{align}
		while the relating equation (\ref{underline relating tensor}) yields for the scalar $ h_{\underline{\alpha}}:= r^2 \dd\DD (\underline{\alpha}_{\star}) $,	after we project to the $ E_{\ell}$ eigenspace,
		\begin{align} \label{scalar relating underline}
			\dfrac{M}{r \cdot D^2}e_{4^{\star}} \left(r D^2 h_{\underline{\alpha}_{_{\ell}}}\right) = \dfrac{2M}{r}(1-4\sqrt{D}) h_{\underline{\mathfrak{f}}_{_{\ell}}} + r \dfrac{\mu^{2}}{2} h_{\underline{\tilde{\beta}}_{_{\ell}}}, \hspace{1cm} \mu^{2}= (\ell-1)(\ell+2)
		\end{align}
		One final rescaling of the above scalars allows us to write all equations to follow more concisely and perform computations more easily.  In particular, let \begin{align}
			\underline{f}:= r^2 \cdot h_{\underline{\mathfrak{f}}}, \hspace{1cm} \underline{b}:=r^4 \cdot h_{\tilde{\underline{\beta}}}, \hspace{1cm}  \underline{a}:= r^2 h_{\underline{\alpha}}
		\end{align}
		then we may write (\ref{scalar transport underline}) as \begin{align}
		\begin{aligned} \label{reduced transport negative}
				\dfrac{1}{r}\cdot \underline{\phi} &= \partial_{v}(2\underline{f}) + \partial_r\left(D\cdot \underline{f}\right) \\
			\dfrac{1}{r}\cdot \underline{\psi} &= \partial_{v}(2\underline{b}) + \partial_r\left(D\cdot \underline{b}\right) 
		\end{aligned}
		\end{align}
		Moreover, the relating equation (\ref{scalar relating underline}) now reads \begin{align}
			\dfrac{1}{r}\dfrac{M}{r} \Big( 2\partial_v(\underline{a}_{_{\ell}})+D\partial_r\underline{a}_{_{\ell}}+2D'\underline{a}_{_{\ell}}\Big) =  \dfrac{\mu^{2}}{2r^3} \underline{b}_{_{\ell}} + \dfrac{2M}{r^3}(1-4\sqrt{D})\underline{f}_{_{\ell}} + D \dfrac{M}{r^3} \underline{a}_{_{\ell}}. \label{scalar relating l frequency}
		\end{align}
		
		\subsection{Negative spin Teukolsky system} We write the Teukolsky equations of Section \ref{Teukolsky section} with the coefficients expanded in the $ \mathcal{N}_{\star} $ frame and we obtain\begin{align}
			\begin{aligned} \label{Tensor Teukolsky}
				\square(r\underline{\mathfrak{f}_{\star}}) &= D'(r) \slas{\nabla}_{3^{\star}}(r\underline{\mathfrak{f}_{\star}}) - \dfrac{2}{r}\slas{\nabla}_{4^{\star}}(r\underline{\mathfrak{f}_{\star}}) + \left(\dfrac{2D}{r^2}-D''(r)\right) r\underline{\mathfrak{f}_{\star}} + \dfrac{M}{r}\left(\slas{\nabla}_{4^{\star}}\underline{\alpha}_{\star}+ \left(\dfrac{2D}{r} + 2D'(r) \right)\underline{\alpha}_{\star}\right) \\
				\square(\underline{\alpha}_{\star}) &= 2 D'(r) \slas{\nabla}_{3^{\star}} (\underline{\alpha}_{\star}) - \dfrac{4}{r} \slas{\nabla}_{4^{\star}}\underline{\alpha}_{\star} - \dfrac{2\sqrt{D}}{r^2}(2-\sqrt{D})\underline{\alpha}_{\star} - \dfrac{4M}{r^2}\left(\slas{\nabla}_{3^{\star}}\underline{\mathfrak{f}}_{\star} - \dfrac{2}{r}\underline{\mathfrak{f}}_{\star}\right) \\
				\square(r^3\tilde{\underline{\beta}}_{\star}) &= D'(r)\slas{\nabla}_{3^{\star}}(r^3\tilde{\underline{\beta}}_{\star})-\dfrac{2}{r}\slas{\nabla}_{4^{\star}} (r^3\tilde{\underline{\beta}}_{\star}) + \dfrac{1}{r^2}\left(3-8\sqrt{D}+4D\right)r^3\tilde{\underline{\beta}}_{\star} \\ & \ \  + 8r \left(\dfrac{M}{r}\right)^2 \slas{div}\underline{\mathfrak{f}}_{\star} +D\cdot \mathcal{A}[\underline{\tilde{\beta}}, \underline{\alpha}],
			\end{aligned}
		\end{align}
		where $ \mathcal{A}[\underline{\tilde{\beta}}, \underline{\alpha}] $ is an expression depending on $ \slas{div}\underline{\alpha}_{\star} $ and up to one $ e_{3^{\star}} $--derivative of $ \underline{\tilde{\beta}}_{\star} $; see p. 41, \cite{giorgi2019boundedness}, pp. 133-134 \cite{giorgi2020linear}.
		Using elliptic identities from Section \ref{Geometry}, we obtain the corresponding Teukolsky system for the scalar rescaled quantities $ \underline{a}, \underline{f} $ and $ \underline{b}, $ when supported on the fixed frequency $ \ell $.  With respect to the ingoing Finkelstein-Eddington  coordinates the system reads
		\begin{align}
			\begin{aligned} \label{Teukolsky scalar rescaled}
				\square(\underline{f}_{\ell}) &= -\dfrac{6}{r^2} \left(\dfrac{M}{r}\right)^2 \underline{f}_{\ell} - \dfrac{2}{r}\partial_v(\underline{f}_{\ell}) - D'(r) \partial_r\underline{f}_{\ell} + 	\dfrac{1}{r}\dfrac{M}{r} \Big( 2\partial_v(\underline{a}_{\ell})+D\partial_r\underline{a}_{\ell}+2D'\underline{a}_{\ell}\Big) \\
				\square(\underline{a}_{\ell}) &= -\dfrac{4}{r^2} \left(\dfrac{M}{r}\right)^2 \underline{a}_{\ell}
				- \dfrac{4}{r}\partial_v(\underline{a}_{\ell}) - 2D'(r) \partial_r\underline{a}_{\ell} + \dfrac{4}{r}\dfrac{M}{r}\partial_r\underline{f}_{\ell} \\
				\square(\underline{b}_{\ell}) &= \dfrac{2}{r^2} \left(\dfrac{M}{r}\right)^2 \underline{b}_{\ell}+ \dfrac{8}{r} \left(\dfrac{M}{r}\right) ^2 \underline{f}_{\ell}- \dfrac{2}{r}\partial_v(\underline{b}_{\ell}) - D'(r) \partial_r\underline{b}_{\ell} + D \cdot  \mathcal{A}[\underline{b}_{\ell},\underline{a}_{\ell}].
			\end{aligned}
		\end{align}
		
		Using the relating equation (\ref{scalar relating l frequency}) we can rewrite the Teukolsky equation for $ \underline{f} $ so that it is coupled only with $ \underline{b} $ modulo a term with a factor of $ D(r) $ in it, i.e. \begin{align}
			\square(\underline{f}_{\ell}) &= -\dfrac{2}{r^2} \left(\dfrac{M}{r}\right)\left(2+\sqrt{D}\right)\underline{f}_{\ell}  + \dfrac{\mu^{2}}{2r^3} \underline{b}_{\ell} - \dfrac{2}{r}\partial_v(\underline{f}_{\ell}) - D'(r) \partial_r\underline{f}_{\ell}  + D\dfrac{M}{r} \dfrac{1}{r^2} \underline{a}_{\ell}.
		\end{align}

		\subsection{Decay for \texorpdfstring{$ \underline{\mathfrak{f}}_{\star}, \underline{\tilde{\beta}}_{\star}$ and $ \underline{\alpha}_{\star} $ away from the horizon $ \h. $ }{PDFstring}} 
		
	We repeat the ideas of Section \ref{f,beta estimates} for the transport equations (\ref{trans-underline}) and the relating equation (\ref{underline relating tensor}). This time we omit most of the details involved as the procedure is identical to that of Section \ref{f,beta estimates}, however, we examine more carefully the two main differences that manifest when controlling the negative spin quantities. First, we prove the following lemma.

\begin{lemma} \label{lemma negative control of transport}
	Let $ f, F $ be two scalar functions satisfying \begin{align}
		e_{4^{\star}}\left(r^{m}\kappa_{\star}^{n}f\right)=\kappa_{\star}\cdot F, \hspace{1cm} \ m,n \geq 0\label{outgoing transport}
	\end{align}
then the following energy estimates hold near future null infinity $ \I $ \begin{align}
	\int_{\I} r^{2(m-n)-k-3} f^{2} \leq C \int_{\I\cap \br{r=R}} r^{2(m-n)-k-2} f^{2}  +  C \int_{\I} \dfrac{1}{r^{k+3}} F^{2},\hspace{1cm} \forall \  k>0
\end{align}
for a constant $ C>0 $ depending on $ M, k $ and $ R $ of Section \ref{r-p}.
\end{lemma}
\begin{proof}
	We multiply relation (\ref{outgoing transport}) with $ r^{m}\kappa_{\star}^{n} f$ and we obtain \begin{align*}
		e_{4^{\star}}\left(r^{2m}\kappa_{\star}^{2n} f^{2}\right) = 2\kappa_{\star}^{n+1} r^{m} f\cdot F
	\end{align*}
Now, dividing by $ r^{q} $ for $ q>0 $ we have \begin{align*}
	e_{4^{\star}}\left(r^{2m-q}\kappa_{\star}^{2n}f^{2}\right) - e_{4^{\star}}(r^{-q}) r^{2m}\kappa_{\star}^{2n}f^{2} = 2\kappa_{\star}^{n+1} r^{m-q}f\cdot F
\end{align*}
However, $ div (e_{4^{\star}}) =\frac{2\sqrt{D}}{r} $, and note $ e_{4^{\star}}= 2\partial_v + D\partial_r $, in ingoing coordinates, thus we write \begin{align*}
	div\left(r^{2m-q}\kappa_{\star}^{2n}f^{2} \cdot e_{4^{\star}}\right) + \sqrt{D}\left(q\sqrt{D}-2\right) r^{2m-q-1} \kappa_{\star}^{2n}f^{2} &=  2 \kappa_{\star} r^{-q}  \left(r^{m}\kappa_{\star}^{n}f\right)\cdot F  \\
	&\leq \epsilon r^{2m-q}\kappa_{\star}^{2n} \kappa_{\star}f^{2} + \dfrac{1}{\epsilon} r^{-q} \kappa_{\star} F^{2}
\end{align*}
Let $ q = k + 2  $, for $ k>0 $, and pick $ \epsilon = \frac{k}{4} $, then obtain \begin{align*}
		div\left(r^{2m-k-2}\kappa_{\star}^{2n}f^{2}\cdot e_{4^{\star}} \right) +  \sqrt{D}\left(\dfrac{k}{2}\sqrt{D}-2\dfrac{M}{r}\right) r^{2m-k-3} \kappa_{\star}^{2n}f^{2} \leq \dfrac{4}{k} \dfrac{\kappa_{\star}}{r^{k+2}} \cdot F^{2}, 
\end{align*}
Note $ \kappa_{\star} = \frac{2D}{r} $, and thus there exists $ R $ large enough such that \begin{align*}
		div\left(r^{2(m-n)-k-2}f^{2}\cdot e_{4^{\star}} \right) +   r^{2(m-n)-k-3} f^{2} \leq C\dfrac{1}{r^{k+3}} F^{2}
\end{align*}
for a constant $ C>0 $ depending on $ k,M $ and $ R. $ We apply the divergence theorem on the $ \I $ region to conclude the proof. Note, the null vector $ e_{4^{\star}} $ is normal to the boundary null hypersurfaces $ \check{O}_{\tau} $, and thus no such terms appear in the estimate. \\
\end{proof}

\paragraph{Estimates near future null infinity $ \I. $}		
Already, comparing the above lemma to the corresponding Lemma \ref{transport-3} of Section \ref{Teukolsky positive section}, we see that $ k $ appears with the opposite sign. This leads to weaker decay rates in $ r $ for the negative spin Teukolsky solutions. In particular, the transport equations (\ref{trans-underline}) yield for the corresponding scalars \begin{align}
	\begin{aligned}
		e_{4^{\star}}\left(r^{3}\kappa_{\star}h_{\underline{\mathfrak{f}}}\right) &= \kappa_{\star} \cdot \underline{\phi} \\
			e_{4^{\star}}\left(r^{5}\kappa_{\star}h_{\underline{\tilde{\beta}}}\right) &= \kappa_{\star} \cdot \underline{\psi}.
	\end{aligned}
\end{align} We see that for $ h_{\underline{\mathfrak{f}}} $ we have $ m=3, n=1, $ and thus, applying the above lemma for $ k=\delta $, where $ 0<\delta\ll1, $ we obtain \begin{align}
\int_{\I} r^{3-\delta} \dfrac{1}{r^{2}}h_{\underline{\mathfrak{f}}}^{2} \leq C\int_{\I\cap \br{r=R}} r^{2-\delta} h_{\underline{\mathfrak{f}}}^{2} + C \int_{\I} \dfrac{1}{r^{3+\delta}}\underline{\phi}^{2}
\end{align}
The second term of the right-hand side above is controlled  in Proposition \ref{control N flux at null infinity} for any $ \delta>0 $, while the timelike boundary term is treated by applying the divergence theorem in the region $ \mathcal{D}(\tau_1,\tau_2):=\mathcal{R}(\tau_1,\tau_2) \cap \br{r_c \leq r\leq R}. $ Using also the transport equations for $ h_{\underline{\tilde{\beta}}}, $ $ h_{\underline{\alpha}} $, and coarea formula we obtain \begin{align}
\begin{aligned}
			\int_{\tau_1}^{\tau_2} &\left( \int_{\check{O}_{\tilde{\tau}}} r^{3-\delta}  \dfrac{1}{r^{2}}h_{\underline{\mathfrak{f}}}^{2} + r^{7-\delta} \dfrac{ 1}{r^{2}}h_{\underline{\tilde{\beta}}}^{2} + r^{1-\delta} \dfrac{1}{r^{2}}h_{\underline{\alpha}}^{2}\right) d\tilde{\tau}\ \leq \   \\ & C  \sum_{i=1}^{2}\left( \int_{\check{\Sigma}_1} J_{\mu}^T[\underline{\Psi}_i]n_{\Sigma_{\tau_1}}^{\mu}+J_{\mu}^T[T\underline{\Psi}_i]n_{\check{\Sigma}_1}^{\mu} + \int_{\check{O}_{\tau_1}} \dfrac{1}{r}(\partial_v\underline{\Phi}_i)^2\right) 
		  + \int_{\check{\Sigma}_1\cap \br{r_c\leq r\leq R}} h_{\underline{\mathfrak{f}}}^{2} + h_{\underline{\tilde{\beta}}}^{2} + h_{\underline{\alpha}}^{2}  
\end{aligned}
\end{align}   
 The remaining first-order derivatives are controlled similarly as in Section \ref{f,beta estimates}, and so do we obtain estimates in the intermediate region $ \mathcal{D}(\tau_1,\tau_2).  $ In the following paragraph, we examine a degeneracy that manifests close to the event horizon $ \h $.

\paragraph{Estimates near the event horizon $ \h. $} The main issue we are faced with near the horizon can be already seen in the proof of Lemma \ref{lemma negative control of transport}, where the estimates degenerate on $ \h. $ This is due to the occurrence of stronger trapping on $ \h $ for the negative spin Teukolsky solutions, which also leads to a stronger instability as we will see in the coming Section \ref{horizon instability negative spin}.  Nevertheless, we show how to obtain pointwise estimates on the exterior which degenerate, however, on the horizon.  In this paragraph, we will be working with the transport equations (\ref{reduced transport negative}). Fix $ M<r_c<2M $ and let $ \mathcal{A}_c := \mathcal{R}(\tau_1,\tau_2) \cap \br{M\leq r\leq r_c} $, then we have the following lemma. 
\begin{lemma}\label{horizon lemma negative spin}
	Let $ f,F $ be two scalar functions satisfying \begin{align} \label{horizon transport negative spin}
		2\partial_r(f) + \partial_r(D\cdot f) = F,
	\end{align}
with respect to the ingoing coordinates $ (v,r,\theta,\phi), $ then for any $ a\in \left[\frac{1}{2},1\right ) $ and $ \tau_1\leq \tau_2 $,  we have \begin{align} 
	 \int_{\check{\Sigma}_{\tau_2}\cap \mathcal{A}_{c}}D^{a}f^{2} + \int_{\mathcal{A}_{c}} D^{\frac{1}{2}+a}f^{2} \leq C\int_{\check{\Sigma}_{\tau_1}\cap \mathcal{A}_{c}}D^{a}f^{2} + C \int_{\mathcal{A}_c} F^{2},
\end{align}
for a constant $ C>0 $ depending on $ M, r_c, a $ and $ \Sigma_{\tau_1}. $ 
\end{lemma}
\begin{proof}
	Let $a\in \left[\frac{1}{2},1\right ) $, then 
	we multiply relation (\ref{horizon transport negative spin}) with $ D^{a} \cdot f$ and we obtain \begin{align*}
		\partial_v\left(D^{a} f^{2}\right) + \dfrac{2M}{r^{2}}D^{\frac{1}{2}+ a} f^{2} + \partial_r\left(\dfrac{D^{1+a}}{2}f^{2}\right) - \partial_r\left(\frac{D^{1+a}}{2}\right)f^{2} = D^{a} f \cdot F
	\end{align*}
Now, using the fact that $ div(\partial_v) =0 $ and $ div(\partial_r) = \frac{2}{r} $, we obtain \begin{align} \label{lemma variation negative spin}
	div \left(D^{a}f^{2}\partial_v + \frac{D^{1+a}}{2}f^{2}\partial_r\right) + \left(\frac{2M}{r^{2}} D^{\frac{1}{2}+a} -\frac{1}{2} \partial_rD^{1+a} - \frac{D^{1+a}}{r} \right)f^{2} = D^{a}f \cdot F
\end{align} 
By simplifying the above expression and applying Cauchy-Schwarz  on the right-hand side, we get \begin{align*}
		div \left(D^{a}f^{2}\partial_v + \frac{D^{1+a}}{2}f^{2}\partial_r\right) + D^{\frac{1}{2}+a}\frac{1}{r}\left(1-a -\epsilon D^{a-\frac{1}{2}}+ (a-2)\sqrt{D}\right) f^{2} \leq \frac{r}{\epsilon} F^{2
		}
\end{align*}
Thus, choosing $ \epsilon = \frac{a}{3}>0 $ and  $ r_c $ close enough to $ r=M $, we apply the divergence theorem in $ \mathcal{A}_{c} $  and we find a positive constant $ C $ depending on $ M, r_c, a $ such that the estimate of the assumption holds.\\
\end{proof}

\begin{remark} Note, $ a=\frac{1}{2} $ is the smallest value for which the above lemma holds, for if we take $ a<\frac{1}{2} $, then $ D^{2a} > D^{\frac{1}{2}+a} $ near $ \h $ and the term we obtain after applying Cauchy-Schwarz in the last line cannot be absorbed in left-hand side one.
\end{remark}

With the help of the second Hardy inequality (Lemma \ref{second hardy}) we show a slightly stronger estimate which allows us to obtain better pointwise decay for the negative spin Teukolsky solutions on the exterior.

\begin{lemma} \label{improved horizon lemma for negative spin}
	Let $ f,F $ be as in the assumptions of Lemma \ref{horizon lemma negative spin}, then there exists $ \delta_0>0 $ such that for any $ 0<\delta\leq\delta_0 $ the following estimate hold \begin{align*}
		 \int_{\check{\Sigma}_{\tau_2}\cap \mathcal{A}_{c}}D^{\frac{1}{2}-\delta}f^{2} + \int_{\mathcal{A}_{c}} D^{1-\delta}f^{2} \leq C\int_{\check{\Sigma}_{\tau_1}\cap \mathcal{A}_{c}}D^{\frac{1}{2}-\delta}f^{2} + C \int_{\mathcal{D}} F^{2} + C \int_{\mathcal{A}_c\cup\mathcal{D}} D^{1-\delta}\left[(\partial_vF)^{2} + (\partial_rF)^{2}\right],
	\end{align*}
for a positive constant $ C $ depending on $ M,r_c,r_d, \delta_0, $ and $ \Sigma_{\tau_1}. $
\end{lemma}
\begin{proof} Let $ \delta>0 $ to be determined in the end, then
	repeating the steps of the previous lemma for $ a=\frac{1}{2}-\delta $ we arrive at \begin{align*}
			div \left(D^{\frac{1}{2}-\delta}f^{2}\partial_v + \frac{D^{1-\delta}}{2}f^{2}\partial_r\right) + \left(\frac{2M}{r^{2}} D^{\frac{1}{2}-\delta} -\frac{1}{2} \partial_rD^{1-\delta} - \frac{D^{1-\delta}}{r} \right)f^{2} &= D^{\frac{1}{2}-\frac{\delta}{2}} f \cdot D^{-\frac{\delta}{2}} F \\
			&\leq \epsilon D^{1-\delta} f^{2} + \dfrac{1}{\epsilon} \left(\dfrac{F}{D^{\frac{\delta}{2}}}\right)^{2}
	\end{align*}
Now, choosing $ \epsilon>0 $ small enough we can absorb the first term of the right-hand side in the left one. However, we still need to treat the second term of the right-hand side. For that, we apply the second Hardy inequality for the regions $ \mathcal{A}_c $, $ \mathcal{D}(\tau_1,\tau_2) $ and we obtain 
 \begin{align} \label{mario}
	\int_{\mathcal{A}_c} \left(\frac{F}{D(r)^{\frac{\delta}{2}}} \right)^{2} \leq \tilde{C} \int_{\mathcal{D}} F^{2} + \tilde{C} \int_{\mathcal{A}_c\cup \mathcal{D}} D^{1-\delta}(\partial_vF)^{2} + D \left[\partial_r \left(\frac{F}{D(r)^{\frac{\delta}{2}}} \right)\right]^{2}
\end{align}
where the constant $ \tilde{C} $ depends only on $ M,r_c,r_d=R  $ and $ \Sigma_0 $.
However, we have \[ \partial_r \left(\frac{F}{D(r)^{a}} \right) = D^{-a}\partial_rF - a D^{-a-1}\partial_rD\cdot F = D^{-a}\partial_rF - \dfrac{2aM}{r^{2}}D^{-a-\frac{1}{2}} F \]
Thus going back to (\ref{mario}) we obtain \begin{align}\begin{aligned}
		\int_{\mathcal{A}_c} \left(\frac{F}{D(r)^{\frac{\delta}{2}}} \right)^{2} \leq \tilde{C}\int_{\mathcal{B}} F^{2} + \tilde{C}\int_{\mathcal{A}_c\cup \mathcal{D}} D^{1-\delta}(\partial_vF)^{2} + D^{1-\delta}(\partial_rF)^{2} + (\delta/2)^{2} \left(\dfrac{F}{D^{\frac{\delta}{2}}}\right)^{2} 
		\\ \Rightarrow \ \ \ \ 
		\int_{\mathcal{A}_c} \left (1-\tilde{C}\left(\delta/2\right)^{2}\right )\left(\frac{F}{D(r)^{a}} \right)^{2} \leq \tilde{C}\int_{\mathcal{D}} F^{2} + \tilde{C}\int_{\mathcal{A}_c\cup \mathcal{D}} D^{1-\delta}(\partial_vF)^{2} + D^{1-\delta}(\partial_rF)^{2}
	\end{aligned}
\end{align}
Thus, there exists a small enough $ \delta_0>0 $ such that $ \left (1-\tilde{C}\left(\delta/2\right)^{2}\right ) $  is uniformly positive definite for any $ 0<\delta\leq \delta_0 $. Hence, for any $ 0<\delta\leq \delta_0 $, applying the divergence theorem for the relation at the beginning concludes the proof.

\end{proof}

In view of the transport equations (\ref{reduced transport negative}) for $ \underline{f} $ and $ \underline{b} $, we apply Lemma \ref{horizon lemma negative spin}  for $ a=1-\delta $,  $ 0<\delta \ll 1  $, and using coarea formula we obtain \begin{align}
	\begin{aligned}
			\int_{\tau_1}^{\tau_2}  \left(\int_{\check{\Sigma}_{\tilde{\tau}}\cap \mathcal{A}_c} D^{\frac{3}{2}-\delta} \underline{f}^{2} + D^{\frac{3}{2}-\delta} \underline{b}^{2} \right) & d\tilde{\tau}
		 \leq C \int_{\check{\Sigma}_{\tau_1}\cap \mathcal{A}_c} D^{1-\delta} \underline{f}^{2} + D^{1-\delta} \underline{b}^{2} +  \int_{\mathcal{A}_c} \phi^{2} + \psi^{2} \\
		& \hspace{-1.5cm} \leq C \int_{\check{\Sigma}_{\tau_1}\cap \mathcal{A}_c} D^{1-\delta} \underline{f}^{2} + D^{1-\delta} \underline{b}^{2} +  C  \sum_{i=1}^{2}\left( \int_{\check{\Sigma}_1} J_{\mu}^T[\underline{\Psi}_i]n_{\Sigma_{\tau_1}}^{\mu} \right) 
	\end{aligned}
\end{align}
Next, we apply Lemma \ref{improved horizon lemma for negative spin}, coarea formula, and Theorem \ref{Spacetime non degenerate photon estimate}  to get \begin{align}
	\begin{aligned}
		\int_{\tau_1}^{\tau_2} \left(\int_{\check{\Sigma}_{\tilde{\tau}}\cap \mathcal{A}_c} D^{1-\delta} \underline{f}^{2} + D^{1-\delta} \underline{b}^{2} \right)  d\tilde{\tau} \leq C & \int_{\check{\Sigma}_{\tau_1}\cap \mathcal{A}_c} D^{\frac{1}{2}-\delta} \underline{f}^{2} + D^{\frac{1}{2}-\delta} \underline{b}^{2} \\ &+   C  \sum_{i,j=1}^{2}\left( \int_{\check{\Sigma}_1} J_{\mu}^T[\underline{\Psi}_i]n_{\Sigma_{\tau_1}}^{\mu}+J_{\mu}^T[T^{j}\underline{\Psi}_i]n_{\check{\Sigma}_1}^{\mu} \right) 
	\end{aligned}
\end{align}
Using the transport equations (\ref{reduced transport negative}) and
 commuting them with the killing vector field $ T $ and the angular derivatives we derive estimates for the first-order derivatives of $ \underline{f},\underline{b}. $ With these local integrated energy estimates  in hand, we use the above lemmas (slightly modified) for the transport equation (\ref{scalar relating l frequency}) to show energy estimates for $ \underline{a} $ as well. Using as ingredients the energy estimates in each respected region  above and applying the dyadic sequence argument we show quadratic decay for the energy of $ h_{\underline{\mathfrak{f}}}, h_{\underline{\tilde{\beta}}} $ and $ h_{\underline{\alpha}}. $

 \begin{corollary} \label{Decay teukolsky -}
Let $ 0<\delta \ll 1 $, then there exists $ C>0 $ depending on $ M, \check{\Sigma}_0, \delta $ and norms of initial data such that \begin{align}
	\abs{r^{\frac{3-\delta}{2}} D^{\frac{3}{4}-\frac{\delta}{2}} \cdot h_{\underline{\mathfrak{f}}_{\ell}}}\leq C \dfrac{1}{\tau}, \hspace{1cm} 	\abs{r^{\frac{7-\delta}{2}}D^{\frac{3}{4}-\frac{\delta}{2}}\cdot h_{\underline{\tilde{\beta}}_{\ell}}}\leq C \dfrac{1}{\tau}, \hspace{1cm} 	\abs{r^{\frac{1-\delta}{2}} D^{\frac{7}{4}-\frac{\delta}{2}}  \cdot h_{\underline{\alpha}_{\ell}}} \leq C \dfrac{1}{\tau},  \label{a,f,b decay negative spin}
\end{align}
\end{corollary}
for all $ r> M, \ \tau \geq 1,$ and all admissible frequencies $ \ell \in \mathbb{N}. $ In addition, fix $ r_0> M $, then using standard elliptic identities of Section \ref{Geometry}, and Sobolev inequalities we obtain the following pointwise decay estimates away from the horizon $ \h $
\begin{align}
\norm{\underline{\mathfrak{f}}_{\star}}_{L^{\infty}(S^{2}_{v,r})} \leq C_{r_0} \dfrac{1}{v \cdot r^{\frac{3-\delta}{2}}}, \hspace{0.5cm} \norm{\underline{\tilde{\beta}}_{\star}}_{L^{\infty}(S^{2}_{v,r})} \leq C_{r_0} \dfrac{1}{v \cdot r^{\frac{7-\delta}{2}}}, \hspace{0.5cm} \norm{\underline{\alpha}_{\star}}_{L^{\infty}(S^{2}_{v,r})} \leq C_{r_0} \dfrac{1}{v \cdot r^{\frac{1-\delta}{2}}}
\end{align}
for all $ r\geq r_0$ and $ v>0.  $

		\subsection{Estimates for \texorpdfstring{$ \underline{\mathfrak{f}}_{\star}, \underline{\tilde{\beta}}_{\star}$ and $ \underline{\alpha}_{\star} $ along the event horizon $ \h $}{PDFstring}.}  \label{horizon instability negative spin}
		First, we prove estimates for $ \underline{\mathfrak{f}}_{\star}, \tilde{\underline{
				\beta}}_{\star} $ along the event horizon $ \h. $ For that, we use the Teukolsky equations they satisfy and the transport equations relating them to the Regge--Wheeler  solutions $ \underline{\bm{q}}^{F}, \ \underline{\bm{p}} $, for which we have already obtained estimates in the previous sections. In view of $  \underline{\bm{q}}^{F}, \ \underline{\bm{p}}  $ satisfying the same equations as the positive spin equivalent, $ \bm{q}^{F}, \bm{p}$, the same estimates holds for the negative spin scalar quantities $ \underline{\phi} , \underline{\psi}$ as in Section \ref{Section - Regge Wheeler estimates}.  As far as the tensor $ \underline{\alpha}_{\star} $ is concerned, we obtain estimates by using the induced relating equation (\ref{scalar relating l frequency}) in addition to the above.
		
		In order to understand the dominant behavior of Teukolsky solutions of \textit{negative spin} asymptotically along the event horizon $ \h $, it suffices to only study the quantities $ \underline{f}_{\ell=2}, \underline{b}_{\ell=2}$ and $ \underline{a} _{\ell=2}$ supported on the fixed harmonic frequency $ \ell=2.$ \textbf{For shortness, we will often neglect the ``$ \ell=2 $" subscript in the equations.}

		\begin{theorem} \label{f,b undeline estimates}
		
			Let $ \underline{\mathfrak{f}}_{\star}, \underline{\tilde{\beta}}_{\star} $ be solutions to the generalized Teukolsky equations of negative spin, then for generic initial data the following estimates hold asymptotically along $ \h $, for all $ \tau \geq 1 $ 
			\begin{itemize}
				\item Decay
				\begin{align*}
					\norm{\underline{\mathfrak{f}}_{\star}}_{L^{\infty}(S^{2}_{\tau,M})} &\lesssim_{_{M}} \tau^{-\frac{1}{4}}, \\
					\norm{\underline{\tilde{\beta}}_{\star}}_{L^{\infty}(S^{2}_{\tau,M})} &\lesssim_{_{M}} \tau^{-\frac{1}{4}}
				\end{align*}
				\item Non-Decay and Blow up \begin{align*}
					\norm{\slas{\nabla}_{\partial_r}^{k+1}\underline{\mathfrak{f}}_{\star}}_{S^{2}_{\tau,M}} &= \dfrac{1}{\sqrt{12}M^{2}}f_k \cdot  \mathcal{H}_{2}[\underline{\Psi}_2] \cdot \tau^{k} + \mathcal{O}\left(\tau^{k-\frac{1}{4}}\right) \\
					\norm{\slas{\nabla}_{\partial_r}^{k+1}\underline{\tilde{\beta}}_{\star}}_{S^{2}_{\tau,M}} &= \dfrac{1}{\sqrt{6}M^{4}}b_k \cdot \mathcal{H}_{2}[\underline{\Psi}_2] \cdot \tau^{k} + \mathcal{O}\left(\tau^{k-\frac{1}{4}}\right), 
				\end{align*}
			\end{itemize}
			for all $ k\geq 0, $ where $ \mathcal{H}_{2}[\underline{\Psi}_2] :=\norm{H_{\ell=2}[\underline{\Psi}_2]}_{S^{2}_{\tau,M}}$, $ \forall \ \tau \geq 1 $, and  the coefficients $ f_k,b_k $ are given by \[ f_k:= \dfrac{1}{M}\cdot b_k, \hspace{1cm} b_k := \dfrac{1}{20}\dfrac{1}{(2M^2)^k} \dfrac{(k+3)!}{3!}. \]
		\end{theorem}
		\begin{proof} We begin by studying the induced scalars $ \underline{f}
			_{\ell=2} $ and $ \underline{b}_{\ell=2} $.
			First, we show the decay estimates. We evaluate the Teukolsky equation for $ \underline{f} $ on the horizon $ \h $ and we obtain \begin{align*}
				T\left (2\partial_r \underline{f} + \frac{4}{M}\underline{f} \right )- \dfrac{6}{M^2}\underline{f} = - \dfrac{4}{M^2}\underline{f} + \dfrac{2}{M^3}\underline{b}
			\end{align*}
			On the other hand, differentiating the transport equation for $ \underline{f} $ and evaluating it on $ \h $ yields \begin{align*}
				T(2\partial_r\underline{f}) = M^{-1} \partial_r \underline{\phi} - M^{-2}\underline{\phi} - \dfrac{2}{M^2}\underline{f} \label{f teukolsky}
			\end{align*}
			Combining the equations above and the transport equation for $ \underline{f} $ produces \begin{align}
				\dfrac{4}{M^2}\underline{f}+\dfrac{2}{M^3}\underline{b}= M^{-1} \partial_r \underline{\phi} + M^{-2}\underline{\phi}. 
			\end{align}
			Following the same procedure for the Teukolsky and transport equations satisfied by $ \underline{b} $, we obtain \begin{align}
				\begin{aligned}
					T\left (2\partial_r \underline{b} + \frac{4}{M}\underline{b} \right )- \dfrac{6}{M^2}\underline{b} &= \frac{8}{M}\underline{f} + \frac{2}{M^2} \underline{b} \\
					\Rightarrow \ \ \ M^{-1}\partial_r\underline{\psi} + M^{-2} \underline{\psi} &= \frac{8}{M}\underline{f} + \frac{10}{M^2} \underline{b} \label{b teukolsky}
				\end{aligned}
			\end{align}
			Thus, solving the system (\ref{f teukolsky},\ref{b teukolsky}) yields the expressions \begin{align*}
				\dfrac{6}{M^2}\underline{b} &= \dfrac{1}{M}\partial_r\underline{\psi} - 2 \partial_r \underline{\phi} + \dfrac{1}{M^2}\underline{\psi} - \dfrac{2}{M}\underline{\phi} \\
				\dfrac{12}{M}\underline{f} & = 5 \partial_r \underline{\phi} - \dfrac{1}{M}\partial_r \underline{\psi} + \dfrac{5}{M} \underline{\phi} - \dfrac{1}{M^2}\underline{\psi},
			\end{align*}
			and using Corollary \ref{psi-phi l study } we obtain the decay rates of $ \underline{f}, \underline{b} $ along the event horizon $ \h. $
			\begin{flushleft}
				\hrulefill
			\end{flushleft}
			
			To show the non-decay estimates we consider a $ \partial_r$-derivative of the Teukolsky system and we consider sufficiently many $  \partial_r-$derivatives  of  the transport equations to express all components in terms of $ \underline{f}, \underline{b}, \underline{\phi} $ and $ \underline{\psi}. $ In particular, the equation for $ \underline{f} $ yields \begin{align*}
				T&\left(2\partial_r^{2}\underline{f} + \dfrac{4}{M}\partial_r\underline{f} - \frac{4}{M^2}\underline{f}\right) - \frac{6}{M^2}\partial_r\underline{f} + \frac{12}{M^3} \underline{f} + \dfrac{2}{M^2}\partial_r \underline{f} \ = \\ & \ \ - \dfrac{4}{M^2}\partial_r \underline{f} + \dfrac{10}{M^3} \underline{f} + \dfrac{2}{M^2}\partial_r\underline{b} - \dfrac{6}{M^4} \underline{b} - \dfrac{2}{M^2} \partial_r \underline{f}
			\end{align*}
			\begin{align}
				\Rightarrow \hspace{0.5cm} \dfrac{4}{M}\partial_r \underline{f} + \dfrac{2}{M^2} \partial_r \underline{b} = \partial_r^{2}\underline{\phi} - \dfrac{2}{M^2} \underline{\phi} + \dfrac{10}{M^2}\underline{f} + \dfrac{6}{M^3}\underline{b}
			\end{align}
			Using the same approach for the equations of $ \underline{b} $ we obtain \begin{align}
				\dfrac{8}{M}\partial_r \underline{f} + \dfrac{10}{M^2} 
				\partial_r\underline{b} = \frac{1}{M} \partial_r^{2}\underline{\psi} - \dfrac{2}{M^3}\underline{\psi} + \dfrac{24}{M^2}\underline{f}+  \dfrac{30}{M^3}\underline{b}
			\end{align}
			Note that the coefficients of $ \partial_r\underline{f}, \partial_r\underline{b} $ terms in the equations above are the same with the top order ones in (\ref{f teukolsky},\ref{b teukolsky}). Solving the system once again yields \begin{align}
				\begin{aligned}
					\dfrac{6}{M^2}\partial_r\underline{b}\ =\ &	\dfrac{1}{M}\partial_r^{2}\underline{\psi} - 2\partial_r^{2}\underline{\phi}  - \dfrac{2}{M^3}\underline{\psi}  + \dfrac{4}{M^2} \underline{\phi} +  \dfrac{4}{M^2}\underline{f}+   \dfrac{18}{M^3}\underline{b} \\ 
					\dfrac{12}{M}\partial_r \underline{f}\ = \ &  5\partial_r^{2}\underline{\phi}- \frac{1}{M} \partial_r^{2}\underline{\psi} - \dfrac{10}{M^2} \underline{\phi} + \dfrac{2}{M^3}\underline{\psi}+ \dfrac{26}{M^2}\underline{f} \label{partial_r f,b}
				\end{aligned}
			\end{align}
			Consider the decomposition of $ \underline{\phi}, \underline{\psi} $ in terms of $ \Psi_1^{^{(\ell=2)}}, \Psi_2^{^{(\ell=2)}} $
			\begin{align*}\begin{aligned}
					\underline{\phi}_{\ell=2} &= \dfrac{1}{10M^2} \Psi_2^{^{(2)}}+ \dfrac{1}{20M^2}\Psi_1^{^{(2)}}\\
					\underline{\psi}_{\ell=2} &= -\dfrac{1}{10M} \Psi_2^{^{(2)}}+ \dfrac{1}{5M}\Psi_1^{^{(2)}} 
				\end{aligned}
			\end{align*}
			then, using the conservation laws of Theorem \ref{conservation horizon} and the decay estimates of Theorem \ref{PWD-estimates} for the expressions in (\ref{partial_r f,b}) we obtain\begin{align*}
				\partial_r{\underline{b}} &= - \dfrac{1}{20} H_{2}[\underline{\Psi}_2] +	\mathcal{O}(\tau^{-\frac{1}{4}})\\
				\partial_r\underline{f} &= \dfrac{1}{20M} H_{2}[\underline{\Psi}_2] +	\mathcal{O}(\tau^{-\frac{1}{4}}),
			\end{align*}
			and thus concluding the non-decay estimates of the proposition.
			\begin{flushleft}
				\hrulefill
			\end{flushleft}
			
			Finally, to obtain the blow-up estimates we use induction on the number of derivatives. Note that the $ k=0 $ case corresponds to the non-decay estimates shown above. Assume the expression of the assumption holds for all $ s\leq k $, $ k\geq 1, $ then we consider $ \partial_r^{k+1}-$derivatives of the Teukolsky system and we use the same procedure as above. In particular, for the equations of $ \underline{f} $ we obtain \begin{align}\begin{aligned}
					&T\left(2\partial_r^{k+2}\underline{f}+ \partial_r^{k+1}\left(\dfrac{2}{r}\underline{f}\right)\Big|_{r=M}  \right) -  \dfrac{6}{M^2}\partial_r^{k+1}\underline{f} + \left[\binom{k+1}{2}D''(M)+
					\binom{k+1}{1}R'(M)\right]\partial_r^{k+1}\underline{f} \\ & \hspace{0.5cm}=\ -\dfrac{4}{M^2}\partial_r^{k+1}\underline{f} + \dfrac{2}{M^3}\partial_r^{k+1}\underline{b} - \binom{k+1}{1}D''(M)\partial_r^{k+1}\underline{f} + D''(M) \dfrac{1}{M^2}\partial_r^{k-1}\underline{a} + \mathcal{L}[\underline{f}^{^{\leq k}},\underline{b}^{^{\leq k}}, \underline{a}^{^{\leq k-2}}] \label{k-teukolsky f}
				\end{aligned}
			\end{align}
			where $  \mathcal{L}[\underline{f}^{^{\leq k}},\underline{b}^{^{\leq k}}, \underline{a}^{^{\leq k-2}}]$ is an expression involving $ \partial_r-$
			derivatives of the underline quantities up to the order of their respective superscript.
			Differentiating the transport equation for $ \underline{f} $ as well we obtain \begin{align*}
				T(2\partial_r^{s}\underline{f}) = \partial_r^{s}\left(\dfrac{1}{r}\underline{\phi}\right)- \partial_r^{s+1}\left(D\cdot \underline{f}\right), \hspace{1cm} \forall \ s\geq 0
			\end{align*}
			and thus the previous equation yields \begin{align}
				\dfrac{4}{M^2}\partial_r^{k+1}\underline{f} 	+\dfrac{2}{M^3}\partial_r^{k+1}\underline{b} = \dfrac{1}{M}\partial_r^{k+2}\underline{\phi} + \tilde{\mathcal{L}}[\underline{\phi}^{^{\leq k+1}}] +  \mathcal{L}[\underline{f}^{^{\leq k}},\underline{b}^{^{\leq k}}]. \label{k-f}
			\end{align}
			Note that the term of $ \partial_r^{k-1} \underline{a}$ in (\ref{k-teukolsky f}) was absorbed in the $  \mathcal{L}[\underline{f}^{^{\leq k}},\underline{b}^{^{\leq k}}]  $ term. We can do this by taking $ 
			\partial_r^{k-1}- $ derivatives of  the Teukolsky equation for $ \underline{a} $ and using the relating equation (\ref{scalar relating l frequency})  we express $  \partial_r^{k-1} \underline{a} $ in terms of lower derivatives of itself and derivatives up to order $ k $ of both $ \underline{f},\underline{b}. $ By consecutively repeating the same procedure for the remaining lower order derivative of $ \underline{a} $ we eventually write it only in terms of derivatives of $ \underline{f},\underline{b}. $
			
			We repeat these steps for the equations of $ \underline{b} $ and we obtain \begin{align}
				\frac{8}{M}\partial_r^{k+1}\underline{f} +  \dfrac{10}{M^2}\partial_r^{k+1}b &= \dfrac{1}{M}\partial_r^{k+2}\underline{\psi}   +  \tilde{\mathcal{L}}[\underline{\psi}^{^{\leq k+1}}] +  \mathcal{L}[\underline{f}^{^{\leq k}},\underline{b}^{^{\leq k}}] \label{k - b} 
			\end{align}
			Once again, after solving the system formed by (\ref{k-f},\ref{k - b}), using the induction assumption for all derivatives up to order $ k $ and using Theorem \ref{Scalar blow up}, we obtain \begin{align*}
				\begin{aligned}
					\partial_r^{k+1}\underline{b}(\tau,\omega) &= -\dfrac{1}{M}\dfrac{M^3}{20M^2} \partial_r^{k+2}\underline{\Psi}_2^{^{(2)}}(\tau,\omega) + \mathcal{O}(\tau^{k-1}) \\ &=
					(-1)^{k+1} b_k\  H_2[\underline{\Psi}_2](\omega) \cdot \tau^{k} + \mathcal{O}(\tau^{k-\frac{1}{4}}) \\ 
					\partial_r^{k+1}\underline{f}(\tau,\omega) \ &=\ \dfrac{1}{20M}\partial_r^{k+2}\underline{\Psi}_2^{^{(2)}} + \tilde{\mathcal{L}}[\underline{\Psi_2}^{^{\leq k+1}}, \underline{\Psi_1}^{^{\leq k+2}} ] +  \mathcal{L}[\underline{f}^{^{\leq k}},\underline{b}^{^{\leq k}}] \\ &= \
					(-1)^{k} f_k\  H_2[\underline{\Psi}_2](\omega) \cdot \tau^{k} + \mathcal{O}(\tau^{k-\frac{1}{4}})
				\end{aligned}
			\end{align*}
			asymptotically along the event horizon $ \h,  $ for any $ k\geq 1. $
			
			\begin{flushleft}
				\hrulefill
			\end{flushleft}
			
			One can proceed similarly to show estimates for $\underline{\mathfrak{f}}_{\ell}$ and $ \underline{b}_{\ell} $, for all  frequencies $ \ell $. Of course, the higher the frequency $ \ell $ the more $ \partial_r- $derivatives we must take for the non-decay and blow estimates to manifest. Nevertheless, using standard elliptic identities and the results of the dominant frequency projection $ \ell=2 $ we conclude the estimates of the assumption.	
			
		\end{proof}

		\subsection*{Non-decay and blow up estimates for $ \underline{\alpha}_{\star} $ on $ \h. $}
		Now that we understand the behavior of $ \underline{f}, \underline{b} $ asymptotically along the event horizon $ \h $, we can prove estimates for the  rescaled  scalar extreme curvature component $ \underline{a} $ of spin -2. In particular, we see that it does \textbf{not} decay along the event horizon, and any transversal invariant derivative we consider leads to blow-up asymptotically on $ \h. $ Then, estimates for $ \underline{\alpha}_{\star} $ follow by standard elliptic identities.
		
		\begin{theorem} \label{underline a estimate}
			Let $ \underline{\alpha}_{\star} $ be solutions to the generalized Teukolsky equation of  -2--spin, then for generic initial data the following estimates hold asymptotically along $ \h $, for all $ \tau \geq 1 $ 	
			\begin{itemize}\item Non-decay
				\begin{align*}
					\norm{\underline{\alpha}_{\star}}_{S^{2}_{\tau,M}} &=\  \dfrac{1}{\sqrt{12}M^{2}} \dfrac{1}{30}\cdot  \mathcal{H}_{2}[\underline{\Psi}_2]  + \mathcal{O}\left(\tau^{-\frac{1}{4}}\right), 
				\end{align*}
				\item Blow-up
				\begin{align*}
					\hspace{2.5cm}	\norm{\slas{\nabla}_{\partial_r}^{k}\underline{\alpha}_{\star}}_{S^{2}_{\tau,M}} &= \dfrac{1}{\sqrt{12}M^{2}}c_k \cdot  \mathcal{H}_{2}[\underline{\Psi}_2] \cdot \tau^{k} + \mathcal{O}\left(\tau^{k-\frac{1}{4}}\right), \hspace{1cm} \forall \ k\geq 1.
				\end{align*}
			\end{itemize}
			where $\mathcal{H}_{2}[\underline{\Psi}_2] :=\norm{H_{\ell=2}[\underline{\Psi}_2]}_{S^{2}_{\tau,M}}, \ \forall \ \tau \geq 1$ and  \[ c_k := \dfrac{1}{30}\dfrac{1}{(2M)^k} \frac{(k+3)!}{3!}, \hspace{1cm} k\geq 1.\] 
		\end{theorem}
		\begin{proof}First, we show the corresponding estimates for the induced scalar $ \underline{a} $ supported on the fixed frequency $ \ell=2. $
			For the non-decay estimate, we begin by taking one $ \partial_r- $derivative of the relating equation (\ref{scalar relating l frequency}) and after we evaluate it on the horizon $ \br{r=M} $ we obtain \begin{align}
				T(2\partial_r\underline{a}) + \dfrac{4}{M^2}\underline{a} \ = \ \dfrac{2}{M}\partial_r\underline{f} +\dfrac{2}{M^2}\partial_r\underline{b} -\dfrac{10}{M^2}\underline{f}-\dfrac{2}{M^3}\underline{b} \label{one derivative of relating}
			\end{align}
			On the other hand, we write down the Teukolsky equation for $ \underline{a} $ and evaluate it on the horizon $ \h $ to get \begin{align*}
				T\left(2\partial_r\underline{a}+\dfrac{6}{M}\underline{a}\right) = \dfrac{2}{M^2}\underline{a} + \dfrac{4}{M}\partial_r\underline{f},
			\end{align*}
			while evaluating the Teukolsky equation for $ \underline{f}  $ on $ \h $,  written as in (\ref{Teukolsky scalar rescaled}), yields \begin{align*}
				T\left(\frac{2}{M}\underline{a}\right) = T\left(2\partial_r\underline{f}+ \dfrac{4}{M}\underline{f}\right).
			\end{align*}
			Thus, plugging in the last two relations to (\ref{one derivative of relating}) and using the transport equation for $ \underline{f} $ gives us \begin{align}
				\begin{aligned}
					\dfrac{6}{M^2}\underline{a} = - \dfrac{2}{M}\partial_r\underline{f} + \dfrac{2}{M^2}\partial_r\underline{b} + \dfrac{16}{M^2}\underline{f}- \dfrac{2}{M^3}\underline{b} + \dfrac{3}{M}\partial_r\underline{\phi} + \dfrac{3}{M^2}\underline{\phi} 
				\end{aligned}
			\end{align}
			However, from proposition \ref{f,b undeline estimates} we have that the last four terms decay; taking the limit  for the first two yields \begin{align*}
				&\underline{a} = - \dfrac{M}{3}\partial_r\underline{f} + \dfrac{1}{3}\partial_r \underline{b} + \mathcal{O}(\tau^{-\frac{1}{4}}) = \dfrac{1}{3}\left(-\dfrac{1}{20}-\dfrac{1}{20}\right) H_2[\underline{\Psi}_2]  + \mathcal{O}(\tau^{-\frac{1}{4}}) \\
				\Rightarrow \hspace{0.5cm} &\underline{a}(\tau, \omega)  \xrightarrow{\tau \to \infty} -\dfrac{1}{30}H_2[\underline{\Psi}_2]( \omega) .
			\end{align*}
			In view of  $ H_{2}[\Psi_2] $ being almost everywhere non-zero on $ S^{2}_{v,M} $ for generic initial data,  we obtain the non-decay result of $ \underline{a}. $
			\begin{flushleft}
				\hrulefill
			\end{flushleft}
			To obtain the higher-order blow-up estimates we use induction on the number of derivatives. The base case $ k=0 $ corresponds to the non-decay results proved earlier.   Assume the proposition holds for all $ s\leq k-1 $, $ k\geq 1  $, then we consider $ k $ many $ \partial_r- $derivatives of the Teukolsky equation for $ \underline{a} $ and evaluate it on the horizon $ \h  $ to obtain \begin{align}
				\begin{aligned}
					T&\left(2\partial_r^{k+1}\underline{a}+ \partial_r^{k}\left(\dfrac{2}{r}\underline{a}\right)\Big|_{r=M}  \right) -  \dfrac{6}{M^2}\partial_r^{k}\underline{a} + \left[\binom{k}{2}D''(M)+
					\binom{k}{1}R'(M)\right]\partial_r^{k}\underline{a} \\
					=& \ -\frac{4}{M^2}\partial_r^{k}\underline{a} - 2\binom{k}{1} D''(M)\partial_r^{k}\underline{a} + \dfrac{4}{M} \partial_r^{k+1}\underline{f} +
					\mathcal{L}[\underline{a}^{^{\leq k-1}},\underline{f}^{^{\leq k}}]   \label{a- kderivatives}
				\end{aligned}
			\end{align}
			On the other hand, we may write the relating equation (\ref{scalar relating l frequency}) as \begin{align*}
				T(2\underline{a}) = \frac{2}{Mr}\underline{b} + \dfrac{2}{r}\left(1-4\sqrt{D}\right) \underline{f} - 2D'(r) \underline{a} -D\partial_r\underline{a}  + \dfrac{D}{r}\underline{a}
			\end{align*}
			and by taking $ k+1 $-many derivatives of the above and evaluating on $ \h $ we obtain \begin{align*}
				T(2\partial_r^{s}\underline{a}) =\dfrac{2}{M^2}\partial_r^{s}\underline{b} + \dfrac{2}{M}\partial_r^s\underline{f} - \left[ 2\binom{s}{1}D''(M)+ \binom{s}{2}D''(M) \right]\partial_r^{s-1}\underline{a} + \tilde{\mathcal{L}}[\underline{f}^{^{\leq s-1}},\underline{b}^{^{\leq s-1}}, \underline{a}^{^{\leq s-2}}]
			\end{align*}
			Thus, plugging the above expression in (\ref{a- kderivatives}) for all $ s\leq k+1 $ yields \begin{align*}
				\dfrac{6}{M^2}\partial_r^{k}\underline{a}=	\dfrac{2}{M^2}\partial_r^{k+1}\underline{b} - \dfrac{2}{M}\partial_r^{k+1}\underline{f} +\mathcal{L}[\underline{f}^{^{\leq k}},\underline{b}^{^{\leq k}}, \underline{a}^{^{\leq k-1}}]
			\end{align*}
			Now we can use the blow-up estimates of proposition \ref{f,b undeline estimates} and the inductive hypothesis to get  \begin{align*}
				\partial_r^{k}\underline{a}(\omega)&= (-1)^{k} \left( \dfrac{1}{3}  b_k  -M f_k \right) H_2[\underline{\Psi}_2](\omega) \cdot \tau^{k} + \mathcal{O}(\tau^{k-\frac{1}{4}}) \\
				&=(-1)^{k+1}\dfrac{2}{3}b_k \ H_2[\underline{\Psi}_2](\omega) \cdot \tau^{k} + \mathcal{O}(\tau^{k-\frac{1}{4}})
			\end{align*}
			and in view of $ \frac{2}{3}  b_k = c_k$ we conclude the estimates for $ \underline{a}. $ Finally, using standard elliptic identities we show the estimates of the assumption for the tensor $ \underline{\alpha}_{\star}. $\\
		\end{proof}

		\setstretch{1}

		\appendix

\printbibliography

	\end{document}